\documentclass[aps,pra,nofootinbib,superscriptaddress,reprint,longbibliography]{revtex4-1}
\pdfoutput=1

\usepackage{amsmath}
\usepackage{amssymb}
\usepackage{amsthm}

\usepackage{graphicx}
\usepackage{mathrsfs}
\usepackage{mathtools}
\usepackage{enumitem}
\usepackage{microtype}
\usepackage{bbm}

\usepackage{phfqit}
\usepackage{phfthm}

\usepackage[backtick=true]{phfparen}
\usepackage[preset=reset]{phfnote}
\usepackage{xcolor}
\usepackage{hyperref}
\definecolor{docnotelinkcolor}{rgb}{0,0,0.4}%
\hypersetup{unicode=true,%
  bookmarksnumbered=true,bookmarksopen=false,bookmarksopenlevel=1,%
  breaklinks=true,pdfborder={0 0 0},colorlinks=true}%
\hypersetup{%
  anchorcolor=docnotelinkcolor,citecolor=docnotelinkcolor,%
  filecolor=docnotelinkcolor,linkcolor=docnotelinkcolor,%
  menucolor=docnotelinkcolor,runcolor=docnotelinkcolor,%
  urlcolor=docnotelinkcolor}%

\newtheorem{mainproposition}{Proposition}

\newtheorem*{mainresult}{Main Result}

\ifpdf%
\begingroup \catcode`\_=12\relax
\immediate\pdfobj stream attr{/N 3}  file{sRGB_IEC61966-2-1_black_scaled.icc}
\pdfcatalog{%
/OutputIntents [ <<
/Type /OutputIntent
/S/GTS_PDFA1
/DestOutputProfile \the\pdflastobj\space 0 R
/OutputConditionIdentifier (sRGB IEC61966-2.1 black scaled)
/Info(sRGB IEC61966-2.1 black scaled)
>> ]
}
\endgroup
\fi

\makeatletter
\def\section{%
  \@startsection
    {section}%
    {1}%
    {\z@}%
    {1.1cm \@plus1ex \@minus .2ex}%
    {.6cm}%
    {%
      \normalfont\small\bfseries
      \centering
    }%
}%
\def\subsection{%
  \@startsection
    {subsection}%
    {2}%
    {\z@}%
    {0.95cm \@plus1ex \@minus .2ex}%
    {.55cm}%
    {%
     \normalfont\small\bfseries
     \centering
    }%
}%
\makeatother

\makeatletter
\renewcommand*{\eqref}[1]{%
  \hyperref[{#1}]{\textup{\tagform@{\ref*{#1}}}}%
}
\DefineTunedQitObject{DCohz}{\DCohbaseParse}{\DCohbaseRender}%
{{\hat\DSym}}{%
  \let\DCohbaseStateSubscripts\DCohxStateSubscripts%
}
\newcommand\Hminz{\@HHbase{\HHSym}{\mathrm{min},0}}
\makeatother

\phfMakeThmheadingEnvironment[thmstyle=remark]{thmheadingremark}

\makeatletter
\newcounter{appendix}

\def\appendixlabel{%
  \setcounter{appendix}{\numexpr\value{section}-1\relax}%
  \refstepcounter{appendix}%
  \label}
\apptocmd\appendix{%
}
\makeatother

\usepackage[p,osf,swashQ]{cochineal}
\usepackage[liby,libertine]{newtxmath}

\usepackage[french,german,english]{babel}

\usepackage{csquotes}
\DeclareAutoPunct{.,;:!?}

\def\qq{\textquote}

\makeatletter
\newcommand\bibalias[2]{%
  \@namedef{bibali@#1}{#2}%
}

\newtoks\biba@toks
\let\bibalias@oldcite\cite
\def\cite{%
  \@ifnextchar[{%
    \biba@cite@optarg%
  }{%
    \biba@cite{}%
  }%
}
\newcommand\biba@cite@optarg[2][]{%
  \biba@cite{[#1]}{#2}%
}
\newcommand\biba@cite[2]{%
  \biba@toks{\bibalias@oldcite#1}%
  \def\biba@comma{}%
  \def\biba@all{}%
  \@for\biba@one@:=#2\do{%
    \edef\biba@one{\expandafter\@firstofone\biba@one@\@empty}%
    \@ifundefined{bibali@\biba@one}{%
      \edef\biba@all{\biba@all\biba@comma\biba@one}%
    }{%
      \PackageInfo{bibalias}{%
        Replacing citation `\biba@one' with `\@nameuse{bibali@\biba@one}'
      }%
      \edef\biba@all{\biba@all\biba@comma\@nameuse{bibali@\biba@one}}%
    }%
    \def\biba@comma{,}%
  }%
  \immediate\write\@auxout{\noexpand\bgroup\noexpand\renewcommand\noexpand\citation[1]{}\noexpand\citation{#2}\noexpand\egroup}%
  \edef\biba@tmp{\the\biba@toks{\biba@all}}%
  \biba@tmp%
}
\makeatother

\bibalias{Meer2017arXiv_smoothed}{Meer2017PRX_smoothed}
\bibalias{Ito2016arXiv_performance}{Ito2018PRE_performance}
\bibalias{Aberg2016arXiv_fluctuation}{Aberg2018PRX_fluctuation}
\bibalias{Mosonyi2014CMP_hypothesis}{Mosonyi2015CMP_hypothesis}
\bibalias{Bera2016arXiv_laws}{Bera2017NComm_laws}

\begin{document}

\title{Fundamental work cost of quantum processes}

\author{Philippe Faist}
\email[]{phfaist@caltech.edu}
\affiliation{Institute for Theoretical Physics, ETH Zurich, 8093 Switzerland}
\affiliation{Institute for Quantum Information and Matter, Caltech,
  Pasadena CA, 91125 USA}
\author{Renato Renner}
\email[]{renner@phys.ethz.ch}
\affiliation{Institute for Theoretical Physics, ETH Zurich, 8093 Switzerland}

\date{May 1, 2018}

\begin{abstract}
  Information-theoretic approaches provide a promising avenue for extending the
  laws of thermodynamics to the nanoscale.
  Here, we provide a general fundamental lower limit, valid for systems with an
  arbitrary Hamiltonian and in contact with any thermodynamic bath, on the work
  cost for the implementation of any logical process.  This limit is given by a
  new information measure---the coherent relative entropy---which accounts for
  the Gibbs weight of each microstate.
  The coherent relative entropy enjoys a collection of natural properties
  justifying its interpretation as a measure of information, and can be
  understood as a generalization of a quantum relative entropy difference.
  As an application, we show that the standard first and second laws of
  thermodynamics emerge from our microscopic picture in the macroscopic limit.
  Finally, our results have an impact on understanding the role of the observer
  in thermodynamics: Our approach may be applied at any level of knowledge---for
  instance at the microscopic, mesoscopic or macroscopic scales---thus providing
  a formulation of thermodynamics that is inherently relative to the observer.
  We obtain a precise criterion for when the laws of thermodynamics can be
  applied, thus making a step forward in determining the exact extent of the
  universality of thermodynamics and enabling a systematic treatment of
  Maxwell-demon-like situations.
\end{abstract}

\maketitle

\section{Introduction}

Thermodynamics enjoys an extraordinary universality---applying to heat engines,
chemical reactions, electromagnetic radiation, and even to black holes.  Thus, we are naturally
led to further apply it to small-scale quantum systems.
In such a context, the information content of a system plays a key role:
Landauer's principle states that logically irreversible information processing
incurs an unavoidable thermodynamic cost~\cite{Landauer1961_5392446Erasure}.
Landauer's principle has generated a new line of research in which information
and thermodynamic entropy are treated on an equal
footing~\cite{Bennett1982IJTP_ThermodynOfComp}, in turn providing a resolution
to the paradox of Maxwell's demon~\cite{Bennett2003_NotesLP}.
In the context of statistical mechanics, a significant effort has also been
made to elucidate the role of the second law~\cite{%
  Lenard1978_thermodynamical,%
  Kurchan2000arXiv_fluctuation,%
  Tasaki2000arXiv_Jarzynski,%
  Tasaki2000arXiv_dynamics,%
  Ikeda2015AoP_unitary,%
  Iyoda2017PRL_manybody}.  Statistical mechanics has further provided important
contributions to understanding the interplay between information and
thermodynamics~\cite{%
  Jaynes1957PR_InfThStatMech,Jaynes1957PR_InfThStatMech2,%
  Shizume1995_HeatGeneration,%
  Piechocinska2000PRA,%
  Popescu2006NPhys_entanglement,%
  Hanggi2009_brownian,%
  Anders2010_LPQuantDomain,%
  Sagawa2012PRL_FluctThm,%
  Goold2015PRL_nonequilibrium}, with works studying the energy requirements of
information processing~\cite{Sagawa2008PRL_discrete,Sagawa2009PRL_minimal,%
  Parrondo2015NPhys_thermodynamics}.
This has also led to an improved understanding of nanoengines and
information-driven thermodynamic
devices~\cite{Szilard1929ZeitschriftFuerPhysik,%
  Scovil1959_masers,%
  Geusic1967_quantum,%
  Alicki1979_davies,%
  Geva1992_classical,%
  Lloyd1997_demon,%
  Lloyd2000_ultimate,%
  Linden2010PRL_SmallFridge,%
  Abah2012PRL_single,%
  Roulet2016PRE_rotor}, %
paving the way for experimental demonstrations~\cite{Baugh2005_algocool,%
  Toyabe2010NPhys_experimental,%
  Vidrighin2016PRL_photonic}.

When studying the thermodynamics of small-scale quantum systems, it is
particularly relevant to define the thermodynamic framework precisely.  A
customary approach, the resource theory approach, is to investigate the state
transformations that are possible after imposing a restriction on the types of
elementary physical operations that are allowed.
Such frameworks have enabled us to understand general conditions under which it is
possible to transform one state into another%
~\cite{Janzing2000_cost,%
  Horodecki2003PRA_NoisyOps,%
  Brandao2013_resource,%
  Horodecki2013_ThermoMaj,%
  Skrzypczyk2014NComm_individual,%
  Gour2015PR_resource,%
  Brandao2015PNAS_secondlaws,%
  Brandao2015PRL_resource} and %
to study erasure and work extraction in the single-shot
regime~\cite{Dahlsten2011NJP_inadequacy,%
  Aberg2013_worklike,%
  Egloff2015NJP_measure%
}.  Such results have been extended to the case where quantum side information
is available~\cite{delRio2011Nature,%
  delRio2014arXiv_relative}, to situations with multiple thermodynamic
reservoirs~\cite{Barnett2013_beyond,%
  YungerHalpern2016PRE_beyond,%
  PerarnauLlobet2016NJP_GGE,%
  Lostaglio2017NJP_noncommutativity,%
  Guryanova2016NatComm_multiple,%
  YungerHalpern2016NatComm_NATSandNATO}, and to the case of a finite bath
size~\cite{Reeb2014NJP_improved,Richens2017arXiv_finite,%
  Tajima2017PRE_finite,Ito2016arXiv_performance}.  The role of coherence and
catalysis has been underscored~\cite{Aberg2014PRL_catalytic,%
  Cwiklinski2015PRL_limitations,%
  NellyNg2015NJP_limits,%
  Lostaglio2015PRX_coherence,%
  Lostaglio2015NC_beyond,%
  Korzekwa2016NJP_extraction,%
  Winter2016PRL_coherence,%
  Dana2017PRA_beyond,%
  Cirstioiu2017arXiv_gauge}, %
the effect of correlations studied~\cite{Oppenheim2002PRL_thermodynamical,%
  Bruschi2015PRE_creating,%
  Majenz2017PRL_catdecoupling,%
  Lostaglio2015PRL_absence,%
  Bera2016arXiv_laws}, and the efficiency of nanoengines
investigated~\cite{Gallego2014NJP_correlated,%
  Woods2015arXiv_more,%
  Hayashi2017PRA_measurement,%
  Ito2016arXiv_performance}. %
Fully quantum fluctuation relations~\cite{Aberg2016arXiv_fluctuation} and a
second-law equality~\cite{Alhambra2016PRX_equality} have been derived, and
further connections to the recoverability of quantum information have been
exhibited~\cite{Wehner2015arXiv_reversibility}.  Furthermore, fully quantum state
transformations were
characterized~\cite{Buscemi2017PRA_relative,Gour2017arXiv_entropic}.  We refer
to ref.~\cite{Goold2016JPA_review} for a more comprehensive review covering these
approaches to quantum information thermodynamics.

Our main result is a fundamental limit to the work cost of any logical process
implemented on a system with any Hamiltonian and in contact with any type of
thermodynamic reservoir.
It accounts for the necessary changes in the energy level populations in
the system, as well as for the thermodynamic cost of resetting any information
that needs to be discarded by the logical process.  It is valid for a single
instance of the process and ignores unlikely events, thus capturing statistical
fluctuations of the work cost.

Our thermodynamic framework is specified by imposing a restriction on the
operations which can be carried out, along with introducing a battery system
allowing us to
invest resources to overcome this restriction.  The restriction we consider here
is to impose that the allowed operations must be \emph{Gibbs-preserving maps},
that is, mappings for which the thermal state is a fixed point.  This framework
is a natural generalization of the setup in ref.~\cite{Faist2015NatComm} and has
close ties to resource theory approaches~\cite{Horodecki2003PRA_NoisyOps,%
  Horodecki2013_ThermoMaj,%
  Brandao2015PNAS_secondlaws}.  Gibbs-preserving maps are the most generous set
of physical evolutions that can be allowed for free, in the sense that if any
non-Gibbs-preserving map is allowed for free, arbitrary work can be extracted,
rendering the framework trivial.  Since in most existing thermodynamic
frameworks the allowed free operations preserve the thermal state, our bound
still holds in other standard settings such as the framework of thermal
operations~\cite{Horodecki2013_ThermoMaj,%
  Brandao2015PNAS_secondlaws}.  (However, if one considers catalytical
processes, more general transformations can be carried out, and hence additional
care has to be taken in order to apply our framework, e.g., by including the
catalyst explicitly as part of the process~\cite{Brandao2015PNAS_secondlaws,%
  NellyNg2015NJP_limits,Lostaglio2015PRL_absence,%
  Wehner2015arXiv_reversibility}.)
As a battery system, we consider an \emph{information battery}, that is, a
memory register of qubits that are all individually either in a pure state or
in a maximally mixed state.  The pure qubits are a resource that can be
invested in order to implement logical processes that are not Gibbs preserving.

Our main result is expressed in terms of a new purely information-theoretic
quantity, the \emph{coherent relative entropy}.  The coherent relative entropy
observes several natural properties, such as a data processing inequality,
invariance under isometries, and a chain rule, justifying its interpretation as
an entropy measure.  It is a generalization of both the min- and max-relative
entropy as well as the conditional min- and max-entropy.  In the asymptotic
limit of many independent repetitions of the process (the i.i.d.\@ limit), the
coherent relative entropy converges to the difference of the usual quantum
relative entropy of the input state and the output state relative to the Gibbs
state.  Our quantity hence adds structure to the collection of entropy measures
forming the smooth entropy framework~\cite{PhDRenner2005_SQKD,%
  Datta2009IEEE_minmax,%
  BookTomamichel2016_Finite}.

In fact, our result may be phrased in purely information-theoretic terms,
abstracting out physical notions such as energy or temperature in an operator
$\Gamma$, which may be interpreted as assigning abstract ``weights'' to
individual quantum states.  In the case of a system in contact with a heat bath,
these weights are simply the Gibbs weights, where at inverse temperature $\beta$,
the value $e^{-\beta E}$ is assigned to each energy level of energy $E$.
Our main result then quantifies how many pure qubits need to be invested, or how
many pure qubits may be distilled, while carrying out a specific logical process
given as a completely positive, trace-preserving map, subject to the restriction
that the implementation must globally preserve the joint $\Gamma$ operator of
the system and the battery.
In this picture, the coherent relative entropy intuitively measures the amount
of information ``forgotten'' by the logical process, conditioned on the output
of the process, and counted relative to the ``weights'' encoded in the $\Gamma$
operator.

Our framework can be applied to the macroscopic limit, to study transitions
between thermodynamic states of a large system.  (For instance, an isolated gas
in a box that is in a microcanonical state may undergo a process that brings
the gas to another microcanonical state of different energy and volume.)
Remarkably, it turns out that the work cost of any mapping relating two
thermodynamic states, as given by the coherent relative entropy, is equal to the
difference of a potential evaluated on the input and the output state,
regardless of the details of the logical process.  For an isolated system, we
show that this potential is precisely the thermodynamic entropy.  By coupling
the system to another system that plays the role of a piston, i.e., that is
capable of reversibly furnishing work to the system, we recover the standard
second law of thermodynamics relating the entropy change of the system to the
dissipated heat.

Our framework naturally treats thermodynamics as a subjective theory, where a
system can be described from the viewpoint of different observers.  One may thus
account for varying levels of knowledge about a quantum system.  This feature
allows us to systematically analyze Maxwell-demon-like situations.  Furthermore, we
find a criterion that certifies that the laws of thermodynamics hold in a
coarse-grained picture.  For instance, this criterion is not fulfilled in the
case of Maxwell's demon, signaling that a naive application of the laws of
thermodynamics to the gas may be disrupted by the presence of the demon.  We
hence obtain a precise notion of when the laws of thermodynamics can be applied,
contributing to the long-standing open question of the exact extent of the
universality of thermodynamics.

The results presented in this paper have been, to a large extent, reported in the
recent Ph.D. thesis of one of the authors~\cite{PhDPhF2016}.

The remainder of the paper is structured as follows.  In
\autoref{main-sec:the-framework}, we present the general setup in which our results
are derived. %
In \autoref{main-sec:results}, %
we explain our main result, the work cost of any process in contact with any
type of reservoir (\autoref{main-sec:results-main}); we then provide a collection of
properties of our new entropy measure (\autoref{main-sec:results-cohrelentr}), a
study of a special class of states whose properties make them suitable ``battery
states'' for storing extracted work (\autoref{main-sec:results-batteries}), a discussion
of how the macroscopic laws of thermodynamics emerge from our microscopic
considerations (\autoref{main-sec:results-macro}), and an analysis of how to relate
the views of different observers in our framework (\autoref{main-sec:results-obs}).
\autoref{main-sec:discussion} concludes with a discussion and an outlook.

\section{A framework of restricted operations}
\label{main-sec:the-framework}

Consider a system $S$ described by a Hamiltonian $H_S$.  In the framework of
Gibbs-preserving maps, an operation $\Phi(\cdot)$ is forbidden if it does not
satisfy $\Phi(e^{-\beta H_S}/Z) = e^{-\beta H_S}/Z$, where $\beta$ is a given
fixed inverse temperature and $Z=\tr[e^{-\beta H_S}]$. In other words, $\Phi(\cdot)$ is
forbidden if it does not preserve the thermal state.  Now, observe that the
condition on $\Phi(\cdot)$ depends on $\beta$ and $H_S$ only via the thermal
state, so we can rewrite the condition in a more general, but abstract, way as
follows: An operation $\Phi(\cdot)$ is forbidden if it does not preserve some
given fixed operator $\Gamma$, that is, if it does not satisfy
$\Phi(\Gamma) = \Gamma$.  We trivially recover Gibbs-preserving maps by setting
$\Gamma=e^{-\beta H_S}$.  For technical reasons and for convenience, we choose
to loosen the condition on $\Phi$ from being trace preserving to being trace
nonincreasing; correspondingly, we only require  that
$\Phi(\Gamma)\leqslant\Gamma$, instead of demanding strict equality.
By enlarging the class of allowed operations, we
can only obtain a more general bound.  The advantage of this abstract version of
the Gibbs-preserving-maps model is that our framework and its corresponding
results may be potentially applied to any setting where a restriction of the
form $\Phi(\Gamma)\leqslant\Gamma$ applies,
for any given $\Gamma$ which does not
necessarily have to be related to a Gibbs state.
The way $\Gamma$ should be defined is determined by which restriction of the
form $\Phi(\Gamma)\leqslant\Gamma$ makes sense to require in the particular
setting considered.
Finally, it proves convenient to consider non-normalized $\Gamma$ operators
(this becomes especially relevant if we consider different input and output
systems).  For instance, in the case of a system with Hamiltonian $H$ in contact
with a heat bath at inverse temperature $\beta$, the trace of
$\Gamma=e^{-\beta H}$ actually encodes the canonical partition function of the
system.

Our framework is defined in its full generality as follows.  To each system $S$
corresponds an operator $\Gamma_S$, which may be any positive semidefinite
operator.
We then define as \emph{free operations} those completely positive,
trace-nonincreasing maps $\Phi_{A\to B}$, mapping operators on a system $A$ to
operators on another system $B$, which satisfy
\begin{align}
  \Phi_{A\to B}(\Gamma_A) \leqslant \Gamma_B\ .
  \label{main-eq:infframework-model-admissible-operations-definition}
\end{align}
One may think of the $\Gamma$ operator as assigning to each state in a certain
basis a ``weight'' characterizing how ``useless'' it is.  As a convention, if
$\Gamma_S$ has eigenvalues equal to zero, then the corresponding eigenstates are
considered to be impossible to prepare---these states will never be observed.
In the following, a map
obeying~\eqref{main-eq:infframework-model-admissible-operations-definition} will be
referred to as a \emph{$\Gamma$-sub-preserving map}.

As mentioned above, in the case of a system $S$ with Hamiltonian $H_S$ in
contact with a single heat bath at inverse temperature $\beta$, we essentially
recover the usual model of Gibbs-preserving maps by setting
$\Gamma=e^{-\beta H_S}$.
In the case of multiple conserved charges such as a Hamiltonian $H_S$, number
operator $N_S$, etc., we recover the relevant Gibbs-preserving maps model by
setting $\Gamma=e^{-\beta (H_S - \mu N_S + \ldots )}$, with the corresponding
chemical potentials, as expected; furthermore, the physical charges do not
have to commute~\cite{Guryanova2016NatComm_multiple,%
  YungerHalpern2016NatComm_NATSandNATO}.

Our framework is designed to be as tolerant as possible (to the extent that our
allowed operations are ultimately a set of quantum channels), so as to result in
the strongest possible fundamental limit.  We start with this observation in the
case of thermodynamics with a single heat bath: If we allow any physical
evolution for free that does not preserve the thermal state, then we may create
an arbitrary number of copies of a nonequilibrium quantum state for free;
however, this renders our theory trivial since usual thermodynamical models allow us
to extract work from many copies of a nonequilibrium state.  Accordingly,
quantum thermodynamics models that can be written as a set of allowed physical
maps (such as thermal operations) necessarily have the Gibbs state as fixed
point, ensuring that our fundamental limit applies for those models as well.
We note that models in which catalysis is permitted allow for more general state
transformations~\cite{Brandao2015PNAS_secondlaws,%
  NellyNg2015NJP_limits,Lostaglio2015PRL_absence,%
  Wehner2015arXiv_reversibility}, exploiting the fact that, for a forbidden
transition $\sigma\not\to\rho$, there might exist some state $\zeta$ such that
$\sigma\otimes\zeta\to\rho\otimes\zeta$ (where $\zeta$ may be chosen
suitably depending on $\sigma$ and $\rho$).  In order to apply our framework in
such a context, we can consider the catalyst explicitly.  For instance, in the
context of catalytic thermal operations~\cite{Brandao2015PNAS_secondlaws}, after
the catalyst has been included in the picture, the physical evolution that is
applied is a thermal operation and thus has to be Gibbs preserving.  Ultimately,
the correct choice of framework depends on the underlying physical model:
For instance, in a macroscopic isolated gas, the whole system evolves according
to an energy-preserving unitary, and under suitable independence assumptions, the
evolution of an individual particle is well modeled by a thermal operation;
however, other situations might warrant the inclusion of a catalyst, for
instance, in a paranoid adversarial setting in which an eavesdropper may
manipulate a thermodynamic system.  In the first case, our ultimate limits apply
straightforwardly, whereas in the second, one would need to include the catalyst
explicitly.

Work storage systems are often modeled explicitly but are mostly equivalent in
terms of how they account for work~\cite{Bennett1982IJTP_ThermodynOfComp,%
  BookFeynmanLecturesOnComputation1996,%
  Horodecki2013_ThermoMaj,%
  Aberg2014PRL_catalytic,%
  Skrzypczyk2014NComm_individual,%
  Faist2015NatComm}.  Among these, the \emph{information battery} is easily
generalized to our abstract setting.
An information battery is a register $A$ of $n$ qubits whose $\Gamma$ operator
is $\Gamma_A=\Ident_A$. (If $\Gamma_A=e^{-\beta H_A}$ for an inverse
temperature $\beta$ and a Hamiltonian $H_A$, the requirement that
$\Gamma_A=\Ident_A$ is fulfilled by choosing the completely degenerate
Hamiltonian $H_A=0$.)  The register starts in a state where $\lambda_1$ qubits
are maximally mixed and $n-\lambda_1$ qubits are in a pure state.  In the final
state, we require that $\lambda_2$ qubits are maximally mixed and $n-\lambda_2$
are in a pure state.  The difference $\lambda = \lambda_1-\lambda_2$ is the
number of pure qubits extracted or ``distilled.''  In this way, we may invest a
number of pure qubits in order to enable a process that is not a free
operation, or we may try to extract pure qubits from a process that is already
a free operation.

Depending on the physical setup, the $\lambda$ pure battery qubits can
themselves be converted explicitly to some physical resource, such as mechanical
work.  In the case where we have access to a single heat bath at temperature
$T$, a pure qubit can be reversibly converted to and from $kT\ln2$ work using a
Szil\'ard engine~\cite{Szilard1929ZeitschriftFuerPhysik}, where $k$ is
Boltzmann's constant; thus, a process from which we can extract $\lambda$ pure
qubits is a process from which we can extract $\lambda\cdot kT\ln(2)$ work using
the heat bath.
More generally, we may replace the information battery entirely by other battery
models, such as corresponding generalizations to our framework of the work bit
(the ``wit'')~\cite{Brandao2015PNAS_secondlaws}, or the ``weight
system''~\cite{Skrzypczyk2014NComm_individual,Aberg2014PRL_catalytic}.  These
work storage models are known to be
equivalent~\cite{Brandao2015PNAS_secondlaws}; the equivalence persists in our
framework, with a suitable generalization of the ``extracted resource''
$\lambda$.
In the presence of several physical conserved charges, and corresponding
thermodynamic baths, the number $\lambda$ of pure qubits extracted acts as a
common currency that allows us to convert between the different resources. Hence,
a number $\lambda$ of extracted pure qubits may be stored in different forms of
physical batteries, corresponding to different forms of work, such as
\emph{chemical work}~\cite{Guryanova2016NatComm_multiple,%
  YungerHalpern2016NatComm_NATSandNATO}.
Hence, the quantity $\lambda$ should be thought of as a dimensionless value,
expressed in number of qubits, characterizing the ``extracted resource value''
of the logical process independently of which type of battery is actually used
in the implementation, in the same spirit as the \emph{free entropy} of
ref.~\cite{Guryanova2016NatComm_multiple}, and bearing some similarity to
currencies in general resource
theories~\cite{Coecke2016IC_mathematical,Kraemer2016arXiv_currencies}.

The main question we address may thus be reduced to the following form
(\autoref{main-fig:setup}).
\begin{figure}
  \centering
  \includegraphics{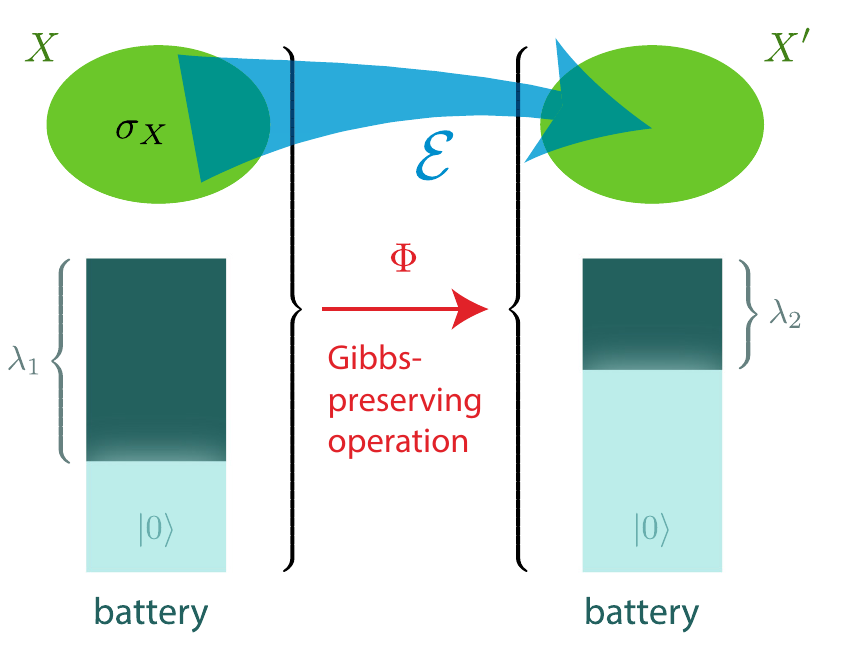}
  \caption{Implementation of a logical process $\mathcal{E}$ (any quantum
    process) using thermodynamic operations.  The process acts on $X$ and has
    output on $X'$, and is implemented by acting on the system and the battery
    with a joint Gibbs-preserving operation.  The battery starts with a depletion
    state $\lambda_1$ and finishes with a depletion state $\lambda_2$.  The
    overall extracted work is given by the difference $\lambda_1-\lambda_2$.}
  \label{main-fig:setup}
\end{figure}
Given operators $\Gamma_X,\Gamma_{X'}\geqslant 0$, an input state $\sigma_X$,
and a logical process $\mathcal{E}_{X\to X'}$ (that is, a trace-nonincreasing,
completely positive map), the task is to find the maximum number of qubits that
can be extracted, or the minimum number of qubits that need to be invested, in
order to implement the logical process on the given input state.  Note that we
require the correlations between the input and the output to match those specified by 
$\mathcal{E}_{X\to X'}$, a condition that is not equivalent to just requiring
that the given input state $\sigma_X$ is transformed into the given output state
$\mathcal{E}_{X\to X'}(\sigma_X)$.  Equivalently, we require that the
implementation acts as the process
$(\mathcal{E}_{X\to X'}\otimes\operatorname{id}_{R_X})$ on a purified state
$\ket{\sigma}_{XR_X}$ of the input, where $\operatorname{id}_{R_X}$ denotes the
identity process on $R_X$.

Finally, we ignore improbable events with total probability $\epsilon$, which is
necessary in order to obtain meaningful physical
results~\cite{Alhambra2016PRX_probability}.  Indeed, in textbook thermodynamics
when calculating the work cost of compressing an ideal gas, for instance, one
ignores the exceedingly unlikely event where all gas particles conspire to hit
against the piston at much greater force than on average, a situation that
would require more work for the compression but that happens with
overwhelmingly negligible probability.  For our purposes, we may optimize the
zero-error work cost over states that are $\epsilon$-approximations of the
required state~\cite{Faist2015NatComm}, which is a standard approach in quantum
information and cryptography~\cite{PhDRenner2005_SQKD,PhDTomamichel2012}.

At this point, it is useful to introduce the notion of the \emph{process matrix}
associated with the pair $(\mathcal{E}_{X\to X'},\sigma_X)$ of the logical
process and input state.  First, we define a reference system $R_X$ of the same
dimension as $X$ and choose some fixed bases $\{\ket k_X\}$ and
$\{\ket k_{R_X} \}$ of $X$ and $R_X$.  Then, we define the \emph{process matrix}
of the pair $(\mathcal{E}_{X\to X'},\sigma_X)$ as the bipartite quantum state
$\rho_{X'R_X} = (\mathcal{E}_{X\to
  X'}\otimes\operatorname{id}_{R_X})(\proj{\sigma}_{X:R_X})$, where
$\ket{\sigma}_{X:R_X} = \sigma_X^{1/2}\,(\sum \ket{k}_X\otimes\ket{k}_{R_X})$.
The process matrix corresponds to the Choi matrix of $\mathcal{E}_{X\to X'}$,
yet it is ``weighted'' by the input state $\sigma_X$ in the sense that the reference
state is $\sigma_{XR_X}$ instead of a maximally entangled state.  The process
matrix is in one-to-one correspondence with the pair
$(\mathcal{E}_{X\to X'}, \sigma_X)$ except for the part of
$\mathcal{E}_{X\to X'}$ that acts outside the support of $\sigma_X$; i.e., the
specification of $\rho_{X'R_X}$ uniquely determines $\sigma_X$ as well as the
logical process $\mathcal{E}_{X\to X'}$ on the support of $\sigma_X$.  The
reduced states $\sigma_X$ and $\sigma_{R_X}$ of $\ket\sigma_{XR_X}$ are related
by a partial transpose, $\sigma_{R_X} = \sigma_X^T$.  Intuitively, the reference
system $R_X$ may be thought of as a ``mirror system'' which ``remembers'' what
the input state to the process was.

As a further remark,
one might be worried that the relaxation of the set
of allowed operations from $\Gamma$-preserving and trace-preserving maps to
$\Gamma$-sub-preserving and trace-nonincreasing maps is too drastic. Indeed,
while yielding a valid bound, the relaxed set of operations is unphysical and we
might thus obtain a looser bound than necessary.
In fact, this is not the case.  Rather, trace-nonincreasing,
$\Gamma$-sub-preserving processes are a technical convenience, which allows for
more flexibility in the characterization of what the process effectively does in
the situations of interest to us while ignoring other irrelevant situations;
yet, ultimately, we show that an equivalent implementation can be carried out as a
single trace-preserving, $\Gamma$-preserving map.
For instance, consider a box separated into two equal-volume compartments, one
of which contains a single-particle gas (a setup known as a Szil\'ard
engine~\cite{Szilard1929ZeitschriftFuerPhysik}).  The particle may be in one of
two states, $\ket{\mathrm{L}},\ket{\mathrm{R}}$, representing the particle being
located in either the left or right compartment. If the particle is located in the
left compartment, then work can be extracted by attaching a piston to the
separator and letting the gas expand in contact with a heat bath. Yet, if we
know the particle to be initially in the left compartment, it makes no
difference what the process would have done had the particle been in the right
compartment---that situation is irrelevant.  Hence, we may define the
corresponding ``effective process'' as the trace-nonincreasing map, which maps
$\ket{\mathrm{L}}$ to the maximally mixed state (allowing us to extract work) and
which maps $\ket{\mathrm{R}}$ to the zero vector.  Evidently, the full actual
physical implementation is a trace-preserving process, yet it is convenient to
represent the ``relevant part'' of this process using a trace-nonincreasing map.
Crucially, both mappings have the same process matrix, given that the input state
is $\ket{\mathrm L}$.
This picture is justified on a formal level: We show that any
trace-nonincreasing, $\Gamma$-sub-preserving map $\tilde\Phi$ can be dilated in
the following way. There exists a trace-preserving, $\Gamma$-preserving map over
an additional ancilla whose process matrix is as close to a given $\rho_{X'R_X}$
as the process matrix of $\tilde\Phi$ combined with a transition on the ancilla
between two eigenstates of the $\Gamma$ operator
(\autoref{prop:dilation-of-Gsp-to-Gp} in the Appendix).

\section{Results}
\label{main-sec:results}

\subsection{Fundamental work cost of a process}
\label{main-sec:results-main}

Consider two systems $X$ and $X'$ with corresponding operators
$\Gamma_X$ and $\Gamma_{X'}$, respectively, as described above and as imposed by the
appropriate thermodynamic bath~\cite{Barnett2013_beyond,%
  YungerHalpern2016PRE_beyond,%
  Guryanova2016NatComm_multiple,%
  YungerHalpern2016NatComm_NATSandNATO}.
We consider any input state $\sigma_X$ as well as any logical process
$\mathcal{E}_{X\to X'}$, i.e., any completely positive, trace-preserving map.
With a reference system $R_X$ of the same dimension as $X$, which purifies the
input state as $\ket\sigma_{XR_X}$, the logical process and the input state
jointly define the process matrix
$\rho_{X'R_X} =
(\mathcal{E}_{X\to{}X'}\otimes\operatorname{id}_{R_X})(\sigma_{XR_X})$.

Our main result is phrased in terms of the \emph{coherent relative entropy},
defined as
\begin{align}
  \hat{D}_{X\to X'}^\epsilon(\rho_{X'R_X} \Vert \Gamma_X, \Gamma_{X'})
  ~=~
  \max_{\substack{
  \mathcal{T}(\Gamma_X) \leqslant 2^{-\lambda}\Gamma_{X'}\\
  \mathcal{T}(\sigma_{XR_X}) \approx_\epsilon \rho_{X'R_X}}} \lambda\ ,
  \label{main-eq:coh-rel-entr-optimization-def}
\end{align}
where the optimization ranges over completely positive, trace-nonincreasing maps
$\mathcal{T}_{X\to X'}$.
The notation `$\approx_\epsilon$' signifies proximity of the quantum states in
terms of the \emph{purified distance}, a distance measure derived from the
fidelity of the quantum states related to the ability to distinguish the two
states by a measurement~\cite{Tomamichel2010IEEE_Duality,%
  PhDTomamichel2012,%
  BookTomamichel2016_Finite}, which is closely related to the \emph{quantum
  angle}, \emph{Bures distance} and \emph{infidelity} distance
measures~\cite{BookBengtssonZyczkowski2006_Geometry,%
  BookNielsenChuang2000}.

The definition~\eqref{main-eq:coh-rel-entr-optimization-def} is independent of which
purification $\ket\sigma_{XR_X}$ is chosen on $R_X$, noting that $\rho_{X'R_X}$
also depends on this choice.
Furthermore, we use the shorthand
$\hat{D}_{X\to X'}(\rho_{X'R_X} \Vert \Gamma_X, \Gamma_{X'}) :=
\hat{D}_{X\to{}X'}^{\epsilon=0}(\rho_{X'R_X} \Vert \Gamma_X, \Gamma_{X'})$.

At this point, we may formulate our main contribution:
\begin{mainresult}
  The optimal implementation of the process $\mathcal{E}_{X\to X'}$ on the input
  state $\sigma_X,$ with free operations acting jointly on the system $X$ and an
  information battery, can extract a number $\lambda_{\mathrm{optimal}}$ of pure
  qubits given by the coherent relative entropy,
  \begin{align}
    \label{main-eq:main-result-qubits}
    \lambda_{\mathrm{optimal}}
    = \hat{D}_{X\to X'}^\epsilon(\rho_{X'R_X} \Vert \Gamma_X, \Gamma_{X'})\ .
  \end{align}
  If $\lambda_{\mathrm{optimal}}<0$, then the implementation needs to invest at
  least $-\lambda_{\mathrm{optimal}}$ pure qubits.
\end{mainresult}

The resources required to carry out the process, counted in terms of
$\lambda_{\mathrm{optimal}}$ pure qubits, may be converted into physical work.
For instance, if we have access to a heat bath at temperature $T$, we may
convert each pure qubit into $kT\ln(2)$ work and vice versa, and thus the work
extracted by an optimal implementation of the process is
\begin{align}
  \label{main-eq:main-result}
  W = kT\ln(2)\cdot
  \hat{D}_{X\to X'}^\epsilon(\rho_{X'R_X} \Vert \Gamma_X, \Gamma_{X'})\ .
\end{align}
In fact, it is not necessary to implement the process using the information
battery at all, and the resources may be directly supplied by a variety of other
battery models.  The work can even be supplied by a macroscopic pistonlike
system, as we will see later.

Here, we provide  the main technical ingredients to understand the idea of the
proof of our main result, while deferring details to
\autoref{appx:sec:properties-of-the-framework} and
\autoref{appx:sec:coh-rel-entr-def-and-props}.

A central step in our proof is a characterization of how much battery charge
needs to be invested in order to implement exactly any completely positive,
trace-nonincreasing map $\mathcal{T}_{X\to X'}$.  Such maps are those over which
we optimize in~\eqref{main-eq:coh-rel-entr-optimization-def} to define the coherent
relative entropy.
The work yield, or negative work cost, of performing $\mathcal{T}_{X\to X'}$
with $\Gamma$-sub-preserving processes using an information battery is given by
``how $\Gamma$-sub-preserving'' the process is:
\begin{mainproposition}
  \label{main-prop:main-equiv-battery-models}
  Let $\mathcal{T}_{X\to X'}$ be a completely positive, trace-nonincreasing map
  and let $y\in\mathbb{R}$.  Then, the following are equivalent:
  \begin{enumerate}[label=(\alph*)]
  \item The map $\mathcal{T}_{X\to X'}$ satisfies
    \begin{align}
      \mathcal{T}_{X\to X'}(\Gamma_X) \leqslant 2^{-y}\,\Gamma_{X'}\ ;
    \end{align}
  \item For a large enough battery $A$ (with $\Gamma_A=\Ident_A$) and for any
    $\lambda_1,\lambda_2\geqslant 0$ such that
    $\lambda_1-\lambda_2 \leqslant y$, there exists a trace-nonincreasing,
    $\Gamma$-sub-preserving map $\Phi_{XA\to X'A}$ satisfying for all
    $\omega_X$,
    \begin{multline}
      \Phi_{XA\to X'A}\bigl(
      \omega_X \otimes \bigl(2^{-\lambda_1}\Ident_{2^{\lambda_1}}\bigr)
      \bigr)
      \\
      = \mathcal{T}_{X\to X'}(\omega_X) \otimes
      \bigl(2^{-\lambda_2}\Ident_{2^{\lambda_2}}\bigr)
      \ ,
    \end{multline}
    where $2^{-\lambda}\Ident_{2^{\lambda}}$ denotes a uniform mixed state of
    rank $2^{\lambda}$ on system $A$.
  \end{enumerate}
\end{mainproposition}

\autoref{main-prop:main-equiv-battery-models} shows that if there is an allowed
operation in our framework which implements a given completely positive,
trace-nonincreasing map $\mathcal{T}$ exactly while charging the battery by an
amount $\lambda$, then the mapping must necessarily satisfy
$\mathcal{T}(\Gamma)\leqslant 2^{-\lambda}\Gamma$.  Conversely, for any
trace-nonincreasing map $\mathcal{T}$ satisfying
$\mathcal{T}(\Gamma)\leqslant2^{-\lambda}\Gamma$ for some value $\lambda$, there
exists an operation in our framework acting on the system and a battery system
which implements $\mathcal{T}$ while charging the battery by some value
$\lambda$.  This operation is a trace-nonincreasing, $\Gamma$-sub-preserving map
acting on the system and the battery.  From this operation, we can then construct
a fully $\Gamma$-preserving, trace-preserving map that implements
$\mathcal{T}$, as argued at the end of the previous section.

Our main result then exploits \autoref{main-prop:main-equiv-battery-models} in order
to answer the original question, that is, to find the optimal battery charge
extraction when implementing approximately a logical process $\mathcal{E}$ on an
input state $\sigma$.  In effect, one needs to optimize the implementation cost
over all maps $\mathcal{T}$ whose process matrix is $\epsilon$-close to the
required process matrix.  This optimization corresponds precisely to the one
carried out in the definition of the coherent relative entropy
in~\eqref{main-eq:coh-rel-entr-optimization-def}.  (If $\sigma_X$ is full rank and if
$\epsilon=0$, then necessarily $\mathcal{T} = \mathcal{E}$; in general, however,
a better candidate $\mathcal{T}$ may be found.)

\subsection{The coherent relative entropy and its properties}
\label{main-sec:results-cohrelentr}

The coherent relative entropy
$\hat{D}_{X\to X'}^\epsilon(\rho_{X'R_X} \Vert \Gamma_X, \Gamma_{X'})$ defined
in~\eqref{main-eq:coh-rel-entr-optimization-def} intuitively measures the amount of
information discarded during the process, relative to the weights represented in
$\Gamma_X$ and $\Gamma_{X'}$. It ignores unlikely events of total probability
$\epsilon$, a parameter that can be chosen freely.
Its interpretation as a measure of information is justified by the collection of
properties it satisfies, which are natural for such measures, and since it
reproduces known results in special cases.  We provide an overview of the
properties of this quantity here, and refer to
\autoref{appx:sec:coh-rel-entr-def-and-props} for the technical details.

\paragraph{Elementary properties.}
The coherent relative entropy obeys some trivial bounds.  Specifically,
\begin{align}
  &-\log_2\tr(\Gamma_{X}) - \log_2\,\norm[\big]{\Gamma_{X'}^{-1}}_\infty
  \nonumber\\[0.5ex]
  &\hspace*{3em}\leqslant~
    \hat{D}_{X\to X'}^\epsilon(\rho_{X'R_X} \Vert \Gamma_X, \Gamma_{X'})
    + \log_2(1-\epsilon^2)
    \nonumber\\[0.5ex]
  &\hspace*{3em}\leqslant~
  \log_2\,\norm[\big]{\Gamma_{X}^{-1}}_\infty + \log_2\tr(\Gamma_{X'})
    \ .
\end{align}
These bounds have a natural interpretation in the context of a single heat bath
at inverse temperature $\beta=1/(kT)$.  The extracted work may never exceed an
amount corresponding to starting in the highest energy level of the system and
finishing in the Gibbs state; similarly, it may never be less than the amount
corresponding to starting in the Gibbs state and finishing in the highest
excited energy level. (A correction is added to account for additional work
that can be extracted by exploiting the $\epsilon$ accuracy tolerance.)

Under scaling of the $\Gamma$ operators, the coherent relative entropy simply
acquires a constant shift: For any $a,b>0$,
\begin{multline}
  \hat{D}_{X\to X'}^\epsilon(\rho_{X'R_X} \Vert a\Gamma_X, b\Gamma_{X'})
  \\
  =
  \hat{D}_{X\to X'}^\epsilon(\rho_{X'R_X} \Vert \Gamma_X, \Gamma_{X'})
  + \log_2\frac{b}{a}\ .
  \label{main-eq:coh-rel-entr-scaling}
\end{multline}
In the case of a single heat bath at inverse temperature $\beta=1/(kT)$, this
property simply corresponds to the fact that if the Hamiltonians of the input
and output systems are translated by some constant energy shifts, then the
difference in the shifts should simply be accounted for in the work cost.
Indeed, if $H_X\to H_X + \Delta E_X$ and $H_{X'}\to H_{X'} + \Delta E_{X'}$,
then $\Gamma_X\to e^{-\beta\Delta E_X}\,\Gamma_X$,
$\Gamma_{X'}\to e^{-\beta\Delta E_{X'}}\,\Gamma_{X'}$ and the optimal extracted
work of a process, given by $kT\ln(2)$ times the coherent relative entropy, has
to be adjusted according to~\eqref{main-eq:coh-rel-entr-scaling} by
$kT\ln(2)\log_2(e^{-\beta\Delta E_{X'}}/e^{-\beta\Delta E_X}) = \Delta
E_{X}-\Delta E_{X'}$.

\paragraph{Recovering known entropy measures.}

In special cases we recover known results in single-shot quantum thermodynamics,
reproducing existing entropy measures from the smooth entropy
framework~\cite{PhDRenner2005_SQKD,BookTomamichel2016_Finite}.

In the case of a system described by a trivial Hamiltonian, the
work cost of erasing a state to a pure state is given by the
max-entropy~\cite{Dahlsten2011NJP_inadequacy}, a measure that characterizes
data compression or information reconciliation~\cite{Koenig2009IEEE_OpMeaning};
similarly, preparing a state from a pure state allows us to extract an amount of
work given by the min-entropy of the state, a measure that characterizes the
amount of uniform randomness that can be extracted from the state.  These results turn
out to be special cases of considering the work cost of any arbitrary quantum
process for systems with a trivial Hamiltonian~\cite{Faist2015NatComm}, which is
given by the conditional max-entropy of the discarded information conditioned on
the output of the process:
\begin{multline}
  \hat{D}_{X\to X'}^\epsilon(\rho_{X'R_X} \Vert \Ident_X, \Ident_{X'})
  \\
  \approx -\hat{H}_\mathrm{max}^\epsilon(E|X')
  = \hat{H}_\mathrm{min}^\epsilon(E|R_X)\ ,
  \label{main-eq:coh-rel-entr-degenerate-Hamiltonian-Hmax}
\end{multline}
where $\ket\rho_{EX'R_X}$ is a purification of $\rho_{X'R_X}$ and where
$\hat{H}_\mathrm{max}^\epsilon(E|X')$ and $\hat{H}_\mathrm{min}^\epsilon(E|R_X)$
are the smooth conditional max-entropy and min-entropy which were introduced in
ref.~\cite{PhDRenner2005_SQKD}, and are also known as the alternative
conditional max-entropy and
min-entropy~\cite{Tomamichel2011TIT_LeftoverHashing}.  A precise meaning of the
approximation in~\eqref{main-eq:coh-rel-entr-degenerate-Hamiltonian-Hmax} is provided
in \autoref{appx:sec:coh-rel-entr-def-and-props}.

We recover more known results with an arbitrary Hamiltonian in contact with a
heat bath by considering state formation and work extraction of a quantum
state~\cite{Aberg2013_worklike,Horodecki2013_ThermoMaj}.  It is known that the
work that can be extracted from a quantum state, or that is required to form a
quantum state, is given by the min-relative entropy and the max-relative
entropy, respectively; these single-shot relative entropies were introduced in
ref.~\cite{Datta2009IEEE_minmax} and are related to hypothesis testing~\cite{%
  Buscemi2010IEEETIT_capacity,%
  Brandao2011IEEETIT_oneshot,%
  Tomamichel2013_hierarchy,%
  Wang2012PRL_oneshot,%
  Matthews2014IEEETIT_blocklength,%
  Mosonyi2014CMP_hypothesis}.
We show that if the input or output system is trivial, then
\begin{subequations}
\begin{align}
  \hat{D}_{X\to\varnothing}^\epsilon(\rho_{R_X} \Vert \Gamma_X, 1)
  &\approx D_\mathrm{min,0}^\epsilon(\rho_X\Vert\Gamma_X)\ ; \\
  \hat{D}_{\varnothing\to X'}^\epsilon(\rho_{X'} \Vert 1, \Gamma_{X'})
  &\approx -D_\mathrm{max}^\epsilon(\rho_{X'}\Vert \Gamma_{X'})\ ,
\end{align}
\end{subequations}
matching the previously known results.  We note that a trivial system as output
or input of a process is equivalent to mapping to or from a pure, zero-energy
eigenstate; this is because the coherent relative entropy is insensitive to
energy eigenstates (or more generally, eigenstates of the $\Gamma$ operator)
that have no overlap with the corresponding input or output state.

\paragraph{Data processing inequality and chain rule.}
The coherent relative entropy satisfies a data processing inequality: If an
additional channel is applied to the output, mapping the Gibbs weights to other
Gibbs weights, then the coherent relative entropy may only increase.  In other
words, for any channel $\mathcal{F}_{X'\to X''}$,
\begin{multline}
  \hat{D}_{X\to X'}^\epsilon(\rho_{X'R_X} \Vert \Gamma_X, \Gamma_{X'})
\\  \leqslant
  \hat{D}_{X\to X''}^\epsilon(\mathcal{F}_{X'\to X''}(\rho_{X'R_X})
  \Vert \Gamma_X, \mathcal{F}_{X'\to X''}(\Gamma_{X'}))
  \ .
\end{multline}
Intuitively, this holds because the final state after the application of
$\mathcal{F}_{X'\to X''}$ is less valuable as it is closer to the Gibbs state,
and hence more work can be extracted by the optimal process realizing the
total operation $X\to X''$.

The coherent relative entropy also obeys a natural chain rule: The work
extracted during two consecutive processes may only be less than an optimal
implementation of the total effective process.  We refer to
\autoref{appx:sec:coh-rel-entr-def-and-props} for a technically precise
formulation.

\paragraph{Asymptotic equipartition.}
An important property of the coherent relative entropy is its asymptotic
behavior in the limit of many independent copies of the process (known as the
\emph{i.i.d. limit}).  In this regime, the coherent relative entropy converges
to the difference in the quantum relative entropies of the input state to the
output state, which is consistent with previous results in quantum
thermodynamics~\cite{Brandao2013_resource,%
  Brandao2015PNAS_secondlaws}:
\begin{multline}
  \lim_{n\to\infty}
  \frac1n
  \hat{D}_{X^n\to X'^n}^\epsilon(\rho_{X'R_X}^{\otimes n}
  \Vert \Gamma_X^{\otimes n}, \Gamma_{X'}^{\otimes n})
  \\
  = D(\sigma_X\Vert\Gamma_X) - D(\rho_{X'}\Vert\Gamma_{X'})\ ,
\end{multline}
recalling that $\sigma_X$ is the input state of the process and $\rho_{X'}$ the
resulting output state, and where $\epsilon$ is small and either kept constant
or taken to zero slower than exponentially in $n$.
Crucially, the average work cost of performing a process in the i.i.d.\@ regime
with Gibbs-preserving operations does not depend on the details of the process,
but only on the input and output states, as was already the case for systems
described by a trivial Hamiltonian~\cite{Faist2015NatComm}.

\paragraph{Miscellaneous properties.}
We show a collection of further properties, including the following: The
coherent relative entropy is equal
to zero for a pure process matrix, which corresponds to an identity mapping, for
any input state and for $\epsilon=0$; the smooth coherent relative entropy can
be bounded in both directions as differences of known entropy measures; the
coherent relative entropy does not depend on the details of the process if the
input state is of the form $\Gamma_X/\tr(\Gamma_X)$ (e.g., a Gibbs state), and
it reduces, in this case, to a difference of input and output relative entropies and
hence only depends on the output of the process.

\subsection{Battery states and robustness to smoothing}
\label{main-sec:results-batteries}

Previous work has already shown the equivalence of several battery models known
in the literature~\cite{Brandao2015PNAS_secondlaws}, notably the information
battery, the work bit
(``wit'')~\cite{Horodecki2013_ThermoMaj,Brandao2015PNAS_secondlaws}, and the
``weight''
system~\cite{Skrzypczyk2014NComm_individual,Alhambra2016PRX_equality}.
Our framework allows us to make this equivalence manifest, by singling out a class of
states on any system for which the system can act as a battery.
These states exhibit the property that they are reversibly interconvertible (as
in ref.~\cite{Gallego2016NJP_work})---the resources invested in a transition
from one battery state to another can be recovered entirely and
deterministically by carrying out the reverse transition.

For any system $W$ with a corresponding $\Gamma_W$, we consider as battery
states those states of the form
\begin{align}
  \label{main-eq:battery-state}
  \tau(P) = \frac{P\Gamma_W P}{\tr(P\Gamma_W)}\ ,
\end{align}
where $P$ is a projector such that $[P,\Gamma_W]=0$.  In the presence of a
single heat bath at inverse temperature $\beta$, this class of states includes,
for instance, individual energy eigenstates or also maximally mixed states on a
subspace of an energy eigenspace.  We define the value of a particular battery
state $\tau(P)$ as
\begin{align}
  \label{main-eq:battery-state-F}
  \Lambda(\tau(P)) = -\log_2\tr(P\Gamma_W)\ .
\end{align}
We require the system $W$ to start in such a battery state $\tau(P)$ and
to end in another such state $\tau(P')$ corresponding to another
projector $P'$ with $[P',\Gamma_W]=0$.  The following proposition asserts that
the system $W$ can act as a battery enabling exactly the same state transitions on
another system $S$ as an information battery with charge difference
$\lambda_1-\lambda_2 = \Lambda(\tau(P')) - \Lambda(\tau(P))$ (see
\autoref{appx:sec:properties-of-the-framework} for proofs):
\begin{mainproposition}
  \label{main-prop:main-equiv-battery-models-rest}
  Let $\mathcal{T}_{X\to X'}$ be a completely positive, trace-nonincreasing map,
  and let $y\in\mathbb{R}$.  Then, statements (a) and (b) in
  \autoref{main-prop:main-equiv-battery-models} are further equivalent to the following:
  \begin{enumerate}[label=(\alph*)]\setcounter{enumi}{2}
  \item For any quantum system $W$ with corresponding $\Gamma_W$, and for any
    projectors $P, P'$ satisfying $[ P, \Gamma_{W} ] = [ P', \Gamma_{W} ] = 0$
    such that $\Lambda(\tau(P')) - \Lambda(\tau(P)) \leqslant y$, there
    exists a $\Gamma$-sub-preserving, trace-nonincreasing map $\Phi_{XW\to X'W}$
    such that for all $\omega_X$,
    \begin{align}
      \Phi_{XW\to X'W}\bigl( \omega_X \otimes \tau(P) \bigr)
      = \mathcal{T}_{X\to X'}(\omega_X)
      \otimes \tau(P')
      \ .
    \end{align}
  \end{enumerate}
\end{mainproposition}

The information battery, the wit as well as the weight system are themselves
special cases of this general battery system.  Indeed, the states
$2^{-\lambda_i}\Ident_{2^{\lambda_i}}$ of the information battery can be cast in
the form~\eqref{main-eq:battery-state}, with $P=\Ident_{2^{\lambda_i}}$ since
$\Gamma=\Ident$ for the information battery; the corresponding value of the
state is indeed $\Lambda(\tau(P)) = -\lambda_i$.  Similarly, in the case of the
wit and of the weight system, and in the presence of a single heat bath at
inverse temperature $\beta$ such that $\Gamma_W=e^{-\beta H_W}$, the relevant
states are energy eigenstates $\ket{E}_W$, whose value is precisely their energy,
up to a factor $\beta$: $\Lambda(\tau(\proj{E}_W)) = \beta E$.  The
equivalence of these models is thereby manifest.

As can be expected, the battery states of the general form $\tau(P)$ are
reversibly interconvertible, implying that for any process that maps $\tau(P)$
to $\tau(P')$ on a system, the coherent relative entropy is equal to the
difference $\Lambda(\tau(P))-\Lambda(\tau(P'))$.

This general formulation enables us to prove an interesting property of these
battery states---they are robust to small imperfections.  Indeed, when
implementing a process on a system $S$ using a battery $W$, it makes no
difference whether one optimizes over $\epsilon$-approximations of the overall
process on the joint system $S\otimes W$, or over $\epsilon$-approximations on
$S$ only with no imperfections on the battery state (as the smooth coherent
relative entropy is defined above).  More precisely, we prove that the smooth
coherent relative entropy is exactly the optimal difference in the charge state
of the battery while capturing all implementations that include slight
imperfections on the battery for any battery system:
\begin{multline}
  \hat{D}_{X\to X'}^\epsilon(\rho_{X'R_X} \Vert \Gamma_X, \Gamma_{X'})
  \\
  = \max_{\substack{W, P_W,P'_{W},\\ \Phi_{XW\to X'W}}}
  -\log_2\frac{\tr(P'_{W}\Gamma_{W})}{\tr(P_{W}\Gamma_{W})}\ ,
\end{multline}
where the optimization ranges over all battery systems $W$ with corresponding
$\Gamma_W$, over all battery states corresponding to projectors $P_W, P'_{W}$
with $[P_W,\Gamma_W]=[P'_{W},\Gamma_{W}]=0$, and over all free operations
$\Phi_{XW\to X'W}$ which are an $\epsilon$-approximation of a joint process
$XW\to X'W$, with a resulting process matrix on the system of interest given
by $\rho_{X'R_X}$ and which induces a transition on the battery from $\tau(P_W)$
to $\tau(P'_{W})$ (see \autoref{appx:sec:robustness-battery-states}).

\subsection{Emergence of macroscopic thermodynamics}
\label{main-sec:results-macro}

We now apply our general framework to the case of macroscopic systems, and
recover the standard laws of thermodynamics as emergent from our model.
On one hand, the goal of this section is to show that our framework behaves as
expected in the macroscopic limit, further justifying it as a model for
thermodynamics.  On the other hand, the arguments presented here reinforce the
picture of the macroscopic laws of thermodynamics as emergent from microscopic
dynamics, in line with common knowledge and existing
literature~\cite{Lieb1999_secondlaw,%
  Lieb2004_guide_secondlaw,%
  Brandao2013_resource,%
  Brandao2015PNAS_secondlaws,%
  Tajima2016arXiv_large,%
  Chubb2017arXiv_beyond}, by providing an alternative explanation of this
emergence based on $\Gamma$-sub-preserving maps.  (In fact, this
emergence may be understood as defining the order relation in
refs.~\cite{Lieb1999_secondlaw,%
  Lieb2004_guide_secondlaw,%
  Lieb2013_entropy_noneq,%
  Lieb2014PRSA_meter,%
  Weilenmann2016PRL_axiomatic} as the ordering induced by transformation by
$\Gamma$-sub-preserving maps).

\paragraph{The general mechanism.}
The macroscopic theory of thermodynamics is recovered when it is possible to
single out a class of states that obey a reversible interconversion property.
More precisely, suppose there are a class of states
$\{\tau^{z_1,z_2,\ldots,z_m}\}$ specified by $m$ parameters $z_1,\ldots,z_m$,
and suppose there exists a potential $\Lambda(z_1,\ldots,z_m)$ such that for
any pair of states $\tau_{X}^{z_1,\ldots,z_m}$ and
$\tau_{X'}^{z'_1,\ldots,z'_m}$ from this class, we have, for any process
matrix $\rho_{X'R_X}$ mapping one state to the other,
\begin{multline}
  \ln(2)\cdot
  \hat{D}_{X\to X'}(\rho_{X'R_X} \Vert \Gamma_X, \Gamma_{X'})
  \\
  = \Lambda(z_1,\ldots,z_m) - \Lambda(z'_1,\ldots,z'_m)\ .
  \label{main-eq:thermo-states-dcoh-natural-thermodynamic-potential}
\end{multline}
The $\ln(2)$ factor merely serves to change the units of the coherent relative
entropy from bits, which is standard in information theory, to nats, which will
prove convenient to recover the standard laws of thermodynamics.  We call the
function $\Lambda(z_1,\ldots,z_m)$ the \emph{natural thermodynamic potential}
corresponding to the physics encoded in the $\Gamma$ operators.  In other words,
the two states $\tau^{z_1,\ldots,z_m}$ and $\tau^{z'_1,\ldots,z'_m}$ may be
reversibly interconverted, as any work invested when going in one direction may
be recovered when returning to the initial state, and this is irrespective of which
precise logical process is effectively carried out during the transition.  An
obvious choice of states with this property are states of the same form as the battery
states introduced above, which motivates recycling the same symbols $\tau$ and
$\Lambda$.
(We have set $\epsilon=0$
in~\eqref{main-eq:thermo-states-dcoh-natural-thermodynamic-potential} because
smoothing such battery-type states has no significant effect.%
)

Suppose that the parameters are sufficiently well approximated by continuous values.
This would typically be the case for a large system such as a macroscopic gas.
Consider an infinitesimal change of a state
$(z_1,\ldots,z_m)\to(z_1+d z_1, \ldots, z_m + d z_m)$.  If there is a free
operation that can perform this transition, then necessarily, the coherent
relative entropy is positive; hence,
$\Lambda(z_1+dz_1,\ldots,z_m+dz_m)\leqslant\Lambda(z_1,\ldots,z_m)$.
Conversely, if the coherent relative entropy is positive, then there necessarily
exists a free operation implementing the said transition.  We deduce that the
infinitesimal transition $z\to z+d z$ is possible with a free operation if and
only if
\begin{align}
  \label{main-eq:macro-thermo-dLambda-leq-zero}
  d\Lambda \leqslant 0\ .
\end{align}
This condition expresses the macroscopic second law of thermodynamics, as we
will see below.

We may define the generalized chemical potentials
\begin{align}
  \label{main-eq:macro-thermo-generalized-chemical-potential}
  \mu_i = \left(\frac{\partial\Lambda}{\partial z_i}\right)_{z_1,\ldots,z_{i-1},z_{i+1},\ldots,z_m}\ ,
\end{align}
where the notation $(\partial f/\partial x)_{y,z}$ denotes the partial
derivative with respect to $x$ of a function $f$, as $y$ and $z$ are kept
constant.  We may then write the differential of $\Lambda$ as
\begin{align}
  d\Lambda = \sum \mu_i\,dz_i\ .
\end{align}
The generalized potentials $\mu_i$ are often directly related to physical
properties of the system in question, such as temperature, pressure, or chemical
potential.

Under external constraints on the variables $z_1,z_2,\ldots,z_m$, we may ask
what the ``most useless thermodynamic state'' compatible with those conditions
is.  The answer is given by minimizing the potential $\Lambda$ subject to those
constraints---this is a variational principle.
For instance, if two systems with natural thermodynamic potentials
$\Lambda_1(z_1,\ldots,z_m)$ and $\Lambda_2(z_1',\ldots,z_m')$ are put into
contact under the constraints that for all $i$, $z_i+z_i'$ must be kept constant
(such as for extensive variables in thermodynamics), then we may write
$dz_i=-dz_i'$ and minimize $\Lambda = \Lambda_1+\Lambda_2$ by requiring that
\begin{align}
  0 = d\Lambda = \sum(\mu_i-\mu_i')\,dz_i\ ,
\end{align}
and we see that the minimum is attained when $\mu_i=\mu_i'$.  If the system is
undergoing suitable thermalizing dynamics, then its evolution will naturally
converge towards that point.

\paragraph{The textbook thermodynamic gas.}
We proceed to recover the usual laws of thermodynamics in this fashion for a
macroscopic isolated gas $S$ composed of many particles
(\autoref{main-fig:emergence-macro-gas}).
\begin{figure}
  \centering
  \includegraphics{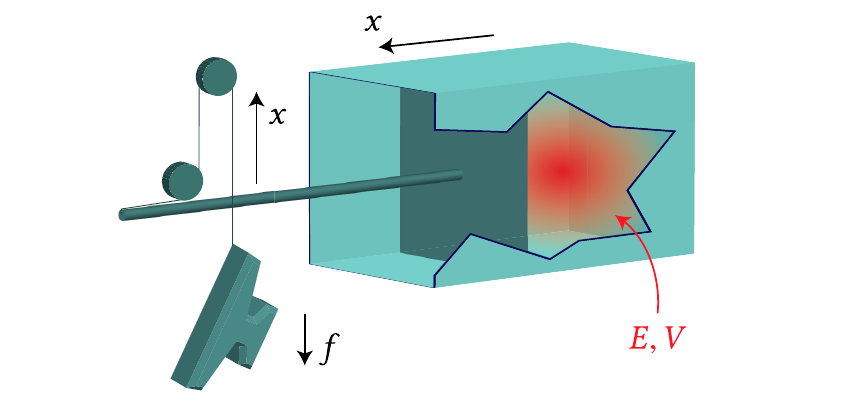}
  \caption{Macroscopic thermodynamics emerges from our framework when singling
    out a set of states that can be parametrized by continuous parameters to a
    good approximation and can be reversibly interconverted into one
    another.  We consider the case of a textbook thermodynamic gas confined in a
    box, with a piston capable of furnishing work.  In this setting, we recover
    the usual second law of thermodynamics, $dS \geqslant \delta Q/T$, relating
    the change in entropy, the dissipated heat, and the temperature.}
  \label{main-fig:emergence-macro-gas}
\end{figure}
The Hamiltonian of the gas is denoted by $H^{(V)}$, where the volume $V$
occupied by the gas is a classical parameter of the Hamiltonian that determines,
for instance, the width of a confining potential.  We assume, for simplicity, that
the number $N$ of particles constituting the gas is kept at a fixed value
throughout, restricting our considerations to the corresponding subspace.

Let us first consider the case of an isolated gas at fixed parameters $E,V$.  In
order to apply our framework, we must identify the $\Gamma$ operator, which
encodes the relevant restrictions imposed by the physics of our system.  Recall
that our restriction is meant to explicitly forbid certain types of processes,
without worrying whether a nonforbidden operation is achievable.  Here, we
assume that at fixed $E,V$, the system is isolated and hence evolves unitarily.
In particular, the projector $P^{E,V}_S$ onto the eigenspace of $H^{(V)}$
corresponding to energy $E$ is preserved.
Hence, the $\Gamma$ operator characterizing the gas alone for fixed $E,V$
can be taken as
\begin{align}
  \Gamma_S^{E,V} = P^{E,V}_S\ .
\end{align}
This is compatible with standard considerations in statistical mechanics, which
identify the state of the gas in such conditions as the maximally mixed state in
the subspace projected onto by $P^{E,V}_S$ (the microcanonical state), which we
denote by $\tau_S^{E,V} = P^{E,V}_S/\tr(P^{E,V}_S)$.  Indeed, at fixed $E,V$ on
the control system, an allowed transformation may not change this state.

Now, we would like to account for changes in $E,V$.  It is convenient to
introduce a physical control system $C$, which plays the following roles: It
stores the information about all the controlled external parameters of the state
in which the gas was prepared---here, the parameters are $E,V$; furthermore, it
provides the necessary physical constraints on the gas and physical resources
necessary for transformations, taking on the role of a battery.  In our case,
the control system includes a piston that confines the gas to a volume $V$ and
is capable of furnishing the energy required to change the state of the gas.
For concreteness, we imagine that the piston is balanced by a weight, causing
the piston to exert a force $f$ on the gas.  The force $f$ may be tuned by
varying the weight.
The states of the control system are $\ket{e,x}_C$, where $e$ is the energy
stored in the control system and $x$ the position of the piston.  The energy $e$
is the potential energy of the weight, and it must be equal to
$e = E_\mathrm{tot} - E$ as enforced by total energy conservation, where
$E_\mathrm{tot}$ is the fixed total energy of the joint $CS$ system.
Furthermore, $x$ determines the volume of the gas as $V=A\cdot x$, where $A$ is
the surface of the piston.
If the control system were isolated and not coupled to the gas, then the
nonforbidden operations on the control system would be those preserving the
operator $\Gamma_C^0 = \sum_{e,x} g_{e,x}\proj{e,x}_C$, where $g_{e,x}$ encodes
the relevant physics of the control system: It decreases as either $e$ increases
or $x$ increases, meaning that a state $\ket{e,x}_C$ cannot be brought to the
state $\ket{e',x}_C$ with $e'>e$ or $\ket{e,x'}_C$ with $x'>x$.  In other words,
we do not forbid reducing the weight charge or lowering it.

The coupling between the control system and the gas can be enforced with a
$\Gamma$ operator of the form
\begin{align}
  \Gamma_{CS} = \sum_{e,x}
  g_{e,x} \proj{e,x}_C\otimes P^{E=E_{\mathrm{tot}}-e,V=Ax}_S\ .
  \label{main-eq:emergence-Gamma-CS}
\end{align}
If the control system is the state $\ket{e,x}_C$, then any allowed operation
must preserve the operator $\Gamma^{E,V}_S$ for the corresponding
$E=E_\mathrm{tot}-e$ and $V=Ax$.  Furthermore~\eqref{main-eq:emergence-Gamma-CS}
accounts for the physics of the control system itself with the coefficient
$g_{e,x}$.

The states
$\tau_{CS}^{e,x} = \proj{e,x}_C\otimes\tau_S^{E=E_{\mathrm{tot}}-e,V=Ax}$ are of
the form~\eqref{main-eq:battery-state}; hence, they are reversibly
interconvertable as
per~\eqref{main-eq:thermo-states-dcoh-natural-thermodynamic-potential} and they are a
valid class of states for our macroscopic description.  The corresponding
natural thermodynamic potential is given as per~\eqref{main-eq:battery-state-F},
\begin{align}
  \Lambda_{CS}(e,x)
  &= \Lambda_C(e,x) + \Lambda_S(E_{\mathrm{tot}}-e,Ax)\ ,
\end{align}
where we have defined $ \Lambda_C(e,x) = -\ln g_{e,x} $ and
$\Lambda_S(E,V) = -\ln \tr(P^{E,V}_S)$.  Observe that
$\tr P^{E,V}_S = \Omega_S(E,V)$ is the microcanonical partition function, and
hence $\Lambda_S(E,V)$ is, up to Boltzmann's constant $k$ and a minus sign,
the quantity $S(E,V) = k\ln\Omega_S(E,V)$, which is known as the thermodynamic
entropy of the gas:
\begin{align}
  \Lambda_S(E,V) = -k^{-1} S(E,V)\ .
\end{align}

As the gas is macroscopic, we assume that the parameters $E,V$ are well
approximated by continuous variables.  It is useful to define the conjugate
variables to $e,x$ and $E,V$ via the differentials of $\Lambda_C$ and
$\Lambda_S$:
\begin{subequations}
  \begin{align}
    d\Lambda_C &= \nu_e\, de + \nu_x\, dx\ ;
    \label{main-eq:emergence-diff-Lambda-C}\\
    d\Lambda_S &= \mu_E\, dE + \mu_V\, dV\ ,
    \label{main-eq:emergence-diff-Lambda-S}
  \end{align}
\end{subequations}
with the coupling inducing the relations $dE=-de$ and $dV = A\,dx$.  The force
$f$ exerted by the piston onto the gas is given by
$f=(\partial e/\partial x)_{\Lambda_C}$.
Using~\eqref{main-eq:emergence-diff-Lambda-C} we see that
$de = \nu_e^{-1}(d\Lambda_C - \nu_x dx)$, and hence $f = -\nu_x/\nu_e$.  The
thermodynamic work provided by the piston is the mechanical work performed by
the weight,
\begin{align}
  \delta W = -f\cdot dx = \frac{\nu_x}{\nu_e}\,dx\ .
\end{align}
Any operation mapping two states $\tau_{CS}^{e,x}\to\tau_{CS}^{e+de,x+dx}$ which
obeys our global restriction, i.e.\@ which preserves the
operator~\eqref{main-eq:emergence-Gamma-CS}, must
obey~\eqref{main-eq:macro-thermo-dLambda-leq-zero} or, equivalently,
$d\Lambda_S\leqslant -d\Lambda_C$; hence,
\begin{align}
  d\Lambda_S \leqslant -\nu_e\,de - \nu_x\,dx
  = \nu_e\,(dE - \delta W) = \nu_e\,\delta Q\ ,
  \label{main-eq:emergence-second-law-for-LambdaS-dQ}
\end{align}
where we have defined the change in energy of the gas that is not due to
thermodynamic work as \emph{heat}: $\delta Q = dE - \delta W$.

The temperature of the gas is defined as
$T_\mathrm{gas} = (\partial S/\partial E)^{-1} = -(k\mu_E)^{-1}$ as in standard
textbooks, as the conjugate variable corresponding to entropy.  The control
system also acts as a heat bath, so we define its temperature $T$ as the
temperature of a gas that it would be ``in equilibrium'' with, in the sense that our
variational principle is achieved.  The potential $\Lambda_{CS}$ attains its
minimum under the constraints $dE=-de$ and $dV=A\cdot dx$ if
$0=d\Lambda_{CS}=(\mu_E-\nu_e)\,dE + (\mu_V+A^{-1}\nu_x)\,dV$, implying that
$\mu_E=\nu_e$ and hence $T=-(k\nu_e)^{-1}$.
We may now write~\eqref{main-eq:emergence-second-law-for-LambdaS-dQ} in its more
traditional form,
\begin{align}
  dS\geqslant \frac{\delta Q}{T}\ .
  \label{main-eq:macro-second-law-traditional}
\end{align}

Our control system is in fact another example of a battery system.  Indeed, it
can convert another form of a useful resource, mechanical work, into the
equivalent of pure qubits for enabling processes on the system, while still
working under the relevant global constraints such as conservation of energy.

The thermodynamic gas illustrates a situation in which the macroscopic second
law of thermodynamics is recovered as emergent.  Note that the argument can also be
applied  to a system with different relevant physical quantities, such as
magnetic field and magnetization of a medium.

\subsection{Observers in thermodynamics}
\label{main-sec:results-obs}

In standard thermodynamics, one describes systems from the macroscopic point of
view.  This point of view is usually assumed only implicitly, to the point that
notions such as thermal equilibrium or the thermodynamic entropy function are
often thought of as objective properties of the system.
Yet, a closer look reveals that they can be thought of as observer-dependent
quantities, which can be extended to observers with different amounts of
knowledge about the
system~\cite{Jaynes1992_Gibbs,delRio2011Nature,delRio2014arXiv_relative}.  This
observation is at the core of a modern understanding of Maxwell's demon.

The present section begins with a brief motivation, reviewing a variant of
Maxwell's demon.  Then, we show that our framework is well suited for describing
different observers and that it provides a natural notion of coarse-graining.  Indeed,
the framework itself, thanks to the abstraction provided by the $\Gamma$
operator, is scale agnostic and can be applied consistently from any level of
knowledge about the system.
More precisely, we show how to relate two descriptions from the viewpoints of
two observers, where one observer sees a coarse-grained version of another
observer's knowledge.  The coarse-graining is given by any completely positive,
trace-preserving map.  We define a sense in which we can carry out the reverse
transformation, where one recovers the fine-grained information, given the
coarse-grained information, with the help of a recovery map.  This allows us to
relate the laws of thermodynamics in either observer's picture, where by the
``laws of thermodynamics'' in an observer's picture, we mean that the evolution
of the system is governed in their picture by $\Gamma$-sub-preserving maps.
This provides a precise criterion that can guarantee, in a given setting, that
the laws of thermodynamics hold in the coarse-grained picture or, intuitively,
that ``no Maxwell-demon-type cheating'' is happening.  Namely, if the
fine-grained picture has no more information than what can be recovered from the
coarse-grained picture, then our framework may be applied consistently from
either picture, with both observers agreeing on the class of possible processes.

Consider the variant of Maxwell's demon depicted in \autoref{main-fig:MaxwellDemon}.
\begin{figure}
  \centering
  \includegraphics{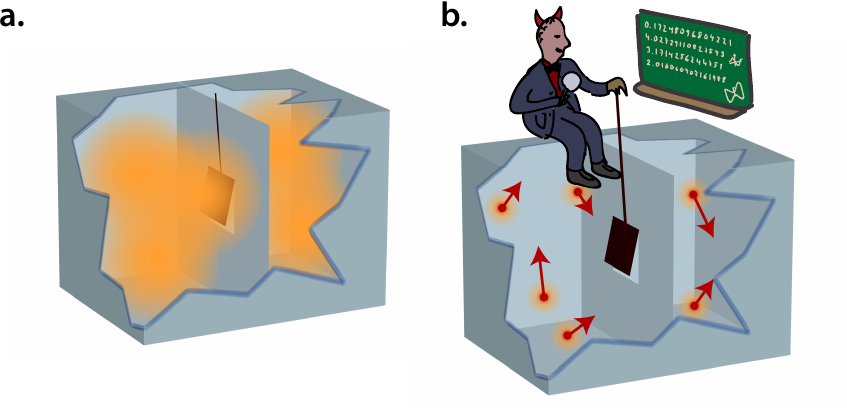}
  \caption{Maxwell's demon concentrates all particles on one side of the box by
    opening the trap door at appropriate times. \textbf{a.}~A macroscopic
    observer describing only the gas sees its entropy decrease, in apparent
    violation of the macroscopic observer's idea of the second law of thermodynamics.
    \textbf{b.}~The demon observes no entropy change, as the state of the gas is
    conditioned on his knowledge.  By modeling his memory as an explicit system,
    originally in a pure state, we may understand his actions as simply
    correlating his memory with the state of the gas.  In doing so, a macroscopic
    observer may be induced into witnessing a violation of a macroscopic second
    law.  If the demon wishes to operate cyclically, he needs to reset his
    memory register back to a pure state, which costs work according to
    Landauer's
    principle~\protect\cite{Bennett1982IJTP_ThermodynOfComp,Bennett2003_NotesLP};
    any work he might have extracted using his scheme is paid back at this
    point.}
  \label{main-fig:MaxwellDemon}
\end{figure}
A gas is enclosed in a box separated into two equal volume compartments, which
communicate only through a small trap door controlled by a demon.  The demon is
able to observe individual particles and activates the trap door at appropriate
times, letting a single particle through each time, in order to concentrate all
particles on one side of the box.
From a macroscopic perspective, and looking only at the gas, one observes an
apparent entropy decrease as the gas now occupies a smaller volume.
However, from a microscopic perspective, the demon is essentially transferring
entropy from the gas into a memory register, which is initially in a pure
state~\cite{Bennett1982IJTP_ThermodynOfComp,Bennett2003_NotesLP}.  Consider in
more detail the following process: The demon performs a series of \textsc{cnot}
gates using the gas degrees of freedom as controls and his memory qubits as
targets, which ``replicates'' the information about the gas particles into his
memory.  Since this process is unitary, it preserves the joint entropy of the
memory and the gas.  The result is a classically correlated state between the
memory register and the gas.
So, \emph{what is the entropy of the gas?}  It is now clear that the answer
depends on the observer.  The macroscopic observer sees the gas with its usual
macroscopic thermodynamic entropy, while the demon has engineered a state where
the gas has zero entropy conditioned on the side information stored in his
memory---he knows all there is to know about the gas.  Conceptually, the
thermodynamic reason for this difference is that the demon is able to extract
work from the gas, whereas the macroscopic observer is not.
Indeed, the demon can exploit the side information stored in his memory to
design a perfect trap-door opening schedule which, when executed, concentrates
all the particles on one side of the box.  (This process can itself be thought
of as \textsc{cnot} gates acting in the other direction.)  With all particles
concentrated on one side of the box, the demon can now extract work by replacing the
separator by a piston and letting the gas expand isothermally.  (Of course, the
memory register is still littered with all the information about the gas;
resetting the register costs work according to Landauer's principle, which is
where the demon pays back his extracted work if he wishes to operate
cyclically~\cite{Bennett1982IJTP_ThermodynOfComp,Bennett2003_NotesLP}.)

The above example shows that a fully general framework of thermodynamics should
be universally applicable from the point of view of any observer, accounting for
any level of knowledge one might possess about a system.  One also expects that
if an observer sees a violation of their laws of thermodynamics, while knowing
that in a finer-grained picture the corresponding laws are obeyed, then they may
attribute this effect to lack of knowledge about microscopic degrees of freedom
which the observed process exploits.
In the following, we show that our framework displays these desired properties.

Consider two observers, Alice and Bob, who have distinct degrees of knowledge
about a system.  We assume that the system's microscopic state space
$\mathscr{H}_A$, which Alice has access to, is transformed by a completely
positive, trace-preserving map
$\mathcal{F}^{\mathcal{A}\to\mathcal{B}}_{A\to B}$ to a state space
$\mathscr{H}_B$ which is used by Bob to describe the situation
(\autoref{main-fig:ObserversGamma}).
\begin{figure}
  \centering
  \includegraphics{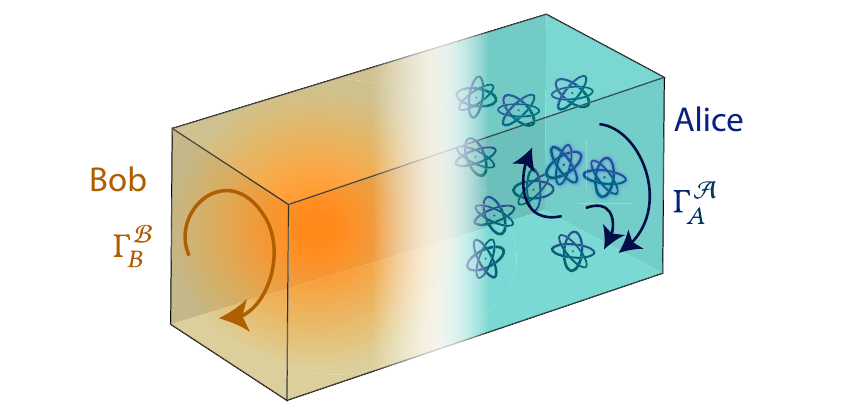}
  \caption[]{Observers in thermodynamics.  Alice has access to microscopic
    degrees of freedom of a gas, while Bob can only observe its coarse
    macroscopic properties, such as its temperature $T$, volume $V$ and pressure
    $p$.  Alice describes the evolution of the gas using Gibbs-preserving maps,
    with a Gibbs state $\Gamma_A^{\mathcal{A}}$ on the full state space of the
    many particles of the gas.  On the other hand, Bob describes the gas using
    his own knowledge---for instance, the macroscopic variables $T,V,p$---which
    in full generality we can represent as a quantum state in a state space
    $\mathscr{H}_B$ which is obtained by applying a given mapping
    $\mathcal{F}_{A\to B}^{\mathcal{A}\to\mathcal{B}}(\cdot)$ on Alice's
    state. (For instance, this map may trace out the inaccessible microscopic
    information.) States of the gas described by Bob may be transformed to
    Alice's picture by applying a suitable recovery map, such as the Petz
    map~\protect\cite{%
      Petz1986CMP_sufficient,%
      Petz1988QJM_sufficiency,%
      Petz2003RMathPhys_monotonicity,%
      Wilde2015RSPA_recoverability,%
      Beigi2016JMP_Petz,%
      Wehner2015arXiv_reversibility}.  Then, Alice's
    $\Gamma_A^\mathcal{A}$-preserving maps appear to Bob as
    $\Gamma_B^{\mathcal{B}}$-preserving maps, where Bob's
    $\Gamma_B^{\mathcal{B}}$ operator is taken to be
    $\Gamma_B^{\mathcal{B}} =
    \mathcal{F}_{A\to{}B}^{\mathcal{A}\to\mathcal{B}}(\Gamma_A^{\mathcal{A}})$.
    Conversely, operations that preserve $\Gamma_B^{\mathcal{B}}$ for Bob may
    be described by Alice as preserving $\Gamma_A^{\mathcal{A}}$.}
  \label{main-fig:ObserversGamma}
\end{figure}
For instance, Alice might have access to individual position and momenta of all
the particles of a gas, while Bob only has access to partial information given
by macroscopic physical quantities such as temperature, pressure, volume, etc.
More generally, if the microscopic system can be embedded in a bipartite
system $\mathscr{H}_{K}\otimes\mathscr{H}_{N}$ that stores, respectively, the
macroscopic information (available to both Bob and Alice) and the microscopic
information (available to Alice only), then Bob's observations can be
related to Alice's simply by tracing out the $\mathscr{H}_N$ system.

Suppose that Alice observes some microscopic dynamics happening within
$\mathscr{H}_A$ and that this evolution is $\Gamma$-preserving with a
particular operator $\Gamma_A^{\mathcal{A}}$.  How does this evolution appear to
Bob?  It turns out that for Bob, these maps are also $\Gamma$-preserving maps,
but they are relative to \emph{his} $\Gamma$ operator, which is simply given as
$\Gamma_B^{\mathcal{B}} =
\mathcal{F}^{\mathcal{A}\to\mathcal{B}}_{A\to{}B}(\Gamma_A^{\mathcal{A}})$, that
is, by transforming Alice's $\Gamma$ operator into Bob's picture.  Conversely, a
map that appears as $\Gamma^{\mathcal{B}}$-preserving to Bob is observed
by Alice as being $\Gamma^{\mathcal{A}}$-preserving.

In order to give a precise meaning to the above statements, it is necessary to
specify how a state described by Bob can be translated back to Alice's picture.
Indeed, there can be several possible states for Alice that are compatible with
Bob's state.  We describe this ``recovery process'' using a \emph{recovery map},
which gives, in a sense, the ``best guess'' of what the state on $\mathscr{H}_A$
could be, given only knowledge of Bob's state on $\mathscr{H}_B$.  More
precisely, we define the state transformation from Bob's picture to Alice's
picture as the application of a completely positive, trace-preserving map
$\mathcal{R}_{B\to A}^{\mathcal{B}\to\mathcal{A}}(\cdot)$, with the property
that
$\mathcal{R}_{B\to A}^{\mathcal{B}\to\mathcal{A}}(\Gamma_B^\mathcal{B}) =
\Gamma_A^\mathcal{A}$, recalling that
$\Gamma^{\mathcal{B}}_B = \mathcal{F}^{\mathcal{A}\to\mathcal{B}}_{A\to B}(
\Gamma^{\mathcal{A}}_{A} ) $.  This ensures that the completely useless state in
Bob's picture is mapped back to the completely useless state in Alice's picture.
An example of a suitable recovery map is the \emph{Petz recovery map}~\cite{%
  Petz1986CMP_sufficient,%
  Petz1988QJM_sufficiency,%
  Petz2003RMathPhys_monotonicity,%
  Wilde2015RSPA_recoverability,%
  Beigi2016JMP_Petz,%
  Wehner2015arXiv_reversibility}, defined as
\begin{align}
  \mathcal{R}_{B\to A}^{\mathcal{B}\to\mathcal{A}}(\cdot)
  = \Gamma^{\mathcal{A}\,1/2}_{A}\,
  \mathcal{F}^{\mathcal{A}\leftarrow\mathcal{B}\,\dagger}_{A\leftarrow B}\bigl( 
  \Gamma^{\mathcal{B}\,-1/2}_B\,(\cdot)\,\Gamma^{\mathcal{B}\,-1/2}_B \bigr)
  \, \Gamma^{\mathcal{A}\,1/2}_{A}\ ,
\end{align}
where
$\mathcal{F}^{\mathcal{A}\leftarrow\mathcal{B}\,\dagger}_{A\leftarrow B}$ is the
adjoint of the superoperator $\mathcal{F}^{\mathcal{A}\to\mathcal{B}}_{A\to B}$.
The Petz recovery map is completely positive and trace preserving, and satisfies
$\mathcal{R}_{B\to A}^{\mathcal{B}\to\mathcal{A}}(\Gamma_{B}^{\mathcal{B}}) =
\Gamma_A^{\mathcal{A}}$ (assuming that $\Gamma_B^{\mathcal{B}}$ is full rank).

Hence, given a trace-nonincreasing mapping $\mathcal{E}^{\mathcal{A}}_{A}$ in
Alice's picture, we define Bob's description of the mapping as the composed map
of transforming into Alice's picture, applying the map, and transforming back to
Bob's picture:
\begin{align}
\mathcal{E}^{\mathcal{B}}_B = \mathcal{F}^{\mathcal{A}\to\mathcal{B}}_{A\to B}
\circ \mathcal{E}^{\mathcal{A}}_A \circ
\mathcal{R}_{B\to{A}}^{\mathcal{B}\to\mathcal{A}}\ .
\end{align}
Our claim is the following: If $\mathcal{E}^{\mathcal{A}}_{A}$ satisfies
$\mathcal{E}^{\mathcal{A}}_{A}(\Gamma^{\mathcal{A}}_A) \leqslant
\Gamma^{\mathcal{A}}_A$, then $\mathcal{E}^{\mathcal{B}}_B$ satisfies
$\mathcal{E}^{\mathcal{B}}_B(\Gamma^{\mathcal{B}}_B)\leqslant
\Gamma^{\mathcal{B}}_B$.
Conversely, if we are given a trace-nonincreasing mapping
$\mathcal{E}^{\mathcal{B}}_{B}$ in Bob's picture, then this map is described in
Alice's picture as the composed map of transforming to Bob's picture, applying
the map, and transforming back:
\begin{align}
\mathcal{E}^{\mathcal{A}}_{A} =
\mathcal{R}_{B\to{A}}^{\mathcal{B}\to\mathcal{A}} \circ
\mathcal{E}^{\mathcal{B}}_{B} \circ
\mathcal{F}^{\mathcal{A}\to\mathcal{B}}_{A\to{}B}\ ;
\end{align}
we assert that if
$\mathcal{E}^{\mathcal{B}}_{B}(\Gamma^{\mathcal{B}}_B)\leqslant
\Gamma^{\mathcal{B}}_B$, then
$\mathcal{E}^{\mathcal{A}}_{A}(\Gamma^{\mathcal{A}}_A) \leqslant
\Gamma^{\mathcal{A}}_A$.

The proof of both claims is straightforward, using
$\mathcal{F}_{A\to{B}}^{\mathcal{A}\to\mathcal{B}}(\Gamma_A^{\mathcal{A}}) =
\Gamma_B^{\mathcal{B}}$ and
$\mathcal{R}_{B\to{A}}^{\mathcal{B}\to\mathcal{A}}(\Gamma^{\mathcal{B}}_B) =
\Gamma^{\mathcal{A}}_A$.  More generally, these claims hold as well for any
trace-nonincreasing, completely positive maps
$\mathcal{F}_{A\to{}B}^{\mathcal{A}\to\mathcal{B}}$,
$\mathcal{R}_{B\to{}A}^{\mathcal{B}\to\mathcal{A}}$ satisfying
$\mathcal{F}_{A\to{}B}^{\mathcal{A}\to\mathcal{B}}(\Gamma_A^{\mathcal{A}})
\leqslant \Gamma_B^\mathcal{B}$ and
$\mathcal{R}_{B\to{}A}^{\mathcal{B}\to\mathcal{A}}(\Gamma_B^{\mathcal{B}})
\leqslant \Gamma_A^\mathcal{A}$, in which case $\Gamma_B^{\mathcal{B}}$ does not
have to be full rank.

The above provides a general criterion that is able to guarantee that the laws
of thermodynamics in the coarse-grained picture are valid: If the state of the
system in Alice's picture is one that can be recovered from Bob using a fixed
recovery map, then Alice's free operations correspond to free operations in
Bob's picture, and hence Alice's laws of thermodynamics indeed translate to
Bob's idea of what the laws of thermodynamics are.

A simple example is the relation of the microcanonical to the canonical
ensemble.  (This is also known as Gibbs-rescaling, an essential tool to relate
thermal operations to noisy operations~\cite{Egloff2015NJP_measure,%
  Horodecki2013_ThermoMaj,Brandao2015PNAS_secondlaws}.)  If Alice describes
unitary dynamics within an energy eigenspace of the joint system and a large
heat bath, then Bob describes the dynamics of the system alone as
Gibbs-preserving maps.
Consider a system $S$ and a
heat bath $R$, with respective Hamiltonians $H_S$ and $H_R$ and total
Hamiltonian $H_{SR} = H_S+H_R$.  Suppose that Alice has microscopic access to
the heat bath and hence describes the situation using the state space
$A=S\otimes R$.  Assume that the global state and evolution are constrained to
unitaries within a subspace of fixed total energy $E$.  This evolution is, in
particular, $\Gamma$-sub-preserving if we choose
$\Gamma^{\mathcal{A}}_{A} = P_{SR}^{E}$, where $P_{SR}^E$ is the projector onto
the eigenspace of $H_{SR}$ corresponding to the energy $E$.
On the other hand, Bob only has access to the system $B=S$.  The mapping
$\mathcal{F}^{\mathcal{A}\to\mathcal{B}}$, which relates Alice's point of view
to Bob's, simply traces out the heat bath $R$.
Bob then describes the operator $\Gamma^{\mathcal{A}}_{A}$ as
\begin{align}
  \Gamma^{\mathcal{B}}_{S} = \tr_R(\Gamma^{\mathcal{A}}_{SR})
  = \sum_{E_S,k} g(E-E_S)\,\proj{E_S,k}_S\ ,
  \label{main-eq:relating-obs-GammaS-canonical}
\end{align}
where $g(E_R)$ is the degeneracy of the energy eigenspace of the heat bath
corresponding to the energy $E_R$, and where the vectors $\{\ket{E_S,k}_S\}$ are
the energy eigenstates on $S$ with a possible degeneracy index $k$.  Following
standard arguments in statistical mechanics, and as argued in
ref.~\cite{Horodecki2013_ThermoMaj}, we have, in typical situations and under
mild assumptions, $g(E-E_S)\propto e^{-\beta E_S}$, and we hence recover
in~\eqref{main-eq:relating-obs-GammaS-canonical} the standard canonical form of the
thermal state.
In other words, Bob describes the dynamics on $S$ as maps that preserve the
Gibbs state.

The above reasoning can be seen as a rule for transforming one observer's
picture into another; it remains important to analyze the situation in the
picture that accurately describes the state of knowledge of the input state of
the agent carrying out the operations.  The pictures are equivalent when Alice's
state of knowledge of $A$ is no more than what $B$ can recover using the
recovery map, i.e., when her input state is exactly of the form
$\mathcal{R}_{B\to A}^{\mathcal{B}\to\mathcal{A}}(\rho_B^\mathcal{B})$, where
$\rho_B^{\mathcal{B}}$ is the state of the system in Bob's picture.  However,
not all actions that Alice can perform using
$\Gamma^{\mathcal{A}}_A$-sub-preserving maps must induce a
$\Gamma^{\mathcal{B}}_B$-sub-preserving effective map on $B$.  Indeed, if
Alice's input state is more refined, i.e., if she has more fine-grained
information about the microscopic initial state than what Bob can infer, then
her actions might appear to Bob as violating his idea of the second law of
thermodynamics.  In this case, Alice may indeed perform
$\Gamma^\mathcal{A}_A$-sub-preserving operations that result in an effective
mapping on $B$ that is not $\Gamma^{\mathcal{B}}_B$-sub-preserving.  Enter
Maxwell's demon.

Our framework hence allows us to systematically analyze a variety of settings
inspired by Maxwell's demon.  Returning to our example depicted in
\autoref{main-fig:MaxwellDemon}, we identify Alice as possessing a microscopic
description of the gas and the demon, and Bob as the macroscopic observer.  The
demon, as described by Alice, can perform Gibbs-preserving operations on the
joint system of the gas $S$ and the demon's memory register $M$, which,
for simplicity, we choose to have a completely degenerate Hamiltonian $H_M=0$ and
thus $\Gamma_M = \Ident_M$.
Bob, on the other hand, describes the gas alone using standard thermodynamic
variables, say, the energy $E$, the volume $V$, and the number of particles $N$.  To relate
both points of view, we write the gas system (including a possible control
system to fix macroscopic thermodynamic variables) as a bipartite system
$S=K\otimes N$ with states of the form $\proj{E,V,N}_K\otimes\tau^{E,V,N}_N$,
where $\tau^{E,V,N}_N$ is the microcanonical state corresponding to the
macroscopic variables $E,V,N$.  We have
$\tau^{E,V,N}_N = P^{E,V,N}_N/\Omega(E,V,N)$, where $P^{E,V,N}_N$ projects onto
the subspace of the microscopic system corresponding to fixed $E,V,N$, and where
the partition function is $\Omega(E,V,N) = \tr[P^{E,V,N}_N]$.  Then, Bob's
picture is obtained from Alice's by disregarding the memory register as well as
the microscopic information, which corresponds to the mapping
$\mathcal{F}_{KNM\to K}^{\mathcal{A}\to\mathcal{B}}(\cdot) = \tr_{MN}(\cdot)$.
Alice uses the description
$\Gamma^{\mathcal{A}}_{KNM} = \sum_{E,V,N} \proj{E,V,N}_K\otimes P^{E,V,N}_N
\otimes \Ident_M$ (see previous section).  Bob, on the other hand, describes the
gas using
$\Gamma^{\mathcal{B}}_{K} =
\mathcal{F}_{KNM\to{}K}^{\mathcal{A}\to\mathcal{B}}(\Gamma^{\mathcal{A}}_{KNM})
= d_M \sum \Omega(E,V,N)\,\proj{E,V,N}_K$, where $d_M$ is the dimension of the
system $M$.
Using the fact that
$\mathcal{F}^{\mathcal{A}\leftarrow\mathcal{B}\;\dagger}_{KNM\leftarrow{}K}(\cdot)
= (\cdot)\otimes\Ident_{NM}$, the Petz recovery map corresponding to
$\mathcal{F}^{\mathcal{A}\to\mathcal{B}}_{KNM\to{}K}$ is determined to be
  \begin{align}
    \mathcal{R}_{K\to KNM}^{\mathcal{B}\to\mathcal{A}}(\cdot)
    &=  \bigl(R_{K\to KN}\,\bigl[(\cdot)\otimes\Ident_N\bigr]
      \,R_{K\leftarrow KN}^\dagger\bigr) \otimes
      \frac{\Ident_M}{d_M}
      \ ,
  \end{align}
where we have defined the operator
\begin{align}
  R_{K\to KN} = \sum_{E,V,N}\proj{E,V,N}_K\otimes
  \frac{P^{E,V,N}_N}{\sqrt{\Omega(E,V,N)}}\ .
\end{align}
Importantly, the recovery map applied to any state of the form $\ket{E,V,N}_K$
gives
\begin{multline}
  \mathcal{R}_{K\to KNM}^{\mathcal{B}\to\mathcal{A}}(
  \proj{E,V,N}_K )
  \\
  = \proj{E,V,N}_K \otimes \tau^{E,V,N}_N \otimes \frac{\Ident_M}{d_M}\ ,
\end{multline}
i.e., Bob assigns a standard thermal state to all systems that he cannot otherwise
access.
From Alice's perspective (the demon's), the memory register $M$ starts in a pure
state $\ket0_M$, in order to store the future results from observations of the
gas.
On the other hand, Bob has no way to infer this state from his macroscopic
information.  Because of this, Alice can design processes that are perfectly
$\Gamma$-sub-preserving from her perspective but which can trick Bob into
thinking he is observing a violation of the second law (as described in
\autoref{main-fig:MaxwellDemon}).  Consider, for concreteness, the following procedure:
Alice performs a unitary process mapping the state
$\proj{E,V,N}_K\otimes\tau_N^{E,V,N}\otimes\proj0_M$ to
$\proj{E,V/2,N}_K\otimes\tau_N^{E,V/2,N}\otimes(d_M^{-1}\Ident_M)$, where we
assume that the system $M$ has just the right dimension to store all the entropy
resulting from mapping a state $\tau_N^{E,V,N}$ to the state $\tau_N^{E,V/2,N}$
of lower rank (we assume, for simplicity, that the rank of $\tau_N^{E,V/2,N}$
divides that of $\tau_N^{E,V,N}$, and thus
$\Omega(E,V,N) = d_M\,\Omega(E,V/2,N)$).  Alice's process is fully $\Gamma$
preserving because it is unitary and commutes with
$\Gamma^{\mathcal{A}}_{KNM}$.
However, from Bob's perspective, the gas changed its state from $\ket{E,V,N}_K$
to $\ket{E,V/2,N}_K$, in a blatant violation of his idea of the second law of
thermodynamics!
Of course, a clever Bob would be led to infer that there exists some system
($M$) that has interacted with the gas and absorbed the surplus
entropy.  The point is, however, that Bob can still very well apply his laws of
thermodynamics (in the form of the restriction imposed by
$\Gamma$-sub-preserving maps) as long as Alice does not ``actively mess with
him.''  In other words, any observer can consistently apply the laws of thermodynamics
(in the form of our framework) from their perspective, using the restriction of
$\Gamma$-sub-preserving maps for appropriately chosen $\Gamma$ operators as long
as this restriction indeed holds.  A $\Gamma$-sub-preserving restriction
inferred from coarse-graining a finer $\Gamma$-sub-preserving restriction fails
exactly when the finer-grained observer actively makes use of their privileged
microscopic access.

A further example illustrating the necessity of treating thermodynamics as an
observer-dependent framework, where our framework could be applied, is
provided by Jaynes' beautiful treatment of the Gibbs
paradox~\cite{Jaynes1992_Gibbs}.

\section{Discussion}
\label{main-sec:discussion}

One might think that thermodynamics, as a physical theory in essence, would
require physical concepts, such as energy or number of particles, to be built
into the theory, as is done in usual textbooks.  Our results align with the
opposite view, where thermodynamics is a generic framework itself, agnostic of
any physical quantities such as ``energy,'' which can be applied to different
physical situations, in the same spirit as previously proposed
approaches~\cite{BookGiles1964_thermodynamics,%
  BookGyftopoulos2005_thermodynamics,%
  Lieb1999_secondlaw,%
  Lieb2004_guide_secondlaw,%
  Weilenmann2016PRL_axiomatic,%
  Zanchini2011_thermodynamics}.  The physical properties of the system, such as
energy, temperature, or number of particles, are accounted for in our framework
only through the abstract $\Gamma$ operator.

Our results provide an additional step in understanding the core ingredients of
thermodynamics and hence the extent of its universality.  Our approach reveals
the following picture: Given any situation where the system obeys some physical
laws that imply the restriction that the evolution must preserve (or
sub-preserve) a certain operator $\Gamma$, then purity may be invested to lift
the restriction on any process, as quantified by the coherent relative entropy;
depending on how $\Gamma$ is defined, one may express this abstract resource
in terms of a physical resource such as mechanical work.  Furthermore, if the
states of interest of our system form a class of states that happen to be
reversibly interconvertible, the macroscopic laws of thermodynamics emerge,
along with the relevant thermodynamic potential.
In a coarse-grained picture, the thermodynamic laws apply as long as our
thermodynamic coarse-graining criterion is fulfilled, namely, if the
fine-grained state is not more informative than what can be recovered from the
coarse-grained information.

The notion of macroscopic limit considered here is more general than assuming
that the state of the system is a product state $\rho^{\otimes n}$, where each
particle or subsystem is independent and identically distributed (i.i.d.).
While typical thermodynamic systems are indeed close to an i.i.d.\@ state (for
instance, the Gibbs state of many noninteracting particles is an i.i.d.\@
state), we only rely on a notion of ``thermodynamic states,'' defined by their
ability to be interconverted reversibly and with certainty.  Thermodynamic
states may include arbitrary interaction between the particles, or, in fact, may
even be defined on a small system of a few particles.  More precisely, our notion
of thermodynamic states coincides with our definition of battery states and
corresponds to a state of the form $P\Gamma P/\tr(P\Gamma)$ for a
projector $P$ that commutes with $\Gamma$.  These states can be reversibly
interconverted in our framework, and usual statistical mechanical states are
precisely of this form.  The thermodynamic states may be used as reference
charge states of a battery system, in the sense that they enable the same
processes.

The core of the framework is the $\Gamma$-sub-preserving restriction imposed on
the free operations.  The $\Gamma$ operator\ encodes all the relevant physics of
the system considered.  The restriction may be due to any physical reason---for
instance, by assuming that the evolution is modeled by thermal operations on the
microscopic level, or by otherwise justifying or assuming that the spontaneous
dynamics are thermalizing in an appropriate sense.  Furthermore,
$\Gamma$-sub-preservation may come about in any situation where one or several
conserved physical quantities are being exchanged with a corresponding
thermodynamic bath, in a natural generalization of thermal
operations~\cite{YungerHalpern2016PRE_beyond,%
  Guryanova2016NatComm_multiple,%
  YungerHalpern2016NatComm_NATSandNATO}.

Our framework is not limited to usual thermodynamics: By considering the
$\Gamma$ operator as an abstract entity, all considerations in our framework are
of a purely quantum information theoretic nature and make no explicit reference
to any physical quantity.
For instance, one can consider purity as a resource and
impose that operations sub-preserve the identity operator; our framework applies
by taking $\Gamma=\Ident$; in this way, one can recover the max-entropy as the
number of pure qubits required to perform data compression of a given state.
We might further expect connections with single-shot notions of conditional
mutual information~\cite{Fawzi2015CMP_Markov,%
  Berta2015JMP_mutual,Berta2016TIT_bounds,%
  Majenz2017PRL_catdecoupling}, which in the i.i.d.\@ case
can also be expressed as a difference of quantum relative entropies.  Our
approach is also promising for calculating remainder terms in recovery of
quantum information~\cite{Wilde2015RSPA_recoverability,%
  Datta2015JPA_Markov,Wehner2015arXiv_reversibility,%
  Buscemi2016PRA_reversibility,Sutter2016IEEETIT_pinched,%
  Sutter2016PRSA_markov}.
Furthermore, being a $\Gamma$-sub-preserving map is a semidefinite constraint,
and thus optimization problems over free operations may often be formulated as
semidefinite programs, which exhibit a rich structure and can be solved
efficiently.

Although the goal of our paper is to derive a fundamental limitation on
operations in quantum thermodynamics, one can also ask the question of whether
this limit can be achieved within a physically well-motivated set of operations.
Because our bound is given by an optimization over Gibbs-preserving maps, it is
clear that there is one such map that will attain that bound (or get
arbitrarily close).  However, it is not clear under which conditions our bound
can be approximately attained in a more practical or realistic regime such as
thermal operations (possibly combined with additional resources), as is the case
for a system described by a fully degenerate Hamiltonian~\cite{Faist2015NatComm}
or for classical systems~\cite{Faist2015NJP_Gibbs}.

The question of achievability is related to coherence in the context of
thermodynamic transformations, an issue of significant recent
interest~\cite{Korzekwa2016NJP_extraction,%
  Aberg2014PRL_catalytic,%
  Cwiklinski2015PRL_limitations,%
  NellyNg2015NJP_limits,%
  Lostaglio2015NC_beyond,%
  Lostaglio2015PRX_coherence}. %
In particular, thermal operations do not allow the generation of a coherent
superposition of energy levels, while this is allowed to some extent by
Gibbs-preserving maps, which are hence not necessarily covariant under time
translation~\cite{Faist2015NJP_Gibbs}.
Our approach suggests a possible interpretation for why this is the case: With
$\Gamma$-sub-preserving operations, one requires no assumption that the system
in question is isolated---for instance, $\Gamma$ could be the reduced state on
one party of a joint Gibbs state of a strongly interacting bipartite system.
Indeed, the example in ref.~\cite{Faist2015NJP_Gibbs} can be explained in this
way~\cite[Section~4.4.4]{PhDPhF2016}.
Still, the question of whether Gibbs-preserving maps may be implemented
approximately using a more practical framework, such as thermal operations
(perhaps under certain conditions), remains an open question.  We note, though,
that the coherence resources required in order to implement a process can be
determined using the techniques of ref.~\cite{Cirstioiu2017arXiv_gauge}.  These
general tools might thus clarify the precise coherence requirements of
implementing Gibbs-preserving maps with covariant operations.
In a similar vein, one could study the effect of catalysis in our
framework~\cite{Vidal2002PRL_Catalysis,%
  NellyNg2015NJP_limits,Lostaglio2015PRL_absence}, presumably in the context of
state transitions rather than logical processes.
A closer study of this type of situation is expected to reveal connections with
smoothed, generalized, free energies~\cite{Meer2017arXiv_smoothed} and the notion
of approximate majorization~\cite{Horodecki2017arXiv_approximate}.
Furthermore, we expect tight connections with recent results that provide a
complete set of entropic conditions for fully quantum state transformations
under either general Gibbs-preserving maps or time-covariant Gibbs-preserving
maps~\cite{Gour2017arXiv_entropic}.  As a condition on state transformations, it
automatically provides an upper bound to the amount of work one can extract when
implementing a specific process, which, in particular, implements a specific
state transformation.  Furthermore, the way the covariance constraint is
enforced in ref.~\cite{Gour2017arXiv_entropic} provides a promising approach for
including the covariance constraint in our framework as well and tightening our
fundamental bound in the context of operations which are restricted to be time
covariant.  Finally, the conditions of ref.~\cite{Gour2017arXiv_entropic} may be
used to prove the achievability of state transformations with a covariant
mapping; one could expect a suitable generalization of both frameworks to
simultaneously handle possible symmetry constraints and logical processes as
well as state transformations, and a tolerance against unlikely events using
$\epsilon$-approximations.

Finally, our framework can describe a system at any degree of coarse-graining,
including intermediate scales between the microscopic and macroscopic regimes.
We can consider, for instance, a small-scale classical memory element that stores
information using many electrons or many spins (such as everyday hard drives):
The electrons may need to be treated thermodynamically, but not the system as a
whole, since we have control over the information-bearing degrees of freedom on
a relatively small scale.
Other such examples include Maxwell-demon-type scenarios, which our framework
allows to treat systematically.
Our framework is also suitable for describing agents who possess a quantum
memory containing quantum side information about the system in question.
In other words, we provide a self-contained framework of thermodynamics, which
allows us to make the dependence on the observer explicit, underscoring the idea
that thermodynamics is a theory that is relative to the
observer~\cite{Jaynes1992_Gibbs}.

\begin{acknowledgments}
We are grateful to
Mario Berta,
Fernando Brand\~ao,
Fr\'ed\'eric Dupuis,
Lea Kr\"amer Gabriel,
David Jennings,
and
Jonathan Oppenheim
for discussions.
We acknowledge contributions from the Swiss National Science Foundation (SNSF)
via the NCCR QSIT as well as Project No.~{200020\_165843}.  PhF acknowledges
support from the SNSF through the Early PostDoc.Mobility Fellowship
No.~{P2EZP2\_165239} hosted by the Institute for Quantum Information and Matter
(IQIM) at Caltech, as well as from the National Science Foundation.
\end{acknowledgments}

\onecolumngrid
\vspace{2cm}
\section*{APPENDICES}
\vspace{1cm}
\twocolumngrid

\appendix

%
%

The appendices are structured as
follows. \autoref{appx:sec:preliminaries} offers some preliminary definitions
and notation conventions. In \autoref{appx:sec:properties-of-the-framework} we
prove the properties of our framework outlined in the main text, namely that any
trace-nonincreasing, $\Gamma$-sub-preserving map can be dilated to a
trace-preserving, $\Gamma$-preserving map, as well as the equivalence of a class
of battery models. \autoref{appx:sec:coh-rel-entr-def-and-props} is dedicated to
the definition and properties of the coherent relative
entropy. \autoref{appx:sec:robustness-battery-states} discusses the robustness
of battery states to small perturbations.  Finally,
\autoref{appx:sec:technical-utilities} provides a selection of miscellaneous
technical tools which are used in the rest of the paper.

\section{Technical Preliminaries}
\appendixlabel{appx:sec:preliminaries}

Let us first fix some notation.  The state space of a quantum system $S$ is a
Hilbert space $\Hs_S$ (in this work, we deal exclusively with finite-dimensional
spaces), the dimension of which we denote by $\abs{S}$.  A quantum state
$\rho_S$ of $S$ is a positive semidefinite operator of unit trace acting on
$\Hs_S$.  A subnormalized quantum state $\rho_S$ is defined as satisfying
$\tr\rho_S\leqslant1$.  In this work, quantum states are normalized to unit
trace unless otherwise stated.  We use the notation $A\geqslant 0$ to indicate
that an operator $A$ is positive semidefinite, and $A\geqslant B$ to indicate
that $A-B\geqslant 0$.  For any positive semidefinite operator $A_S$ acting on
$\Hs_S$ corresponding to a system $S$, we denote by $\Pi^{A_S}_S$ the projector
onto the support of $A_S$.  Furthermore, all projectors considered in this work
are Hermitian.  For each system $S$ with Hilbert space $\Hs_S$, we fix a basis
which we denote by $`{\ket k_S}$.  Between any two systems $A$ and $B$ of same
dimension (which we denote by $\Hs_A\simeq\Hs_B$ or $A\simeq B$), we may define
a reference (not normalized) entangled ket
$\ket\Phi_{A:B} := \sum_k \ket k_A\otimes\ket k_B$, as well as the partial
transpose operation
$t_{A\to B}`*(\cdot) = \tr_A`*[\Phi_{A:B}\,(\cdot)] = \sum_{kk'}
\matrixel{k}\cdot{k'}_A\,\ketbra{k'}{k}_B$ with
$\Phi_{A:B} = \proj{\Phi}_{A:B}$.  Furthermore, for any operator
$\Xi_A\geqslant 0$, a ket $\ket{\Xi}_{A:B}$ is a purification of $\Xi_A$ if and
only if there exists a ket $\ket{\Phi^\Xi}_{A:B}$ of the form
$\ket{\Phi^\Xi}_{A:B} = \sum_j \ket{\chi_j}_A\ket{\chi_j}_B$ with orthonormal
sets $`{\ket{\chi_j}_A},`{\ket{\chi_j}_B}$ such that
$\ket{\Xi}_{A:B} = \Xi_A^{1/2}\ket{\Phi^\Xi}_{A:B} =
\Xi_B^{1/2}\ket{\Phi^\Xi}_{A:B}$ with $\Xi_A = \tr_B \proj{\Xi}_{A:B}$ and
$\Xi_B = \tr_A \proj{\Xi}_{A:B}$ (Schmidt decomposition); the ket
$\ket{\Xi}_{A:B}$ is normalized if and only if $\tr\Xi_A=1$.

Throughout this paper, `$\log$' denotes the logarithm in base~2.

\subsection{Logical process and process matrix}

We denote by a \emph{logical process} a full description of a logical mapping of input
states to output states:

\begin{thmheading}{Logical process}
  A \emph{logical process} $\mathcal{E}_{X\to X'}$ is a completely positive,
  trace-preserving map, mapping Hermitian operators on $\Hs_X$ to Hermitian
  operators on $\Hs_{X'}$.
\end{thmheading}

A logical process along with an input state may be characterized by their
\emph{process matrix}, defined as the Choi-Jamio\l{}kowski map of the completely
positive map, weighted by the input state.

\begin{thmheading}{Process matrix}
  \label{defn:process-matrix}
  Let $\mathcal{E}_{X\to X'}$ be a logical process, and let $\sigma_X$ be a
  quantum state. Let $R_X$ be a system described by a Hilbert space
  $\Hs_{R_X}\simeq\Hs_X$, and let
  $\ket\sigma_{XR_X} = \sigma_X^{1/2}\,\ket\Phi_{X:R_X}$ be a purification of
  $\sigma_X$.  Then the \emph{process matrix} corresponding to
  $\mathcal{E}_{X\to X'}$ and $\sigma_X$ is defined as
  $ \rho_{X'R_X} = \mathcal{E}_{X\to X'}`\big(\proj\sigma_{XR_X}) $, where the
  identity process is understood on $R_X$.  The process matrix is itself a
  normalized quantum state.  The (unnormalized) \emph{Choi matrix} of
  $\mathcal{E}_{X\to X'}$ is
  $E_{X'R_X} = \mathcal{E}_{X\to X'}`*(\Phi_{X:R_X})$, and satisfies
  $\tr_{X'}(E_{X'R_X}) = \Ident_{R_X}$.
\end{thmheading}

The reduced states $\sigma_X$ and $\sigma_{R_X}$ of $\ket\sigma_{X:R_X}$ on
$R_X$ and $X$, respectively, are related by a partial transpose operation:
$\sigma_{R_X} = \tr_X(\sigma_{XR_X}) = t_{X\to R_X}`*(\sigma_X)$.  Furthermore,
we have the properties
$\rho_{X'R_X} = \sigma_{R_X}^{1/2}\,E_{X'R_X}\,\sigma_{R_X}^{1/2}$ and
$\rho_{R_X}=\tr_{X'}(\rho_{X'R_X})=\sigma_{R_X}$.

The process matrix in return fully determines the channel
$\mathcal{E}_{X\to X'}$ on the support of $\sigma_X$, allowing for a full
characterization of the input state as well as the logical process on the
support of the input.

\subsection{Distance measures on states}

For two quantum states $\rho,\sigma$, the trace distance is given by
$D(\rho,\sigma) = \frac12\norm{\rho - \sigma}_1$, and their fidelity is defined
as $F`*(\rho,\sigma) = \tr`\big[`(\rho^{1/2}\,\sigma\,\rho^{1/2})^{1/2}]$.  From
the fidelity one can define the \emph{purified distance}\footnote{The purified
  distance is also called \emph{Bures distance} (up to a factor of
  $2$)~\cite{BookBengtssonZyczkowski2006_Geometry} and coincides to second order
  with the quantum \emph{angle}~\cite{BookNielsenChuang2000}.} as
$P(\rho,\sigma) = \sqrt{1 -
  F^2`(\rho,\sigma)}$~\cite{Tomamichel2010IEEE_Duality,PhDTomamichel2012,%
  BookTomamichel2016_Finite}.

It will also prove convenient to work with subnormalized quantum states.
Following Refs.~\cite{Tomamichel2010IEEE_Duality,PhDTomamichel2012,%
  BookTomamichel2016_Finite}, for any two subnormalized states $\rho,\sigma$, we
define the \emph{(generalized) trace distance}
$D(\rho,\sigma) = \frac12\norm{\rho-\sigma}_1 + \frac12\abs{\tr\rho-\tr\sigma}$,
the \emph{(generalized) fidelity}
$F(\rho,\sigma) = \tr`\big[`(\rho^{1/2}\,\sigma\,\rho^{1/2})^{1/2}] +
\sqrt{(1-\tr\rho)(1-\tr\sigma)}$ and the \emph{(generalized) purified distance}
$P(\rho,\sigma) = \sqrt{1 - F^2`(\rho,\sigma)}$.
For any two subnormalized states $\rho,\sigma$, we have the useful relation
$D`(\rho,\sigma)\leqslant P`(\rho,\sigma)\leqslant
\sqrt{2\,D`(\rho,\sigma)}$.

\subsection{Semidefinite programming}

Semidefinite programming is a useful toolbox which brings a rich structure to a
certain class of optimization problems.  We follow the notation of
Refs.~\cite{Watrous2009_sdps,Watrous2011LectureNotesQIT}, where proofs to the
statements given here may also be found.

Let $A$ and $B$ be {Hermitian matrices}, let $\Phi\,`(\cdot)$ be a
Hermiticity-preserving superoperator, and let $X\geqslant 0$ be the optimization
variable, which is a Hermitian matrix constrained to the cone of positive
semidefinite matrices.  The prototypical {semidefinite program} is an
optimization problem of the following form:\footnote{Several equivalent
  prototypical forms for semidefinite programs exist in the literature.}
\begin{subequations}
  \label{eq:formalism-SDP-minform-primal-min}
  \begin{align}
    \mathrm{minimize:}\quad& \tr`(A\,X) \\
    \mathrm{subject~to:}\quad& \Phi\,`(X) \geqslant B\ .
  \end{align}
\end{subequations}
To any such problem corresponds another, related problem in terms of a different
variable $Y\geqslant 0$:
\begin{subequations}
  \label{eq:formalism-SDP-minform-dual-max}
  \begin{align}
    \mathrm{maximize:}\quad& \tr`(B\,Y) \\
    \mathrm{subject~to:}\quad& \Phi^\dagger`(Y) \leqslant A\ .
  \end{align}
\end{subequations}
The first problem is called the \emph{primal problem}, and the second,
\emph{dual problem}.  Either problem is deemed \emph{feasible} if there exists a
valid choice of the optimization variable satisfying the corresponding
constraint.  If there exists a $X\geqslant0$ such that $\Phi`(X)-B$ is positive
definite, the primal problem is said to be \emph{strictly feasible}; the dual is
\emph{strictly feasible} if there is a $Y\geqslant0$ such that
$A-\Phi^\dagger`(Y)$ is positive definite.
For these two problems, we define their optimal attained values
\begin{subequations}
  \begin{align}
    \alpha &= \inf`\big{ \tr`(A\,X) : \Phi\,`(X) \geqslant B, X\geqslant 0 }\ ; \\
    \beta &= \sup`\big{ \tr`(B\,Y) : \Phi^\dagger`(Y) \leqslant A, Y\geqslant 0 }\ ,
  \end{align}
\end{subequations}
with the convention that $\alpha=-\infty$ if the primal problem is not feasible
and $\beta=+\infty$ if the dual problem is not feasible.

For any semidefinite program, we have $\alpha\geqslant\beta$, a property called
\emph{weak duality}.  This convenient relation allows us to immediately bound
the optimal attained value of one of the two problems by picking any valid
candidate in the other.

For some pairs of problems, we may have $\alpha=\beta$.  In those cases we speak
of \emph{strong duality}.  This is often the case in practice.  A useful result
here is Slater's theorem, providing sufficient conditions for strong
duality~\cite[Theorem~2.2]{Watrous2009_sdps}.
\begin{theorem}[Slater's conditions for strong duality]
  \label{thm:Slater}\noproofref
  Consider any semidefinite program written in the
  form~\eqref{eq:formalism-SDP-minform-primal-min}, and let its dual problem be
  given by~\eqref{eq:formalism-SDP-minform-dual-max}.  Then:
  \begin{enumerate}[label=(\roman*)]
  \item if the primal problem is feasible and the dual is strictly feasible,
    then strong duality holds and there exists a valid choice $X$ for the primal
    problem with $\tr\,`(A\,X)=\alpha$;
  \item if the dual problem is feasible and the primal is strictly feasible,
    then strong duality holds and there exists a valid choice $Y$ for the dual
    problem with $\tr\,`(B\,Y)=\beta$.
  \end{enumerate}
\end{theorem}

We note that strong duality in itself doesn't necessarily imply the existence of
an optimal choice of variables attaining the infimum or supremum.  The existence
of optimal primal or dual choices may be explicitly stated by Slater's
conditions, or may be deduced by an auxiliary argument such as if the
constraints force the optimization region to be compact.

\section{Properties of our framework}
\appendixlabel{appx:sec:properties-of-the-framework}

\subsection{Dilation of $\Gamma$-sub-preserving maps
  to $\Gamma$-preserving maps}

For two systems $X$, $Y$, and corresponding operators
$\Gamma_X,\Gamma_Y\geqslant 0$, We say that a completely positive map
$\Phi_{X\to Y}$ is \emph{$\Gamma$-sub-preserving} if it satisfies
$\Phi`(\Gamma_X)\leqslant\Gamma_Y$.  Similarly, $\Phi_{X\to Y}$ is
\emph{$\Gamma$-preserving} if it satisfies $\Phi`(\Gamma_X)=\Gamma_Y$.

From a technical point of view, trace-preserving $\Gamma$-preserving maps don't
handle nicely systems of varying sizes or with different $\Gamma$ operators.
For example, if $X$ and $Y$ are systems with $\tr\Gamma_X\neq\tr\Gamma_Y$, there
may clearly be no $\Gamma$-preserving map from $X$ to $Y$ which is also trace
preserving.
It turns out that, by focusing on trace-nonincreasing $\Gamma$-sub-preserving
maps instead, we may circumvent the issue in a physically justified way: A
trace-nonincreasing $\Gamma$-sub-preserving map can always be seen as a
restriction of a $\Gamma$-preserving map on a larger system.  Furthermore, the
ancillas we have to include in this dilation are prepared in, or finish up in,
eigenstates of the respective $\Gamma$ operators.

\begin{proposition}[Dilation of $\Gamma$-sub-preserving maps]
  \index{dilation of $\Gamma$-sub-preserving to $\Gamma$-preserving maps}
  \label{prop:dilation-of-Gsp-to-Gp}
  Let $K$ and $L$ be quantum systems with corresponding $\Gamma_K$ and
  $\Gamma_L$.  Let $\tilde\Phi_{K\to L}$ be a trace-nonincreasing,
  $\Gamma$-sub-preserving map. Choose two arbitrary eigenvectors
  $\ket{\mathrm k}_K$ and $\ket{\mathrm l}_L$ of $\Gamma_K$ and $\Gamma_L$,
  respectively.  Then there exists a qubit system $\Hs_Q$ with corresponding
  $\Gamma_Q$ diagonal in a basis composed of two orthogonal states
  $`{\ket{\mathrm{i}}_Q,\ket{\mathrm{f}}_Q}$, such that there exists a
  trace-preserving, $\Gamma$-preserving map $\Phi_{KLQ\to KLQ}$ satisfying
  \begin{align}
    \label{eq:dilation-of-Gsp-to-Gp-condition-recover-correct-mapping}
    \tilde\Phi_{K\to L}\left(\cdot\right)
    = \matrixel{\mathrm{k}\,\mathrm{f}}{\;
    \Phi_{KLQ\to KLQ}`\big(\left(\cdot\right) \otimes
    \proj{\mathrm{l}\,\mathrm{i}}_{LQ})
    \;
    }{\mathrm{k}\,\mathrm{f}}_{KQ}\ .
  \end{align}
  Here, the joint $\Gamma$ operator on $K,L,Q$ is
  $\Gamma_{KLQ}=\Gamma_K\otimes\Gamma_L\otimes\Gamma_Q$.  Furthermore, the
  corresponding eigenvalues satisfy
  \begin{align}
    \label{eq:dilation-of-Gsp-to-Gp-condition-on-Gamma-Q}
    \matrixel{\mathrm l}{\Gamma_L}{\mathrm l}_L
    \matrixel{\mathrm i}{\Gamma_Q}{\mathrm i}_Q
    = \matrixel{\mathrm k}{\Gamma_K}{\mathrm k}_K
    \matrixel{\mathrm f}{\Gamma_Q}{\mathrm f}_Q\ .
  \end{align}
\end{proposition}

This means that for any trace-nonincreasing, $\Gamma$-sub-preserving map
$\tilde\Phi_{K\to L}$, we may find a larger system and a trace-preserving,
$\Gamma$-preserving map $\Phi_{KLQ}$ such that $\tilde\Phi_{K\to L}$ is seen as
the restriction of $\Phi_{KLQ}$ to the case where the input is fixed to
$\ket{\mathrm l\,\mathrm i}_{LQ}$ on $LQ$, where we only consider the subspace
of the output in the support of $\ket{\mathrm k\,\mathrm f}_{KQ}$ on $KQ$.

If the operators $\Gamma_{K},\Gamma_L,\Gamma_Q$ come from Hamiltonians
$H_K,H_L,H_Q$ as $\Gamma_i = \ee^{-\beta H_i}$ for a fixed inverse temperature
$\beta$, then the ancillas are prepared and left in pure energy eigenstates,
specifically $\ket{\mathrm l\,\mathrm i}_{LQ}$ for the input and
$\ket{\mathrm k\,\mathrm f}_{KQ}$ for the output.  Furthermore
condition~\eqref{eq:dilation-of-Gsp-to-Gp-condition-on-Gamma-Q} ensures that the
total energy of the ancillas remains the same:
\begin{align}
  \dmatrixel{\mathrm l}{H_L}_L + \dmatrixel{\mathrm i}{H_Q}_Q
  = \dmatrixel{\mathrm k}{H_K}_K + \dmatrixel{\mathrm f}{H_Q}_Q\ .
\end{align}

Note that the apparent post-selection
in~\eqref{eq:dilation-of-Gsp-to-Gp-condition-recover-correct-mapping} is simply
a statement about the output of $\Phi$.  This is made clear by the following
corollary.  For instance, if $\tilde{\Phi}$ is trace-preserving on a certain
subspace, then as long as the input state is in that subspace, no post-selection
occurs in effect because the output state on $KQ$ is already exactly
$\ket{\mathrm{k\,f}}_{KQ}$, i.e., if we were to project the output onto that
state the projection would succeed with certainty. More generally, we show that
performing the dilated mapping with the correct input states on the ancillary
systems and without any post-selection at all, yields a process matrix which is
just as close to the ideal process matrix as the one which would have been
achieved with the original trace-decreasing map.

\begin{corollary}
  \label{cor:dilation-of-Gsp-to-Gp-cases}
  Consider the setting of \autoref{prop:dilation-of-Gsp-to-Gp}.  Then all the
  following statements hold.
  \begin{enumerate}[label=(\alph*)]
  \item Let $P$ be any projector on $\Hs_K$ and assume that $\tilde{\Phi}$ is
    trace-preserving on the support of $P$, i.e., for any state $\tau$ supported
    on $P$, it holds that $\tr(\tilde\Phi(\tau))=1$.  Then the mapping
    $\Phi_{KLQ}$ given by \autoref{prop:dilation-of-Gsp-to-Gp} satisfies
    \begin{align}
      \Phi_{KLQ\to KLQ}(\tau\otimes\proj{\mathrm{l\,i}}_{LQ}) =
      \tilde\Phi_{K\to L}(\tau)\otimes\proj{\mathrm{k\,f}}_{KQ}\ ,
    \end{align}
    for any quantum state $\tau$ supported on $P$.
  \item Let $\sigma_{KR}$ be any pure state between $K$ and a reference system
    $R$.  Assume that $\tilde{\Phi}$ satisfies
    $\tr(\tilde{\Phi}(\sigma_{KR}))=1$.  Then
    \begin{multline}
      \Phi_{KLQ\to KLQ}(\sigma_{KR}\otimes\proj{\mathrm{l\,i}}_{LQ})
      \\
      = \tilde\Phi_{K\to L}(\sigma_{KR})\otimes\proj{\mathrm{k\,f}}_{KQ}\ .
    \end{multline}
  \item Let $\sigma_{KR}$ be any pure state between $K$ and a reference system
    $R$, and let $\rho_{LR}$ be any quantum state.  Then the mapping
    $\Phi_{KLQ\to KLQ}$ provided by \autoref{prop:dilation-of-Gsp-to-Gp}
    satisfies
    \begin{multline}
      P`*(\Phi_{KLQ\to KLQ}(\sigma_{KR}\otimes \proj{\mathrm{l\,i}}_Q) ,
      \rho_{LR}\otimes\proj{\mathrm{k\,f}}_{LQ})
      \\
      = P`*(\tilde{\Phi}_{K\to L}(\sigma_{KR}), \rho_{LR})\ .
    \end{multline}
  \end{enumerate}
\end{corollary}

\subsection{Equivalence of battery models}

Consider a logical process $\mathcal{E}_{X\to X'}$ which is not itself a free
operation (i.e., $\mathcal{E}_{X\to X'}`(\Gamma_X) \not\leqslant \Gamma_{X'}$).
It turns out that it is possible to implement this process by investing a
certain amount of resources by means of an explicit battery system.

One example of such a battery system is the \emph{information battery}.  The
{information battery} is a quantum system $A$ of dimension which we denote by
$\abs A$, and for which $\Gamma_A=\Ident_A$.  We require the battery to
initially be prepared in a state $2^{-\lambda_1}\Ident_{2^{\lambda_1}}$ and to
finish in a state $2^{-\lambda_2}\Ident_{2^{\lambda_2}}$ at the end, where both
states are simply a state with a flat spectrum of rank $2^{\lambda_1}$ or
$2^{\lambda_2}$, and where we require that $\lambda_1,\lambda_2\geqslant0$ and
that $2^{\lambda_1},2^{\lambda_2}$ are integers.  If $\lambda_1,\lambda_2$ are
themselves integers, this corresponds exactly to having $\lambda_1$ or
$\lambda_2$ qubits in a fully mixed state and the remaining qubits in a pure
state.

It is known that this model is equivalent to several other battery models known
in the literature~\cite{Brandao2015PNAS_secondlaws}, notably the work bit (or
\qq{wit})~\cite{Horodecki2013_ThermoMaj,Brandao2015PNAS_secondlaws}, or a
\qq{weight}
system~\cite{Skrzypczyk2014NComm_individual,Alhambra2016PRX_equality}.  Here, we
point out that these models are in fact different instances of a more general
description, making their equivalence manifest.

The most general system we have shown to be usable as a battery system is simply
any system $W$ with a arbitrary $\Gamma_W$ operator, which is restricted to be
in states of the form $\sigma = (P\Gamma_W P)/\tr P\Gamma_W$, where $P$ is a
projector which commutes with $\Gamma_W$.  The \qq{value} or \qq{uselessness} of
this state is given by the quantity $\log\tr`(P\Gamma_W)$. The wit, the weight, as
well as the information battery are all special cases of this general model.

The following proposition gives a necessary and sufficient condition as to when
it is possible to overcome the $\Gamma$-sub-preservation restriction by
exploiting a particular charge state change of the battery, and shows how the
different battery systems are equivalent.  This proves
Propositions~\ref*{main-prop:main-equiv-battery-models}
and~\ref*{main-prop:main-equiv-battery-models-rest} of the main text.

\begin{proposition}
  \label{prop:equiv-battery-models}
  Let $\mathcal{T}_{X\to X'}$ be a completely positive, trace-nonincreasing map.  Let
  $y\in\mathbb{R}$.  Then, the following are equivalent:
  \begin{enumerate}[label=(\roman*)]
  \item\label{item:prop-equiv-battery-models-equivstatement-criterionMapEnorm}%
    The map $\mathcal{T}_{X\to X'}$ satisfies
    \begin{align}
      \mathcal{T}_{X\to X'}`*(\Gamma_X) \leqslant 2^{-y}\,\Gamma_{X'}\ ;
    \end{align}
  \item\label{item:prop-equiv-battery-models-equivstatement-infbattery}%
    For any $\lambda_1,\lambda_2\geqslant 0$ such that
    $2^{\lambda_1},2^{\lambda_2}$ are integers and
    $\lambda_1-\lambda_2 \leqslant y$, there exists a large enough system $A$
    with $\Gamma_A=\Ident_A$ as well as a trace-nonincreasing,
    $\Gamma$-sub-preserving map $\Phi_{XA\to X'A}$ satisfying for all
    $\omega_X$,
    \begin{multline}
      \Phi_{XA\to X'A}`*(
      \omega_X \otimes `\big(2^{-\lambda_1}\Ident_{2^{\lambda_1}})
      )
      \\
      = \mathcal{T}_{X\to X'}`*(\omega_X)
      \otimes
      `\big(2^{-\lambda_2}\Ident_{2^{\lambda_2}}) \ ;
    \end{multline}
  \item\label{item:prop-equiv-battery-models-equivstatement-wit}%
    For a two-level system $Q$ with two orthonormal states
    $\ket{\mathrm 1}_Q,\ket{\mathrm 2}_Q$, and with
    $\Gamma_Q=g_{1}\proj{1}_Q + g_{2}\proj{2}_Q$ chosen such that $g_{2}/g_{1}\geqslant 2^{-y}$, there
    exists a trace-nonincreasing, $\Gamma$-sub-preserving map $\Phi'_{XQ\to X'Q}$
    satisfying for all $\omega_X$,
    \begin{align}
      \Phi'_{XQ\to X'Q}`*( \omega_X\otimes\proj{1}_Q )
      = \mathcal{T}_{X\to X'}`*(\omega_X)\otimes\proj 2_Q\ ;
    \end{align}
  \item\label{item:prop-equiv-battery-models-equivstatement-weight}%
    Let $\tilde Q$ be any system and choose two orthogonal states
    $\ket{\mathrm 1}_{\tilde Q},\ket{\mathrm 2}_{\tilde Q}$ which are eigenstates of
    $\Gamma_{\tilde Q}$ corresponding to respective eigenvalues $g_1,g_2$ which satisfy
    $g_2/g_1\geqslant 2^{-y}$.  Then there exists a trace-nonincreasing,
    $\Gamma$-sub-preserving map $\Phi'_{X\tilde Q\to X'\tilde Q}$ satisfying for all
    $\omega_X$,
    \begin{align}
      \Phi'_{X\tilde Q\to X'\tilde Q}`*( \omega_X\otimes\proj{1}_{\tilde Q} )
      = \mathcal{T}_{X\to X'}`*(\omega_X)\otimes\proj 2_{\tilde Q}\ ;
    \end{align}
  \item\label{item:prop-equiv-battery-models-equivstatement-projGamma}%
    Let $W_1,W_2$ be quantum systems with respective corresponding $\Gamma$
    operators $\Gamma_{W_1},\Gamma_{W_2}$, and let $P_{W_1}, P'_{W_2}$ be
    projectors satisfying $[ P_{W_1}, \Gamma_{W_1} ] = 0$ and
    $[ P'_{W_2}, \Gamma_{W_2} ] = 0$, such that
    \begin{align}
      \frac{\tr P'_{W_2}\Gamma_{W_2}}{\tr P_{W_1}\Gamma_{W_1}} \geqslant 2^{-y}\ .
    \end{align}
    Then there exists a $\Gamma$-sub-preserving, trace-nonincreasing map
    $\Phi''_{XW_1\to X'W_2}$ such that for all $\omega_X$,
    \begin{multline}
      \Phi''_{XW_1\to X'W_2}`*(
      \omega_X \otimes
      \frac{P_{W_1} \Gamma_{W_1} P_{W_1}}{\tr`(P_{W_1}\Gamma_{W_1})})
      \\[1ex]
      = \mathcal{T}_{X\to X'}`*(\omega_X)
      \otimes
      \frac{P'_{W_2} \Gamma_{W_2} P'_{W_2}}{\tr`(P'_{W_2}\Gamma_{W_2})}
      \ .
      \label{eq:prop-equiv-battery-models-projGamma-condition-correct-mapping}
    \end{multline}
  \end{enumerate}
\end{proposition}

\subsection{Proofs}

\begin{proof}[*prop:dilation-of-Gsp-to-Gp]
  By definition, $\tilde\Phi_{K\to L}$ satisfies both
  $\tilde\Phi_{K\to L}`(\Gamma_K) \leqslant \Gamma_L$ and
  $\tilde\Phi^\dagger_{K\leftarrow L}`(\Ident_L) \leqslant \Ident_K$. Hence, let
  $F_K,G_L\geqslant 0$ such that
  \begin{subequations}
    \begin{align}
      \tilde\Phi_{K\to L}\left(\Gamma_K\right) &= \Gamma_L - G_L\ ; \\
      \tilde\Phi_{K\leftarrow L}^\dagger\left(\Ident_L\right) &= \Ident_K - F_K \ .
    \end{align}
    Let $\Pi^\Gamma_L$ be the projector onto the support of $\Gamma_L$. We have
    $\Pi^\Gamma_L\leqslant\Ident_L$ and thus
    $\tilde\Phi_{K\leftarrow L}^\dagger\left(\Pi^\Gamma_L\right) \leqslant
    \tilde\Phi_{K\leftarrow L}^\dagger\left(\Ident_L\right) \leqslant \Ident_K$.
    So define $F_K' \geqslant 0$ such that
    \begin{align}
      \tilde\Phi_{K\leftarrow L}^\dagger\left(\Pi^\Gamma_L\right) = \Ident_K - F_K'\ .
    \end{align}
  \end{subequations}
  
  Let the system $Q$ be as in the claim, with $\Gamma_Q$ diagonal in the basis
  $`{\ket{\mathrm i}_Q,\ket{\mathrm f}_Q}$.  Define now the completely positive map
  \begin{align}
    \hspace*{1ex}&\hspace*{-1ex}\Phi_{KLQ\to KLQ}\left(\cdot\right) =\\
    &\hspace*{0.5ex}
      \tilde\Phi_{K\to L}`*(\matrixel{\mathrm l\,\mathrm i}{\,\cdot\,}{\mathrm l\,\mathrm i}_{LQ})
      \otimes \proj{\mathrm k\,\mathrm f}_{KQ}
      \nonumber\\
    &+ \Gamma_K^{1/2} \tilde\Phi_{K\leftarrow L}^\dagger`*(
      \bigl(\Gamma_L^{-1/2}\bra{\mathrm k\,\mathrm f}_{KQ}\bigr) \left(\cdot\right)
      \bigl(\Gamma_L^{-1/2}\ket{\mathrm k\,\mathrm f}_{KQ}\bigr)
      )\Gamma_K^{1/2}\otimes\proj{\mathrm l\,\mathrm i}_{LQ}
      \nonumber\\
    &+\Xi_{KL\to KL}`*(\matrixel{\mathrm i}{\,\cdot\,}{\mathrm i}_Q) \otimes \proj{\mathrm i}_Q
      \nonumber\\
    &+\Omega_{KL\to KL}`*(\matrixel{\mathrm f}{\,\cdot\,}{\mathrm f}_Q) \otimes \proj{\mathrm f}_Q
      \ ,
  \end{align}
  with some completely positive maps $\Xi_{KL\to KL}$ and $\Omega_{KL\to KL}$ yet to be
  determined.

  First, note that the
  property~\eqref{eq:dilation-of-Gsp-to-Gp-condition-recover-correct-mapping} is obvious for
  this $\Phi_{KLQ}$, simply because $\ket{\mathrm i}_Q$ and $\ket{\mathrm f}_Q$ are
  orthogonal. It remains to exhibit explicit $\Xi_{KL\to KL}$ and $\Omega_{KL\to KL}$ such
  that $\Phi_{KLQ}$ is trace-preserving and $\Gamma$-preserving.  Define as shorthands
  \begin{equation}
    \begin{aligned}
      g_\mathrm{k} &= \matrixel{\mathrm k}{\Gamma_K}{\mathrm k}_K\ ; & &\hspace{2em} &
      g_\mathrm{l} &= \matrixel{\mathrm l}{\Gamma_L}{\mathrm l}_L\ ;
      \\
      g_\mathrm{i} &= \matrixel{\mathrm i}{\Gamma_Q}{\mathrm i}_Q\ ; & & &
      g_\mathrm{f} &= \matrixel{\mathrm f}{\Gamma_Q}{\mathrm f}_Q\ .
    \end{aligned}
  \end{equation}
  Note that Condition~\eqref{eq:dilation-of-Gsp-to-Gp-condition-on-Gamma-Q} is then equivalent to
  \begin{align}
    \label{eq:dilation-of-Gsp-to-Gp-proof-condition-gk-gl-gi-gf}
    g_{\mathrm l}\cdot g_{\mathrm i} = g_{\mathrm k}\cdot g_{\mathrm f}\ ,
  \end{align}
  and that this is straightforwardly satisfied for an appropriate choice of $\Gamma_Q$
  (and hence of $g_{\mathrm i},g_{\mathrm f}$).

  At this point, we'll derive conditions that $\Xi_{KL\to KL}$ and $\Omega_{KL\to KL}$
  need to satisfy in order for $\Phi_{KLQ\to KLQ}$ to map $\Gamma_{KLQ}$ onto itself and
  to be trace-preserving.  Calculate
  \begin{align}
    \hspace*{3em}&\hspace{-3em}\Phi_{KLQ\to KLQ}\left(\Gamma_{KLQ}\right) \nonumber\\
    &= g_{\mathrm l} g_{\mathrm i}\,
    \tilde\Phi_{K\to L}\left(\Gamma_K\right)\otimes \proj{\mathrm k\,\mathrm f}_{KQ}
    \nonumber\\
    &\quad+ g_{\mathrm k} g_{\mathrm f}\,\Gamma_K^{1/2}
    \tilde\Phi_{K\leftarrow L}^\dagger\left(\Pi^\Gamma_L\right)\Gamma_K^{1/2}
    \otimes \proj{\mathrm l\,\mathrm i}_{LQ}
      \nonumber\\
    &\quad+ g_{\mathrm i} \Xi_{KL\to KL}\left(\Gamma_{KL}\right)\otimes\proj{\mathrm i}_Q
      + g_{\mathrm f} \Omega_{KL\to KL}\left(\Gamma_{KL}\right)\otimes\proj{\mathrm f}_Q
      \nonumber\\
    &= \proj{\mathrm f}_Q\otimes\left[
      g_{\mathrm l} g_{\mathrm i}\left(\Gamma_L - G_L\right)\otimes\proj{\mathrm k}_{K}
      + g_{\mathrm f}\Omega_{KL\to KL}\left(\Gamma_{KL}\right)
      \right]
      \nonumber\\
    &\quad+ \proj{\mathrm i}_Q\otimes\Bigl[
      g_{\mathrm k} g_{\mathrm f}\Gamma_K^{1/2}\left(\Ident_K - F_K'\right)\Gamma_K^{1/2}
      \otimes\proj{\mathrm l}_{KQ}
      \nonumber\\
    &\quad\hspace*{13em}
      + g_{\mathrm i}\Xi_{KL\to KL}\left(\Gamma_{KL}\right)
      \Bigr]\ .
  \end{align}
  We see that in order for this last expression to equal
  $\Gamma_{KLQ} = g_\mathrm{f}\proj{\mathrm f}_Q\otimes\Gamma_{KL} +
  g_\mathrm{i}\proj{\mathrm i}_Q\otimes\Gamma_{KL}$,
  we need that the terms in square brackets above obey
  \begin{subequations}
    \begin{gather}
      g_{\mathrm l} g_{\mathrm i}\left(\Gamma_L - G_L\right)\otimes\proj{\mathrm k}_{K}
      + g_{\mathrm f}\Omega_{KL\to KL}\left(\Gamma_{KL}\right) = g_{\mathrm f}\,\Gamma_{KL}\ ;
      \label{eq:dilation-of-Gsp-to-Gp-proof-calc-cond-Omega-fwd}\\
      g_{\mathrm k} g_{\mathrm f}\Gamma_K^{1/2}\left(\Ident_K - F_K'\right)\Gamma_K^{1/2}
      \otimes\proj{\mathrm l}_{KQ} + g_{\mathrm i}\Xi_{KL\to KL}\left(\Gamma_{KL}\right) =
      g_{\mathrm i}\,\Gamma_{KL}\ .
      \label{eq:dilation-of-Gsp-to-Gp-proof-calc-cond-Xi-fwd}
    \end{gather}
  \end{subequations}

  On the other hand, the adjoint map of $\Phi_{KLQ\to KLQ}$ is relatively straightforward
  to identify as
  \begin{align}
    \hspace*{1ex}
    &\hspace{-1ex}
       \Phi_{KLQ\leftarrow KLQ}^\dagger\left(\cdot\right) =
      \nonumber\\
    &\hspace*{0.5ex}\tilde\Phi_{K\leftarrow L}^\dagger\left(
      \matrixel{\mathrm k\,\mathrm f}{\,\cdot\,}{\mathrm k\,\mathrm f}_{KQ}
      \right)\otimes \proj{\mathrm l\,\mathrm i}_{LQ}
      \nonumber\\
    &+
      \Gamma_L^{-1/2}
      \tilde\Phi_{K\to L}\left(
      \bigl(\Gamma_K^{1/2}\bra{\mathrm l\,\mathrm i}_{LQ}\bigr)\left(\cdot\right)
      \bigl(\Gamma_K^{1/2}\ket{\mathrm l\,\mathrm i}_{LQ}\bigr)
      \right) \Gamma_L^{-1/2}
      \otimes\proj{\mathrm k\,\mathrm f}_{KQ}
      \nonumber\\
    &+
      \Xi_{KL\leftarrow KL}^\dagger\left(\matrixel{\mathrm i}{\,\cdot\,}{\mathrm i}_Q\right)
      \otimes\proj{\mathrm i}_Q
      \nonumber\\
    &+
      \Omega_{KL\leftarrow KL}^\dagger\left(\matrixel{\mathrm f}{\,\cdot\,}{\mathrm f}_Q\right)
      \otimes\proj{\mathrm f}_Q\ .
  \end{align}
  We may thus now derive the conditions on $\Xi_{KL\to KL}$ and $\Omega_{KL\to KL}$ for
  $\Phi_{KLQ\to KLQ}$ to be trace-preserving.  Specifically, we need to ensure that
  $\Phi_{KLQ\leftarrow KLQ}^\dagger\left(\Ident_{KLQ}\right) = \Ident_{KLQ}$.  A
  calculation gives us
  \begin{align}
    \hspace*{3ex}
    &\hspace{-3ex} \Phi_{KLQ\leftarrow KLQ}^\dagger\left(\Ident_{KLQ}\right)
      \nonumber\\
    &= \tilde\Phi_{K\leftarrow L}^\dagger\left(\Ident_L\right)\otimes
      \proj{\mathrm l\,\mathrm i}_{LQ}
      \nonumber\\
    &\quad+ \Gamma_L^{-1/2}\tilde\Phi_{K\to L}\left(\Gamma_K\right)\Gamma_L^{-1/2}
      \otimes\proj{\mathrm k\,\mathrm f}_{KQ}
      \nonumber\\
    &\quad+\Xi_{KL\leftarrow KL}^\dagger\left(\Ident_{KL}\right)\otimes\proj{\mathrm i}_Q
      \nonumber\\
    &\quad+ \Omega_{KL\leftarrow KL}^\dagger\left(\Ident_{KL}\right)\otimes\proj{\mathrm f}_Q\ .
      \nonumber\\
    &= \proj{\mathrm f}_Q\otimes\Bigl[
      \Gamma_L^{-1/2}\left(\Gamma_L-G_L\right)\Gamma_L^{-1/2}\otimes\proj{\mathrm k}_K
      \nonumber\\
    &\quad+ \Omega_{KL\leftarrow KL}^\dagger`*(\Ident_{KL})
      \Bigr] \nonumber\\
    &\quad+ \proj{\mathrm i}_Q\otimes\Bigl[
      `*(\Ident_K-F_K)\otimes\proj{\mathrm l}_L
      \nonumber\\
    &\quad+ \Xi_{KL\leftarrow KL}^\dagger`*(\Ident_{KL}) \Bigr]\ .
  \end{align}
  Thus, for $\Phi_{KLQ\to KLQ}$ to be trace-preserving we must have
  \begin{subequations}
    \begin{gather}
      \Gamma_L^{-1/2}\left(\Gamma_L-G_L\right)\Gamma_L^{-1/2}\otimes\proj{\mathrm k}_K +
      \Omega_{KL\leftarrow KL}^\dagger\left(\Ident_{KL}\right) = \Ident_{KL} \ ;
      \label{eq:dilation-of-Gsp-to-Gp-proof-calc-cond-Omega-rev}\\
      \left(\Ident_K-F_K\right)\otimes\proj{\mathrm l}_L
      + \Xi_{KL\leftarrow KL}^\dagger\left(\Ident_{KL}\right) = \Ident_{KL}\ .
      \label{eq:dilation-of-Gsp-to-Gp-proof-calc-cond-Xi-rev}
    \end{gather}
  \end{subequations}

  Let us now explicitly construct an $\Xi_{KL\to KL}$ which satisfies
  both~\eqref{eq:dilation-of-Gsp-to-Gp-proof-calc-cond-Xi-fwd}
  and~\eqref{eq:dilation-of-Gsp-to-Gp-proof-calc-cond-Xi-rev}. These conditions may be
  written as
  \begin{subequations}
    \begin{align}
      \Xi_{KL\to KL}\left(\Gamma_{KL}\right)
      &= \Gamma_{KL} - g_\mathrm{l}\,\Gamma_K^{1/2}\left(\Ident_K-F_K'\right)\Gamma_K^{1/2}
        \otimes\proj{\mathrm l}_L =: A_{KL}\ ;
        \label{eq:dilation-of-Gsp-to-Gp-proof-calc-condition-Xi-A_KL}
      \\
      \Xi_{KL\leftarrow KL}^\dagger\left(\Ident_{KL}\right)
      &= \Ident_{KL} - \left(\Ident_K-F_K\right)\otimes\proj{\mathrm l}_L =: B_{KL}
        \label{eq:dilation-of-Gsp-to-Gp-proof-calc-condition-Xi-B_KL}
    \end{align}
  \end{subequations}
  where we have
  used~\eqref{eq:dilation-of-Gsp-to-Gp-proof-condition-gk-gl-gi-gf} and defined
  two new operators $A_{KL}$ and $B_{KL}$. Observe now that since
  $g_\mathrm{l}\,\Gamma_K^{1/2}\left(\Ident_K-F_K'\right)\Gamma_K^{1/2}
  \otimes\proj{\mathrm l}_L \leqslant \Gamma_K\otimes\left(g_{\mathrm
      l}\,\proj{\mathrm l}_L\right) \leqslant \Gamma_{KL}$, we have that
  $A_{KL}\geqslant 0$. Similarly,
  $\left(\Ident_K-F_K\right)\otimes\proj{\mathrm l}_L \leqslant \Ident_{KL}$ and
  hence $B_{KL}\geqslant 0$.  Let $\xi_{KL}$ be a quantum state defined as
  follows: If $\tr A_{KL}\neq0$, then $\xi_{KL}=A_{KL}/\tr A_{KL}$; else
  $\xi_{KL}=\Ident_{KL}/\abs{KL}$.  Then define
  \begin{align}
    \Xi_{KL\to KL}\left(\cdot\right) = \tr\left(B_{KL}\,\left(\cdot\right)\right)\,\xi_{KL}\ .
  \end{align}
  We then have
  \begin{align}
    \Xi_{KL\leftarrow KL}^\dagger\left(\Ident_{KL}\right)
    = \tr\left(\xi_{KL}\,\Ident_{KL}\right)\,B_{KL} = B_{KL}\ ,
  \end{align}
  thus satisfying condition~\eqref{eq:dilation-of-Gsp-to-Gp-proof-calc-condition-Xi-B_KL}.
  On the other hand we have
  \begin{align}
    \Xi_{KL\to KL}\left(\Gamma_{KL}\right) = \tr\left[B_{KL}\,\Gamma_{KL}\right]\,\xi_{KL}\ ,
    \label{eq:dilation-of-Gsp-to-Gp-proof-calc-condition-Xi-A_KL-calc-1}
  \end{align}
  which we need to show equals $A_{KL}$ to satisfy
  condition~\eqref{eq:dilation-of-Gsp-to-Gp-proof-calc-condition-Xi-A_KL}.
  Consider first the case where $\tr A_{KL} = 0$ and hence $A_{KL}=0$.  Then
  $\Gamma_{KL} = g_\mathrm{l}
  \Gamma_K^{1/2}\left(\Ident_K-F_K'\right)\Gamma_K^{1/2}\otimes
  \proj{\mathrm{l}}_L$, and hence $\Gamma_L=g_\mathrm{l}\proj{\mathrm l}_L$ and
  $F_K'=0$.  Since
  $\tilde\Phi_{K\leftarrow L}^\dagger\bigl(\Pi^{\Gamma}_L\bigr) \leqslant
  \tilde\Phi_{K\leftarrow L}^\dagger\left(\Ident_L\right)$, we have
  $F_K\leqslant F_K'$ and thus $F_K=0$.  Then
  $B_{KL}=\Ident_K\otimes\left(\Ident_L-\proj{\mathrm l}_L\right)$.  Thus,
  $B_{KL}$ has no overlap with
  $\Gamma_{KL} = \Gamma_K\otimes\left(g_\mathrm{l}\proj{\mathrm{l}}_L\right)$
  and
  $\text{\eqref{eq:dilation-of-Gsp-to-Gp-proof-calc-condition-Xi-A_KL-calc-1}}=0=A_{KL}$
  as required.  Now consider the case where $\tr A_{KL}\neq 0$.  We have
  \begin{multline}
    \tr A_{KL} = \tr\Gamma_{KL} - g_\mathrm{l}\tr\left[\left(\Ident_K-F_K'\right)\Gamma_K\right]
    \\
    = \tr\Gamma_{KL}
    - g_\mathrm{l}\tr\left[\tilde\Phi_{K\leftarrow L}^\dagger\left(\Pi^\Gamma_L\right)\,\Gamma_K\right]
    \\
    = \tr\Gamma_{KL}
    - g_\mathrm{l}\tr\left[\Pi^\Gamma_L\,\tilde\Phi_{K\to L}\left(\Gamma_K\right)\right]\ .
    \label{eq:dilation-of-Gsp-to-Gp-proof-calc-constr-xi-tr-1}
  \end{multline}
  Now, because $\tilde\Phi_{K\to L}\left(\Gamma_K\right) \leqslant \Gamma_L$, the operator
  $\tilde\Phi_{K\to L}\left(\Gamma_K\right)$ must lie within the support of
  $\Gamma_L$. Thus the projector in the last term
  of~\eqref{eq:dilation-of-Gsp-to-Gp-proof-calc-constr-xi-tr-1} has no effect and can be
  replaced by an identity operator. We then have
  \begin{align}
    \text{\eqref{eq:dilation-of-Gsp-to-Gp-proof-calc-constr-xi-tr-1}}
    &= \tr\Gamma_{KL} - g_\mathrm{l}\tr\left[\Ident_{L}\tilde\Phi_{K\to L}\left(\Gamma_K\right)\right]
    \nonumber\\
    &= \tr\Gamma_{KL} - g_\mathrm{l}\tr\left[\tilde\Phi_{K\leftarrow L}^\dagger\left(\Ident_L\right)\Gamma_K\right]
    \nonumber\\
    &= \tr\Gamma_{KL} - g_\mathrm{l}\tr\left[\left(\Ident_K - F_K\right)\Gamma_K\right]
    \nonumber\\
    &= \tr\Gamma_{KL} - \tr\left[\left(\Ident_K - F_K\right)\otimes\proj{\mathrm{l}}_L\Gamma_{KL}\right]
      \nonumber\\
    &= \tr\left(B_{KL}\Gamma_{KL}\right)\ .
  \end{align}
  Since $\tr(B_{KL}\Gamma_{KL})=\tr(A_{KL})$, we have
  $\text{\eqref{eq:dilation-of-Gsp-to-Gp-proof-calc-condition-Xi-A_KL-calc-1}}=A_{KL}$
  as required.  We have thus constructed $\Xi_{KL\to KL}$ such that it satisfies
  conditions~\eqref{eq:dilation-of-Gsp-to-Gp-proof-calc-cond-Xi-fwd}
  and~\eqref{eq:dilation-of-Gsp-to-Gp-proof-calc-cond-Xi-rev}.

  Let's now proceed analogously for $\Omega_{KL\to KL}$. We can rewrite
  conditions~\eqref{eq:dilation-of-Gsp-to-Gp-proof-calc-cond-Omega-fwd}
  and~\eqref{eq:dilation-of-Gsp-to-Gp-proof-calc-cond-Omega-rev} as
  \begin{align}
    \Omega_{KL\to KL}\left(\Gamma_{KL}\right)
    &= \Gamma_{KL} - g_\mathrm{k}\proj{\mathrm k}_K\otimes\left(\Gamma_L-G_L\right)
      =: C_{KL}\ ;
      \label{eq:dilation-of-Gsp-to-Gp-proof-calc-condition-Omega-C_KL} \\
    \Omega_{KL\leftarrow KL}^\dagger\left(\Ident_{KL}\right)
    &= \Ident_{KL}
      - \proj{\mathrm k}_K\otimes \Gamma_L^{-1/2}\left(\Gamma_L-G_L\right)\Gamma_L^{-1/2}
      =: D_{KL}\ ,
      \label{eq:dilation-of-Gsp-to-Gp-proof-calc-condition-Omega-D_KL}
  \end{align}
  defining the operators $C_{KL}$ and $D_{KL}$. We have
  $g_\mathrm{k}\proj{\mathrm k}_K\otimes\left(\Gamma_L-G_L\right)\leqslant
  \Gamma_{KL}$ and thus $C_{KL}\geqslant 0$. Also
  $\Gamma_L^{-1/2}\left(\Gamma_L - G_L\right)\Gamma_L^{-1/2}\leqslant\Ident_L$
  and thus $D_{KL}\geqslant 0$. Proceeding as for $\Xi_{KL\to KL}$, let
  $\omega_{KL}$ be a quantum state defined as $\omega_{KL}=C_{KL}/\tr C_{KL}$ if
  $\tr C_{KL}\neq 0$ or $\omega_{KL}=\Ident_{KL}/\abs{KL}$ otherwise.  Define
  \begin{align}
    \Omega_{KL\to KL}\left(\cdot\right)
    = \tr\left(D_{KL}\,\left(\cdot\right)\right) \omega_{KL} \ .
  \end{align}
  Then
  \begin{align}
    \Omega_{KL\leftarrow KL}^\dagger\left(\Ident_{KL}\right)
    = \tr\left(\omega_{KL}\Ident_{KL}\right)\,D_{KL} = D_{KL}\ ,
  \end{align}
  which satisfies~\eqref{eq:dilation-of-Gsp-to-Gp-proof-calc-condition-Omega-D_KL}. On the
  other hand, we have
  \begin{align}
    \Omega_{KL\to KL}\left(\Gamma_{KL}\right)
    = \tr\left(D_{KL}\Gamma_{KL}\right)\,\omega_{KL}\ ,
    \label{eq:dilation-of-Gsp-to-Gp-proof-calc-condition-Omega-C_KL-calc-1}
  \end{align}
  which we need to show is equal to $C_{KL}$.  First consider the case where
  $\tr C_{KL}=0$, i.e.\@ $C_{KL}=0$.  Then
  $\Gamma_{KL}=g_{\mathrm k}\proj{\mathrm k}_K\otimes\left(\Gamma_L-G_L\right)$,
  implying that $\Gamma_K=g_{\mathrm k}\proj{\mathrm k}_K$ and $G_L=0$.  Then
  $D_{KL}=\Ident_{KL} - \proj{\mathrm k}_K\otimes\Pi_L^{\Gamma_L} = \Ident_{KL}
  - \Pi^{\Gamma_{KL}}_{KL}$, and thus $D_{KL}$ has no overlap with
  $\Gamma_{KL}$.  It follows that
  $\text{\eqref{eq:dilation-of-Gsp-to-Gp-proof-calc-condition-Omega-C_KL-calc-1}}
  = 0 = C_{KL}$ as required.  Now assume that $\tr C_{KL}\neq 0$.  Then
  \begin{multline}
    \tr\left(D_{KL}\Gamma_{KL}\right)
    = \tr\Gamma_{KL} 
    - g_\mathrm{k} \tr\left(\left(\Gamma_L-G_L\right)\Pi^\Gamma_L\right)
    \\
    = \tr\Gamma_{KL} 
    - g_\mathrm{k} \tr\left(\Gamma_L-G_L\right)
    = \tr C_{KL}\ ,
  \end{multline}
  where the projector $\Pi^\Gamma_L$ has no effect in the second expression
  since $\Gamma_L-G_L$ is entirely contained within the support of $\Gamma_L$.
  Then again
  $\text{\eqref{eq:dilation-of-Gsp-to-Gp-proof-calc-condition-Omega-C_KL-calc-1}}
  = C_{KL}$ as required.

  We have thus constructed a completely positive, trace preserving map $\Phi_{KLQ\to KLQ}$
  which maps $\Gamma_{KLQ}$ onto itself and which
  satisfies~\eqref{eq:dilation-of-Gsp-to-Gp-condition-recover-correct-mapping}. This
  concludes the proof.
\end{proof}

\begin{proof}[*cor:dilation-of-Gsp-to-Gp-cases]
  The proofs of~(a) and~(b) exploit the following fact: If a bipartite
  (normalized) quantum state $\zeta_{AB}$ satisfies
  $\dmatrixel{\chi}{\zeta_{AB}}_{B} = \zeta'_A$ for some pure state
  $\ket{\chi}_B$ and a (normalized) quantum state $\zeta'_A$, then
  $\zeta_{AB} = \zeta'_A\otimes\proj{x}_B$.  [Indeed, $\zeta_{AB}$ must lie
  within the support of the projector $\Ident_A\otimes\proj\chi_B$ since
  $\tr({(\Ident_{AB}-(\Ident_A\otimes\proj\chi_B))}\,\zeta_{AB}) =
  \tr(\zeta_{AB})-\tr(\zeta'_A)=0$, and hence
  $\zeta_{AB} =
  (\Ident_A\otimes\proj{\chi}_B)\zeta_{AB}(\Ident_A\otimes\proj{\chi}_B) =
  \zeta_A'\otimes\proj{\chi}_B$.]

  \emph{Proof of (a):} For any quantum state $\tau$
  supported on $P$, we have by assumption
  $\dmatrixel{\mathrm{k,f}}{\Phi_{KLQ\to{}KLQ}
    (\tau_K\otimes\proj{\mathrm{l\,i}}_{LQ})} = \tilde{\Phi}_{K\to L}(\tau_K)$
  with
  $\tr(\dmatrixel{\mathrm{k,f}}{\Phi_{KLQ\to{}KLQ}
    (\tau_K\otimes\proj{\mathrm{l\,i}}_{LQ})}) =
  \tr(\tilde{\Phi}_{K\to{}L}(\tau_K)) = 1$.  Using the fact above we conclude
  that
  $\Phi_{KLQ\to KLQ}(\tau_K\otimes\proj{\mathrm{l\,i}}_{LQ}) =
  \tilde{\Phi}_{K\to L}(\tau_K)\otimes\proj{\mathrm{k\,f}}_{KQ}$.

  \emph{Proof of (b):} By assumption,
  $\dmatrixel{\mathrm{k,f}}{\Phi_{KLQ\to{}KLQ}
    (\sigma_{KR}\otimes\proj{\mathrm{l\,i}}_{LQ})} =
  \tilde{\Phi}_{K\to{}L}(\sigma_{KR})$ with
  $\tr(\dmatrixel{\mathrm{k,f}}{\Phi_{KLQ\to{}KLQ}
    (\sigma_{KR}\otimes\proj{\mathrm{l\,i}}_{LQ})}) =
  \tr(\tilde{\Phi}_{K\to{}L}(\sigma_{KR})) = 1$.  Again a straightforward
  application of the above fact yields
  $\Phi_{KLQ\to KLQ}(\sigma_{KR}\otimes\proj{\mathrm{l\,i}}_{LQ}) =
  \tilde{\Phi}_{K\to L}(\sigma_{KR})\otimes\proj{\mathrm{k\,f}}_{KQ}$.

  \emph{Proof of (c):} We know that the mapping $\Phi_{KLQ}$ provided by
  \autoref{prop:dilation-of-Gsp-to-Gp} is such that
  $\dmatrixel{\mathrm{k\,f}}{\Phi_{KLQ}
    (\sigma_{KR}\otimes\proj{\mathrm{l\,i}}_{LQ})} =
  \tilde{\Phi}_{K\to{}L}(\sigma_{KR})$.  We exploit the fact that the fidelity
  does not change if we project one state onto the support of the other state
  [indeed, we have
  $F(\sigma,\rho) = \tr`\big[`\big(\sigma^{1/2}\rho\sigma^{1/2})^{1/2}] =
  \tr`\big[`\big(\sigma^{1/2}\Pi^\sigma\, \rho\, \Pi^\sigma\sigma^{1/2})^{1/2}]
  = F(\sigma,\Pi^\sigma\,\rho\Pi^\sigma)$].  This means in turn that
  $F`\big(\Phi_{KLQ}(\sigma_{KR}\otimes\proj{\mathrm{l\,i}}_{LQ}),
  \rho_{LR}\otimes\proj{\mathrm{k\,f}}_{KQ}) =
  F\bigl(`(\Ident_{LR}\otimes\proj{\mathrm{k\,f}}_{KQ}) \;
  {\Phi_{KLQ}(\sigma_{KR}\otimes\proj{\mathrm{l\,i}}_{LQ})} \;
  `(\Ident_{LR}\otimes\proj{\mathrm{k\,f}}_{KQ}) \;\mathrel{,~}\;
  {\rho_{LR}\otimes\proj{\mathrm{k\,f}}_{KQ}}\bigr) =
  F`\big(\tilde{\Phi}_{K\to{}L}(\sigma_{KR})\otimes\proj{\mathrm{k\,f}}_{KQ}, 
  \rho_{LR}\otimes\proj{\mathrm{k\,f}}_{KQ}) =
  F`\big(\tilde{\Phi}_{K\to{}L}(\sigma_{KR}), \rho_{LR})$.  This proves the claim
  since the purified distance is defined in terms of the fidelity.
\end{proof}

\begin{proof}[*prop:equiv-battery-models]
  The proof consists in showing
  \ref{item:prop-equiv-battery-models-equivstatement-criterionMapEnorm}%
  $\,\Rightarrow$%
  \ref{item:prop-equiv-battery-models-equivstatement-projGamma}%
  $\,\Rightarrow$%
  \ref{item:prop-equiv-battery-models-equivstatement-weight}%
  $\,\Rightarrow$%
  \ref{item:prop-equiv-battery-models-equivstatement-wit}%
  $\,\Rightarrow$%
  \ref{item:prop-equiv-battery-models-equivstatement-criterionMapEnorm}
  as well as
  \ref{item:prop-equiv-battery-models-equivstatement-projGamma}%
  $\,\Rightarrow$%
  \ref{item:prop-equiv-battery-models-equivstatement-infbattery}%
  $\,\Rightarrow$%
  \ref{item:prop-equiv-battery-models-equivstatement-criterionMapEnorm}.

  {%
    \ref{item:prop-equiv-battery-models-equivstatement-criterionMapEnorm}%
    $\,\Rightarrow$%
    \ref{item:prop-equiv-battery-models-equivstatement-projGamma}:}%
  \hspace{1ex}%
  By assumption we have $\mathcal{T}_{X\to X'}`*(\Gamma_X)\leqslant 2^{-y}\Gamma_{X'}$.
  Let $\Gamma_{W_1},\Gamma_{W_2}$ and $P_{W_1},P'_{W_2}$
  satisfy the assumptions in the claim
  \ref{item:prop-equiv-battery-models-equivstatement-projGamma}, and
  define the shorthands
  \begin{align}
    \sigma^{(1)}_{W_1}
    &= \frac{P_{W_1} \Gamma_{W_1} P_{W_1}}{\tr`(P_{W_1}\Gamma_{W_1})}\ ;
    &\sigma^{(2)}_{W_2}
    &= \frac{P'_{W_2} \Gamma_{W_2} P'_{W_2}}{\tr`(P'_{W_2}\Gamma_{W_2})}\ .
  \end{align}
  Define the map
  \begin{align}
    \Phi''_{XW_1\to X'W_2}`*(\cdot)
    =  \mathcal{T}_{X\to X'}`*[
    \tr_{W_1}`*( P_{W_1} `(\cdot) )] \otimes \sigma^{(2)}_{W_2} \ .
  \end{align}
  This map is completely positive by construction, and is trace nonincreasing because it
  is a composition of trace nonincreasing maps.  We need to show that it is
  $\Gamma$-sub-preserving. We have
  \begin{align}
    \Phi''_{XW_1\to X'W_2}`*(\Gamma_X\otimes\Gamma_{W_1})
    &= `\big(\tr P_{W_1}\Gamma_{W_1}) \cdot
      \mathcal{T}_{X\to X'}`(\Gamma_{X}) \otimes \sigma^{(2)}_{W_2}
      \nonumber\\
    &\leqslant 2^{-y} \frac{\tr P_{W_1}\Gamma_{W_1}}{\tr P'_{W_2}\Gamma_{W_2}}
      \cdot \Gamma_{X'} \otimes
      `*(P'_{W_2}\Gamma_{W_2}P'_{W_2}) 
      \nonumber\\
    &\leqslant \Gamma_{X'}\otimes\Gamma_{W_2}\ ,
  \end{align}
  using the fact that $P'_{W_2}\Gamma_{W_2}P'_{W_2}\leqslant\Gamma_{W_2}$ since
  $\Gamma_{W_2}$ commutes with $P'_{W_2}$.

  {%
    \ref{item:prop-equiv-battery-models-equivstatement-projGamma}%
    $\,\Rightarrow$%
    \ref{item:prop-equiv-battery-models-equivstatement-weight}:}%
  \hspace{1ex}%
  This special case follows directly from
  \ref{item:prop-equiv-battery-models-equivstatement-projGamma} with
  $W_1= W_2= \tilde Q$, $\Gamma_{W_1}=\Gamma_{W_2}=\Gamma_{\tilde Q}$ and by
  choosing $P_{W_1}=\proj 1_{\tilde Q}$, $P'_{W_2}=\proj 2_{\tilde Q}$. Note that $g_1=\tr
  P_{W_1}\Gamma_{W_1}$ and $g_2=\tr P'_{W_2}\Gamma_{W_2}$ and hence indeed
  $`*(\tr P'_{W_2}\Gamma_{W_2})/`*(\tr P_{W_1}\Gamma_{W_1}) = g_2/g_1 \geqslant 2^{-y}$.

  {%
    \ref{item:prop-equiv-battery-models-equivstatement-weight}%
    $\,\Rightarrow$%
    \ref{item:prop-equiv-battery-models-equivstatement-wit}:}%
  \hspace{1ex}%
  This is a trivial special case of
  \ref{item:prop-equiv-battery-models-equivstatement-weight}.

  {%
    \ref{item:prop-equiv-battery-models-equivstatement-wit}%
    $\,\Rightarrow$%
    \ref{item:prop-equiv-battery-models-equivstatement-criterionMapEnorm}:}%
  \hspace{1ex}%
  Pick $\Gamma_Q$, $\ket 1_Q, \ket 2_Q$, $g_1, g_2$ such that they satisfy the
  assumptions of \ref{item:prop-equiv-battery-models-equivstatement-wit} as well
  as $g_2/g_1=2^{-y}$ and let $\Phi'_{XQ\to X'Q}$ be the corresponding mapping.
  Observe that for any $\omega_X$
  \begin{align}
    \mathcal{T}_{X\to X'}`(\omega_X) = \dmatrixel[\big]{2}{%
      \Phi'_{XQ\to X'Q}`\big(\omega_X\otimes\proj 1_Q)%
      }_Q\ .
  \end{align}
  Plugging in $\omega_X=\Gamma_X$, and using the fact that $g_1\proj 1_Q \leqslant
  \Gamma_Q$ and that $\Phi'_{XQ\to X'Q}$ is $\Gamma$-sub-preserving,
  \begin{align}
    \mathcal{T}_{X\to X'}`(\Gamma_X)
    &\leqslant \dmatrixel[\big]{2}{%
      g_1^{-1}\cdot\Phi'_{XQ\to X'Q}`\big(\Gamma_X\otimes\Gamma_Q)%
    }_Q\ .
    \nonumber\\
    &\leqslant \dmatrixel[\big]{2}{%
      g_1^{-1}\cdot\Gamma_{X'}\otimes\Gamma_Q%
    }_Q\ .
    \nonumber\\
    &= \frac{g_2}{g_1}\cdot\Gamma_{X'}
    = 2^{-y}\,\Gamma_{X'}\ .
  \end{align}

  {%
    \ref{item:prop-equiv-battery-models-equivstatement-projGamma}%
    $\,\Rightarrow$%
    \ref{item:prop-equiv-battery-models-equivstatement-infbattery}:}%
  \hspace{1ex}%
  This is in fact another special case of
  \ref{item:prop-equiv-battery-models-equivstatement-projGamma}.  Let
  $\lambda_1,\lambda_2$ such that $\lambda_1-\lambda_2\leqslant y$ and that
  $2^{\lambda_1},2^{\lambda_2}$ are integers.  Let $A$ be any quantum system of
  dimension at least $\max`{2^{\lambda_1},2^{\lambda_2}}$ and with
  $\Gamma_A=\Ident_A$.  Now we use our assumption that
  \ref{item:prop-equiv-battery-models-equivstatement-projGamma} holds.  Choose
  $W_1=W_2 = A$, $P_{W_1} = \Ident_{2^{\lambda_1}}$,
  $P'_{W_2}=\Ident_{2^{\lambda_2}}$.  Observe that
  $\tr(P_{W_1}\Gamma_{W_1})=\tr(P_{W_1}) = 2^{\lambda_1}$ and
  $\tr(P'_{W_2}\Gamma_{W_2})=\tr(P'_{W_2}) = 2^{\lambda_2}$, and hence the
  assumptions of \ref{item:prop-equiv-battery-models-equivstatement-projGamma}
  are satisfied.  Then we know that there must exist a $\Gamma$-sub-preserving,
  trace-nonincreasing map $\Phi''_{XA\to X'A}$
  obeying~\eqref{eq:prop-equiv-battery-models-projGamma-condition-correct-mapping}.
  The latter condition reads by plugging in our choices
  \begin{align}
    \Phi''_{XA\to X'A}`\big(
    \omega_X \otimes `\big(2^{-\lambda_1}\Ident_{2^{\lambda_1}})
    )
    = \mathcal{T}_{X\to X'}`*(\omega_X) \otimes
    `\big(2^{-\lambda_2}\Ident_{2^{\lambda_2}})\ ,
  \end{align}
  for all $\omega_X$.  This is exactly the condition that $\Phi$ has to fulfill, and hence
  $\Phi$ may be taken equal to the map $\Phi''$.  It follows that
  \ref{item:prop-equiv-battery-models-equivstatement-infbattery} is true.

  {%
    \ref{item:prop-equiv-battery-models-equivstatement-infbattery}%
    $\,\Rightarrow$%
    \ref{item:prop-equiv-battery-models-equivstatement-criterionMapEnorm}:}%
  \hspace{1ex}%
  Consider any $\lambda_1,\lambda_2\geqslant 0$ with $\lambda_1-\lambda_2\leqslant y$.
  Let $\Phi_{XA\to X'A}$ be the corresponding $\Gamma$-sub-preserving map given by the
  assumption that
  \ref{item:prop-equiv-battery-models-equivstatement-infbattery} holds.
  Observe that for all $\omega_X$,
  \begin{align}
    \mathcal{T}_{X\to X'}`*(\omega_X) = \tr_A\,`*{
    \Ident_{2^{\lambda_2}}\;
    \Phi_{XA\to X'A}`\big( \omega_X\otimes
    `\big(2^{-\lambda_1}\Ident_{2^{\lambda_1}}) )
    }\ .
  \end{align}
  Plugging in $\omega_X=\Gamma_X$, and using the fact that $\Phi$ is
  $\Gamma$-sub-preserving,
  \begin{align}
    \mathcal{T}_{X\to X'}`*(\Gamma_X)
    &\leqslant \tr_A\,`*{
      \Ident_{2^{\lambda_2}}\;
      \Phi_{XA\to X'A}`\big( 2^{-\lambda_1}\cdot\Gamma_X\otimes\Gamma_A)
    }
    \nonumber\\
    &\leqslant 2^{-\lambda_1}\cdot\tr_A\,`*{
      \Ident_{2^{\lambda_2}}\;
      \Gamma_{X'}\otimes\Gamma_A
    }
    \nonumber\\
    &= 2^{-`*(\lambda_1-\lambda_2)}\,\Gamma_{X'}
  \end{align}
  Statement~\ref{item:prop-equiv-battery-models-equivstatement-criterionMapEnorm}
  follows by choosing a sequence of $(\lambda_1,\lambda_2)$ with
  $\lambda_1-\lambda_2\rightarrow{}y$.
\end{proof}

\section{The coherent relative entropy}
\appendixlabel{appx:sec:coh-rel-entr-def-and-props}

\subsection{Definition and basic properties}

Consider two quantum systems $X$ and $X'$, described by respective $\Gamma$
operators $\Gamma_X$ and $\Gamma_{X'}$.  We would like to perform a logical
process from $X$ to $X'$ which is described by the process matrix
$\rho_{X'R_X}$, with a reference system $R_X\simeq X$.  As we have seen, the
process matrix uniquely identifies both an input state $\sigma_X$ and a
trace-nonincreasing, completely positive map $\mathcal{E}_{X\to X'}$ on the
support of $\sigma_X$.

Because $\rho_{X'R_X}$ only fixes the mapping on the support of $\sigma_X$,
there may be several trace-nonincreasing, completely positive maps
$\mathcal{T}_{X\to X'}$ which implement this given process matrix.
The coherent relative entropy is defined as the optimal battery usage achieved
by a $\mathcal{T}_{X\to X'}$ with fixed process matrix $\rho_{X'R_X}$, relative
to $\Gamma$ operators $\Gamma_{X}, \Gamma_{X'}$.

In fact, we allow the implementation to fail with some fixed probability
$\epsilon\geqslant 0$ which can be chosen freely.  This allow us to ignore very
improbable events.  Such a practice is standard in the smooth entropy framework,
and it is even necessary in order to make physical statements and recover the
correct asymptotic
behavior~\cite{PhDRenner2005_SQKD,PhDTomamichel2012,BookTomamichel2016_Finite}.
Hence, we allow the process matrix achieved by the optimization variable
$\mathcal{T}_{X\to X'}$ on the given input state to only be $\epsilon$-close to
the requested process matrix $\rho_{X'R_X}$.

By \autoref{prop:equiv-battery-models}, the optimal number of extracted battery
charge $y$ of a fixed $\mathcal{T}_{X\to X'}$ is given by the condition
$\mathcal{T}_{X\to X'}`*(\Gamma_X) \leqslant 2^{-y}\,\Gamma_{X'}$.  We are then
directly led to the following definition.

\begin{thmheading}{Coherent Relative Entropy}
  \label{def:coh-rel-entr}
  For a bipartite quantum normalized state $\rho_{X'R_X}$, two positive
  semidefinite operators $\Gamma_X$ and $\Gamma_{X'}$ such that
  $t_{R_X\to X}(\rho_{X'R_X})$ lies in the support of
  $\Gamma_X\otimes\Gamma_{X'}$, and for $\epsilon\geqslant 0$, the
  \emph{coherent relative entropy} is defined as
  \begin{multline}
    \DCohz[\epsilon]{\rho}{X}{X'}{\Gamma_{X}}{\Gamma_{X'}}
    = 
      \max_{\substack{
        \mathcal{T}_{X\to X'}`(\Gamma_X) \leqslant 2^{-y}\,\Gamma_{X'} \\
        \mathcal{T}_{X\leftarrow X'}^\dagger`(\Ident_{X'}) \leqslant \Ident_{R_X} \\
        P`(\mathcal{T}_{X\to X'}`(\sigma_{XR_X}), \rho_{X'R_X}) \leqslant\epsilon
      }}
      y\quad ,
    \label{eq:coh-rel-entr-def}
  \end{multline}
  where the optimization ranges over all $y\in\mathbb{R}$ and over all
  completely positive maps $\mathcal{T}_{X\to X'}$ satisfying the given
  conditions, and where we use the shorthand
  $\ket\sigma_{XR_X} = \rho_{R_X}^{1/2}\,\ket\Phi_{X:R_X}$.
\end{thmheading}

If $\epsilon=0$, we may omit the $\epsilon$ superscript altogether:
\begin{align}
  \DCohz{\rho}{X}{X'}{\Gamma_{X}}{\Gamma_{X'}}
  =
  \DCohz[\epsilon=0]{\rho}{X}{X'}{\Gamma_{X}}{\Gamma_{X'}}\ .
\end{align}

Clearly, the coherent relative entropy is monotonously increasing in $\epsilon$, as the
optimization set gets larger.

We now introduce the variable $\alpha = 2^{-y}$ and denote by $T_{X'R_X}$ the
Choi matrix of $\mathcal{T}_{X\to X'}$, allowing us to write the coherent
relative entropy as a semidefinite program.

\begin{proposition}[Semidefinite program]
  \label{prop:coh-rel-entr-SDP}
  For a bipartite quantum normalized state $\rho_{X'R_X}$, two positive
  semidefinite operators $\Gamma_X$ and $\Gamma_{X'}$ such that
  $t_{R_X\to X}(\rho_{X'R_X})$ lies in the support of
  $\Gamma_X\otimes\Gamma_{X'}$, and for $\epsilon\geqslant 0$, the
  coherent relative entropy may be written as
  \begin{align}
    \DCohz[\epsilon]{\rho}{X}{X'}{\Gamma_{X}}{\Gamma_{X'}}
    = -\log \alpha\ ,
  \end{align}
  where $\alpha$ is the optimal solution to the following semidefinite program
  in terms of the variables $T_{X'R_XE}\geqslant 0, \alpha\geqslant 0$, and dual
  variables $\mu, \omega_{X'}, X_{R_X}\geqslant 0$, with $\ket\rho_{X'R_XE}$
  being an arbitrary but fixed purification of $\rho_{X'R_X}$ into an
  environment system $E$ of dimension at least $\abs{E}\geqslant\abs{X'R_X}$:
 \medskip\par\noindent\textit{Primal problem:} \colorlet{thesisSDPDualVarColor}{black!50!white}
  \begin{subequations}
    \label{eq:SDP-coh-rel-entr-primal}
    \begin{alignat}{3}
      \text{\normalfont minimize:}&~& \alpha \quad &&& \nonumber\\
      \text{\normalfont subject to:}&~~&
      \tr_{X'}`*[T_{X'R_X}] &\leqslant \Ident_{R_X}
      &~&\textcolor{thesisSDPDualVarColor}{{}: X_{R_X}}
      \label{eq:SDP-coh-rel-entr-cond-trnoninc} \\
      & &
      \tr_{R_X}`*[T_{X'R_X}\,\Gamma_{R_X}] &\leqslant \alpha\,\Gamma_{X'}
      &~&\textcolor{thesisSDPDualVarColor}{{}:\omega_{X'}}
      \label{eq:SDP-coh-rel-entr-cond-alphaSGPM} \\
      & &
      \hspace*{-2em}
      \tr`( \rho_{R_X}^{1/2}\,T_{X'R_XE}\,\rho_{R_X}^{1/2}\,\rho_{X'R_XE} )
      &\geqslant 1 - \epsilon^2
        &~&\textcolor{thesisSDPDualVarColor}{{}: \mu}
        \label{eq:SDP-coh-rel-entr-cond-processmatrixclose}
    \end{alignat}
  \end{subequations}
  \par\noindent\textit{Dual problem:}\index{dual problem}
  \begin{subequations}
    \label{eq:SDP-coh-rel-entr-dual}
    \begin{alignat}{1}
      \text{\normalfont maximize:}&\quad
      \makebox[13em][c]{$\displaystyle
        \mu\,`(1-\epsilon^2) - \tr`(X_{R_X})
        $} \nonumber\\
      \text{\normalfont subject to:}&\quad
      \makebox[13em][c]{$\displaystyle
        \tr\,`*[\omega_{X'}\Gamma_{X'}] \leqslant 1
        $}~\textcolor{thesisSDPDualVarColor}{{}:\alpha}
      \label{eq:SDP-coh-rel-entr-conddual-tromegaGamma}\\
      &\hspace{-1em}
        \mu\, \rho_{R_X}^{1/2}\,\rho_{X'R_XE}\,\rho_{R_X}^{1/2} \leqslant
        \Ident_E\otimes`\big(
        \omega_{X'}\otimes \Gamma_{R_X} + \Ident_{X'}\otimes X_{R_X}
        )
        \nonumber\\[-0.5ex]
        &\hspace{-1em}
        \rule{15em}{0pt}~\textcolor{thesisSDPDualVarColor}{{}: T_{X'R_X}}
      \label{eq:SDP-coh-rel-entr-conddual-Z}
    \end{alignat}
  \end{subequations}
  using the shorthand $\Gamma_{R_X} = t_{X\to R_X}`(\Gamma_X)$.
\end{proposition}

In the above, the reference system $R_X$ may be understood as a ``mirror
system'' which allows us to compare how the output and the input of the process
are correlated. A classical analogue of $R_X$ would be a memory register which
stores a copy of the input.  Crucially, in the semidefinite program the ``mirror
images'' $\Gamma_{R_X}$ and $\sigma_{R_X}$ of $\Gamma_X$ and $\sigma_X$ must be
constructed consistently, using the same reference basis on $R_X$, as
encoded in the ket $\ket{\Phi}_{X:R_X}$ and the partial transpose operation
$t_{X\to R_X}(\cdot)$.
In the semidefinite program, $\Gamma_X$ needs to be represented on $R_X$, and
general Choi matrices of processes $\mathcal{T}_{X\to X'}$ need to be
represented on $X'R_X$, so in general we need $R_X\simeq X$ even if a smaller
system could hold a purification of $\sigma_X$ (for instance, if $\sigma_X$ is
already pure).  By contrast, in the definition~\eqref{eq:coh-rel-entr-def} one
could actually choose a more general $R_X$ system: Given $\sigma_{X}$ and
$\mathcal{E}_{X\to X'}$, one may choose any purification $\ket{\sigma}_{XR_X}$
and correspondingly define
$\rho_{X'R_X} = \mathcal{E}_{X\to X'}(\sigma_{XR_X})$.

The dual problem~\eqref{eq:SDP-coh-rel-entr-dual} is strictly feasible (choose,
e.g., $\omega_{X'}=\Ident_{X'}/(2\tr(\Gamma_{X'}))$, $X_{R_X}=\Ident_{R_X}$ and
$\mu=1/2$), and $T_{X'R_X}=\rho_{R_X}^{-1/2}\rho_{X'R_X}\rho_{R_X}^{-1/2}$ is a
feasible primal candidiate, and hence by Slater's sufficiency conditions
(\autoref{thm:Slater}) we have that {strong duality} holds and there always
exists optimal primal candidates. For $\epsilon>0$, the primal problem is also
strictly feasible (choose
$T_{X'R_X} = (1-\epsilon^2/2)\, \rho_{R_X}^{-1/2}\rho_{X'R_X}\rho_{R_X}^{-1/2} +
(\epsilon^2/4) \Ident_{X'R_X}/\abs{X'}$), and there always exists optimal dual
candidates as well.  However, note that for $\epsilon=0$ the primal problem is
not always strictly feasible (indeed,
constraint~\eqref{eq:SDP-coh-rel-entr-cond-processmatrixclose} is very strong
and fixes the mapping $T_{X'R_X}$ on a subspace; because it must be
trace-preserving on that subspace then~\eqref{eq:SDP-coh-rel-entr-cond-trnoninc}
cannot be satisfied strictly). This means that there is a possibility that there
is no choice of optimal dual variables.  However, since strong duality holds,
there is always a sequence of choices for dual variables whose attained
objective value will converge to the optimal solution of the semidefinite
program.

Here are first some basic properties of the coherent relative entropy.

\begin{proposition}[Trivial bounds]
  \label{prop:coh-rel-entr-trivial-bounds}
  For any $0\leqslant\epsilon<1$,  we have 
  \begin{subequations}
    \label{eq:prop-coh-rel-entr-trivial-bounds}
  \begin{align}
    &
    \DCohz[\epsilon]{\rho}{X}{X'}{\Gamma_X}{\Gamma_{X'}}
      \nonumber\\
    &\hspace*{4em}
      \geqslant - \log \tr\Gamma_{X} - \log\, \norm{\Gamma_{X'}^{-1}}_\infty
      - \log`(1-\epsilon^2)\ ;
    \label{eq:prop-coh-rel-entr-trivial-bounds-upper-bound}\\
    &\DCohz[\epsilon]{\rho}{X}{X'}{\Gamma_X}{\Gamma_{X'}}
      \nonumber\\
    &\hspace*{4em}
      \leqslant \log\,\norm{\Gamma_X^{-1}}_\infty + \log\tr\Gamma_{X'}
      - \log`(1-\epsilon^2)\ .
    \label{eq:prop-coh-rel-entr-trivial-bounds-lower-bound}
  \end{align}
  \end{subequations}
\end{proposition}

In the thermodynamic version of the framework, these bounds can be understood in
terms of work extraction.  Suppose $\Gamma_X = \Gamma_{X'} = \ee^{-\beta H_X}$
with a Hamiltonian $H_X$ and an inverse temperature $\beta$.  Then
$\log\,\norm{\Gamma_X^{-1}}_\infty$ (resp.\@
$\log\,\norm{\Gamma_{X'}^{-1}}_\infty$) is $\beta$ times the maximum {energy} of
$R$ (resp.\@ $X'$), and similarly, $\tr\Gamma_X$ (resp.\@ $\tr\Gamma_{X'}$) is
the partition function of $X$ (resp.\@ $X'$).  The partition function is
directly related to the work cost of {erasure} (resp.\@ formation) of a {thermal
  state} to (resp.\@ from) a pure energy eigenstate of zero energy.  In this
case, the bounds~\eqref{eq:prop-coh-rel-entr-trivial-bounds} correspond to the
ultimate worst and best cases respectively.  The ultimate worst case is that we
start off in a thermal state and end up in the highest energy level, whereas the
absolute best case would be to start in the highest energy eigenstate and finish
in the {Gibbs state}.

Much like the conditional entropy and relative entropy, the coherent relative
entropy is invariant under {partial isometries} of which $\rho_{X'R}$ and
$\Gamma$ operators lie in the support.  In particular, the coherent relative
entropy is completely oblivious to dimensions of the Hilbert spaces which are
not spanned by $\Gamma_R$ and $\Gamma_{X'}$.

\begin{proposition}[Invariance under isometries]
  \label{prop:coh-rel-entr-invar-under-isometries}
  Let $\tilde X$, $\tilde X'$ be new systems.  Suppose there exist partial
  isometries $V_{X\to\tilde X}$ and $V'_{X'\to \tilde X'}$ such that both
  $t_{R_X\to X}`(\rho_{R_X})$ and $\Gamma_X$ are in the support of
  $V_{X\to\tilde X}$, and both $\rho_{X'}$ and $\Gamma_{X'}$ are in the support
  of $V'_{X'\to\tilde X'}$.  Then
  \begin{multline}
    \DCohz[\epsilon]{*`(V'\otimes V)\, \rho_{X'R_X}\, `(V'\otimes V)^\dagger}%
    {\tilde X}{\tilde X'}{V\Gamma_X V^\dagger}{V'\Gamma_{X'} V'^\dagger}
    \\
    = \DCohz[\epsilon]{\rho}{X}{X'}{\Gamma_X}{\Gamma_{X'}}\ .
  \end{multline}
\end{proposition}

This proposition allows us to embed states in larger dimensions, as well as to show that
it is invariant under simultaneous action of unitaries on the states and the $\Gamma$
operators.

We may also check the behavior of the coherent relative entropy under re-scaling of the
$\Gamma$ operators (as the latter need not conform to any normalization).
Intuitively, in the thermodynamic case where $\Gamma = \ee^{-\beta H}$ for a Hamiltonian
$H$ and an inverse temperature $\beta$, the transformation $\Gamma\to a\Gamma$ for a
constant factor $a$ yields the $\Gamma$ operator corresponding to the modified Hamiltonian
$H\to H-\beta^{-1}\ln a$, that is, a constant energy shift of all levels.  Consequently,
we expect that scaling the $\Gamma$ operators introduces a constant shift in the coherent
relative entropy, which would correspond to providing the required energy to compensate
for the global change in energy.

\begin{proposition}[Scaling the $\Gamma$ operators]
  \label{prop:coh-rel-entr-scaling-Gamma}
  For any $0\leqslant\epsilon<1$, and for real numbers $a,b>0$,
  \begin{multline}
    \DCohz[\epsilon]\rho{X}{X'}{a\Gamma_X}{b\Gamma_{X'}}
\\
    = \DCohz[\epsilon]\rho{X}{X'}{\Gamma_X}{\Gamma_{X'}} + \log \frac b a\ .
  \end{multline}
\end{proposition}

The coherent relative entropy furthermore obeys a super\-additivity rule, expressing the
fact that a joint implementation of two parallel independent processes cannot be worse
than two separate implementations of each process.

\begin{proposition}[Superadditivity for tensor products]
  \label{prop:coh-rel-entr-superadditive-tensor-products}
  Let systems $X_1$, $X_1'$, $X_2$, $X_2'$ have respective $\Gamma$ operators
  $\Gamma_{X_1}$, $\Gamma_{X_1'}$, $\Gamma_{X_2}$, $\Gamma_{X_2'}$.  Let
  $\rho_{X_1'R_{X_1}}$ and $\zeta_{X_2'R_{X_2}}$ be two quantum states.  Then
  for any $\epsilon,\epsilon'\geqslant 0$,
  \begin{align}
    \hspace*{3em}&\hspace*{-3em}
    \DCohz[\epsilon'']{*\rho_{X_1'R_{X_1}}\otimes\zeta_{X_2'R_{X_2}}}%
      {X_1X_2}{X_1'X_2'}%
      {\Gamma_{X_1}\otimes\Gamma_{X_2}}{\Gamma_{X_1'}\otimes\Gamma_{X_2'}}
      \nonumber\\
    &\geqslant \DCohz[\epsilon]{\rho}{X_1}{X_1'}{\Gamma_{X_1}}{\Gamma_{X_1'}}
      \nonumber\\
    &\hspace*{5em}
      + \DCohz[\epsilon']{\zeta}{X_2}{X_2'}{\Gamma_{X_2}}{\Gamma_{X_2'}}\ ,
  \end{align}
  where $\epsilon''=\sqrt{\epsilon^2+\epsilon'^2}$.
\end{proposition}

In contrast to measures like the min-entropy and the max-entropy, we do not have
equality in general in
\autoref{prop:coh-rel-entr-superadditive-tensor-products}.  One may see this
with a simple example analogous to that in Ref.~\cite{Faist2015NJP_Gibbs}.
Consider two qubit systems $Q_i$ with $\Gamma_{Q_i} = g_0\proj 0 + g_1\proj 1$
(with $i=1,2$; $g_0> g_1$).  On a single system, performing the logical process
$\ket0\to\ket+ =`*(\ket 0 + \ket 1)/\sqrt2$ has a different cost than the yield
of $\ket+ \to \ket0$.\footnote{That the processes $\ket0\to\ket+$ and
  $\ket+\to\ket0$ have different work cost and yield respectively follows from
  \autoref{cor:coh-rel-entr-pureev-input-output} below.  We have
  $\Dminz{\proj+}{\Gamma} = -\log\dmatrixel+\Gamma=-\log\,`[`(g_0+g_1)/2]$ and
  $-\Dmax{\proj+}{\Gamma} = -\log\norm{\Gamma^{-1/2}\proj+\Gamma^{-1/2}}_\infty
  = -\log\dmatrixel+{\Gamma^{-1}} = -\log\,`[`(g_0^{-1}+g_1^{-1})/2]$ (the
  argument of the norm is a pure state).} However, the transition
$\ket0\otimes\ket+ \to \ket+\otimes\ket0$ can be achieved with a swap operation,
which is perfectly $\Gamma$-preserving and hence costs no pure qubits.

A further property of the coherent relative entropy can be derived in the case
where the $\Gamma$ operators are restricted by projecting them onto selected
eigenkets, while still having the process matrix lying in their support.  Then
the coherent relative entropy remains unchanged.

\begin{proposition}[Restricting the $\Gamma$ operators]
  \label{prop:coh-rel-entr-restrict-Gamma-eigenspaces}
  Let $P_X$ and $P'_{X'}$ be projectors such that $[P_X,\Gamma_X]=0$ and
  $[P'_{X'},\Gamma_{X'}]=0$.  Define $\Gamma'_X=P_X\Gamma_XP_X$ and
  $\Gamma'_{X'}=P'_{X'}\Gamma_{X'}P'_{X'}$.  Let $\rho_{X'R_X}$ be any quantum
  state with support inside that of $\Gamma'_{X'}\otimes \Gamma'_{R_X}$.  Then
  \begin{align}
    \DCohz[\epsilon]\rho{X}{X'}{\Gamma'_X}{\Gamma'_{X'}}
    = \DCohz[\epsilon]\rho{X}{X'}{\Gamma_X}{\Gamma_{X'}}\ .
  \end{align}
\end{proposition}

Another property relates the coherent relative entropy to that with respect to
different $\Gamma$ operators which represent \qq{at least or at most as much
  weight on each state}, as represented as an operator inequality.  Intuitively,
this proposition states that if we raise the energy levels at the input and
lower the levels at the output, then the process is easier to carry out.
\begin{proposition}
  \label{prop:coh-rel-entr-relate-Gamma-operator-less-than}%
  Let $\tilde\Gamma_X\geqslant0$ and $\tilde\Gamma_{X'}\geqslant0$ be such that
  $\tilde\Gamma_X\leqslant\Gamma_X$ and $\Gamma_{X'}\leqslant\tilde\Gamma_{X'}$.
  Then for any $\epsilon\geqslant 0$,
  \begin{align}
    \DCohz[\epsilon]\rho{X}{X'}{\tilde\Gamma_X}{\tilde\Gamma_{X'}}
    \geqslant \DCohz[\epsilon]\rho{X}{X'}{\Gamma_R}{\Gamma_{X'}}\ .
  \end{align}
\end{proposition}

We further note that it is possible to rewrite the definition of the coherent
relative entropy in a slightly alternative form.
\begin{proposition}
  \label{prop:coh-rel-entr-rewrite-primal-problem-infinity-norm}
  The optimization problem defining the coherent relative entropy can be rewritten as
  \begin{multline}
    2^{-\DCohz[\epsilon]{\rho}{X}{X'}{\Gamma_X}{\Gamma_{X'}}}
    \\
    = \min_{T_{X'R_X}}\, \norm[\big]{
      \Gamma_{X'}^{-1/2}\,\tr_{R_X}`*[T_{X'R_X}\,t_{X\to R_X}(\Gamma_{X})]\,\Gamma_{X'}^{-1/2}
    }_\infty\ ,
    \label{eq:prop-coh-rel-entr-rewrite-primal-problem-infinity-norm-TXpR}
  \end{multline}
  where the minimization is taken over all positive semidefinite $T_{X'R_X}$
  satisfying both conditions~\eqref{eq:SDP-coh-rel-entr-cond-trnoninc}
  and~\eqref{eq:SDP-coh-rel-entr-cond-processmatrixclose}, and for which the
  operator $\tr_{R_X}\left(T_{X'R_X}\,t_{X\to R_X}(\Gamma_X)\right)$ lies within
  the support of $\Gamma_{X'}$.  Equivalently,
  \begin{multline}
    2^{-\DCohz[\epsilon]{\rho}{X}{X'}{\Gamma_X}{\Gamma_{X'}}}
    \\
    = \min_{\mathcal{T}_{X\to X'}}\, \norm[\big]{
      \Gamma_{X'}^{-1/2}\,\mathcal{T}_{X\to X'}`*[\Gamma_X]\,\Gamma_{X'}^{-1/2}
    }_\infty\ ,
    \label{eq:prop-coh-rel-entr-rewrite-primal-problem-infinity-norm-mathcalT}
  \end{multline}
  where the minimization is taken over all trace nonincreasing, completely
  positive maps $\mathcal{T}_{X\to X'}$ which satisfy
  $P`(\mathcal{T}_{X\to X'}`[\sigma_{XR_X}], \rho_{X'R_X}) \leqslant
  \epsilon$ and for which $\mathcal{T}_{X\to X'}\left(\Gamma_X\right)$ lies
  within the support of $\Gamma_{X'}$.
\end{proposition}

Finally, we present an alternative form of the semidefinite program for the
non-smooth coherent relative entropy, i.e., in the case where $\epsilon=0$.
This version of the semidefinite program will prove useful in some later proofs.

\begin{proposition}[Non-smooth specialized semidefinite program]
  \label{prop:coh-rel-entr-nonsmooth-SDP}
  For a bipartite quantum state $\rho_{X'R_X}$, and two positive semidefinite
  operators $\Gamma_X$ and $\Gamma_{X'}$ such that $t_{R_X\to X}(\rho_{X'R_X})$
  lies in the support of $\Gamma_X\otimes\Gamma_{X'}$, the non-smooth coherent
  relative entropy can be written as
  \begin{align}
    \DCohz{\rho}{X}{X'}{\Gamma_{X}}{\Gamma_{X'}}
    &= -\log\alpha \ ;
    \label{eq:coh-rel-entr-nonsmooth-SDP-logvalue}
  \end{align}
  where $\alpha$ is the optimal solution to the following semidefinite program
  in terms of the variables $T_{X'R_X}\geqslant 0, \alpha\geqslant 0$, and dual
  variables
  $Z_{X'R_X}=Z_{X'R_X}^\dagger, \omega_{X'}\geqslant 0, X_{R_X}\geqslant 0$:

  \medskip\par\noindent\textit{Primal problem:}
  \colorlet{thesisSDPDualVarColor}{black!50!white}
  \begin{subequations}
    \begin{alignat}{3}
      \text{\normalfont minimize:}&~& \alpha &&& \nonumber\\
      \text{\normalfont subject to:}&~~&
      \tr_{X'}`*[T_{X'R_X}] &\leqslant \Ident_{R_X}
      &~&\textcolor{thesisSDPDualVarColor}{{}: X_{R_X}}
      \label{eq:SDP-coh-rel-entr-nonsmooth-cond-trnoninc} \\
      & &
      \tr_{R_X}`*[T_{X'R_X}\,t_{X\to R_X}`(\Gamma_X)] &\leqslant \alpha\,\Gamma_{X'}
      &~&\textcolor{thesisSDPDualVarColor}{{}:\omega_{X'}}
      \label{eq:SDP-coh-rel-entr-nonsmooth-cond-alphaSGPM} \\
      & &\rho_{R_X}^{1/2}\,T_{X'R_X}\,\rho_{R_X}^{1/2} &= \rho_{X'R_X}
        &~&\textcolor{thesisSDPDualVarColor}{{}: Z_{X'R_X}}
        \label{eq:SDP-coh-rel-entr-nonsmooth-cond-correctmapping}
    \end{alignat}
  \end{subequations}
  \par\noindent\textit{Dual problem:}\index{dual problem}
  \begin{subequations}
    \begin{alignat}{1}
      \text{\normalfont maximize:}&\quad
      \makebox[13em][c]{$\displaystyle
        \tr\,`*[Z_{X'R_X}\rho_{X'R_X}] - \tr X_{R_X}
        $} \nonumber\\
      \text{\normalfont subject to:}&\quad
      \makebox[13em][c]{$\displaystyle
        \tr\,`*[\omega_{X'}\Gamma_{X'}] \leqslant 1
        $}~\textcolor{thesisSDPDualVarColor}{{}:\alpha}
      \label{eq:SDP-coh-rel-entr-nonsmooth-conddual-tromegaGamma}\\
      &\hspace{-1em}
        \rho_{R_X}^{1/2}\,Z_{X'R_X}\,\rho_{R_X}^{1/2} \leqslant
        t_{X\to R_X}`(\Gamma_X)\otimes\omega_{X'} + X_{R_X}\otimes\Ident_{X'}
        \nonumber\\[-0.5ex]
        &\hspace{-1em}
        \rule{15em}{0pt}~\textcolor{thesisSDPDualVarColor}{{}: T_{X'R_X}}
      \label{eq:SDP-coh-rel-entr-nonsmooth-conddual-Z}
    \end{alignat}
\end{subequations}
\end{proposition}

Here are the proofs corresponding to this section's propositions.

\begin{proof}[*prop:coh-rel-entr-SDP]
  Write $\ket\sigma_{XR} = \rho_{R_X}^{1/2}\ket{\Phi}_{X:R_X}$.  Let
  $\ket\rho_{X'R_XE}$ be any fixed purification of $\rho_{X'R_X}$ in an
  environment system $E$ with dimension $\abs{E}\geqslant\abs{X'R_X}$.
  
  First, consider any feasible candidates $T_{X'RE}, \alpha$
  for~\eqref{eq:SDP-coh-rel-entr-primal}.  Then, setting
  $\mathcal{T}_{X\to X'}`(\cdot) = \tr_E`( T_{X'R_XE}\,t_{X\to R_X}`(\cdot) )$
  and $y=-\log\alpha$ satisfies the requirements
  of~\eqref{eq:coh-rel-entr-def}, in particular,
  $F^2`(\mathcal{T}_{X\to X'}`(\sigma_{XR_X}),\rho_{X'R_X}) \geqslant
  \tr`(\rho_{R_X}^{1/2}\,T_{X'R_XE}\,\rho_{R_X}^{1/2}\,\rho_{X'R_XE}) \geqslant
  1-\epsilon^2$ by Uhlmann's theorem because
  $\rho_{R_X}^{1/2}\,T_{X'R_XE}\rho_{R_X}^{1/2}$ is a purification of
  $\mathcal{T}_{X\to X'}`(\sigma_{XR_X})$.

  Let $\mathcal{T}_{X\to X'}$ and $y$ be valid candidates
  in~\eqref{eq:coh-rel-entr-def}.  Thanks to Uhlmann's theorem, there exists a
  pure quantum state $\ket\tau_{X'R_XE}$ such that
  $F^2`(\mathcal{T}_{X\to X'}(\sigma_{XR_X}),\rho_{X'R_X}) =
  \tr\,`(\tau_{X'R_XE}\,\rho_{X'R_XE})$.  Let $V_{X\to X'E}$ be a Stinespring
  dilation of $\mathcal{T}_{X\to X'}$, i.e., let $V_{X\to{}X'E}$ satisfy
  $V^\dagger V\leqslant\Ident_{X}$ and
  $\mathcal{T}_{X\to X'}`(\cdot) = \tr_E`(V_{X\to{}X'E} \, `(\cdot) \,
  V^\dagger)$.  There exists a unitary $W_E$ such that
  $\ket\tau_{X'R_XE} = W_E\,V_{X\to X'E} \ket\sigma_{X:R_X}$, since those two
  states are both purifications of $\mathcal{T}_{X\to{}X'}`(\sigma_{XR_X})$.
  Now let $\ket{T}_{X'R_XE} = W_E\,V_{X\to X'E}\ket\Phi_{X:R_X}$ and
  $\alpha=2^{-y}$.  Then,
  $\tr_{X'E}`(T_{X'R_XE}) = \tr_{X}`(V^\dagger\,V\,\Phi_{X:R_X}) \leqslant
  \Ident_{R_X}$. Also,
  $\tr_{R_XE}`[T_{X'R_XE}\,\Gamma_{R_X}] = \mathcal{T}_{X\to X'}`(\Gamma_{X})
  \leqslant 2^{-y}\,\Gamma_{X'} = \alpha\,\Gamma_{X'}$.  Finally,
  $\tr`(\rho_{R_X}^{1/2}\,T_{X'R_XE}\,\rho_{R_X}^{1/2}\,\rho_{X'R_XE}) =
  \tr`(W_E\,V_{X\to X'E}\,\sigma_{XR}\,V^\dagger W_E^\dagger\,\rho_{X'R_XE}) =
  \tr`(\tau_{X'R_XE}\,\rho_{X'R_XE}) =
  F^2`(\mathcal{T}_{X\to{}X'}(\sigma_{XR_X}),\rho_{X'R_X}) \geqslant
  1-\epsilon^2$.
\end{proof}

\begin{proof}[*prop:coh-rel-entr-trivial-bounds]
  Let
  $T_{X'R_XE}=`(1-\epsilon^2) \,
  \rho_{R_X}^{-1/2}\rho_{X'R_XE}\rho_{R_X}^{-1/2}$ and note that the
  condition~\eqref{eq:SDP-coh-rel-entr-cond-processmatrixclose} is fulfilled.
  On the other hand,
  $\tr_{X'E} T_{X'RE} = `(1-\epsilon^2)\,\Pi^{\rho_{R_X}}_{R_X}\leqslant
  \Ident_{R_X}$ fulfilling~\eqref{eq:SDP-coh-rel-entr-cond-trnoninc}.  Now
  observe that
  \begin{align}
    \tr`*(T_{X'R_X} \, \Gamma_{R_X})
    = `(1-\epsilon^2) \tr`\big(\Pi^{\rho_{R_X}}_{R_X}\,\Gamma_{R_X})
    \leqslant `(1-\epsilon^2) \tr`*(\Gamma_{R_X})\ ,
  \end{align}
  and hence
  $`[`(1-\epsilon^2)\tr(\Gamma_{R_X})]^{-1} \tr_{R_X}`(T_{X'R_X}\,\Gamma_{R_X})$
  is a subnormalized quantum state, which moreover lives within the support of
  $\Gamma_{X'}$ by assumption.  Hence,
  \begin{align}
    `[`(1-\epsilon^2)\tr(\Gamma_{R_X})]^{-1}
    \tr_{R_X}`(T_{X'R_X} \, \Gamma_{R_X})
    \leqslant \Pi^{\Gamma_{X'}}_{X'}
    \leqslant \norm{\Gamma_{X'}^{-1}}_\infty \, \Gamma_{X'}\ ,
  \end{align}
  noting that $\norm{\Gamma_{X'}^{-1}}_\infty^{-1}$ is the minimal nonzero
  eigenvalue of $\Gamma_{X'}$.  Thus, taking
  $\alpha= `(1-\epsilon^2) \tr(\Gamma_{R_X}) \, \norm{\Gamma_{X'}^{-1}}_\infty$
  satisfies~\eqref{eq:SDP-coh-rel-entr-cond-alphaSGPM} yielding feasible
  primal candidates, which
  proves~\eqref{eq:prop-coh-rel-entr-trivial-bounds-upper-bound}.

  Now consider the dual problem.  Choosing
  $\omega_{X'}=`*(\tr\Gamma_{X'})^{-1}\Ident_{X'}$ immediately
  satisfies~\eqref{eq:SDP-coh-rel-entr-conddual-tromegaGamma}.  Using
  $\rho_{X'R_XE}\leqslant\Ident_{X'R_XE}$ and
  $\rho_{R_X}\leqslant\Pi^{\Gamma_{R_X}}_{R_X}$, we have
  \begin{align}
  \mu\,\rho_{R_X}^{1/2}\rho_{X'R_XE}\rho_{R_X}^{1/2}
    &\leqslant \mu\,\Pi^{\Gamma_{R_X}}_{R_X}\otimes\Ident_{X'E}
      \nonumber\\
    &= \mu\,`*(\tr\Gamma_{X'})\,
    \Ident_E\otimes\omega_{X'}\otimes\Pi_R^{\Gamma_{R_X}}
      \nonumber \\
    &\leqslant
    \mu\,`*(\tr\Gamma_{X'})\,\norm{\Gamma_{R_X}^{-1}}_\infty\,
    \Ident_E\otimes\omega_{X'}\otimes\Gamma_{R_X}\ ,
  \end{align}
  so we choose
  $\mu = `*(\tr\Gamma_{X'})^{-1}\,\norm{\Gamma_{R_X}^{-1}}_\infty^{-1}$ and
  $X_{R_X}=0$ in order to fulfill~\eqref{eq:SDP-coh-rel-entr-conddual-Z},
  which proves~\eqref{eq:prop-coh-rel-entr-trivial-bounds-lower-bound}.
\end{proof}

\begin{proof}[*prop:coh-rel-entr-invar-under-isometries]
  This is clearly the case, because the semidefinite problem lies entirely within the
  support of the isometries.  Formally, any choice of variables for the original problem
  can be mapped in the new spaces through these partial isometries, and vice versa, and
  the attained values remain the same.  Hence the optimal value of the problem is also the
  same.
\end{proof}

\begin{proof}[*prop:coh-rel-entr-scaling-Gamma]
  Consider the optimal primal candidiates $T_{X'R_XE}$ and $\alpha$ for the
  problem defining
  $2^{-\DCohz[\epsilon]{\rho}{X}{X'}{\Gamma_X}{\Gamma_{X'}}}$. Then $T_{X'R_XE}$
  and $a b^{-1}\alpha$ are feasible primal candidates for the semidefinite
  program with the scaled $\Gamma$ operators.  Hence
  \begin{align}
    2^{-\DCohz[\epsilon]{\rho}{X}{X'}{a\Gamma_X}{b\Gamma_{X'}}}
    \leqslant \frac{a}{b}\alpha
    = \frac{a}{b}\,2^{-\DCohz[\epsilon]\rho{X}{X'}{\Gamma_X}{\Gamma_{X'}}}\ .
  \end{align}
  The opposite direction follows by applying the same argument to the reverse
  situation with $\Gamma_X\to a^{-1}\Gamma_X$,
  $\Gamma_{X'}\to b^{-1}\Gamma_{X'}$.
\end{proof}

\begin{proof}[*prop:coh-rel-entr-superadditive-tensor-products]
  Let $T_{X_1'R_{X_1}E_1},\alpha_1$ and $T_{X_2'R_{X_2}E_2},\alpha_2$ be the
  optimal choice of primal variables for
  $2^{-\DCohz[\epsilon]{\rho}{X_1}{X_1'}{\Gamma_{X_1}}{\Gamma_{X_1'}}}$ and
  $2^{-\DCohz[\epsilon']{\zeta}{X_2}{X_2'}{\Gamma_{X_2}}{\Gamma_{X_2'}}}$,
  respectively.
  Now, let
  $\bar T_{X_1'X_2'R_{X_1}R_{X_2}E_1E_2} = T_{X_1'R_{X_1}E_1}\otimes
  T_{X_2'R_{X_2}E_2}$ and $\bar\alpha=\alpha_1\alpha_2$.  Then
  \begin{align}
    \tr_{R_{X_1}R_{X_2}}\bigl[
      \bar T_{X_1'X_2'R_{X_1}R_{X_2}} \Gamma_{R_{X_1}}\otimes\Gamma_{R_{X_2}}
    \bigr]
    &\leqslant \alpha_1\alpha_2\, \Gamma_{X_1'}\otimes\Gamma_{X_2'}\ ;
    \\
    \tr_{X_1'X_2'}\bigl[ \bar T_{X_1'X_2'R_{X_1}R_{X_2}} \bigr]
    &\leqslant \Ident_{R_{X_1}}\otimes\Ident_{R_{X_2}}\ ,
  \end{align}
  and
  \begin{multline}
    \tr\bigl[
    `(\rho_{R_{X_1}}^{1/2}\otimes\zeta_{R_{X_2}}^{1/2})
    \bar T_{X_1'X_2'R_{X_1}R_{X_2}E_1E_2}
    `(\rho_{R_{X_1}}^{1/2}\otimes\zeta_{R_{X_2}}^{1/2})
    \;\cdot \\
    \rho_{X_1'R_{X_1}E_1}\otimes\zeta_{X_2'R_{X_2}E_2}
    \bigr]
    \geqslant `(1 - \epsilon^2)`(1 - \epsilon'^2) \geqslant 1 - \epsilon''^2\ ,
  \end{multline}
  and hence this choice of variables is feasible for the tensor product problem.
  We then have
  \begin{multline}
    2^{-\DCohz[\epsilon'']{*\rho_{X_1'R_{X_1}}\otimes
        \zeta_{X_2'R_{X_2}}}{X_1X_2}{X_1'X_2'}%
       {\Gamma_{X_1}\otimes\Gamma_{X_2}}{\Gamma_{X_1'}\otimes\Gamma_{X_2'}}}
    \leqslant  \alpha_1\alpha_2
    \\
    = 2^{-`\big[\DCohz[\epsilon]{\rho}{X_1}{X_1'}{\Gamma_{X_1}}{\Gamma_{X_1'}}
                +\DCohz[\epsilon']{\zeta}{X_2}{X_2'}{\Gamma_{X_2}}{\Gamma_{X_2'}}]}\ .
            \tag*\qedhere
  \end{multline}
\end{proof}

\begin{proof}[*prop:coh-rel-entr-restrict-Gamma-eigenspaces]
  Let $T_{X'R_XE}$ and $\alpha$ be the optimal feasible candidates for the primal
  semidefinite problem defining
  $2^{-\DCohz[\epsilon]\rho{X}{X'}{\Gamma_X}{\Gamma_{X'}}}$.  Let
  $T'_{X'R_XE} = `(P'_{X'}\otimes P_{R_X})\, T_{X'RE} \, `(P'_{X'}\otimes
  P_{R_X})$ and $\alpha'=\alpha$, writing $P_{R_X} = t_{X\to R_X}`(P_X)$. Then
  \begin{multline}
    \tr_{X'} T'_{X'R_X}
    = P_{R_X} \tr_{X'}\left[P'_{X'} T_{X'R_X}\right] P_{R_X}
    \leqslant P_{R_X} \tr_{X'}\left(T_{X'R_X}\right) P_{R_X}
    \\
    \leqslant P_{R_X}
    \leqslant \Ident_{R_X}\ ,
  \end{multline}
  satisfying~\eqref{eq:SDP-coh-rel-entr-cond-trnoninc}, and
  \begin{multline}
    \tr`[ \rho_{R_X}^{1/2}T'_{X'R_XE}\rho_{R_X}^{1/2} \, \rho_{X'R_XE} ]
    \\    
    = \tr`[ \rho_{R_X}^{1/2}T_{X'R_XE}\rho_{R_X}^{1/2} \, \rho_{X'R_XE} ]
    \geqslant 1-\epsilon^2\ ,
  \end{multline}
  where the first equality holds because $\rho_{R_X}$ and $\rho_{X'R_XE}$
  already lie within the support of $P_{R_X}$ and
  $P'_{X'}\otimes P_{R_X}\otimes\Ident_E$, respectively, and hence those
  projectors have no effect.
  Hence~\eqref{eq:SDP-coh-rel-entr-cond-processmatrixclose} is fulfilled.  Now
  we have
  \begin{align}
    \tr_{R_X}`[T'_{X'R_X}\Gamma'_{R_X}]
    &= \tr_{R_X}`[
    `(P'_{X'}\otimes P_{R_X}) T_{X'R_X} `(P'_{X'}\otimes P_{R_X})\,
    \Gamma_{R_X} ]
    \nonumber\\
    &\leqslant P'_{X'} \, \tr_{R_X} `[ T_{X'R_X} \Gamma_{R_X} ] P'_{X'}
    \nonumber\\
    &\leqslant P'_{X'}\, `(\alpha\Gamma_{X'})\, P'_{X'}
    = \alpha' \Gamma'_{X'}\ ,
  \end{align}
  using the fact that $\Gamma'_{R_X}\leqslant\Gamma_{R_X}$ (because
  $[P_{R_X},\Gamma_{R_X}]=0$).  Hence
  \begin{align}
    2^{-\DCohz[\epsilon]\rho{X}{X'}{\Gamma'_X}{\Gamma'_{X'}}}
    \leqslant 2^{-\DCohz[\epsilon]\rho{X}{X'}{\Gamma_X}{\Gamma_{X'}}} \ .
  \end{align}
  
  Let $\mu$, $X_{R_X}$ and $\omega_{X'}$ be any dual feasible candidates for
  $2^{-\DCohz[\epsilon]\rho{X}{X'}{\Gamma_X}{\Gamma_{X'}}}$.  Now let
  $\mu'=\mu$, $X'_{R_X} = P_{R_X}\,X_{R_X}\,P_{R_X}$ and
  $\omega_{X'} = P'_{X'}\,\omega'_{X'}\,P'_{X'}$.  Then
  $\tr`(\omega'_{X'}\Gamma'_{X'}) = \tr`(\omega_{X'} \Gamma'_{X'}) \leqslant
  \tr`(\omega_{X'}\Gamma_{X'}) \leqslant 1$ (using the fact that
  $\Gamma'_{X'}\leqslant\Gamma_{X'}$ since $[\Gamma_{X'},P'_{X'}]=0$), in
  accordance with~\eqref{eq:SDP-coh-rel-entr-conddual-tromegaGamma}.  Also,
  apply $`(P'_{X'}\otimes P_{R_X})`(\cdot) `(P'_{X'}\otimes P_{R_X})$ onto the
  dual constraint~\eqref{eq:SDP-coh-rel-entr-conddual-Z} to immediately see
  that $\mu'$, $\omega'_{X'}$ and $X_{R_X}$ obey the new constraint with
  $\Gamma'_{R_X}$.  Finally, the attained dual value is
  \begin{align}
    \mu'\,`(1-\epsilon^2) - \tr`(X'_{R_X})
    \geqslant \mu\,`(1-\epsilon^2) - \tr`(X_{R_X})\ .
  \end{align}
  Hence, we now have
  \begin{align}
    2^{-\DCohz[\epsilon]\rho{X}{X'}{\Gamma'_X}{\Gamma'_{X'}}}
    \geqslant 2^{-\DCohz[\epsilon]\rho{X}{X'}{\Gamma_X}{\Gamma_{X'}}} \ ,
  \end{align}
  which completes the proof.
\end{proof}

\begin{proof}[*prop:coh-rel-entr-relate-Gamma-operator-less-than]
  Let $T_{X'R_XE}$ and $\alpha$ be the optimal solution to the semidefinite
  program for $2^{-\DCohz[\epsilon]\rho{X}{X'}{\Gamma_X}{\Gamma_{X'}}}$.  They
  are then also feasible candidates for the semidefinite program for
  $2^{-\DCohz[\epsilon]\rho{X}{X'}{\tilde\Gamma_X}{\tilde\Gamma_{X'}}}$, because
  the only condition that changes
  is~\eqref{eq:SDP-coh-rel-entr-cond-alphaSGPM}, which is obviously still
  satisfied.
\end{proof}

\begin{proof}[*prop:coh-rel-entr-rewrite-primal-problem-infinity-norm]
  Let $T_{X'R_X}$ be any candidate in the primal problem. If
  $\tr_R`*(T_{X'R_X})$ does not lie within the support of $\Gamma_{X'}$, then
  condition~\eqref{eq:SDP-coh-rel-entr-cond-alphaSGPM} is not satisfied and the
  candidate is not primal feasible; we can hence ignore it in the minimization.
  Otherwise, by conjugating
  condition~\eqref{eq:SDP-coh-rel-entr-cond-alphaSGPM} by $\Gamma_{X'}^{-1/2}$,
  we see that~\eqref{eq:SDP-coh-rel-entr-cond-alphaSGPM} is equivalent to
  \begin{align}
    \Gamma_{X'}^{-1/2}\tr_{R_X}\left[T_{X'R_X}\,t_{X\to R_X}`(\Gamma_X)\right]
    \Gamma_{X'}^{-1/2}
    \leqslant \alpha\, \Pi^{\Gamma_{X'}}_{X'}\ ,
    \label{eq:prop-coh-rel-entr-rewrite-primal-problem-infinity-norm-calc-1}
  \end{align}
  which in turn is equivalent to
  \begin{align}
    \Gamma_{X'}^{-1/2}\tr_{R_X}\left[T_{X'R_X}\,t_{X\to R_X}`(\Gamma_R)\right]
    \Gamma_{X'}^{-1/2}
    \leqslant \alpha\, \Ident\ ,
  \end{align}
  because the left hand side
  of~\eqref{eq:prop-coh-rel-entr-rewrite-primal-problem-infinity-norm-calc-1}
  is entirely within the support of its right hand side.  Now, the optimal
  $\alpha$ which corresponds to this fixed $T_{X'R_X}$ is given by
  $\norm{\Gamma_{X'}^{-1/2}\tr_{R_X}\left[T_{X'R_X}\,t_{X\to R_X}`(\Gamma_X)\right]
    \Gamma_{X'}^{-1/2}}_\infty$.
  This chain of equivalences may be followed in reverse order, establishing the
  equivalence of the minimization problems.

  The formulation in terms of channels follows immediately from the translation of one
  formalism to the other.
\end{proof}

\begin{proof}[*prop:coh-rel-entr-nonsmooth-SDP]
  In the case $\epsilon=0$, the conditions in~\eqref{eq:coh-rel-entr-def}
  reduce to
  \begin{align*}
    \mathcal{T}_{X\to X'}`(\Gamma_X) &\leqslant 2^{-y}\Gamma_{X'}\ ; \\
    \mathcal{T}_{X\leftarrow X'}^\dagger`(\Ident_{X'}) &\leqslant \Ident_X\ ; \\
    \mathcal{T}_{X\to X'}`(\sigma_{XR_X}) &= \rho_{X'R_X}\ ,
  \end{align*}
  where we write $\ket\sigma_{XR_X} = \rho_{R_X}^{1/2}\,\ket\Phi_{X:R_X}$.
  These conditions, when written in terms of the Choi matrix $T_{X'R_X}$
  corresponding to $\mathcal{T}_{X\to X'}$, yield precisely the semidefinite
  program given in the claim.
\end{proof}

\subsection{Some special cases}

In this section, we look at some instructive special cases where the coherent
relative entropy can be evaluated exactly.

The first proposition concerns identity mappings. It is a property that one
would expect very naturally: If the process matrix corresponds to the identity
mapping on the support of the input, and if the $\Gamma$ operators coincide,
then the process should be a free operation and should not require a battery.
This property may seem like a triviality, but it is in fact not so obvious to
prove: Indeed, because the coherent relative entropy is a function of the
process matrix only, the implementation can choose to implement whatever process
it likes on the complement of the support of the input state.  In other words,
this proposition tells us that there is no way to extract work by exploiting
the freedom on this complementary subspace when performing the identity map on
the support of $\sigma_X$.

\begin{proposition}[Identity mapping]
  \label{prop:coh-rel-entr-identity-mapping}
  Let $\IdentProc[X][X']{}$ be the identity map from a system $X$ to a system
  $X'\simeq X$. Assume that $\Gamma_{X'}=\IdentProc[X][X']{\Gamma_X}$. Let
  $\sigma_{X}$ be any state on $X$, let $R_X\simeq X$ and
  $\ket\sigma_{XR_X} = \sigma_X^{1/2}\,\ket\Phi_{X:R_X}$, and let
  $\ket\rho_{X'R_X}$ be the process matrix of the identity process applied on
  $\sigma_X$, i.e.\@ $\rho_{X'R_X} = \IdentProc[X][X']{\sigma_{XR_X}}$.  Then
  \begin{align}
    \DCohz{\rho}{X}{X'}{\Gamma_X}{\Gamma_{X'}} = 0\ .
  \end{align}
\end{proposition}
\begin{proof}[*prop:coh-rel-entr-identity-mapping]
  Let $\Phi_{X'R_X} = \IdentProc[X][X']{\Phi_{X:R_X}}$ be the unnormalized
  maximally entangled state on $X'$ and $R_X$ such that
  $\rho_{X'R_X}=\rho_{R_X}^{1/2}\Phi_{X'R_X}\rho_{R_X}^{1/2}$.
  
  First we show that $\DCohz{\rho}{X}{X'}{\Gamma_X}{\Gamma_{X'}} \geqslant
  0$. Consider the mapping $\mathcal{T}_{X\to X'}=\IdentProc[X][X']{}$ and $y=0$,
  i.e., consider the identity mapping as an implementation candidate.  This
  clearly satisfies the requirements of the maximization
  in~\eqref{eq:coh-rel-entr-def} for $\epsilon=0$, and thus
  \begin{align}
    \DCohz{\rho}{X}{X'}{\Gamma_X}{\Gamma_{X'}} \geqslant 0\ .
  \end{align}

  We prove the reverse direction by exhibiting dual candidates for the problem
  given in \autoref{prop:coh-rel-entr-nonsmooth-SDP}.  The tricky part is that
  there might not be an optimal choice of dual variables. The best we can do in
  general is to come up with a sequence of choices for dual candidates whose
  attained value converges to $1$. For any $\mu>0$, let
  \begin{align}
    Z_{X'R_X} &= \mu\rho_{R_X}^{-1/2}\Phi_{X'R_X}\rho_{R_X}^{-1/2}\ ;
    & \omega_{X'} = \left(\tr\left[\Pi^{\rho_{X'}}_{X'}\Gamma_{X'}\right]\right)^{-1}
      \, \Pi^{\rho_{X'}}_{X'}\ .
  \end{align}
  Then $\tr\left(\omega_{X'}\Gamma_{X'}\right)=1$, satisfying the dual
  constraint~\eqref{eq:SDP-coh-rel-entr-nonsmooth-conddual-tromegaGamma}. Let's
  now study~\eqref{eq:SDP-coh-rel-entr-nonsmooth-conddual-Z}:
  \begin{multline}
    \rho_{R_X}^{1/2}Z_{X'R_X}\rho_{R_X}^{1/2}
    - \Gamma_{R_X}\otimes\omega_{X'}
    \\
    = \mu\Pi^{\rho_{R_X}}_{R_X}\,\Phi_{X':R_X}\,\Pi^{\rho_{R_X}}_{R_X}
    - \left(\tr\left[\Pi^{\rho_{X'}}_{X'}\Gamma_{X'}\right]\right)^{-1}
    \Gamma_{R_X}\otimes\Pi^{\rho_{X'}}_{X'}\ .
    \label{eq:prop-coh-rel-entr-identity-mapping-calc-1}
  \end{multline}
  The operator $\Pi^{\rho_{R_X}}_{R_X}\Phi_{X'R_X}\Pi^{\rho_{R_X}}_{R_X}$ is a
  rank-1 positive operator with support within
  $\Pi^{\rho_{R_X}}_{R_X}\otimes\Pi^{\rho_{X'}}_{X'}$, and its nonzero
  eigenvalue is given by
  \begin{align}
    \tr\left(\Pi^{\rho_{R_X}}_{R_X}\Phi_{X'R}\Pi^{\rho_{R_X}}_{R_X}\right)
    = \rank\rho_{R_X}\ .
    \label{eq:prop-coh-rel-entr-identity-mapping-calc-tr-PiPhiPi-rankRho}
  \end{align}
  Let $r=\rank\rho_{R_X}$. We then have
  $\Pi^{\rho_{R_X}}_{R_X}\Phi_{X'R}\Pi^{\rho_{R_X}}_{R_X}\leqslant
  r\Pi^{\rho_{R_X}}_{R_X}\otimes\Pi^{\rho_{X'}}_{X'}$ and we may continue our
  calculation:
  \begin{align}
    \text{\eqref{eq:prop-coh-rel-entr-identity-mapping-calc-1}}
    &\leqslant \left(\mu r \Pi^{\rho_{R_X}}_{R_X} -
      \left(\tr\left[\Pi^{\rho_{X'}}_{X'}\Gamma_{X'}\right]\right)^{-1}\Gamma_{R_X} \right)
      \otimes \Pi^{\rho_{X'}}_{X'}\ .
      \label{eq:prop-coh-rel-entr-identity-mapping-calc-2}
  \end{align}
  Now, let $P_{R_X}$ be the projector onto the eigenspaces associated to the
  positive (or null) eigenvalues of the operator
  $\left(\mu r \Pi^{\rho_{R_X}}_{R_X} -
    \left(\tr\left[\Pi^{\rho_{X'}}_{X'}\Gamma_{X'}\right]\right)^{-1}\Gamma_{R_X}
  \right)$, and let
  \begin{align}
    X_{R_X} = P_{R_X} \left(
    \mu r \Pi^{\rho_{R_X}}_{R_X} -
    \left(\tr\left[\Pi^{\rho_{X'}}_{X'}\Gamma_{X'}\right]\right)^{-1} \Gamma_{R_X}
    \right) P_{R_X} \ .
  \end{align}
  Then
  \begin{align}
    \text{\eqref{eq:prop-coh-rel-entr-identity-mapping-calc-2}}
    \leqslant X_{R_X}\otimes\Ident_{X'}\ .
  \end{align}
  Hence, for any $\mu>0$, this choice of dual variables satisfies the dual
  constraints.  The value attained by this choice of variables is given by
  \begin{align}
    \tr\left[Z_{X'R}\rho_{X'R}\right] - \tr X_{R_X}
    = \mu\tr\left[\Pi^{\rho_{R_X}}_{R_X}\Phi_{X'R}\Pi^{\rho_{R_X}}_{R_X}\Phi_{X'R}\right]
    - \tr X_{R_X}\ .
    \label{eq:prop-coh-rel-entr-identity-mapping-calc-3}
  \end{align}
  As the object $\Pi^{\rho_{R_X}}_{R_X}\Phi_{X'R}\Pi^{\rho_{R_X}}_{R_X}$ is
  rank-$1$, we have thanks
  to~\eqref{eq:prop-coh-rel-entr-identity-mapping-calc-tr-PiPhiPi-rankRho} that
  $\tr\left[\left(\Pi^{\rho_{R_X}}_{R_X}\Phi_{X'R}\Pi^{\rho_{R_X}}_{R_X}\right)^2\right]
  = \left(\tr\Pi^{\rho_{R_X}}_{R_X}\Phi_{X'R}\Pi^{\rho_{R_X}}_{R_X}\right)^2 =
  r^2$. Then
  \begin{align}
    \text{\eqref{eq:prop-coh-rel-entr-identity-mapping-calc-3}}
    &= \mu r^2 - \tr X_{R_X}
      \nonumber\\
    &= \mu r^2 - \mu r \tr\left(P_{R_X}\Pi^{\rho_{R_X}}_{R_X}\right)
      + \left(\tr\left[\Pi^{\rho_{X'}}_{X'}\Gamma_{X'}\right]\right)^{-1}
      \tr\left(P_{R_X}\Gamma_{R_X}\right)
      \nonumber\\
    &\geqslant \mu r^2 - \mu r \tr\left(\Pi^{\rho_{R_X}}_{R_X}\right)
      + \left(\tr\left[\Pi^{\rho_{X'}}_{X'}\Gamma_{X'}\right]\right)^{-1}
      \tr\left(P_{R_X}\Gamma_{R_X}\right)
      \nonumber\\
    &\geqslant
      \left(\tr\left[\Pi^{\rho_{X'}}_{X'}\Gamma_{X'}\right]\right)^{-1}
      \tr\left(P_{R_X}\Gamma_{R_X}\right)\ ,
      \label{eq:prop-coh-rel-entr-identity-mapping-calc-4}
  \end{align}
  recalling that $\tr\Pi^{\rho_{R_X}}_{R_X}=\rank\rho_{R_X}=r$.

  Next episode: the Lemma awakens. Take $A=\mu r\Pi^{\rho_{R_X}}_{R_X}$ and
  $B=\left(\tr\left[\Pi^{\rho_{X'}}_{X'}\Gamma_{X'}\right]\right)^{-1}\Gamma_{R_X}$;
  \autoref{util:blow-muA-operator-to-try-cover-B-limit} then asserts that there
  exists a constant $c$ independent of $\mu$ such that
  \begin{align}
    \Pi^{\rho_{R_X}}_{R_X} \leqslant P_{R_X} + \frac{c}{\mu}\Ident\ .
  \end{align}
  Hence,
  \begin{align}
    \text{\eqref{eq:prop-coh-rel-entr-identity-mapping-calc-4}}
    \geqslant \left(\tr\left[\Pi^{\rho_{X'}}_{X'}\Gamma_{X'}\right]\right)^{-1}
    \left(\tr\left[\Pi^{\rho_{R_X}}_{R_X}\Gamma_{R_X}\right] - \frac{c}{\mu}\tr\Gamma_{R_X}\right)
    = 1 - O\left(1/\mu\right)\ .
  \end{align}
  Taking $\mu\to\infty$ yields successive feasible dual candidates with attained
  objective value converging to $1$, hence proving that
  \begin{align}
    \DCohz{\rho}{X}{X'}{\Gamma_{X}}{\Gamma_{X'}} \leqslant 0\ .
    \tag*\qedhere
  \end{align}
\end{proof}

An essentially trivial proposition immediately follows from the fact that
$\Gamma$-sub-preserving maps are admissible operations, and hence don't cost anything in
our framework:

\begin{proposition}
  \label{prop:coh-rel-entr-Gamma-subpreserving-map}
  Let $\sigma_X$ be a quantum state and let $\mathcal{E}_{X\to X'}$ be a
  $\Gamma$-sub-preserving logical process.  With the process matrix
  $\rho_{X'R} =
  \mathcal{E}_{X\to{}X'}`\big(\sigma_X^{1/2}\,\Phi_{X:R_X}\,\sigma_X^{1/2})$, we
  have for any $\epsilon\geqslant 0$,
  \begin{align}
    \DCohz[\epsilon]{\rho}{X}{X'}{\Gamma_X}{\Gamma_{X'}} \geqslant 0\ .
  \end{align}
\end{proposition}
\begin{proof}[*prop:coh-rel-entr-Gamma-subpreserving-map]
  The process $\mathcal{E}_{X\to X'}$ itself is a valid optimization candidate
  in~\eqref{eq:prop-coh-rel-entr-rewrite-primal-problem-infinity-norm-mathcalT}, and
  clearly
  $\norm[\big]{\Gamma_{X'}^{-1/2} \, \mathcal{E}_{X\to{}X'}`\big(\Gamma_X) \,
    \Gamma_{X'}^{-1/2}}_\infty \leqslant
  \norm[\big]{\Pi^{\Gamma_{X'}}_{X'}}_\infty\leqslant 1$
  because $\mathcal{E}_{X\to X'}$ is $\Gamma$-sub-preserving.
\end{proof}

In general, the coherent relative entropy depends on the precise {logical
  process} used to map the input and output states.  However, there are some
classes of states for which the coherent relative entropy depends only on the
input and output state.

The following proposition tells us that one may map the $\Gamma_X/\tr\Gamma_X$
state to the $\Gamma_{X'}/\tr\Gamma_{X'}$ state in however way one wants, i.e.\@
regardless of the logical process, and yet in any case the coherent relative
entropy is given by the ratio $\tr\Gamma_{X'}/\tr\Gamma_X$.  This is a
consequence of allowing any $\Gamma$-preserving maps to be performed for free,
and this ratio comes about from the normalization of the respective input and
output states.

\begin{proposition}
  \label{prop:coh-rel-entr-mapping-gibbs-states-onto-each-other}
  Let $P_X$ and $P'_{X'}$ be projectors with $[P_X, \Gamma_X]=0$ and
  $[P'_{X'},\Gamma_{X'}]=0$.  Let $\rho_{X'R_X}$ be a bipartite quantum state
  with reduced states
  $\rho_{R_X}=t_{X\to R_X}`[`(P_X\Gamma_XP_X)/\tr`(P_X\Gamma_X)]$ and
  $\rho_{X'}=(P'_{X'}\Gamma_{X'}P'_{X'})/\tr`(P'_{X'}\Gamma_{X'})$.  Then, for
  any $\epsilon\geqslant 0$,
  \begin{multline}
    \DCohz[\epsilon]\rho{X}{X'}{\Gamma_X}{\Gamma_{X'}}
    \\
    = \log\tr`(P'_{X'}\Gamma_{X'}) - \log\tr`(P_X\Gamma_X) + \log`[1/`(1-\epsilon^2)]\ .
  \end{multline}
\end{proposition}
\begin{proof}[*prop:coh-rel-entr-mapping-gibbs-states-onto-each-other]
  Let $\ket\rho_{X'R_XE}$ be a purification of $\rho_{X'R_X}$ into a (large
  enough) system $E$, and consider the semidefinite program given by
  \autoref{prop:coh-rel-entr-SDP}.  We give feasible primal and dual candidates
  which achieve the same value.  First, let
  $T_{X'R_XE} =
  `(1-\epsilon^2)\,\rho_{R_X}^{-1/2}\,\rho_{X'R_XE}\,\rho_{R_X}^{-1/2}$.  We
  have
  $\tr_{X'E}`(T_{X'R_XE}) = `(1-\epsilon^2)\,\Pi^{\rho_{R_X}}_{R_X}\leqslant
  \Ident_{R_X}$ as required by~\eqref{eq:SDP-coh-rel-entr-cond-trnoninc}.  Also,
  since $\rho_{R_X} = P_{R_X}\Gamma_{R_X} P_{R_X}/\tr`(P_{R_X}\Gamma_{R_X})$ and
  $\rho_{X'} = P'_{X'}\Gamma_{X'} P'_{X'}/\tr`(P'_{X'}\Gamma_{X'})$, we have
  $\tr_{R_XE}`(T_{X'R_XE}\,\Gamma_{R_X}) =
  `(1-\epsilon^2)\,\tr`(P_{R_X}\Gamma_{R_X})\, \tr_{R_X}`(\rho_{X'R_X}\,P_X) =
  `(1-\epsilon^2)\,\tr`(P_{R_X}\Gamma_{R_X})\, \rho_{X'} \leqslant \alpha\,
  \Gamma_{X'}$, where we have defined
  $\alpha = `(1-\epsilon^2) \tr`(P_{R_X}\Gamma_{R_X}) /
  \tr`(P'_{X'}\Gamma_{X'})$ and noting that $[P'_{X'}, \Gamma_{X'}]=0$, hence
  satisfying~\eqref{eq:SDP-coh-rel-entr-cond-alphaSGPM}.  Finally, we have
  $\tr`\big[ \rho_{R_X}^{1/2}\,T_{X'R_XE}\,\rho_{R_X}^{1/2} \, \rho_{X'R_XE} ] =
  `(1-\epsilon^2)$ which
  satisfies~\eqref{eq:SDP-coh-rel-entr-cond-processmatrixclose}.  This choice of
  primal variables is feasible, and attains the value $\alpha$.

  Now we exhibit feasible dual candidates.  Let
  $\mu = \tr`(P_{R_X}\Gamma_{R_X}) / \tr`(P'_{X'}\Gamma_{X'})$,
  $\omega_{X'} = P'_{X'} / \tr`(P'_{X'}\,\Gamma_{X'})$ and $X_{R_X}=0$, and note
  that~\eqref{eq:SDP-coh-rel-entr-conddual-tromegaGamma} is automatically
  satisfied.  Then, since
  $\rho_{X'R_XE}\leqslant\Ident_E\otimes P'_{X'}\otimes P_{R_X}$, we have
  \begin{multline}
    \mu\,\rho_{R_X}^{1/2}\,\rho_{X'R_XE}\,\rho_{R_X}^{1/2}
    \leqslant \frac{\tr P_{R_X}\Gamma_{R_X}}{\tr P'_{X'}\Gamma_{X'}}
    \,\Ident_E\otimes P'_{X'}\otimes\rho_{R_X}
\\
    \leqslant \Ident_E\otimes\omega_{X'}\otimes\Gamma_{R_X}\ ,
  \end{multline}
  keeping in mind that $[P_{R_X},\Gamma_{R_X}]=0$, and hence
  condition~\eqref{eq:SDP-coh-rel-entr-conddual-Z} is satisfied.  The value
  attained by this choice of variables is simply
  $\mu\,(1-\epsilon^2) - \tr X_{R_X} = \alpha$, hence proving that this is the
  optimal solution of the semidefinite program.  Calculating $\,-\log\alpha$
  completes the proof.
\end{proof}

We note that for this special type of states we have the nice expression for their
relative entropy to $\Gamma$.  

\begin{proposition}
  \label{prop:coh-rel-entr-battery-state-calc-rel-entr}
  If $\Gamma\geqslant0$ and $P$ is a projector with $[P,\Gamma]=0$, then
  \begin{align}
    \DD`*{\frac{P\Gamma P}{\tr P\Gamma}}{\Gamma}
    &= \Dminz`*{\frac{P\Gamma P}{\tr P\Gamma}}{\Gamma}
    = \Dmax`*{\frac{P\Gamma P}{\tr P\Gamma}}{\Gamma}
    \nonumber\\[1ex]
    &= -\log\tr P\Gamma\ .
    \label{eq:coh-rel-entr-expression-relative-entropy-for-projected-Gamma-states}
  \end{align}
\end{proposition}
\begin{proof}[*prop:coh-rel-entr-battery-state-calc-rel-entr]
  Write as shorthand $\rho=P\Gamma P / \tr P\Gamma$.  Then
  \begin{align}
    2^{\Dmax{\rho}{\Gamma}}
    &= \norm{\Gamma^{-1/2}\,\rho\,\Gamma^{-1/2}}_\infty
      \nonumber\\
    &= (\tr P\Gamma)^{-1}\norm{\Gamma^{-1/2}\,P\,\Gamma\,P\,\Gamma^{-1/2}}_\infty
      \nonumber\\
    &= (\tr P\Gamma)^{-1}\norm{\Gamma^{-1/2}\,\Gamma^{1/2}\,P\,\Gamma^{1/2}\,\Gamma^{-1/2}}_\infty
      \nonumber\\
    &= (\tr P\Gamma)^{-1}\ ,
  \end{align}
  since $[P,\Gamma]=0$.  Also, observing that $\Pi^\rho = P$,
  \begin{align}
    2^{-\Dminz{\rho}{\Gamma}}
    &= \tr`(\Pi^\rho\Gamma) = \tr`(P\Gamma)\ .
  \end{align}
  The expression $\DD{\rho}{\Gamma}$ is thus also equal to $-\log\tr P\Gamma$
  since we know that
  $\Dminz{\rho}{\Gamma} \leqslant \DD{\rho}{\Gamma} \leqslant
  \Dmax{\rho}{\Gamma}$~\cite[Lemma~10]{Datta2009IEEE_minmax}.
\end{proof}

Notably, the states of the form $P\Gamma P/\tr(P\Gamma)$ for $[P,\Gamma]=0$ are precisely
those general type of states which we allowed on battery systems in
\autoref{item:prop-equiv-battery-models-equivstatement-projGamma} of
\autoref{prop:equiv-battery-models}.

In fact, we may prove a slightly more general version of
\autoref{prop:coh-rel-entr-mapping-gibbs-states-onto-each-other} for the case
$\epsilon=0$: it suffices that the reduced state on the input is of the form
$\Gamma_X/\tr\Gamma_X$, and then the coherent relative entropy is oblivious to
any correlation between input and output, or equivalently, to which process is
exactly implemented, and depends only on the reduced states on the input and the
output.

\begin{proposition}
  \label{prop:coh-rel-entr-mapping-Gibbs-to-arb}
  Let $\rho_{X'R_X}$ such that
  $\tr_{X'}\rho_{X'R_X} = \Gamma_{R_X}/\tr\Gamma_{R_X}$.  Then
  \begin{align}
    \DCohz\rho{X}{X'}{\Gamma_X}{\Gamma_{X'}}
    = - \log\tr\Gamma_X -\Dmax{\rho_{X'}}{\Gamma_{X'}} \ .
    \label{eq:prop-coh-rel-entr-mapping-Gibbs-to-arb-expression}
  \end{align}
\end{proposition}
\begin{proof}[*prop:coh-rel-entr-mapping-Gibbs-to-arb]
  Take any $T_{X'R_X}$ satisfying
  $\rho_{R_X}^{1/2}T_{X'R_X}\rho_{R_X}^{1/2}=\rho_{X'R_X}$ and
  $\tr_{X'}T_{X'R_X}\leqslant\Ident_{R_X}$. Then since
  $\tr`(\Gamma_{R_X})\,\rho_{R_X} = \Gamma_{R_X}$, we have
  \begin{align}
    \tr_{R_X}`(T_{X'R_X}\Gamma_{R_X})
    &= \tr`(\Gamma_{R_X})\tr_{R_X}`(\rho_{R_X}^{1/2} T_{X'R_X}\rho_{R_X}^{1/2})
    \nonumber\\
    &= \tr`(\Gamma_{R_X}) \tr_{R_X}`(\rho_{X'R_X})
    = \tr\left(\Gamma_{R_X}\right) \rho_{X'}
    \ ,
  \end{align}
  and thus
  \begin{align}
    -\log\,\norm[\big]{
    \Gamma_{X'}^{-1/2} \tr_{R_X}`[T_{X'R_X}\Gamma_{R_X}]\, \Gamma_{X'}^{-1/2}
    }_\infty
    &= -\log\,\norm[\big]{
      \Gamma_{X'}^{-1/2}\rho_{X'}\Gamma_{X'}^{-1/2}
      }_\infty
    \nonumber \\
    &= -\Dmax{\rho_{X'}}{\Gamma_{X'}}\ .
  \end{align}
  This argument holds in particular for the optimal such $T_{X'R_X}$.
\end{proof}

Remarkably, if $\tr_{R_X}\rho_{X'R_X}= \Gamma_{X'}/\tr\Gamma_{X'}$, the coherent
relative entropy may still depend on the exact process, and does not necessarily
reduce to a difference of input and output terms as
in~\eqref{eq:prop-coh-rel-entr-mapping-Gibbs-to-arb-expression}.  This can be
seen by considering the unitary process $\mathcal{U}$ which swaps two levels
$\ket0,\ket1$, choosing $\Gamma=g_0\proj0 + g_1\proj1$ (with $g_0+g_1=1$ and
$g_0>g_1$) for both input and output, and using the input state
$\sigma=g_1\proj0+g_0\proj1$: in this case, $\sigma$ maps to $\Gamma$, but
$-\log\,\norm{\Gamma^{-1/2}\,\mathcal{U}(\Gamma)\,\Gamma^{-1/2}}_\infty =
-\Dmax{\sigma}{\Gamma}$ whereas there are processes which map $\sigma$ to
$\Gamma$, such as $\mathcal{T}`(\cdot) = \tr(\Pi^{\sigma}(\cdot))\,\Gamma$,
which achieve a coherent relative entropy of $\Dminz{\sigma}{\Gamma}$.

\subsection{Data processing inequality}

The data processing inequality is an important property desirable for an information
measure.  Intuitively, it asserts that processing information cannot make it more
\qq{valuable}.

In our case, the data processing inequality asserts that post-processing, or
applying a map to both the output state and output $\Gamma$, may only increase
the coherent relative entropy.

\begin{proposition}[Data processing inequality]
  \label{prop:coh-rel-entr-data-processing-inequality}
  Let $\rho_{X'R_X}$ be a quantum state and let
  $\Gamma_X,\Gamma_{X'}\geqslant 0$.  Let $\mathcal{F}_{X'\to X''}$ be a
  trace-preserving, completely positive map.
  Then, for any $\epsilon\geqslant 0$,
  \begin{multline}
    \DCohz[\epsilon]{\rho}{X}{X'}{\Gamma_X}{\Gamma_{X'}}
    \\
    \leqslant 
    \DCohz[\epsilon]{*\mathcal{F}_{X'\to X''}`(\rho_{X'R_X})}%
    {X}{X''}{\Gamma_{X}}{\mathcal{F}_{X'\to X''}`(\Gamma_{X'})}
    \ .
    \label{eq:prop-coh-rel-entr-data-processing-inequality-postproc}
  \end{multline}
\end{proposition}
\begin{proof}[*prop:coh-rel-entr-data-processing-inequality]
  Let $\mathcal{T}_{X\to{}X'},y$ be optimal candidates for the optimization
  defining $2^{-\DCohz[\epsilon]\rho{X}{X'}{\Gamma_X}{\Gamma_{X'}}}$
  in~\eqref{eq:coh-rel-entr-def}.  We construct an optimization candidate for
  the coherent relative entropy of the post-processed state.
  Let
  $\mathcal{T}'_{X\to X''} =
  \mathcal{F}_{X'\to{}X''}\circ\mathcal{T}_{X\to{}X'}$.  Then
  $\mathcal{T}'^\dagger_{X\leftarrow{}X''}`(\Ident_{X''}) =
  \mathcal{T}^\dagger_{X\leftarrow{}X'}`(
  \mathcal{F}^\dagger_{X'\leftarrow{}X''}(\Ident_{X''}) ) \leqslant
  \Ident_{R_X}$ because $\mathcal{F}_{X'\to X''}$ is trace-preserving. Also,
  $\mathcal{T}'_{X\to{}X''}`(\Gamma_{X}) \leqslant
  \alpha\,\mathcal{F}_{X'\to{}X''}`( \Gamma_{X'} )$.  Finally, writing
  $\ket{\sigma}_{XR_X} = \rho_{R_X}^{1/2}\,\ket\Phi_{X:R_X}$, we have
  $P`( \mathcal{T}'_{X\to{}X''}`(\sigma_{XR}), \mathcal{F}_{X'\to{}X''}`(
  \rho_{X'R_X} ) ) \leqslant P`( \mathcal{T}_{X\to{}X'}`(\sigma_{XR}),
  \rho_{X'R_X} ) \leqslant \epsilon$.
\end{proof}

The case of pre-processing, i.e.\@ when a map is applied to the input before the
actual mapping is carried out, is less clear how to formulate.  Indeed, the
expression
$\DCohz[\epsilon]{*\mathcal{F}_{R_X\to{}R_{\tilde X}}`(\rho_{X'R_X})}%
{R_{\tilde{X}}}{X'}%
{\mathcal{F}_{X\to{}\tilde{X}}`(\Gamma_X)}{\Gamma_{X'}}$ would correspond to the
not-so-natural setting where one implements a process matrix defined by the
state resulting when two logical processes are applied on both the system $X$ of
interest and the reference system $R_X$ on a pure state $\ket\sigma_{XR_X}$.
However, a more general statement about composing processes can be derived in
the form of a chain rule, which is the subject of the next section.

\subsection{Chain rule}

If two individual processes are concatenated, what can be said of the coherent
relative entropy of the combined processes?  As one would expect, it turns out
that the optimal battery use of implementing directly a composition of logical
maps can only be better than the sum of the battery uses of the individual
realizations of each map.

\begin{proposition}[Chain rule]
  \label{prop:coh-rel-entr-a-chain-rule}
  Consider three systems $X, X', X''$ with corresponding
  $\Gamma_X, \Gamma_{X'}, \Gamma_{X''} \geqslant 0$, and let
  $R_X\simeq X, R_{X'}\simeq X'$.
  Let $\sigma_X$ be a quantum state.  Let $\mathcal{E}^{(1)}_{X\to X'}$ and
  $\mathcal{E}^{(2)}_{X'\to X''}$ be two completely positive, trace-nonincreasing maps
  such that
  $\tr`[\mathcal{E}^{(2)}_{X'\to X''}`(\mathcal{E}^{(1)}_{X\to X'}`(\sigma_X))]=1$.
  Let $\epsilon,\epsilon'\geqslant 0$.  Then:
  \begin{align}
    &\DCohz[\epsilon]{*\mathcal{E}^{(1)}_{X\to X'}`(\sigma_{XR_X})}%
      {X}{X'}{\Gamma_X}{\Gamma_{X'}}
      \nonumber\\
    &\hspace*{4em}
      + ~
      \DCohz[\epsilon']{*\mathcal{E}^{(2)}_{X'\to X''}`(\rho'_{X'R_X'})}%
        {X'}{X''}{\Gamma_{X'}}{\Gamma_{X''}}
      \nonumber\\
    &\hspace*{1em}
      \leqslant ~~
      \DCohz[\epsilon+\epsilon']%
        {*\mathcal{E}^{(2)}_{X'\to X''}`(\mathcal{E}^{(1)}_{X\to X'}`(\sigma_{XR_X}))}%
        {X}{X''}{\Gamma_{X}}{\Gamma_{X''}}\ ,
  \end{align}
  where $\ket\sigma_{XR_X} = \sigma_X^{1/2}\,\ket\Phi_{X:R_X}$ and
  $\ket{\rho'}_{X'R_{X'}} =
  `(\mathcal{E}^{(1)}_{X\to{}X'}`(\sigma_X))^{1/2}\,\ket\Phi_{X':R_{X'}}$.
\end{proposition}
\begin{proof}[*prop:coh-rel-entr-a-chain-rule]
  Let $\mathcal{T}^{(1)}_{X\to X'}$, $y_1$ be optimal choices
  in~\eqref{eq:coh-rel-entr-def} for
  $\DCohz[\epsilon]{*\mathcal{E}^{(1)}_{X\to{}X'}`(\sigma_{XR_X})}%
  {X}{X'}{\Gamma_X}{\Gamma_{X'}}$, and let $\mathcal{T}^{(2)}_{X\to X'}$, $y_2$
  be optimal choices for
  $\DCohz[\epsilon']{*\mathcal{E}^{(2)}_{X'\to{}X''}`(\rho'_{X'R_X'})}%
  {X'}{X''}{\Gamma_{X'}}{\Gamma_{X''}}$.
  Let $V_{X\to X'E}$ be a Stinespring dilation of $\mathcal{E}^{(1)}_{X\to X'}$,
  such that
  $\mathcal{E}^{(1)}_{X\to X'}`(\cdot) = \tr_E`[ V_{X\to X'E}\, `(\cdot)
  V^\dagger ]$.  Now, as two different purifications of
  $\mathcal{E}^{(1)}_{X\to X'}`(\sigma_{X})={\rho'}_{X'}$, there must exist a
  partial isometry $W_{R_{X'}\to R_XE}$ such that
  $V_{X\to X'E}\, \ket\sigma_{XR_X} = W_{R_{X'}\to R_XE}
  \ket{\rho'}_{X'R_{X'}}$.  Define
  $\mathcal{F}_{R_{X'}\to R_X}`(\cdot) =
  \tr_E`(W_{R_{X'}\to{}R_XE}\,`(\cdot)\,W^\dagger)$, and note that
  $\mathcal{E}^{(1)}_{X\to X'}`(\sigma_{XR_X}) =
  \mathcal{F}_{R_{X'}\to{}R_X}`(\rho'_{X'R_{X'}})$.
  Now, let
  $\mathcal{T}_{X\to X''} = \mathcal{T}^{(2)}_{X'\to X''} \circ
  \mathcal{T}^{(1)}_{X\to X'}$, and note that
  \begin{align}
    &P`\big[\mathcal{T}_{X\to X''}`(\sigma_{XR_X}),
      \mathcal{E}^{(2)}_{X'\to X''}`(\mathcal{E}^{(1)}_{X\to X'}`(\sigma_{XR_X}))]
      \nonumber\\
    &\hspace*{2em}
      \leqslant P`\big[
        \mathcal{T}^{(2)}_{X'\to X''}`( \mathcal{T}^{(1)}_{X\to X'}`(\sigma_{XR_X}) ),
        \mathcal{T}^{(2)}_{X'\to X''}`( \mathcal{E}^{(1)}_{X\to X'}`(\sigma_{XR_X}) ) ]
      \nonumber\\
    &\hspace*{3em}+ P`\big[
        \mathcal{T}^{(2)}_{X'\to X''}`( \mathcal{E}^{(1)}_{X\to X'}`(\sigma_{XR_X}) ) ,
        \mathcal{E}^{(2)}_{X'\to X''}`( \mathcal{E}^{(1)}_{X\to X'}`(\sigma_{XR_X}) )
      ]
      \nonumber\\
    &\hspace*{2em}
      \leqslant P`\big[
         \mathcal{T}^{(1)}_{X\to X'}`(\sigma_{XR_X}),
         \mathcal{E}^{(1)}_{X\to X'}`(\sigma_{XR_X}) ]
      \nonumber\\
    &\hspace*{3em}+ P`\big[
        \mathcal{T}^{(2)}_{X'\to X''}`( {\rho'}_{X'R_{X'}} ) ,
        \mathcal{E}^{(2)}_{X'\to X''}`( {\rho'}_{X'R_{X'}} )
      ]
      \nonumber\\
    &\hspace*{2em}\leqslant \epsilon + \epsilon'\ .
  \end{align}
  where in second inequality we have used twice the fact that the purified
  distance cannot decrease under application of a completely positive,
  trace-nonincreasing map, and that
  $\mathcal{E}^{(1)}_{X\to X'}`(\sigma_{XR_X}) =
  \mathcal{F}_{R_{X'}\to{}R_X}`({\rho'}_{X'R_{X'}})$.
  Observe finally that
  \begin{multline}
    \mathcal{T}_{X\to X''}(\Gamma_X) =
    \mathcal{T}^{(2)}_{X'\to X''}`(\mathcal{T}^{(1)}_{X\to X'}`(\Gamma_X))
    \leqslant 2^{-y_1}\,\mathcal{T}^{(2)}_{X'\to X''}`(\Gamma_{X'})
    \\
    \leqslant 2^{-y_1-y_2}\,\Gamma_{X''}\ ,
  \end{multline}
  proving that $\mathcal{T}_{X\to X''}$, $y=y_1+y_2$ are valid optimization
  candidates in~\eqref{eq:coh-rel-entr-def} for $\DCohz[\epsilon+\epsilon']%
  {*\mathcal{E}^{(2)}_{X'\to{}X''}`(\mathcal{E}^{(1)}_{X\to{}X'}`(\sigma_{XR_X}))}%
  {X}{X''}{\Gamma_{X}}{\Gamma_{X''}}$, proving the claim.
\end{proof}

\begin{corollary}[Chain rule in terms of states]
  \label{cor:coh-rel-entr-chain-rule-states}
  Consider systems $A,B,C$ and $R_A\simeq A$, $R_B\simeq B$.  Let
  $\Gamma_{C}\geqslant 0$, $\Gamma_{AB}\geqslant 0$ and write
  $\Gamma_A = \tr_B`[\Gamma_{AB}]$.  Let $\tau_{CR_AR_B}$ be any tripartite
  state. Then, for $\epsilon,\epsilon'\geqslant 0$,
  \begin{multline}
    \DCohz[\epsilon]{\rho}{A}{AB}{\Gamma_A}{\Gamma_{AB}}
    + \DCohz[\epsilon']{*\tau_{CR_AR_B}}{AB}{C}{\Gamma_{AB}}{\Gamma_{C}}
    \\
    \leqslant \DCohz[\epsilon+\epsilon']{\tau}{A}{C}{\Gamma_{A}}{\Gamma_{C}}\ ,
  \end{multline}
  where
  $\rho_{ABR_A} =
  \tr_{R_B}`[\tau_{R_AR_B}^{1/2}\,\Phi_{AB:R_AR_B}\,\tau_{R_AR_B}^{1/2}]$.
\end{corollary}
\begin{proof}[*cor:coh-rel-entr-chain-rule-states]
  Define systems $X = A$, $X'=AB$ and $X'' = C$.  Let
  \begin{subequations}
    \begin{align}
      \mathcal{E}^{(1)}_{X\to X'}`(\cdot)
      &=
        \tr_{R_A}`[\rho_{R_A}^{-1/2}\,\rho_{ABR_A}\,\rho_{R_A}^{-1/2}\,
        t_{A\to{}R_A}`(\cdot)]\ ;
        \\
      \mathcal{E}^{(2)}_{X'\to X''}(\cdot)
      &=
        \tr_{R_AR_B}`[\tau_{R_AR_B}^{-1/2}\,\tau_{CR_AR_B}\,\tau_{R_AR_B}^{-1/2}\,
        t_{AB\to R_AR_B}`(\cdot) ]\ .
    \end{align}
  \end{subequations}
  These mappings are trace nonincreasing.  Let
  $\sigma_X = t_{R_X\to X}`(\tau_{R_X}) = t_{R_X\to X}`(\rho_{R_X})$. We see
  that
  $\mathcal{E}^{(2)}_{X'\to X''}(\mathcal{E}^{(1)}_{X\to X'}`(\sigma_X)) =
  \mathcal{E}^{(2)}_{X'\to X''}(\rho_{AB}) = \mathcal{E}^{(2)}_{X'\to{}X''}(
  t_{R_AR_B\to AB}`(\tau_{R_AR_B}) ) = \tau_C$ which has unit trace as required.
  Furthermore, let
  $\ket\sigma_{XR_X} = \sigma_X^{1/2}\,\ket\Phi_{X:R_X} =
  \sigma_A^{1/2}\,\ket\Phi_{A:R_A}$ and
  $\ket{\rho'}_{X'R_{X'}} = (\rho_{AB}^{1/2})\,\ket\Phi_{AB:R_AR_B} =
  (\tau_{R_AR_B}^{1/2})\,\ket\Phi_{ABR_AR_B}$.  Now calculate
  \begin{multline}
    \mathcal{E}^{(1)}_{X\to X'}`(\sigma_{XR_X}) =
    \Pi^{\rho_{R_A}}_{R_A} \,
    \tr_{\tilde{R}_A}`[\rho_{AB\tilde{R}_A}\,t_{A\to\tilde{R}_A}`(\Phi_{A:R_A})]
    \, \Pi^{\rho_{R_A}}_{R_A}
    \\ = \rho_{ABR_A}\ ,
  \end{multline}
  as well as
  \begin{multline}
    \mathcal{E}^{(2)}_{X'\to X''}`(\rho'_{X'R_{X'}}) = \\
    \Pi^{\tau_{R_AR_B}}_{R_AR_B} \,
    \tr_{\tilde{R}_A\tilde{R}_B}`[ \tau_{C\tilde{R}_A\tilde{R}_B} \,
    t_{AB\to\tilde{R}_A\tilde{R}_B}`(\Phi_{AB:R_AR_B}) ] \,
    \Pi^{\tau_{R_AR_B}}_{R_AR_B}
    \\ = \tau_{CR_AR_B}\ ,
  \end{multline}
  and, since
  $\mathcal{E}^{(1)}_{X\to X'}`(\sigma_{XR_X}) = \rho_{ABR_A} =
  \tr_{R_B}`[\rho'_{ABR_AR_B}]$,
  \begin{align}
    &\mathcal{E}^{(2)}_{X'\to X''}`(\mathcal{E}^{(1)}_{X\to X'}`(\sigma_{XR_X}))
    \nonumber\\
    &\hspace*{3em}=
      \tr_{R_B}`\big[ \mathcal{E}^{(2)}_{X'\to X''}`(
      \rho'_{ABR_AR_B}
    ) ]
    \nonumber\\
    &\hspace*{3em}=
      \tau_{CR_A}\ .
  \end{align}
  All conditions for \autoref{prop:coh-rel-entr-a-chain-rule} are fulfilled, and the claim
  follows.
\end{proof}

\subsection{Alternative smoothing of the
  coherent relative entropy}

There is another possible way to define the smooth coherent relative entropy
(i.e., for $\epsilon>0$), based on optimizing its non-smooth version (for
$\epsilon=0$) over all states which are $\epsilon$-close to the requested state.
This smoothing method is the method used traditionally in the smooth entropy
framework~\cite{PhDRenner2005_SQKD,PhDTomamichel2012,%
  BookTomamichel2016_Finite}.  The disadvantage of this alternative definition
is that it can no longer be formulated as a semidefinite program.  However, in
the regime of small $\epsilon$, it turns out that both definitions are equivalent
up to factors which depend only on $\epsilon$, and which do not scale with the
dimension of the system (\autoref{prop:cohrelentr-2-equiv} below).  In
particular, both quantities behave in the same way in the i.i.d.\@ limit.

\begin{thmheading}{Alternative smoothing}
  For a normalized state $\rho_{X'R_X}$, positive semidefinite
  $\Gamma_X, \Gamma_{X'}$, and for $\epsilon\geqslant 0$, we define the quantity
  \begin{multline}
    \DCohx[\epsilon]{\rho}{X}{X'}{\Gamma_X}{\Gamma_{X'}}
    \\= \max_{\hat\rho_{X'R_X}\approx_\epsilon\rho_{X'R_X}}
      \DCohz{\hat\rho}{X}{X'}{\Gamma_X}{\Gamma_{X'}}\ ,
      \label{eq:coh-rel-entr-def-smooth-alt}
  \end{multline}
  where the maximization in~\eqref{eq:coh-rel-entr-def-smooth-alt} is taken
  over (normalized) quantum states which are in the support of
  $\Gamma_X\otimes\Gamma_{X'}$ and which are close to $\rho_{X'R_X}$ in the
  purified distance, $P`(\hat\rho_{X'R_X},\rho_{X'R_X}) \leqslant \epsilon$.
\end{thmheading}

Some properties of $\DCohz[\epsilon]{\rho}{X}{X'}{\Gamma_{X}}{\Gamma_{X'}}$
carry over immediately to
$\DCohx[\epsilon]{\rho}{X}{X'}{\Gamma_X}{\Gamma_{X'}}$, which we summarize here
without explicit proof.  These propositions are straightforwardly proven by
applying the relevant property to the inner coherent relative entropy
in~\eqref{eq:coh-rel-entr-def-smooth-alt}.

\begin{proposition}[cf.~\autoref{prop:coh-rel-entr-trivial-bounds}]
  \noproofref
  \label{prop:coh-rel-entr-smooth-alt-trivial-bounds}
  For any $0\leqslant\epsilon<1$,
  \begin{subequations}
    \label{eq:prop-coh-rel-entr-smooth-alt-trivial-bounds}
  \begin{align}
    \DCohx[\epsilon]{\rho}{X}{X'}{\Gamma_X}{\Gamma_{X'}}
    &\geqslant - \log \tr\Gamma_{X} - \log\, \norm{\Gamma_{X'}^{-1}}_\infty\ ;
    \label{eq:prop-coh-rel-entr-smooth-alt-trivial-bounds-upper-bound}\\
    \DCohx[\epsilon]{\rho}{X}{X'}{\Gamma_X}{\Gamma_{X'}}
    &\leqslant \log\,\norm{\Gamma_X^{-1}}_\infty + \log\tr\Gamma_{X'}\ .
    \label{eq:prop-coh-rel-entr-smooth-alt-trivial-bounds-lower-bound}
  \end{align}
  \end{subequations}
\end{proposition}

\begin{proposition}[cf.~\autoref{prop:coh-rel-entr-scaling-Gamma}]
  \noproofref
  \label{prop:coh-rel-entr-smooth-alt-scaling-Gamma}
  For any $a, b\geqslant 0$, 
  \begin{multline}
    \DCohx[\epsilon]\rho{X}{X'}{a\Gamma_X}{b\Gamma_{X'}}
    \\
    = \DCohx[\epsilon]\rho{X}{X'}{\Gamma_X}{\Gamma_{X'}} + \log \frac b a\ .
  \end{multline}
\end{proposition}

\begin{proposition}[cf.~\autoref{prop:coh-rel-entr-invar-under-isometries}]
  \noproofref
  \label{prop:coh-rel-entr-smooth-alt-invar-under-isometries}
  Let $\tilde X$, $\tilde X'$ be new systems.  Suppose there exist partial
  isometries $V_{X\to\tilde X}$ and $V'_{X'\to \tilde X'}$ such that both
  $t_{R_X\to X}`(\rho_{R_X})$ and $\Gamma_X$ are in the support of
  $V_{X\to\tilde X}$, and both $\rho_{X'}$ and $\Gamma_{X'}$ are in the support
  of $V'_{X'\to\tilde X'}$.  Then
  \begin{multline}
    \DCohx[\epsilon]{*`(V'\otimes V)\, \rho_{X'R_X}\, `(V'\otimes V)^\dagger}%
    {\tilde X}{\tilde X'}{V\Gamma_X V^\dagger}{V'\Gamma_{X'} V'^\dagger}
    \\
    = \DCohx[\epsilon]{\rho}{X}{X'}{\Gamma_X}{\Gamma_{X'}}\ .
  \end{multline}
\end{proposition}

We now give a loose equivalent of
\autoref{prop:coh-rel-entr-mapping-gibbs-states-onto-each-other} for the
alternative smoothing of the coherent relative entropy.  The error term is
relatively loose (it scales proportionally to $n$ and to $\epsilon$), and it
does not disappear in the i.i.d.\@ limit unless the limit $\epsilon\to0$ is
taken explicitly.  For this reason, for small $\epsilon$, it might be
advantageous to use
\autoref{prop:coh-rel-entr-mapping-gibbs-states-onto-each-other} in conjunction
with \autoref{prop:cohrelentr-2-equiv}.

\begin{proposition}
  \label{prop:coh-rel-entr-smooth-cohrelentr-with-reduced-proj-of-Gamma}
  Let $P_X, P'_{X'}$ be projectors such that $[\Gamma_X,P_X]=0$ and
  $[\Gamma_{X'},P'_{X'}]=0$.  Let $\rho_{X'R_X}$ be such that
  $\rho_{R_X} = t_{X\to R_X}`(P_X\Gamma_XP_X/\tr P_X\Gamma_X)$ and
  $\rho_{X'} = P'_{X'}\Gamma_{X'} P'_{X'}/\tr P'_{X'}\Gamma_{X'}$.  Let
  $\epsilon\geqslant 0$. Then
  \begin{align}
    \DCohx[\epsilon]\rho{X}{X'}{\Gamma_X}{\Gamma_{X'}}
    = \log\frac{\tr P'_{X'}\Gamma_{X'}}{\tr P_X\Gamma_X}
    + f`*(\epsilon,\Gamma_X,\Gamma_{X'})\ ,
  \end{align}
  where the error term $f`*(\epsilon,\Gamma_X,\Gamma_{X'})$ is bounded as
  \begin{align}
    0 \leqslant f`*(\epsilon,\Gamma_X,\Gamma_{X'})
    \leqslant f_0`*(\epsilon,\Gamma_X) + f_0`*(\epsilon,\Gamma_{X'})\ ,
  \end{align}
  where
  $f_0`*(\epsilon,\Gamma) = \epsilon\log`*(\rank\Gamma-1) + \epsilon
  \norm{\log\Gamma}_\infty + h`(\epsilon)$ with the binary entropy
  $h`(\epsilon)=-\epsilon\log\epsilon-(1-\epsilon)\log(1-\epsilon)$.
\end{proposition}
\begin{proof}[*prop:coh-rel-entr-smooth-cohrelentr-with-reduced-proj-of-Gamma]
  The lower bound is given simply as
  \begin{multline}
    \DCohx[\epsilon]\rho{X}{X'}{\Gamma_X}{\Gamma_{X'}}
    \\
    \geqslant \DCohx[\epsilon=0]\rho{X}{X'}{\Gamma_X}{\Gamma_{X'}}
    = \log \frac{\tr P'_{X'}\Gamma_{X'}}{\tr P_X\Gamma_X}\ ,
  \end{multline}
  where the latter expression is provided by
  \autoref{prop:coh-rel-entr-mapping-gibbs-states-onto-each-other}, recalling
  that for $\epsilon=0$ both versions of the smooth coherent relative entropy
  coincide exactly.  For the upper bound, let $\hat\rho_{X'R_X}$ be the optimal
  state such that $P`*(\hat\rho_{X'R_X},\rho_{X'R_X})\leqslant\epsilon$ and
  \begin{align}
    \DCohx[\epsilon]\rho{X}{X'}{\Gamma_X}{\Gamma_{X'}}
    &= \DCohx{\hat\rho}{X}{X'}{\Gamma_X}{\Gamma_{X'}}\ ,
      \label{eq:prop-coh-rel-entr-smooth-cohrelentr-with-reduced-proj-of-Gamma-upbound-calc-0b}
  \end{align}
  and invoke \autoref{prop:coh-rel-entr-upper-bound-diff-D} to get
  \begin{align}
    \text{\eqref{eq:prop-coh-rel-entr-smooth-cohrelentr-with-reduced-proj-of-Gamma-upbound-calc-0b}}
    &\leqslant \DD{\hat\rho_X}{\Gamma_X} - \DD{\hat\rho_{X'}}{\Gamma_{X'}}\ .
      \label{eq:prop-coh-rel-entr-smooth-cohrelentr-with-reduced-proj-of-Gamma-upbound-calc-1}
  \end{align}
  We have
  $D`*(\hat\rho_{R_X},\rho_{R_X})\leqslant P`*(\hat\rho_{R_X},\rho_{R_X})
  \leqslant \epsilon$ and analogously
  $D`*(\hat\rho_{X'},\rho_{X'})\leqslant\epsilon$.  By continuity of the
  relative entropy given in
  \autoref{util:continuity-of-relative-entropy-in-1st-arg}, we get
  \begin{subequations}
    \begin{align}
      \abs*{\DD{\hat\rho_{R_X}}{\Gamma_{R_X}} - \DD{\rho_{R_X}}{\Gamma_{R_X}}}
      &\leqslant f_0`*(\epsilon,\Gamma_{R_X})\ ; \nonumber\\
        \abs*{\DD{\hat\rho_{X'}}{\Gamma_{X'}} - \DD{\rho_{X'}}{\Gamma_{X'}}}
      &\leqslant  f_0`*(\epsilon,\Gamma_{X'})\ ,
    \end{align}
  \end{subequations}
  where $f_0`*(\epsilon,\Gamma)$ is as given in the claim.  On the other hand,
  \begin{align}
    \DD{\rho_{R_X}}{\Gamma_{R_X}} - \DD{\rho_{X'}}{\Gamma_{X'}}
    = \log\tr P'_{X'}\Gamma_{X'} - \log\tr P_{R_X}\Gamma_{R_X}\ ,
  \end{align}
  because $\rho_{R_X} = P_{R_X}\Gamma_{R_X} P_{R_X}/\tr P_{R_X}\Gamma_{R_X}$ and
  $\rho_{X'}= P'_{X'}\Gamma_{X'}P'_{X'}/\tr P'_{X'}\Gamma_{X'}$, as given
  by~\eqref{eq:coh-rel-entr-expression-relative-entropy-for-projected-Gamma-states}.
  This means that
  \begin{align}
    \text{\eqref{eq:prop-coh-rel-entr-smooth-cohrelentr-with-reduced-proj-of-Gamma-upbound-calc-1}}
    &\leqslant \log\tr\frac{P'_{X'}\Gamma_{X'}}{P_{R_X}\Gamma_{R_X}}
      + f_0`*(\epsilon,\Gamma_{R_X}) + f_0`*(\epsilon,\Gamma_{X'})
      \ .
      \tag*\qedhere
  \end{align}
\end{proof}

Crucially, this alternative smoothing method does not alter the quantity much in
the regime of small $\epsilon$.  In fact, both versions of the smooth coherent
relative entropy are related by a simple adjustment of the $\epsilon$ parameter,
and up to an error term which depends only on $\epsilon$ and doesn't scale with
the system size.
\begin{proposition}
  \label{prop:cohrelentr-2-equiv}
  Let $\rho_{X'R_X}$ be any quantum state.  Then for any
  $\epsilon\geqslant 0$ with $3\sqrt{\epsilon}<1$,
  \begin{align}
    \DCohx[\epsilon]{\rho}{X}{X'}{\Gamma_X}{\Gamma_{X'}}
    \leqslant
    \DCohz[3\sqrt{\epsilon}]{\rho}{X}{X'}{\Gamma_X}{\Gamma_{X'}}\ .
    \label{eq:prop-cohrelentr-2-equiv-xlez}
  \end{align}
  Conversely, for any $\epsilon>0$ with $9 \epsilon^{1/4}<1$,
  \begin{multline}
    \DCohz[\epsilon]{\rho}{X}{X'}{\Gamma_X}{\Gamma_{X'}}
    \\
    \leqslant
    \DCohx[9\epsilon^{1/4}]{\rho}{X}{X'}{\Gamma_X}{\Gamma_{X'}}
    + \log`(1/\epsilon)\ .
    \label{eq:prop-cohrelentr-2-equiv-zlex}    
  \end{multline}
\end{proposition}

We need to prove the following lemma first.
\begin{lemma}
  \label{lemma:cohrelentr-2-trdecrease-max-work-extract}
  Let $\Gamma_X,\Gamma_{X'}\geqslant 0$.  Let $\mathcal{T}_{X\to X'}$ be a
  completely positive, trace-nonincreasing map.  Let
  $Q_X = \mathcal{T}^\dagger(\Ident_{X'})$.  Assume that the support of $Q_X$
  lies within the support of $\Gamma_X$, and that
  $\mathcal{T}_{X\to X'}`(\Gamma_X)$ lies within the support of
  $\Gamma_{X'}$. Then
  \begin{align}
    \min\,`\big{\alpha:~ \mathcal{T}_{X\to X'}(\Gamma_X) \leqslant \alpha\,\Gamma_{X'} }
    \geqslant \frac{\tr`(Q_X\Gamma_X)}{\tr\Gamma_{X'}}\ .
  \end{align}
\end{lemma}
\begin{proof}[*lemma:cohrelentr-2-trdecrease-max-work-extract]
  The optimal $\alpha$ is given by
  \begin{align}
    \alpha
    &= \norm[\big]{ \Gamma_{X'}^{-1/2}\, \mathcal{T}_{X\to X'}(\Gamma_X)\,
      \Gamma_{X'}^{-1/2} }_\infty
      \nonumber\\
    &\geqslant \tr`*[ `*(\frac{\Gamma_{X'}}{\tr \Gamma_{X'}}) \,
      \Gamma_{X'}^{-1/2}\, \mathcal{T}_{X\to X'}(\Gamma_X)\,
      \Gamma_{X'}^{-1/2} ]
      \nonumber\\
    &= `*(\tr\Gamma_{X'})^{-1}\, \tr`*[ \mathcal{T}_{X\to X'}(\Gamma_X) ]
      = `*(\tr\Gamma_{X'})^{-1}\, \tr`*[ Q_X \Gamma_X ]\ ,
  \end{align}
  where we have used that
  $\norm{\cdot}_\infty = \max_\gamma \tr`[\gamma\, (\cdot)]$ with $\gamma$
  ranging over all density operators.
\end{proof}

\begin{proof}[*prop:cohrelentr-2-equiv]
  First we prove~\eqref{eq:prop-cohrelentr-2-equiv-xlez}.  Let
  $\tilde\rho_{X'R}$ be the state which achieves the optimum in
  $\DCohx[\epsilon]{\rho}{X}{X'}{\Gamma_X}{\Gamma_{X'}}$, and let $T_{X'R_X}$,
  $\alpha$ be optimal primal variables for
  $2^{-\DCohz{\tilde\rho}{X}{X'}{\Gamma_X}{\Gamma_{X'}}}$ for the semidefinite
  program in \autoref{prop:coh-rel-entr-nonsmooth-SDP}, and denote by
  $\mathcal{T}_{X\to X'}$ the completely positive, trace-nonincreasing map
  corresponding to $T_{X'R_X}$.  Write
  $\ket{\sigma}_{XR_X} = \rho_{R_X}^{1/2}\,\ket\Phi_{X:R_X}$ and
  $\ket{\tilde\sigma}_{XR_X} = \tilde\rho_{R_X}^{1/2}\,\ket\Phi_{X:R_X}$.  Since
  $P`(\sigma_{R_X},\tilde\sigma_{R_X})\leqslant\epsilon$, we see using
  \autoref{lemma:cohrelentr-2-std-purifs-close} that
  $P(\sigma_{XR_X}, \tilde\sigma_{XR_X})\leqslant 2\sqrt\epsilon$.  The purified
  distance may not increase under the action of the trace nonincreasing map
  $\mathcal{T}_{X\to X'}$, and hence
  \begin{align}
    \hspace*{3em}&\hspace*{-3em}
    P`*(\mathcal{T}_{X\to X'}`*(\sigma_{XR_X}), \rho_{X'R_X})
    \nonumber\\
    &\leqslant
    P`*(\mathcal{T}_{X\to X'}`*(\sigma_{XR_X}), \tilde\rho_{X'R_X})
    + P(\tilde\rho_{X'R_X}, \rho_{X'R_X})
      \nonumber\\
    &\leqslant
      P`*(\mathcal{T}_{X\to X'}`*(\sigma_{XR_X}),
      \mathcal{T}_{X\to X'}`*(\tilde\sigma_{XR_X}))
      + \epsilon
      \nonumber\\
    &\leqslant 2\sqrt\epsilon + \epsilon \leqslant 3\sqrt\epsilon\ .
  \end{align}
  Hence, $\mathcal{T}_{X\to X'}$ is an optimization candidate for
  $2^{-\DCohz[3\sqrt{\epsilon}]{\rho}{X}{X'}{\Gamma_X}{\Gamma_{X'}}}$ with the
  same achieved value, proving~\eqref{eq:prop-cohrelentr-2-equiv-xlez}.

  Now we prove~\eqref{eq:prop-cohrelentr-2-equiv-zlex}.  In the remainder of
  this proof, we use the shorthand system name $R \equiv R_X$.  Let
  $\hat{T}_{X'RE}$, $\hat\alpha$ be the optimal primal variables for
  $2^{-\DCohz[\epsilon]{*\rho_{X'R}}{X}{X'}{\Gamma_X}{\Gamma_{X'}}}$.  We will
  construct an explicit $\tilde\rho_{X'R}$ close to $\rho_{X'R}$, as well as
  feasible candidates $\tilde{T}_{X'R}$ and $\tilde\alpha$ in the optimization
  for $\DCohx{*\tilde\rho_{X'R}}{X}{X'}{\Gamma_X}{\Gamma_{X'}}$ as given by
  \autoref{prop:coh-rel-entr-nonsmooth-SDP}.  We denote by
  $\hat{\mathcal{T}}_{X\to X'}$ the completely positive, trace nonincreasing map
  corresponding to $\hat{T}_{X'RE}$.
  Let $\sigma_{XR} = \rho_R^{1/2}\,\Phi_{X:R}\,\rho_R^{1/2}$ and define
  \begin{align}
      \hat\rho_{X'R} &= \hat{\mathcal{T}}_{X\to X'}(\sigma_{XR})\ .
  \end{align}
  By assumption, $P`(\hat\rho_{X'R}, \rho_{X'R})\leqslant\epsilon$ and hence
  $D`(\hat\rho_{X'R}, \rho_{X'R})\leqslant\epsilon$.  Using the fact that
  $\hat\rho_{X'R} = \rho_{X'R} + \Delta^+_{X'R} - \Delta^-_{X'R}$ for some
  $\Delta^\pm_{X'R}\geqslant 0$ with
  $\tr\Delta^+_{X'R} = \tr\Delta^-_{X'R} = D(\hat\rho_{X'R},\rho_{X'R})
  \leqslant \epsilon$, we see that
  $\tr\hat\rho_{X'R} \geqslant \tr\rho_{X'R} - \epsilon = 1 - \epsilon$.

  Define $Q = \hat{\mathcal{T}}^\dagger(\Ident_{X'})$ and note that
  $0\leqslant Q\leqslant \Ident$.  For any $0<\eta<1$, let $P^\eta$ be the
  projector onto the eigenspaces of $Q$ for which the corresponding eigenvalues
  are greater or equal to $\eta$; clearly $P^\eta$ and $Q$ commute.  Define
  $R^\eta =P^\eta - P^\eta QP^\eta$, noting that $P^\eta, Q, R^\eta$ all
  commute.  By definition, $\eta P^\eta \leqslant P^\eta QP^\eta$, and hence
  $R^\eta\leqslant (\eta^{-1}-1)\, P^\eta Q P^\eta \leqslant (\eta^{-1}-1)\, Q$.
  We may now define
  \begin{align}
    \tilde{\mathcal{T}}_{X\to X'}(\cdot) =
    \hat{\mathcal{T}}_{X\to X'}(\cdot)
    + \tr`(R^\eta \,(\cdot))\,\frac{\Gamma_{X'}}{\tr\Gamma_{X'}}\ .
  \end{align}
  The mapping $\tilde{\mathcal{T}}_{X\to X'}$ is trace non-increasing,
  \begin{align}
    \tilde{\mathcal{T}}_{X\leftarrow X'}^\dagger(\Ident_{X'})
    = Q + R^\eta 
    = P^\eta + P^{\eta,\perp}\,Q\,P^{\eta,\perp}
    \leqslant \Ident\ ,
  \end{align}
  where $P^{\eta,\perp}=\Ident - P^\eta$, keeping in mind that
  $Q = P^\eta Q P^\eta + P^{\eta,\perp} Q P^{\eta,\perp}$ and that
  $R^\eta + P^\eta Q P^\eta = P^\eta$.  Furthermore
  $\tilde{\mathcal{T}}_{X\to X'}$ is trace-preserving on the subspace spanned by
  $P^\eta$, i.e.\@
  $P^\eta\,\tilde{\mathcal{T}}^\dagger_{X\leftarrow X'}(\Ident_{X'})\,P^\eta =
  P^\eta$.  This means that for any state $\tau$ lying in the support of
  $P^\eta$, it holds that $\tr`[\tilde{\mathcal{T}}_{X\to X'}(\tau)] = 1$.
  The map $\tilde{\mathcal{T}}_{X\to X'}$ moreover satisfies
  \begin{align}
    \tilde{\mathcal{T}}_{X\to X'}`(\Gamma_X)
    &\leqslant \hat\alpha\,\Gamma_{X'} +
    \frac{\tr R^\eta \Gamma_{X}}{\tr\Gamma_{X'}}\,\Gamma_{X'}
    \nonumber\\
    &\leqslant `*(\hat\alpha + 
    `(\eta^{-1}-1)\, \frac{\tr Q\Gamma_{X}}{\tr\Gamma_{X'}})\,\Gamma_{X'}
    \leqslant \eta^{-1}\hat\alpha\,\Gamma_{X'}\ ,
  \end{align}
  where in the last inequality we have used
  \autoref{lemma:cohrelentr-2-trdecrease-max-work-extract} to see that
  $\hat\alpha \geqslant \tr(Q\Gamma_X)/\tr\Gamma_{X'}$.  We are led to define
  (surprise!) $\tilde\alpha=\eta^{-1}\hat\alpha$.
  
  It remains to find a state $\tilde{\rho}_{X'R}$ which is close to $\rho_{X'R}$
  such that
  $\tilde{\rho}_R^{1/2}\, \tilde{\mathcal{T}}_{X\to X'}`(\Phi_{X:R})\,
  \tilde{\rho}_R^{1/2} = \tilde{\rho}_{X'R}$.  First define
  \begin{align}
    \tilde\sigma_X = \frac{ P^\eta\, \sigma_X\, P^\eta }{\tr`*(P^\eta\, \sigma_X)}\ .
  \end{align}
  Observe that
  $\tr`*(P^\eta\sigma_X) \geqslant \tr`*(P^\eta Q P^\eta \sigma_X) =
  \tr`*(Q\,\sigma_X) - \tr`*(P^{\eta,\perp}\,Q\,P^{\eta,\perp} \sigma_X)
  \geqslant 1 - \epsilon - \eta$, where $P^{\eta,\perp}=\Ident - P^\eta$, using
  the fact that all eigenvalues of $Q$ within $P^{\eta,\perp}$ are less than
  $\eta$ and that
  $\tr`*(Q\,\sigma_X) = \tr`\big(\hat{\mathcal{T}}(\sigma_X)) = \tr\hat\rho_{X'}
  \geqslant 1-\epsilon$.  Then,
  using~\autoref{util:purified-distance-to-normalized-projected-normalized-state},
  \begin{align}
    P(\tilde\sigma_X,\sigma_X)
    \leqslant \frac{\sqrt{2`(\epsilon + \eta)}}{\sqrt{1 - \epsilon - \eta}}
    =: \bar\epsilon\ .
  \end{align}
  Write
  $\tilde\sigma_{XR} = \tilde\sigma_X^{1/2}\,\Phi_{X:R}\,\tilde\sigma_X^{1/2}$.
  Using \autoref{lemma:cohrelentr-2-std-purifs-close} we see that
  $P`(\tilde\sigma_{XR} , \sigma_{XR}) \leqslant
  2\sqrt{D`(\tilde\sigma_R,\rho_R)} \leqslant 2\sqrt{P`(\tilde\sigma_R,\rho_R)}
  \leqslant 2\sqrt{\bar\epsilon}$.  At this point, define
  \begin{subequations}
    \begin{align}
      \tilde\rho_{X'R} &= \tilde{\mathcal{T}}_{X\to X'}(\tilde\sigma_{XR})\ ;\\
      \bar\rho_{X'R} &= \tilde{\mathcal{T}}_{X\to X'}(\sigma_{XR})\ .
    \end{align}
  \end{subequations}
  Because $\tilde\sigma_X$ lies within the support of $P^\eta$, we have
  $\tr_{X'}\tilde\rho_{X'R} =
  \tr_{X}`\big(\tilde{\mathcal{T}}^\dagger`*(\Ident_{X'})\,\tilde\sigma_{XR}) =
  \tr_{X}`\big(\tilde{\mathcal{T}}^\dagger`*(\Ident_{X'})\,
  P^\eta\,\tilde\sigma_{XR}\,P^\eta) = \tilde\sigma_R$, and hence we have
  $\tilde{\rho}_R^{1/2}\, \tilde{\mathcal{T}}_{X\to X'}`(\Phi_{X:R})\,
  \tilde{\rho}_R^{1/2} = \tilde{\rho}_{X'R}$ as required.
  Furthermore, the purified distance cannot increase under the action of
  $\tilde{\mathcal{T}}_{X\to X'}$, so we have
  $P(\tilde\rho_{X'R}, \bar\rho_{X'R}) \leqslant 2\sqrt{\bar\epsilon}$.
  Also,
  $\bar\rho_{X'R} = \hat{\mathcal{T}}_{X\to{}X'}(\sigma_{XR}) + D_{X'R} =
  \hat\rho_{X'R} + D_{X'R}$ with
  $D_{X'R} = \tr`(R^\eta\sigma_{XR})\,`(\tr\Gamma_{X'})^{-1}\,\Gamma_{X'}$,
  noting that
  $\tr D_{X'R} \leqslant \tr`(\bar\rho_{X'R}) - \tr`(\hat\rho_{X'R}) \leqslant 1
  - (1-\epsilon) \leqslant \epsilon$; hence
  $D(\bar\rho_{X'R},\hat\rho_{X'R})\leqslant \epsilon$ and thus
  $P`(\bar\rho_{X'R},\hat\rho_{X'R})\leqslant \sqrt{2\epsilon}$.  We deduce that
  $P(\tilde\rho_{X'R},\rho_{X'R}) \leqslant P(\tilde\rho_{X'R},\bar\rho_{X'R}) +
  P(\bar\rho_{X'R},\hat\rho_{X'R}) + P(\hat\rho_{X'R},\rho_{X'R}) \leqslant
  2\sqrt{\bar\epsilon} + \sqrt{2\epsilon} + \epsilon$.

  Let's summarize: We now have a state $\tilde\rho_{X'R}$ satisfying
  $P(\tilde\rho_{X'R}, \rho_{X'R})\leqslant
  2\sqrt{\bar\epsilon}+\sqrt{2\epsilon}+\epsilon$, as well as a
  trace-nonincreasing map $\tilde{\mathcal{T}}_{X\to X'}$ satisfying
  $\tilde\rho_R^{1/2}\, \tilde{\mathcal{T}}_{X\to{}X'}`*(\Phi_{X:R})\,
  \tilde\rho_R^{1/2} = \tilde\rho_{X'R}$ and
  $\tilde{\mathcal{T}}_{X\to X'}(\Gamma_X)\leqslant
  \alpha\,\eta^{-1}\,\Gamma_{X'}$.  The claim follows by choosing
  $\eta=\epsilon$ and calculating the bounds
  $\bar\epsilon \leqslant \sqrt{8\epsilon}$ (using the assumption
  $\epsilon<1/4$) as well as
  $2\sqrt{\bar\epsilon} + \sqrt{2\epsilon} + \epsilon \leqslant `\big(4\sqrt2 +
  \sqrt2 + 1)\,\epsilon^{1/4} \leqslant 9\,\epsilon^{1/4}$.
\end{proof}

\subsection{Recovering known entropy measures}
\label{sec:recovering-known-entropies}

An interesting aspect of the coherent relative entropy is that it reduces to various
previously-known entropy measures, including the min- and max-relative
entropies~\cite{Datta2009IEEE_minmax}, as well as the conditional min- and
max-entropy~\cite{PhDRenner2005_SQKD,PhDTomamichel2012}.  These measures are already known
to be relevant in counting the work cost of specific processes in quantum
thermodynamics~\cite{Dahlsten2011NJP_inadequacy,delRio2011Nature,Aberg2013_worklike,Horodecki2013_ThermoMaj,Faist2015NatComm}.

First we present some definitions. Given a (normalized) quantum state
$\rho_{AB}$, we define the \emph{(conditional) von Neumann entropy}, the
\emph{(conditional alternative) max-entropy}, and the \emph{(conditional
  alternative) min-entropy} respectively as,%
\footnote{There exist several different variants of the min- and
  max-entropy~\cite{PhDRenner2005_SQKD,PhDTomamichel2012}; however, all the
  {max-entropies} as well as all the {min-entropies} are equivalent up to terms
  of order $\log\epsilon$ after smoothing with a parameter $\epsilon$.}
\begin{align*}
\HH[\rho]{A}[B] &= -\tr`*(\rho_{AB}\log\rho_{AB}) + \tr`*(\rho_B\log\rho_B)\ ; \\
\Hzero[\rho]{A}[B] &= \log\,\norm[\big]{\tr_A \Pi^{\rho_{AB}}_{AB}}_\infty\ ;\text{ and} \\
\Hminz[\rho]{A}[B] &= -\log\,\norm[\big]{\rho_B^{-1/2}\rho_{AB}\rho_B^{-1/2}}_\infty\ .
\end{align*}
For any $\epsilon>0$, we define the \emph{smooth (conditional alternative) max-entropy}
and \emph{smooth (conditional alternative) min-entropy} respectively as
\begin{align*}
  \Hzero[\rho][\epsilon]{A}[B]
  &= \min_{\hat\rho_{AB}\approx_\epsilon\rho_{AB}} \Hzero[\hat\rho]{A}[B]\ ; \\
  \Hminz[\rho][\epsilon]{A}[B]
  &= \max_{\hat\rho_{AB}\approx_\epsilon\rho_{AB}} \Hminz[\hat\rho]{A}[B]\ ,
\end{align*}
where the optimizations range over (normalized\footnote{One easily notices that
  the normalization of the state doesn't affect these quantities, so smoothing
  may be restricted to normalized states (in contrast to, e.g.,
  Refs.~\cite{PhDTomamichel2012,BookTomamichel2016_Finite}).}) states
$\hat\rho_{AB}$ and where $\hat\rho_{AB}\approx_\epsilon\rho_{AB}$ denotes
proximity in the purified distance, i.e.,
$P`(\hat\rho_{AB},\rho_{AB})\leqslant\epsilon$.

For a (normalized) quantum state $\rho_X$, and any $\Gamma_X\geqslant 0$, we
define the \emph{quantum relative entropy}, the \emph{relative min-entropy}, and
the \emph{relative max-entropy} respectively as,
\begin{align*}
  \DD{\rho_X}{\Gamma_X} &= \tr`*[\rho_X`*(\log_2\rho_X - \log_2\Gamma_X)]\ ; \\
  \Dminz{\rho_X}{\Gamma_X} &= -\log \tr`*[\Pi^{\rho_X}_X\Gamma_X]\ ; \\
  \Dmax{\rho_X}{\Gamma_X} &= \log\,\norm{\Gamma_X^{-1/2}\rho_X\Gamma_X^{-1/2}}_\infty\ ,
\end{align*}
recalling that $\Pi^{\rho_X}_X$ denotes the projector onto the support of
$\rho_X$.  We define the smoothed versions of the relative min- and
max-entropies as
\begin{align*}
  \Dminz[\epsilon]{\rho}{\Gamma}
  &= \max_{\hat\rho\approx_\epsilon\rho} \Dminz{\hat\rho}{\Gamma}\ ; \\
  \Dmax[\epsilon]{\rho}{\Gamma}
  &= \min_{\hat\rho\approx_\epsilon\rho} \Dmax{\hat\rho}{\Gamma}\ .
\end{align*}
where the optimizations range over normalized\footnote{These smooth quantities
  were introduced in Ref.~\cite{Datta2009IEEE_minmax} using the trace distance
  and optimizing over subnormalized states.  The two distances are tightly
  related and a simple adjustment of the $\epsilon$ parameter is required.
  Furthermore we restrict to normalized states for our convenience; the
  $\Dminz[\epsilon]{}{}$ is not affected and the $\Dmax[\epsilon]{}{}$ is at
  most shifted by a factor depending on $\log(1-\epsilon)$ only.}  states
$\hat\rho_{AB}$ such that $P`(\hat\rho_{AB},\rho_{AB})\leqslant\epsilon$.

We furthermore define the \emph{hypothesis testing relative entropy}~%
\cite{Buscemi2010IEEETIT_capacity,%
  Brandao2011IEEETIT_oneshot,%
  Wang2012PRL_oneshot,%
  Tomamichel2013_hierarchy,%
  Dupuis2013_DH,%
  Matthews2014IEEETIT_blocklength%
} for $0<\eta\leqslant1$ as
\begin{align*}
  \DHyp[\eta]{\rho}{\Gamma}
  &= -\frac1\eta \log\min_{\substack{0\leqslant Q\leqslant\Ident\\\tr`[Q\rho]\geqslant\eta}}
  \tr`*[Q\Gamma] \ .
\end{align*}

We now show that we can recover the max-entropy in the case where for both input and
output systems we have $\Gamma=\Ident$.

\begin{proposition}[Recovering the max-entropy]
  \label{prop:coh-rel-entr-recover-max-entropy}
  Let $\ket\rho_{X'R_XE}$ be any pure state on systems $R_X$, $X'$, and $E$
  with $\abs{E}\geqslant\abs{X'R_X}$. Then
  \begin{multline}
    \DCohx[\epsilon]{\rho}{X}{X'}{\Ident_X}{\Ident_{X'}}
    \\ = -\Hzero[\rho][\epsilon]{E}[X'] =
    \Hminz[\rho][\epsilon]{E}[R_X]\ .
  \end{multline}
\end{proposition}
\begin{proof}[*prop:coh-rel-entr-recover-max-entropy]
  Let $\ket{\tilde\rho}_{X'R_XE}$ be any pure quantum state.  Considering the
  semidefinite problem for
  $2^{-\DCohz{\tilde\rho}{X}{X'}{\Gamma_X}{\Gamma_{X'}}}$, let
  $T_{X'RE} =
  \tilde\rho_{R_X}^{-1/2}\tilde\rho_{X'R_XE}\tilde\rho_{R_X}^{-1/2}$.
  Conditions~\eqref{eq:SDP-coh-rel-entr-cond-trnoninc}
  and~\eqref{eq:SDP-coh-rel-entr-cond-processmatrixclose} are automatically
  satisfied.  Choosing
  $\alpha = \norm{\tr_{R_X}`[T_{X'R_X}]}_\infty =
  \norm{\tr_{R_X}\tilde\rho_{R_X}^{-1/2}\tilde\rho_{X'R_X}\tilde\rho_{R_X}^{-1/2}}_\infty$
  ensures that~\eqref{eq:SDP-coh-rel-entr-cond-alphaSGPM} is satisfied, and
  hence
  \begin{align}
    \DCohz{\tilde\rho}{X}{X'}{\Gamma_X}{\Gamma_{X'}} \geqslant 
    -\log\,\norm[\big]{
      \tr_{R_X} \tilde\rho_{R_X}^{-1/2} \tilde\rho_{X'R_X} \tilde\rho_{R_X}^{-1/2}
    }_\infty\ .
    \label{eq:coh-rel-entr-recover-max-entropy-ident-calc-first-inequality}
  \end{align}

  Now let $\omega_{X'}\geqslant0$ with $\tr\omega_{X'}=1$ such that
  $\tr`\big[\omega_{X'}\cdot\tr_R`(
  \tilde\rho_{R_X}^{-1/2}\tilde\rho_{X'R_X}\tilde\rho_{R_X}^{-1/2} )] =
  \norm[\big]{\tr_{R_X}
    \tilde\rho_{R_X}^{-1/2}\tilde\rho_{X'R_X}\tilde\rho_{R_X}^{-1/2}}_\infty$,
  and note that condition~\eqref{eq:SDP-coh-rel-entr-conddual-tromegaGamma} is
  satisfied.  Now let $X_{R_X}=0$ and
  $Z_{X'R_X} = \tilde\rho_{R_X}^{-1}\otimes\omega_{X'}$, and we see that
  \begin{align}
    \tilde\rho_{R_X}^{1/2}Z_{X'R_X}\tilde\rho_{R_X}^{1/2}
    = \Pi^{\tilde\rho_{R_X}}_{R_X}\otimes\omega_{X'}
    \leqslant \Ident_{R_X}\otimes \omega_{X'}\ .
  \end{align}
  The attained value is
  \begin{align*}
    \tr\,`\big[Z_{X'R_X}\tilde\rho_{X'R_X}]
    &= \tr\,`\big[\tilde\rho_{R_X}^{-1}\otimes\omega_{X'}\cdot\tilde\rho_{X'R_X}]\\
    &= \tr\,`\big[\omega_{X'}\cdot\tr_{R_X}`(\tilde\rho_{R_X}^{-1/2}
      \tilde\rho_{X'R_X} \tilde\rho_{R_X}^{-1/2})]\\
    &= \norm[\big]{
      \tr_{R_X}\tilde\rho_{R_X}^{-1/2}\tilde\rho_{X'R_X}\tilde\rho_{R_X}^{-1/2} }_\infty\ ,
  \end{align*}
  providing us with the opposite bound
  to~\eqref{eq:coh-rel-entr-recover-max-entropy-ident-calc-first-inequality},
  and hence proving that
  \begin{align}
    \DCohz{\tilde\rho}{X}{X'}{\Ident_X}{\Ident_{X'}} =
    -\log\,\norm[\big]{\tr_R \tilde\rho_R^{-1/2}\tilde\rho_{X'R}\tilde\rho_R^{-1/2}}_\infty\ .
    \label{eq:coh-rel-entr-recover-max-entropy-ident-calc-value-Hmin-Hmax-infnorm}
  \end{align}
  We now use this expression to show that
  \begin{align}
    \DCohz{\tilde\rho}{X}{X'}{\Ident_X}{\Ident_{X'}}
    = -\Hzero[\tilde\rho]{E}[X'] = \Hminz[\tilde\rho]{E}[R_X]\ .
    \label{eq:coh-rel-entr-recover-max-entropy-ident-calc-value-Hmin-Hmax-rhotilde}
  \end{align}
  Consider the bipartition ${EX':R}$ of the pure state
  $\ket{\tilde\rho}_{EX'R}$, and write the Schmidt decomposition
  $\ket{\tilde\rho}_{EX'R_X} =
  \tilde\rho_{EX'}^{1/2}\ket{\Phi^{\tilde\rho}}_{EX':R_X} =
  \tilde\rho_{R_X}^{1/2}\ket{\Phi^{\tilde\rho}}_{EX':R_X}$, with
  $\tr_{R_X}\Phi^{\tilde\rho}_{EX':R_X} = \Pi^{\tilde\rho_{EX'}}_{EX'}$.  Then
  \begin{align*}
    \text{\eqref{eq:coh-rel-entr-recover-max-entropy-ident-calc-value-Hmin-Hmax-infnorm}}
    &= -\log\,\norm[\big]{
      \tr_{ER_X}\tilde\rho_{R_X}^{-1/2}\tilde\rho_{EX'R_X}\tilde\rho_{R_X}^{-1/2}
      }_\infty \\
    &= -\log\,\norm[\big]{
      \tr_{ER_X} \proj\Phi^{\tilde\rho}_{EX'R_X}
      }_\infty
      \\
    &= -\log\,\norm[\big]{
      \tr_E \Pi^{\tilde\rho_{EX'}}_{EX'}
      }_\infty \\
    &= -\Hzero[\tilde\rho]{E}[X']\ .
  \end{align*}
  Similarly,
  \begin{align*}
    \text{\eqref{eq:coh-rel-entr-recover-max-entropy-ident-calc-value-Hmin-Hmax-infnorm}}
    &= -\log\,\norm[\big]{
      \tr_{ER_X}`(\tilde\rho_{R_X}^{-1/2}\tilde\rho_{EX'R_X}\tilde\rho_{R_X}^{-1/2})
      }_\infty \\
    &= -\log\,\norm[\big]{
      \tr_{X'}`(\tilde\rho_{R_X}^{-1/2}\tilde\rho_{EX'R_X}\tilde\rho_{R_X}^{-1/2})
      }_\infty \\
    &= -\log\,\norm[\big]{
      \tilde\rho_{R_X}^{-1/2}\tilde\rho_{ER_X}\tilde\rho_{R_X}^{-1/2}
      }_\infty
      = \Hminz[\tilde\rho]{E}[R_X]\ ,
  \end{align*}
  where the second equality holds because the argument of the partial trace is
  pure, and hence has the same spectrum on $ER$ as on $X'$ (by Schmidt
  decomposition).

  We now see that
  \begin{align*}
    &\DCohx[\epsilon]{\rho}{X}{X'}{\Ident_X}{\Ident_{X'}}
    \\
    &\hspace*{2em}=
      \max_{P`\big(\tilde\rho_{X'R_X},\rho_{X'R_X})\leqslant\epsilon}
      \DCohz{\tilde\rho}{X}{X'}{\Ident_X}{\Ident_{X'}}
      \\
    &\hspace*{2em}=
      \max_{P`\big(\ket{\tilde\rho}_{X'R_XE},\ket{\rho}_{X'R_XE})\leqslant\epsilon}
      \DCohz{\tilde\rho}{X}{X'}{\Ident_X}{\Ident_{X'}}
      \\
    &\hspace*{2em}=
      \max_{P`\big(\ket{\tilde\rho}_{X'R_XE},\ket{\rho}_{X'R_XE})\leqslant\epsilon}
      \Hzero[\tilde\rho]{E}[X']
      \\
    &\hspace*{2em}=
      \Hzero[\rho][\epsilon]{E}[X']\ ,
  \end{align*}
  where the second equality holds by properties of the purified distance
  (Uhlmann's theorem).  An analogous argument holds for
  $\Hminz[\rho][\epsilon]{E}[R_X]$.
\end{proof}

The min- and max-relative entropies already have known connections to
thermodynamics~\cite{Aberg2013_worklike,Horodecki2013_ThermoMaj,Brandao2015PNAS_secondlaws}
in terms of work cost of erasure and work yield of formation of a state in the presence of
a heat bath. These results are recovered here, in a fully information-theoretic context.

\begin{proposition}[Recovering the min- and max-relative entropies]
  \label{prop:coh-rel-entr-trivial-input-output}
  The min-relative entropy is recovered with a trivial output state:
  \begin{align}
      \DCohx[\epsilon]{\rho}{X}{\emptysystem}{\Gamma_X}{1}
      &= \Dminz[\epsilon]{\sigma_X}{\Gamma_X}\ ,
      \label{eq:prop-coh-rel-entr-trivial-input-output--trivial-output}
  \end{align}
  writing $\sigma_X = t_{R_X\to X}`(\rho_{R_X})$.  Furthermore the max-relative
  entropy is recovered with a trivial input state:
  \begin{align}
    \DCohx[\epsilon]{*\rho_{X'}}{\emptysystem}{X'}{1}{\Gamma_X'}
    &= -\Dmax[\epsilon]{\rho_{X'}}{\Gamma_{X'}}\ .
    \label{eq:prop-coh-rel-entr-trivial-input-output--trivial-input}
  \end{align}
\end{proposition}

\begin{proof}[*prop:coh-rel-entr-trivial-input-output]
  For any state $\tilde\rho_{R_X}$, consider the semidefinite program given in
  \autoref{prop:coh-rel-entr-nonsmooth-SDP} for
  $2^{-\DCohz{\tilde\rho}{X}{\emptysystem}{\Gamma_X}{1}}$.  The choice
  $T_{R_X} = \Pi^{\tilde\rho_{R_X}}_{R_X}$ along with
  $\alpha=\tr`\Big(\Pi^{\tilde\rho_{R_X}}_{R_X}\Gamma_{R_X})$ is primal
  feasible, hence
  \begin{align}
    2^{-\DCohz{\tilde\rho}{X}{\emptysystem}{\Gamma_X}{1}}
    \leqslant 2^{-\Dminz{\tilde\rho_{R_X}}{\Gamma_{R_X}}}\ .
  \end{align}
  In the dual problem, for any $\mu>0$ let $Z_R=\mu\Pi^{\tilde\rho_{R_X}}_{R_X}$
  and $\omega_{X'}=1$. Let $P_{R_X}$ be the projector onto the eigenspaces
  associated with the positive (or null) eigenvalues of
  $`(\mu\tilde\rho_{R_X}-\Gamma_{R_X})$, and let
  $X_{R_X} = P_{R_X}\,`(\mu\tilde\rho_{R_X}-\Gamma_{R_X})\,P_{R_X}$.  Then the
  dual constraints~\eqref{eq:SDP-coh-rel-entr-nonsmooth-conddual-tromegaGamma}
  and~\eqref{eq:SDP-coh-rel-entr-nonsmooth-conddual-Z} are clearly
  satisfied. The attained value is
  \begin{multline}
    \tr`(Z_{R_X}\tilde\rho_{R_X}) - \tr`(X_R)
    = \mu\tr\tilde\rho_{R_X} - \mu\tr`(P_{R_X}\tilde\rho_{R_X})
    + \tr\left(P_{R_X}\Gamma_{R_X}\right)
    \\
    \geqslant \tr`(P_{R_X}\Gamma_{R_X})
    \geqslant \tr`\Big(\Pi^{\tilde\rho_{R_X}}_{R_X}\Gamma_R) - O`(1/\mu)\ ,
  \end{multline}
  where we have used \autoref{util:blow-muA-operator-to-try-cover-B-limit} in
  the last step. If we take $\mu\to\infty$ we get successive feasible dual
  candidates whose attained value approaches
  $2^{-\Dminz{\tilde\rho_R}{\Gamma_R}}$; hence this is the optimal value of the
  semidefinite program.
  Finally, we have
  \begin{align*}
    \DCohx[\epsilon]{\rho}{X}{\emptysystem}{\Gamma_X}{1}
    &= \max_{\tilde\rho_{R_X}\approx_\epsilon\rho_{R_X}}
      \DCohz{\tilde\rho}{X}{\emptysystem}{\Gamma_X}{1}
    \\
    &= \max_{\tilde\rho_{R_X}\approx_\epsilon\rho_{R_X}}
      \Dminz{\tilde\rho_{R_X}}{\Gamma_{R_X}}\ ,
    \\
    &=
      \Dminz[\epsilon]{\sigma_X}{\Gamma_X}\ .
  \end{align*}

  Let's now prove
  equality~\eqref{eq:prop-coh-rel-entr-trivial-input-output--trivial-input}.
  For any state $\tilde\rho_{X'}$, consider the semidefinite program given in
  \autoref{prop:coh-rel-entr-nonsmooth-SDP} for
  $2^{-\DCohz{*\tilde\rho_{X'}}{\emptysystem}{X'}{1}{\Gamma_{X'}}}$.  The choice
  $T_{X'}=\rho_{X'}$ and
  $\alpha =
  \norm[\big]{\Gamma_{X'}^{-1/2}\tilde\rho_{X'}\Gamma_{X'}^{-1/2}}_\infty =
  2^{\Dmax{\tilde\rho_{X'}}{\Gamma_{X'}}}$ clearly satisfies the primal
  constraints, and thus
  \begin{align}
    2^{-\DCohz{*\tilde\rho_{X'}}{\emptysystem}{X'}1{\Gamma_{X'}}}
    \leqslant 2^{\Dmax{\tilde\rho_{X'}}{\Gamma_{X'}}}\ .
  \end{align}
  By properties of the infinity norm, there exists a $\tau_{X'}\geqslant 0$ with
  $\tr\tau_{X'}=1$ such that
  $\norm[\big]{\Gamma_{X'}^{-1/2}\tilde\rho_{X'}\Gamma_{X'}^{-1/2}}_\infty =
  \tr`\big[\tau_{X'} \cdot
  \Gamma_{X'}^{-1/2}\tilde\rho_{X'}\Gamma_{X'}^{-1/2}]$.  Let
  $\omega_{X'} = \Gamma_{X'}^{-1/2}\tau_{X'}\Gamma_{X'}^{-1/2}$,
  $Z_{X'}=\omega_{X'}$ and $X=0$. Then the dual constraints are trivially
  satisfied and the attained value is
  \begin{align}
    \tr`[Z_{X'}\tilde\rho_{X'}] =
    \tr`\big[\Gamma_{X'}^{-1/2}\tau_{X'}\Gamma_{X'}^{-1/2}\tilde\rho_{X'}]
    = 2^{\Dmax{\tilde\rho_{X'}}{\Gamma_{X'}}}\ .
  \end{align}
  The primal and dual candidates achieve the same value, and hence this is the optimal
  solution to the semidefinite program.  We then have
  \begin{align*}
    \DCohx[\epsilon]{*\rho_{X'}}{\emptysystem}{X'}{1}{\Gamma_X'}
    &= \max_{\tilde\rho_{X'}\approx_\epsilon\rho_{X'}}
      \DCohz{\tilde\rho}{\emptysystem}{X'}{1}{\Gamma_X'}
    \\
    &= \max_{\tilde\rho_{X'}\approx_\epsilon\rho_{X'}}
      -\Dmax{\tilde\rho_{X'}}{\Gamma_{X'}}
      \\
    &= -\Dmax[\epsilon]{\rho_{X'}}{\Gamma_{X'}}\ .
      \tag*\qedhere
  \end{align*}
\end{proof}

It is clear that in \autoref{prop:coh-rel-entr-trivial-input-output} in the
case of $\epsilon=0$, we may replace the trivial system with $\Gamma=1$ by a
nontrivial system with arbitrary $\Gamma$, as long as it is in a pure eigenstate
of the $\Gamma$ operator.

\begin{corollary}
  \label{cor:coh-rel-entr-pureev-input-output}
  Let $\Gamma_X, \Gamma_{X'}\geqslant 0$.  Both following statements hold:
  \begin{enumerate}[label=(\alph*)]
  \item Let $\ket{\mathrm f}_{X'}$ be an eigenstate of $\Gamma_{X'}$ with
    eigenvalue $g_\mathrm{f}$, and let $\sigma_{X}$ be any quantum state in the
    support of $\Gamma_X$. Then:
    \begin{multline}
      \DCohz{*t_{X\to R_X}`(\sigma_X)\otimes\proj{\mathrm f}_{X'}}{X}{X'}%
      {\Gamma_X}{\Gamma_{X'}}
      \\
       = \Dminz{\sigma_X}{\Gamma_X} + \log g_\mathrm f\ .
    \end{multline}
  \item Let $\ket{\mathrm i}_X$ be an eigenstate of $\Gamma_{X}$ with eigenvalue
    $g_\mathrm i$, and let $\rho_{X'}$ be any quantum state in the support of
    $\Gamma_{X'}$. Then:
    \begin{multline}
      \DCohz{*t_{X\to R_X}`(\proj{\mathrm i}_X)\otimes\rho_{X'}}{X}{X'}%
      {\Gamma_X}{\Gamma_{X'}}
      \\
       = - \log g_\mathrm i - \Dmax{\rho_{X'}}{\Gamma_{X'}}\ .
    \end{multline}
  \end{enumerate}
\end{corollary}
\begin{proof}[*cor:coh-rel-entr-pureev-input-output]
  First consider claim (a).  Invoking successively
  \autoref{prop:coh-rel-entr-restrict-Gamma-eigenspaces},
  \autoref{prop:coh-rel-entr-scaling-Gamma}, and
  \autoref{prop:coh-rel-entr-invar-under-isometries},
  we have (writing $\sigma_{R_X} = t_{X\to R_X}`(\sigma_X)$):
  \begin{align}
    \hspace*{4em}
    &\hspace*{-4em}
      \DCohz{*\sigma_{R_X}\otimes\proj{\mathrm f}_{X'}}{X}{X'}%
      {\Gamma_X}{\Gamma_{X'}}
      \nonumber\\
    &= \DCohz{*\sigma_{R_X}\otimes\proj{\mathrm f}_{X'}}{X}{X'}%
      {\Gamma_X}{g_\mathrm{f}\proj{\mathrm f}_{X'}}
      \nonumber\\
    &= \DCohz{*\sigma_{R_X}\otimes\proj{\mathrm f}_{X'}}{X}{X'}%
      {\Gamma_X}{\proj{\mathrm f}_{X'}}
      + \log g_\mathrm{f}
      \nonumber\\
    &= \DCohz{\sigma}{X}{\emptysystem}{\Gamma_X}{1} + \log g_\mathrm{f}\ ,
  \end{align}
  at which point we may apply \autoref{prop:coh-rel-entr-trivial-input-output}.
  Claim (b) follows analogously.
\end{proof}

Finally, we will see that the usual quantum relative entropy can also be
recovered in the regime where we consider states of the form
$\rho^{\otimes n}_{X'^nR^n}$ for $n\to\infty$.  We defer this case to
\autoref{sec:AEP}, as the proof of this property requires some additional bounds
we have yet to present.

\subsection{Bounds on the coherent relative entropy}
\label{sec:coh-rel-entr-bounds}

At this point, we further characterize the coherent relative entropy with bounds
in terms of simpler quantities depending only on the input and output states.
The main goal of this section is to prove
\autoref{prop:cohrelentr-smooth-upper-bound-min-max} and
\autoref{prop:coh-rel-entr-smooth-lower-bound-diff-Dmin-Dmax}, which will allow
us to understand our quantity's asymptotic behavior in the i.i.d.\@ regime.

We begin with a few upper bounds on the coherent relative entropy, given in
terms of a difference of relative entropies.

\begin{proposition}
  \label{prop:coh-rel-entr-upper-bound-diff-D}
  We have the upper bound
  \begin{align}
    \DCohz{\rho}{X}{X'}{\Gamma_X}{\Gamma_{X'}}
    &\leqslant \DD{\sigma_X}{\Gamma_X}
      - \DD{\rho_{X'}}{\Gamma_{X'}} \ ,
      \label{eq:coh-rel-entr-upper-bound-diff-D}
  \end{align}
  writing $\sigma_{X} = t_{R_X\to X}`(\rho_{R_X})$
\end{proposition}
\begin{proof}[*prop:coh-rel-entr-upper-bound-diff-D]
  Consider the optimal solution $T_{X'R_X}$ and $\alpha$ to the primal
  semidefinite program of \autoref{prop:coh-rel-entr-nonsmooth-SDP}, and let
  $\mathcal{T}_{X\to X'}$ be the completely positive map corresponding to
  $T_{X'R_X}$, i.e.\@ defined by
  $\mathcal{T}_{X\to X'}`(\cdot) = \tr_{R_X}`[T_{X'R_X}\,
  t_{X\to{}R_X}`*(\cdot)]$.  The mapping defined in this way is completely
  positive since $T_{X'R_X}\geqslant 0$ and is trace-nonincreasing thanks to
  condition~\eqref{eq:SDP-coh-rel-entr-cond-trnoninc}.

  The map $\mathcal{T}_{X\to X'}$ thus satisfies the conditions of
  \autoref{item:prop-equiv-battery-models-equivstatement-criterionMapEnorm} of
  \autoref{prop:equiv-battery-models}.  Hence, invoking
  \autoref{item:prop-equiv-battery-models-equivstatement-infbattery} of that
  proposition, let $\tilde\Phi_{XA\to X'A'}$ be a trace nonincreasing
  $\Gamma$-sub-preserving map for large enough $A$, $A'$, with
  $\Gamma_A=\Ident_A$, $\Gamma_{A'}=\Ident_{A'}$, satisfying
  \begin{align}
    \tilde\Phi_{XA\to X'A'}`*(
    \sigma_{XR_X}\otimes`*(2^{-\lambda_1}\Ident_{2^{\lambda_1}}))
    = \rho_{X'R_X}\otimes`*(2^{-\lambda_2}\Ident_{2^{\lambda_2}})\ ,
  \end{align}
  with $\alpha=2^{-\left(\lambda_1-\lambda_2\right)}$ and
  $\ket\sigma_{XR_X} = \rho_{R_X}^{1/2}\,\ket\Phi_{X:R_X}$. (If $\alpha$ is
  irrational, the following argument may be applied to arbitrary good rational
  approximations to $\alpha$.)

  Now, dilate $\tilde\Phi_{XA\to X'A'}$ using
  \autoref{prop:dilation-of-Gsp-to-Gp} to a trace-preserving,
  $\Gamma$-preserving map $\Phi_{XAX'A'Q\to XAX'A'Q}$ with states
  $\ket{\mathrm x}_{X}, \ket{\mathrm a}_{A}, \ket{\mathrm i}_Q, \ket{\mathrm
    x'}_{X'}, \ket{\mathrm a'}_{A'}, \ket{\mathrm f}_Q$ (all of them being
  eigenstates of the respective $\Gamma$ operators), satisfying
  \begin{subequations}
    \begin{gather}
      \Phi_{XAX'A'Q}`*(\Gamma_{XAX'A'Q})
      = \Gamma_{XAX'A'Q}\ ; \\
      \begin{split}
        &\Phi_{XAX'A'Q}`*(
        \sigma_{XR_X}\otimes`*(2^{-\lambda_1}\Ident_{2^{\lambda_1}}^A)
          \otimes\proj{\mathrm{x'a'i}}_{X'A'Q} ) \\
        &\hspace*{4em}= \rho_{X'R_X}\otimes`*(
        2^{-\lambda_2}\Ident_{2^{\lambda_2}}^{A'}
        ) \otimes \proj{\mathrm{xaf}}_{XAQ}\ ;\quad\text{and}
      \end{split}%
      \label{eq:prop-coh-rel-entr-upper-bound-diff-D-proof-dilatedGPM-correctmapping}\\
      \dmatrixel{\mathrm{x\,a\,f}}{\Gamma_{XAQ}}_{XAQ}
      = \dmatrixel{\mathrm{x'a'i}}{\Gamma_{X'A'Q}}_{X'A'Q}\ .
      \label{eq:prop-coh-rel-entr-upper-bound-diff-D-proof-dilatedGPM-energylevelsdilated}
    \end{gather}
  \end{subequations}
  Using \autoref{prop:coh-rel-entr-battery-state-calc-rel-entr} recalling that
  $\Gamma_A=\Ident_A$, we see that
  \begin{subequations}
    \begin{align}
      \DD`*{2^{-\lambda_1}\Ident_{2^{\lambda_1}}^A}{\Gamma_A}
      &= -\log\tr`*(\Ident_{2^{\lambda_1}}^A\,\Gamma_A) = -\lambda_1\ ;\\
      \DD`*{2^{-\lambda_2}\Ident_{2^{\lambda_2}}^{A'}}{\Gamma_{A'}}
      &= -\log\tr`*(\Ident_{2^{\lambda_2}}^A\,\Gamma_A) = -\lambda_2 \ ,
    \end{align}
    as well as for any pure eigenstate $y$ of any positive semidefinite
    $\Gamma$,
    \begin{align}
      \DD`*{\proj{\mathrm y}}{\Gamma} &= -\log\tr\dmatrixel{\mathrm y}{\Gamma}\ .
    \end{align}
  \end{subequations}
  Then, by the data processing inequality for the relative entropy and
  with~\eqref{eq:prop-coh-rel-entr-upper-bound-diff-D-proof-dilatedGPM-correctmapping},
  \begin{align}
    \begin{split}
      0 &\leqslant
      \DD`*{\sigma_{X}\otimes`\big(2^{-\lambda_1}\Ident_{2^{\lambda_1}}^A)
        \otimes\proj{\mathrm{x'a'i}}_{X'A'Q}}{\Gamma_{XAX'A'Q}}
      \\
      &\quad-
      \DD`*{\rho_{X'}\otimes`\big(2^{-\lambda_2}\Ident_{2^{\lambda_2}}^{A'})
        \otimes\proj{\mathrm{xaf}}_{XAQ}}{\Gamma_{XAX'A'Q}}
      \nonumber\\
    \end{split} \nonumber\\
    \begin{split}
      &= \DD{\sigma_X}{\Gamma_X} +
      \DD`\big{2^{-\lambda_1}\Ident_{2^{\lambda_1}}^A}{\Gamma_A} +
      \DD{\proj{\mathrm{x'a'i}}_{X'A'Q}}{\Gamma_{X'A'Q}}
      \\
      &\quad- \DD{\rho_{X'}}{\Gamma_{X'}}
      - \DD`\big{2^{-\lambda_2}\Ident_{2^{\lambda_2}}^{A'}}{\Gamma_{A'}}
      - \DD{\proj{\mathrm{x\,a\,f}}_{XAQ}}{\Gamma_{XAQ}}
    \end{split}\nonumber\\
    \begin{split}
      &= \DD{\sigma_X}{\Gamma_X}-\DD{\rho_{X'}}{\Gamma_{X'}}
      - \lambda_1 + \lambda_2 
      \\
      &\quad- \log\,\dmatrixel{\mathrm{x'a'i}}{\Gamma_{X'A'Q}}
      + \log\,\dmatrixel{\mathrm{x\,a\,f}}{\Gamma_{XAQ}}
    \end{split}\nonumber\\
    &= \DD{\sigma_X}{\Gamma_X}-\DD{\rho_{X'}}{\Gamma_{X'}}
    - \lambda_1 + \lambda_2 \ ,
  \end{align}
  where we invoked the
  condition~\eqref{eq:prop-coh-rel-entr-upper-bound-diff-D-proof-dilatedGPM-energylevelsdilated}
  in the last step. We then have
  \begin{align}
    \DCohz\rho{X}{X'}{\Gamma_X}{\Gamma_{X'}} = \lambda_1-\lambda_2
    \leqslant \DD{\sigma_X}{\Gamma_X}-\DD{\rho_{X'}}{\Gamma_{X'}}\ .
    \tag*\qedhere
  \end{align}
\end{proof}

The following upper bound is easy to prove, although it has not found tremendous
use.

\begin{proposition}
  \label{prop:coh-rel-entr-upper-bound-diff-Dmax-Dmax}
  The coherent relative entropy may be upper bounded as
  \begin{multline}
    \DCohz{\rho}{X}{X'}{\Gamma_X}{\Gamma_{X'}}
    \\
    \leqslant \Dmax{\sigma_X}{\Gamma_X}
      - \Dmax{\rho_{X'}}{\Gamma_{X'}}\ ,
      \label{eq:coh-rel-entr-upper-bound-diff-Dmax-Dmax}
  \end{multline}
  writing $\sigma_{X} = t_{X\to R_X}`(\rho_{R_X})$
\end{proposition}
\begin{proof}[*prop:coh-rel-entr-upper-bound-diff-Dmax-Dmax]
  Consider an optimal solution $T_{X'R_X}$ and $\alpha$ for the primal
  semidefinite program. Then we have via the semidefinite constraints
  \begin{align}
    \rho_{X'} = \tr_{R_X}`*[T_{X'R_X}\rho_{R_X}]
    &\leqslant 2^{\Dmax{\rho_{R_X}}{\Gamma_{R_X}}} \,
    \tr_{R_X}`*[T_{X'R_X}\Gamma_{R_X}]
    \nonumber\\
    &\leqslant \alpha\, 2^{\Dmax{\rho_{R_X}}{\Gamma_{R_X}}}\, \Gamma_{X'}\ .
  \end{align}
  By definition, we have
  \begin{align}
    2^{\Dmax{\rho_{X'}}{\Gamma_{X'}}}
    = \min\{ \mu:~ \mu\,\Gamma_{X'}\geqslant \rho_{X'}\}\ ,
  \end{align}
  and thus we see that $\alpha 2^{\Dmax{\rho_{R_X}}{\Gamma_{R_X}}}$ is a
  candidate $\mu$ in this minimization. Hence
  $2^{\Dmax{\rho_{X'}}{\Gamma_{X'}}} \leqslant
  \alpha\,2^{\Dmax{\rho_{R_X}}{\Gamma_{R_X}}}$ and
  \begin{align}
    \alpha \geqslant 2^{\Dmax{\rho_{X'}}{\Gamma_{X'}} -
    \Dmax{\rho_{R_X}}{\Gamma_{R_X}}}\ .
  \end{align}
  The claim then follows from
  $\DCohz\rho{X}{X'}{\Gamma_X}{\Gamma_{X'}} = -\log\alpha$.
\end{proof}

The last of the upper bounds holds for the smooth coherent relative entropy.
The present upper bound will be used to prove one direction of the asymptotic
equipartition property.

\begin{proposition}
  \label{prop:cohrelentr-smooth-upper-bound-min-max}
  Let $\rho_{X'R_X}$ be any quantum state, and denote the corresponding input
  state by $\sigma_X=t_{R_X\to X}(\rho_{R_X})$.  Then for any
  $\epsilon,\epsilon',\epsilon''\geqslant 0$ such that
  $\bar\epsilon:=\epsilon+\epsilon'+2\epsilon'' < 1$,
  \begin{multline}
    \DCohx[\epsilon'']{\rho}{X}{X'}{\Gamma_X}{\Gamma_{X'}}
    \\
    \leqslant \Dmax[\epsilon]{\sigma_X}{\Gamma_X} - 
    \Dminz[\epsilon']{\rho_{X'}}{\Gamma_{X'}}
    - \log\,`*(1 - \bar\epsilon)\ .
  \end{multline}
\end{proposition}
\begin{proof}[*prop:cohrelentr-smooth-upper-bound-min-max]
  Let $\bar\rho_{X'R_X}$ be the quantum state which achieves the optimum for
  $\DCohx[\epsilon'']{\rho}{X}{X'}{\Gamma_X}{\Gamma_{X'}}$, i.e., satisfying
  $P`(\bar\rho_{X'R},\rho_{X'R})\leqslant\epsilon''$ and
  $\DCohx[\epsilon'']{\rho}{X}{X'}{\Gamma_X}{\Gamma_{X'}} =
  \DCohx{\bar\rho}{X}{X'}{\Gamma_X}{\Gamma_{X'}}$.  The proof proceeds by
  constructing dual candidates for
  $2^{-\DCohx{\bar\rho}{X}{X'}{\Gamma_X}{\Gamma_{X'}}}$
  in~\eqref{eq:coh-rel-entr-nonsmooth-SDP-logvalue} achieving the value in the
  claim.  Define the quantum states $\tilde\sigma_X$, $\tilde\rho_{X'}$ as the
  optimal ones in the optimizations defining the smooth min and max relative
  entropies, i.e., satisfying $P(\tilde\sigma_X,\sigma_X)\leqslant\epsilon$,
  $P(\tilde\rho_{X'},\rho_{X'})\leqslant\epsilon'$, as well as
  \begin{subequations}
    \begin{align}
      \Dmax[\epsilon]{\sigma_{X}}{\Gamma_{X}}
      &= \Dmax{\tilde\sigma_{X}}{\Gamma_{X}}\ ; \\
      \Dminz[\epsilon']{\rho_{X'}}{\Gamma_{X'}}
      &= \Dminz{\tilde\rho_{X'}}{\Gamma_{X'}}\ .
    \end{align}
  \end{subequations}
  Let
  \begin{subequations}
    \begin{align}
      \mu &= 2^{-\Dmax{\tilde\sigma_X}{\Gamma_X}}\,
            `\big(\tr\Pi^{\tilde\rho_{X'}}_{X'}\,\Gamma_{X'})^{-1}\ ;
            \\
      Z_{X'R_X} &= \mu\,\Pi^{\tilde\rho_{X'}}_{X'}\otimes\Ident_{R_X}\ ;
      \\
      \omega_{X'} &=
                    `\big[\tr`(\Pi^{\tilde\rho_{X'}}_{X'}\,\Gamma_{X'})]^{-1} \,
                    \Pi^{\tilde\rho_{X'}}_{X'} \ .
    \end{align}
  \end{subequations}
  Condition~\eqref{eq:SDP-coh-rel-entr-nonsmooth-conddual-tromegaGamma} is
  automatically satisfied.  Writing
  $\tilde\sigma_{R_X} = t_{X\to R_X}(\tilde\sigma_X)$, we have
  $D`(\bar\rho_{R_X},\tilde\sigma_{R_X}) \leqslant
  P`(\bar\rho_{R_X},\tilde\sigma_{R_X}) \leqslant P`(\bar\rho_{R_X},\rho_{R_X}) +
  P`(\rho_{R_X},\tilde\sigma_{R_X}) \leqslant \epsilon'' + \epsilon$; hence, there
  exists $\Delta_{R_X}\geqslant 0$ such that
  $\bar\rho_{R_X} \leqslant \tilde\sigma_{R_X} + \Delta_{R_X}$ with
  $\tr\Delta_{R_X}\leqslant\epsilon''+\epsilon$.  Then,
  \begin{align}
    \bar\rho_{R_X}^{1/2}\, Z_{X'R_X}\, \bar\rho_{R_X}^{1/2}
    &= \mu\,\Pi^{\tilde\rho_{X'}}_{X'}\otimes\bar\rho_{R_X}
      \nonumber\\
    &\leqslant \mu\,\Pi^{\tilde\rho_{X'}}_{X'}\otimes`(\tilde\sigma_{R_X}+\Delta_{R_X})
      \nonumber\\
    &\leqslant \mu\,\Pi^{\tilde\rho_{X'}}_{X'}\otimes
      `\big(2^{\Dmax{\tilde\sigma_{R_X}}{\Gamma_{R_X}}}\,\Gamma_{R_X} + \Delta_{R_X})
      \nonumber\\
    &\leqslant \omega_{X'}\otimes\Gamma_{R_X} + \mu\,\Ident_{X'}\otimes\Delta_{R_X}\ ,
  \end{align}
  and we may define $X_{R_X} = \mu\,\Delta_R$ in order for
  constraint~\eqref{eq:SDP-coh-rel-entr-nonsmooth-conddual-Z} to be also
  satisfied.  The attained dual objective value is
  \begin{align}
    \text{obj.} = \tr`(Z_{X'R_X}\,\bar\rho_{X'R_X}) - \tr`(X_{R_X})
    = \mu\,`\big(\tr`\big(\Pi^{\tilde\rho_{X'}}_{X'}\,\bar\rho_{X'}) - \epsilon'' - \epsilon)\ .
    \label{eq:prop-cohrelentr-2-upper-bound-min-max-calc-1}
  \end{align}
  Analogously to the input state, now we have for the output state
  $D`(\bar\rho_{X'},\tilde\rho_{X'})\leqslant P`(\bar\rho_{X'},\tilde\rho_{X'})
  \leqslant P`(\bar\rho_{X'},\rho_{X'}) + P`(\rho_{X'},\tilde\rho_{X'})
  \leqslant \epsilon'' + \epsilon'$; there must exist $\Delta_{X'}\geqslant 0$
  with $\bar\rho_{X'}\geqslant\tilde\rho_{X'}-\Delta_{X'}$ and
  $\tr\Delta_{X'}\leqslant\epsilon''+\epsilon'$.  Hence,
  $\tr`\big(\Pi^{\tilde\rho_{X'}}_{X'}\,\bar\rho_{X'}) \geqslant
  \tr`\big(\Pi^{\tilde\rho_{X'}}_{X'}\,\tilde\rho_{X'}) -
  \tr`\big(\Pi^{\tilde\rho_{X'}}_{X'}\,\Delta_{X'}) \geqslant 1 -
  \epsilon''-\epsilon'$.  Thus,
  \begin{align}
    \text{\eqref{eq:prop-cohrelentr-2-upper-bound-min-max-calc-1}}
    &\geqslant \mu\,`(1 - \epsilon - \epsilon' - 2\epsilon'')\ .
  \end{align}
  The claim follows by noting that
  $-\log\mu = \Dmax[\epsilon]{\sigma_X}{\Gamma_X} -
  \Dminz[\epsilon']{\rho_{X'}}{\Gamma_{X'}}$.
\end{proof}

In order to formulate lower bounds on the coherent relative entropy, we
introduce a generalization of the \emph{Rob entropy} or \emph{smooth
  $S$-entropy}~\cite{Vitanov2013_chainrules}:
\begin{align}
  \begin{split}
    \Dr{\rho}{\Gamma}
    &= -\log\,\norm[\big]{\rho^{-1/2}\Gamma\rho^{-1/2}}_\infty
    \\
    &= -\log\,\min`{ \nu:~ \nu\rho \geqslant \Pi^\rho\Gamma\Pi^\rho } \ ;
  \end{split}
  \\[1ex]
  \Dr[\epsilon]{\rho}{\Gamma}
  &= \max_{\hat\rho\approx_\epsilon\rho} \Dr{\hat\rho}{\Gamma}\ .
\end{align}

\begin{proposition}
  \label{prop:coh-rel-entr-lower-bound-diff-dr-dmax}
  We have the lower bound
  \begin{align}
    \DCohz{\rho}{X}{X'}{\Gamma_X}{\Gamma_{X'}}
    \geqslant \Dr{\sigma_X}{\Gamma_X}
    - \Dmax{\rho_{X'}}{\Gamma_{X'}}\ ,
  \end{align}
  with $\sigma_{X} = t_{X\to R_X}`(\rho_{R_X})$.
\end{proposition}
\begin{proof}[*prop:coh-rel-entr-lower-bound-diff-dr-dmax]
  Choose the primal candidate
  $T_{X'R_X} = \rho_{R_X}^{-1/2}\rho_{X'R_X}\rho_{R_X}^{-1/2}$.  We have
  $\tr_{X'}T_{X'R_X} = \rho_{R_X}^{-1/2}\rho_{R_X}\rho_{R_X}^{-1/2} =
  \Pi^{\rho_{R_X}}_{R_X} \leqslant \Ident_{R_X}$ so our candidate
  satisifes~\eqref{eq:SDP-coh-rel-entr-nonsmooth-cond-trnoninc}.
  Also~\eqref{eq:SDP-coh-rel-entr-nonsmooth-cond-correctmapping} is satisfied by
  construction, and $\tr_{R_X}`*(T_{X'R_X}\Gamma_{R_X})$ is in the support of
  $\rho_{X'}$ and hence it lies in the support of $\Gamma_{X'}$. According to
  \autoref{prop:coh-rel-entr-rewrite-primal-problem-infinity-norm} we choose
  $\alpha=\norm[\big]{\Gamma_{X'}^{-1/2} \tr_{R_X}`*[T_{X'R_X}\Gamma_{R_X}]
    \Gamma_{X'}^{-1/2}}_\infty$ and
  \begin{align}
    \hspace*{3em}&\hspace*{-3em} 2^{-\DCohz{\rho}{X}{X'}{\Gamma_X}{\Gamma_{X'}}}
    \nonumber\\
    &\leqslant \alpha =
    \norm[\big]{\Gamma_{X'}^{-1/2} \tr_{R_X}`*[T_{X'R_X}\Gamma_{R_X}]
      \Gamma_{X'}^{-1/2}}_\infty
    \nonumber\\
    &= \norm[\big]{\Gamma_{X'}^{-1/2} \tr_{R_X}`*[
      T_{X'R_X} \Pi^{\rho_{R_X}}_{R_X} \Gamma_{R_X}\Pi^{\rho_{R_X}}_{R_X}
      ] \Gamma_{X'}^{-1/2}}_\infty
    \nonumber\\
    &\leqslant 2^{-\Dr{\rho_{R_X}}{\Gamma_{R_X}}}\,
    \norm[\big]{ \Gamma_{X'}^{-1/2} \tr_{R_X}`*[
      T_{X'R_X} \rho_{R_X} ] \Gamma_{X'}^{-1/2} }_\infty\ ,
      \label{eq:coh-rel-entr-lower-bound-diff-dr-dmax-proof-calc-1}
  \end{align}
  since by definition
  $\rho_{R_X}^{-1/2}\Gamma_{R_X}\rho_{R_X}^{-1/2} \leqslant
  2^{-\Dr{\rho_{R_X}}{\Gamma_{R_X}}}\Ident$ and thus
  $\Pi^{\rho_{R_X}}_{R_X}\Gamma_{R_X} \Pi^{\rho_{R_X}}_{R_X} \leqslant
  2^{-\Dr{\rho_{R_X}}{\Gamma_{R_X}}}\rho_{R_X}$. Then
  \begin{align*}
    \text{\eqref{eq:coh-rel-entr-lower-bound-diff-dr-dmax-proof-calc-1}}
    &= 2^{-\Dr{\rho_{R_X}}{\Gamma_{R_X}}}\,
      \norm[\big]{\Gamma_{X'}^{-1/2} \rho_{X'} \Gamma_{X'}^{-1/2}}_\infty
    \\
    &=  2^{-\Dr{\sigma_X}{\Gamma_X}} \, 2^{\Dmax{\rho_{X'}}{\Gamma_{X'}}}\ .
    \tag*\qedhere
  \end{align*}
\end{proof}

The quantity $\Dr\cdot\cdot$, when {smoothed}, is essentially equal to the
{min-relative entropy}: These two differ by a term which is logarithmic in the
{failure probability}.  In this way, the smooth quantity
$\Dr[\epsilon]\cdot\cdot$ may be related to a better known quantity with an
operational interpretation.

\begin{proposition}
  \label{prop:smooth-Drob-to-Dmin}
  Let $\epsilon>0$. Then
  \begin{align}
    \Dr[\epsilon]{\rho}{\Gamma}
    \geqslant \Dminz{\rho}{\Gamma} + \log\epsilon' \ ,
  \end{align}
  where $\epsilon' = \epsilon^2/\left(2+\epsilon^2\right)$, or equivalently,
  $\epsilon = \sqrt{2\epsilon'/(1-\epsilon')}$.
\end{proposition}
\begin{proof}[*prop:smooth-Drob-to-Dmin]
  The proof of this proposition proceeds via the hypothesis testing relative entropy,
  $\DHyp[\eta]{\rho}{\Gamma}$. Let
  $\epsilon' = \epsilon^2/\left(2+\epsilon^2\right)$ and let $\eta=1-\epsilon'$. The
  hypothesis testing relative entropy can be written as the solution of a semidefinite
  program~\cite{Dupuis2013_DH}. Specifically, there exists $Q\geqslant 0$,
  $\mu\geqslant 0$ and $X\geqslant 0$ such that
  \begin{align}
    \label{eq:prop-smooth-Drob-to-Dmin-proof-DH-sdp-value}
    2^{-\DHyp[\eta]{\rho}{\Gamma}} = \frac1\eta \tr\left[Q\Gamma\right]
    = \mu - \frac{\tr X}{\eta}\ ,
  \end{align}
  with $Q$, $\mu$ and $X$ satisfying the conditions
  \begin{subequations}
    \begin{gather}
      Q \leqslant \Ident\ ;
      \label{eq:prop-smooth-Drob-to-Dmin-proof-DH-constraint-primal-QlessIdent} \\
      \tr\left[Q\rho\right] \geqslant \eta\ ;
      \label{eq:prop-smooth-Drob-to-Dmin-proof-DH-constraint-primal-Qrho} \\
      \mu\rho \leqslant \Gamma + X\ .
      \label{eq:prop-smooth-Drob-to-Dmin-proof-DH-constraint-dual-murho}
    \end{gather}
  \end{subequations}
  In addition, the complementary slackness relations for these variables read
  \begin{subequations}
    \begin{gather}
      XQ = X\ ;
      \label{eq:prop-smooth-Drob-to-Dmin-proof-DH-slackness-QX} \\
      \tr\left(Q\rho\right) = \eta\ ;
      \label{eq:prop-smooth-Drob-to-Dmin-proof-DH-slackness-Qrho} \\
      Q\left(\mu\rho - \Gamma - X\right) = 0\ .
      \label{eq:prop-smooth-Drob-to-Dmin-proof-DH-slackness-QmurhoX}
    \end{gather}
  \end{subequations}

  Define $\bar\rho = \Pi^Q \rho\Pi^Q$, where $\Pi^Q$ is the projector onto the support of
  $Q$. Apply $Q^{-1}\left(\cdot\right)\Pi^Q$
  onto~\eqref{eq:prop-smooth-Drob-to-Dmin-proof-DH-slackness-QmurhoX} to obtain
  \begin{align}
    \mu\bar\rho = \Pi^Q\Gamma\Pi^Q + \Pi^QX\Pi^Q \geqslant \Pi^Q\Gamma\Pi^Q\ .
  \end{align}
  In addition, because $\Pi^Q\Gamma\Pi^Q$ has support on $\Pi^Q$, then $\bar\rho$ must
  also have support on the full of $\Pi^Q$, i.e.\@ $\Pi^{\bar\rho} = \Pi^Q$. So, by
  definition of $\Dr{\bar\rho}{\Gamma}$ have that
  \begin{align}
    2^{-\Dr{\bar\rho}{\Gamma}} \leqslant \mu\ .
  \end{align}
  Also, define $\bar\rho' = \bar\rho / \tr\bar\rho$, and we can see by
  \autoref{util:purified-distance-to-normalized-projected-normalized-state} that
  $P\left(\rho,\bar\rho'\right)\leqslant \sqrt{2\epsilon'/\left(1-\epsilon'\right)} =
  \epsilon$.
  Also, $2^{-\Dr{\bar\rho'}{\Gamma}}\leqslant 2^{-\Dr{\bar\rho}{\Gamma}}$ by definition of
  $\Dr\cdot\cdot$.  Then $\bar\rho'$ is a valid optimization candidate in the definition
  of $\Dr[\tilde\epsilon]{\rho}{\Gamma}$ and
  \begin{align}
    2^{-\Dr[\tilde\epsilon]{\rho}{\Gamma}}
    \leqslant 2^{-\Dr{\bar\rho'}{\Gamma}} \leqslant \mu\ .
  \end{align}

  It thus remains to show that
  $\mu\leqslant\epsilon'^{-1}\,2^{-\Dminz{\rho}{\Gamma}}$.  Apply
  $\tr\left(\Pi^\rho\left(\cdot\right)\right)$ onto the
  constraint~\eqref{eq:prop-smooth-Drob-to-Dmin-proof-DH-constraint-dual-murho} to obtain
  \begin{align}
    \mu \leqslant \tr\left(\Pi^\rho\Gamma\right) + \tr\left(\Pi^\rho X\right)
    \leqslant \tr\left(\Pi^\rho\Gamma\right) + \tr\left(X\right)\ .
    \label{eq:prop-smooth-Drob-to-Dmin-proof-calc-almost-there}
  \end{align}
  Now, because of~\eqref{eq:prop-smooth-Drob-to-Dmin-proof-DH-sdp-value}, we have
  $0\leqslant \tr\left[Q\Gamma\right] = \mu\eta - \tr X$, and thus
  $\tr X\leqslant \mu\eta$. Combining
  with~\eqref{eq:prop-smooth-Drob-to-Dmin-proof-calc-almost-there} gives
  \begin{align}
    \mu\left(1-\eta\right) \leqslant \tr\left(\Pi^\rho\Gamma\right)\ ;
  \end{align}
  since $\epsilon'=1-\eta$ and
  $\tr\left(\Pi^\rho\Gamma\right)=2^{-\Dminz{\rho}{\Gamma}}$ we have
  $\mu\leqslant \left(1/\epsilon'\right) 2^{-\Dminz{\rho}{\Gamma}}$ and the
  claim follows.
\end{proof}

The following proposition gives a lower bound to the smooth coherent relative entropy.  This will
prove crucial to the proof of the asymptotic equipartition theorem.

\begin{proposition}
  \label{prop:coh-rel-entr-smooth-lower-bound-diff-Dmin-Dmax}
  Let $\epsilon',\epsilon''\geqslant 0$ and $\epsilon'''>0$. Let
  $\epsilon\geqslant 2\sqrt{2\epsilon'} +
  2\sqrt{2`*(\epsilon''+\epsilon''')}$. Then
  \begin{align}
    &\DCohx[\epsilon]{\rho}{X}{X'}{\Gamma_X}{\Gamma_{X'}}
      \nonumber \\
    &\qquad \geqslant
    \Dminz[\epsilon'']{\sigma_X}{\Gamma_X}
    - \Dmax[\epsilon']{\rho_{X'}}{\Gamma_{X'}}
    + \log\frac{\epsilon'''^2}{2+\epsilon'''^2}\ ,
    \label{eq:prop-coh-rel-entr-smooth-lower-bound-diff-Dmin-Dmax}
  \end{align}
  where $\sigma_{X} = t_{X\to R_X}`(\rho_{R_X})$.
\end{proposition}
\begin{proof}[*prop:coh-rel-entr-smooth-lower-bound-diff-Dmin-Dmax]
  Let $\tilde\rho_{R_X}, \tilde\rho_{X'}$ be quantum states which are optimal
  smoothed states for the quantities
  \begin{subequations}
    \begin{align}
      \Dminz[\epsilon'']{\rho_{R_X}}{\Gamma_{R_X}}
      &= \Dminz{\tilde\rho_{R_X}}{\Gamma_{R_X}}\ . \\
      \Dmax[\epsilon']{\rho_{X'}}{\Gamma_{X'}}
      &= \Dmax{\tilde\rho_{X'}}{\Gamma_{X'}}\ .
    \end{align}
  \end{subequations}
  With $\epsilon'''> 0$ and using \autoref{prop:smooth-Drob-to-Dmin}, we know
  that
  \begin{align}
    \Dr[\epsilon''']{\tilde\rho_{R_X}}{\Gamma_{R_X}}
    \geqslant
    \Dminz{\tilde\rho_{R_X}}{\Gamma_{R_X}}
    + \log\frac{\epsilon'''^2}{2+\epsilon'''^2}\ .
  \end{align}
  Let $\tilde{\tilde{\rho}}_{R_X}$ be the optimal smoothed state for
  $\Dr[\epsilon''']{\tilde\rho_{R_X}}{\Gamma_{R_X}}$, such that
  \begin{align}
    \Dr{\tilde{\tilde\rho}_{R_X}}{\Gamma_{R_X}}
    = \Dr[\epsilon''']{\tilde\rho_{R_X}}{\Gamma_{R_X}}\ .
  \end{align}
  At this point, we have
  \begin{multline}
    \Dr{\tilde{\tilde\rho}_{R_X}}{\Gamma_{R_X}}
    - \Dmax{\tilde\rho}{\Gamma_{X'}}
    \\
    \geqslant \Dminz[\epsilon'']{\rho_{R_X}}{\Gamma_{R_X}}
    - \Dmax[\epsilon']{\rho_{X'}}{\Gamma_{X'}}
    + \log\frac{\epsilon'''^2}{2+\epsilon'''^2}\ ,
    \label{eq:prop-coh-rel-entr-AEP-proof-lowbound-calc-1}
  \end{multline}
  with
  \begin{align}
    P`(\tilde\rho_{X'},\rho_{X'}) &\leqslant \epsilon'\ ;
    &
      P`(\tilde\rho_{R_X},\rho_{R_X}) &\leqslant \epsilon''\ ;
    &
      P`(\tilde{\tilde\rho}_{R_X},\tilde\rho_{R_X}) &\leqslant \epsilon'''\ .
  \end{align}
  Now, we'll apply \autoref{util:smoothing-part-of-bipartite-state-rest-exact}
  twice to construct a state close to $\rho_{X'R_X}$ which has marginals
  $\tilde\rho_{X'}$ and $\tilde{\tilde\rho}_{R_X}$ exactly. Let $\tau_{X'R_X}$
  be the quantum state given by
  \autoref{util:smoothing-part-of-bipartite-state-rest-exact} satisfying
  \begin{align}
    \tau_{X'} &= \tilde\rho_{X'}\ ;
    &
      \tau_{R_X} &= \rho_{R_X}\ ;
    &
      P`(\tau_{X'R_X},\rho_{X'R_X}) &\leqslant 2\sqrt{2\epsilon'}\ .
  \end{align}
  Applying \autoref{util:smoothing-part-of-bipartite-state-rest-exact} again, let
  $\tau'_{X'R}$ be a quantum state close to $\tau_{X'R}$ such that
  \begin{align}
    \tau'_{X'} &= \tilde\rho_{X'}\ ;
    &
      \tau'_{R_X} &= \tilde{\tilde\rho}_{R_X}\ ;
    &
      P`(\tau'_{X'R_X},\tau_{X'R_X}) &\leqslant 2\sqrt{2`(\epsilon''+\epsilon''')}\ .
  \end{align}
  We thus have by triangle inequality
  \begin{align}
    P`(\tau'_{X'R_X},\rho_{X'R_X})
    &\leqslant 2\sqrt{2\epsilon'} + 2\sqrt{2`(\epsilon''+\epsilon''')}\ .
  \end{align}
  By \autoref{prop:coh-rel-entr-lower-bound-diff-dr-dmax} we can now write
  \begin{multline}
    \DCohz{\tau'}{X}{X'}{\Gamma_X}{\Gamma_{X'}}
    \geqslant \Dr{\tau'_{R_X}}{\Gamma_{R_X}}
      - \Dmax{\tau'_{X'}}{\Gamma_{X'}}
      \\
    =  \Dr{\tilde{\tilde\rho}_{R_X}}{\Gamma_{R_X}}
      - \Dmax{\tilde\rho_{X'}}{\Gamma_{X'}}\ .
    \label{eq:prop-coh-rel-entr-AEP-proof-lowbound-calc-2}
  \end{multline}
  Observe now that $\tau'_{X'R_X}$ is a valid optimization candidate for
  $\DCohz[\epsilon]{\rho}{X}{X'}{\Gamma_X}{\Gamma_{X'}}$. Hence
  \begin{align}
    \DCohz[\epsilon]{\rho}{X}{X'}{\Gamma_X}{\Gamma_{X'}}
    \geqslant \DCohz{\tau'}{X}{X'}{\Gamma_X}{\Gamma_{X'}}\ .
    \label{eq:prop-coh-rel-entr-AEP-proof-lowbound-calc-3}
  \end{align}

  Finally, inequality~\eqref{eq:prop-coh-rel-entr-AEP-proof-lowbound-calc-3} followed
  by~\eqref{eq:prop-coh-rel-entr-AEP-proof-lowbound-calc-2}
  and~\eqref{eq:prop-coh-rel-entr-AEP-proof-lowbound-calc-1} provides us the seeked lower
  bound.
\end{proof}

We also have a bound which applies to product states, given in terms of min- and
max-relative entropies of input and output.  Physically, it asserts that a
possible strategy for implementing the product state process matrix is to
completely {erase} the input state (at a cost given by the {min-relative
  entropy}), and subsequently {prepare} the required output state (at a yield
given by the {max-relative entropy}).

\begin{proposition}[coherent relative entropy for product states]
  \label{prop:coh-rel-entr-for-product-states}
  For states $\sigma_X$ and $\rho_{X'}$, we have
  \begin{multline}
    \DCohz{*t_{X\to{}R_X}`(\sigma_X)\otimes\rho_{X'}}{X}{X'}{\Gamma_X}{\Gamma_{X'}}
    \\
    \geqslant \Dminz{\sigma_X}{\Gamma_X}
    - \Dmax{\rho_{X'}}{\Gamma_{X'}}\ .
  \end{multline}
\end{proposition}
\begin{proof}[*prop:coh-rel-entr-for-product-states]
  Write $\sigma_{R_X} = t_{X\to R_X}`(\sigma_X)$.  Choose
  $T_{X'R_X} = \Pi^{\sigma_{R_X}}_{R_X} \otimes\rho_{X'}$. This choice trivially
  satisfies~\eqref{eq:SDP-coh-rel-entr-nonsmooth-cond-trnoninc}. Also,
  $\sigma_{R_X}^{1/2} T_{X'R_X}\sigma_{R_X}^{1/2} =
  \sigma_{R_X}\otimes\rho_{X'}$
  so~\eqref{eq:SDP-coh-rel-entr-nonsmooth-cond-correctmapping} is also
  satisfied.  We have that $\tr_{R_X} T_{X'R_X}\Gamma_{R_X}$ lies in the support
  of $\Gamma_{X'}$ because $\rho_{X'}$ does so, and as per
  \autoref{prop:coh-rel-entr-rewrite-primal-problem-infinity-norm} the optimal
  value of $\alpha$ corresponding to this $T_{X'R_X}$ is given by
  \begin{align}
    \alpha
    &= \norm[\big]{ \Gamma_{X'}^{-1/2} \tr_{R_X} `\big[ T_{X'R_X}\Gamma_{R_X} ]
      \Gamma_{X'}^{-1/2} }_\infty
      \nonumber\\
    &= \norm[\big]{\Gamma_{X'}^{-1/2} \tr_{R_X}`*[
      `\big(\Pi^{\sigma_{R_X}}_{R_X}\otimes\rho_{X'}) \Gamma_{R_X} ]
      \Gamma_{X'}^{-1/2}}_\infty
      \nonumber\\
    &= \tr_{R_X}`[\Pi^{\sigma_{R_X}}_{R_X} \Gamma_{R_X}] \,
      \norm[\big]{\Gamma_{X'}^{-1/2} \rho_{X'} \Gamma_{X'}^{-1/2}}_\infty
      \nonumber\\
    &= 2^{-\Dminz{\sigma_{R_X}}{\Gamma_{R_X}}}\,
      2^{\Dmax{\rho_{X'}}{\Gamma_{X'}}}\ .
  \end{align}
  This choice of $\alpha$ and $T_{X'R_X}$ is feasible for
  $2^{-\DCohz{*\sigma_{R_X}\otimes\rho_{X'}}{X}{X'}{\Gamma_X}{\Gamma_{X'}}}$,
  hence
  \begin{multline}
    \DCohz{*\sigma_{R_X}\otimes\rho_{X'}}{X}{X'}{\Gamma_X}{\Gamma_{X'}}
    \\
    \geqslant \Dminz{\sigma_X}{\Gamma_X} - \Dmax{\rho_{X'}}{\Gamma_{X'}}\ .
    \tag*\qedhere
  \end{multline}
\end{proof}

\subsection{Asymptotic equipartition property}
\label{sec:AEP}

Finally, the coherent relative entropy also obeys an asymptotic equipartition
property in the i.i.d.\@ limit.  In this limit, the coherent relative entropy
converges to the difference of relative entropies of the input and the output to
the respective $\Gamma$ operators.

Both versions of the coherent relative entropy we have introduced have the same
asymptotic behavior for small $\epsilon$.  For completeness we present the
detailed statements, including the ranges of $\epsilon$ for which the property
is proven for each quantity.

\begin{proposition}[Asymptotic equipartition property]
  \label{prop:coh-rel-entr-AEP}
  For any $\Gamma_X,\Gamma_{X'}\geqslant 0$, for any quantum state
  $\rho_{X'R_X}$, and for any $0<\epsilon<1/2$,
  \begin{multline}
    \lim_{n\to\infty}\frac1n
    \DCohx[\epsilon]`\big{*\rho_{X'R_X}^{\otimes n}}{X^n}{X'^n}%
    {\Gamma_X^{\otimes n}}{\Gamma_{X'}^{\otimes n}}
    \\[1ex]
    = \DD{\sigma_X}{\Gamma_X} - \DD{\rho_{X'}}{\Gamma_{X'}}\ ,
    \label{eq:prop-coh-rel-entr-AEP}
  \end{multline}
  where $\sigma_X = t_{X\to X'}`(\rho_{R_X})$.

  Similarly, for any $0<\epsilon<`*(18)^{-4}$,
  \begin{multline}
    \lim_{n\to\infty}\frac1n
    \DCohz[\epsilon]`\big{*\rho_{X'R_X}^{\otimes n}}{X^n}{X'^n}%
    {\Gamma_X^{\otimes n}}{\Gamma_{X'}^{\otimes n}}
    \\[1ex]
    = \DD{\sigma_X}{\Gamma_X} - \DD{\rho_{X'}}{\Gamma_{X'}}\ .
    \label{eq:prop-coh-rel-entr-AEP-2}
  \end{multline}
\end{proposition}

While the original asymptotic equipartition statements in the context of smooth
entropies (e.g., refs.~\cite{PhDRenner2005_SQKD,Datta2009IEEE_minmax})
considered first the limit $n\to\infty$, and then $\epsilon\to0$, the above
proposition is slightly more general in that the limit $\epsilon\to0$ is not
necessary (in line with, e.g., refs.~\cite{PhDTomamichel2012,%
  Tomamichel2013_hierarchy,Dupuis2013_DH}).  However, one may ask if it is
possible to take $\epsilon\to0$ simultaneously with $n\to\infty$.  We may indeed
prove such a statement using recent results on moderate deviation
analysis~\cite{Chubb2017CMP_moderate,Cheng2017arXiv_moderate}.

\begin{proposition}[Asymptotic equipartition property, take 2]
  \label{prop:coh-rel-entr-AEP-moderate}
  Consider any $\Gamma_X,\Gamma_{X'}\geqslant 0$, and  any quantum state
  $\rho_{X'R_X}$.
  Let $(\epsilon_n)$ be a sequence such that $\epsilon_n\to0$ and
  $-(1/n)\ln(\epsilon_n)\to 0$.  Then\footnote{The condition on the sequence
    $(\epsilon_n)$ corresponds to requiring that $(\epsilon_n)$ results from a
    \emph{moderate sequence} as defined in~\cite{Chubb2017CMP_moderate}.  It is
    equivalent to requiring that the sequence $(\epsilon_n)$ converges to zero
    slower than $\exp(-n)$. (For example, this is satisfied if
    $\epsilon_n\sim 1/\poly(n)$.)}
  \begin{multline}
    \lim_{n\to\infty}\frac1n
    \DCohx[\epsilon_n]`\big{*\rho_{X'R_X}^{\otimes n}}{X^n}{X'^n}%
    {\Gamma_X^{\otimes n}}{\Gamma_{X'}^{\otimes n}}
    \\[1ex]
    = \DD{\sigma_X}{\Gamma_X} - \DD{\rho_{X'}}{\Gamma_{X'}}\ ,
    \label{eq:prop-coh-rel-entr-AEP-moderate}
  \end{multline}
  and
  \begin{multline}
    \lim_{n\to\infty}\frac1n
    \DCohz[\epsilon_n]`\big{*\rho_{X'R_X}^{\otimes n}}{X^n}{X'^n}%
    {\Gamma_X^{\otimes n}}{\Gamma_{X'}^{\otimes n}}
    \\[1ex]
    = \DD{\sigma_X}{\Gamma_X} - \DD{\rho_{X'}}{\Gamma_{X'}}\ .
    \label{eq:prop-coh-rel-entr-AEP-moderate-2}
  \end{multline}
  where $\sigma_X = t_{X\to X'}`(\rho_{R_X})$.
\end{proposition}

The proof of the asymptotic equipartition follows from bounds we have derived
above using the min- and max-relative entropies.  The latter have known
asymptotic behavior, and they converge to the usual quantum relative
entropy~\cite{Datta2009IEEE_minmax}.
Note that our definitions of the smooth min- and max-relative entropy differ in
minor details from the ones originally introduced in
ref.~\cite{Datta2009IEEE_minmax}.  For completeness, we hence provide an adapted
proof of the asymptotic equipartition property for the min- and max-relative
entropy.  Our proof proceeds via the hypothesis testing entropy, whose
asymptotic behavior has been thoroughly studied~\cite{%
  Hiai1991CMP_proper,Ogawa2000IEEETIT_Stein,%
  Bjelakovic2003arXiv_revisted,%
  Tomamichel2013_hierarchy,%
  Dupuis2013_DH,%
  Mosonyi2014CMP_hypothesis}.  This will allow us to prove
\autoref{prop:coh-rel-entr-AEP-moderate} via direct application of the results
in refs.~\cite{Chubb2017CMP_moderate,Cheng2017arXiv_moderate}.

\begin{lemma}[Bounds for min- and max-relative entropy in terms of hypothesis
  testing entropy]
  \label{prop:min-max-rel-entropy-to-DHyp}
  The following bounds hold for any $0<\epsilon<1/2$:
  \begin{subequations}
    \begin{align}
      \Dminz[\epsilon]{\sigma}{\Gamma}
      &\leqslant \DHyp[1-\epsilon]{\sigma}{\Gamma} - \log`(1-\epsilon)\ ; \\
      \Dminz[\epsilon]{\sigma}{\Gamma}
      &\geqslant \DHyp[1-\epsilon']{\sigma}{\Gamma} - \log\frac{1-\epsilon'}{\epsilon'}\ ;\\
      \Dmax[\epsilon]{\rho}{\Gamma}
      &\geqslant \DHyp[2\epsilon]{\rho}{\Gamma} - 1\ ; \\
        \Dmax[\epsilon]{\rho}{\Gamma}
      &\leqslant \DHyp[\epsilon^2/2]{\rho}{\Gamma} - \log`(1-\epsilon)\ ,
    \end{align}
    for any $0<\epsilon' \leqslant \epsilon^2/(4+2\epsilon^2)$; we may choose,
    e.g., $\epsilon'=\epsilon^2/6$.
  \end{subequations}
\end{lemma}

Recall that, as a direct consequence of Quantum Stein's
lemma~\cite{Hiai1991CMP_proper,Ogawa2000IEEETIT_Stein,Dupuis2013_DH}, we have
that for all $0<\epsilon<1$,
\begin{align}
  \lim_{n\to\infty}\frac1n
  \DHyp[\epsilon]{\sigma^{\otimes n}}{\Gamma^{\otimes n}}
  &= \DD{\sigma}{\Gamma} \ .
    \label{eq:AEP-DHyp}
\end{align}
As an immediate consequence of \autoref{prop:min-max-rel-entropy-to-DHyp} and
of~\eqref{eq:AEP-DHyp}, we find that for any $0<\epsilon<1/2$,
\begin{subequations}
  \label{eq:AEP-min-max-rel-entropy}
  \begin{align}
    \lim_{n\to\infty}\frac1n
    \Dminz[\epsilon]{\sigma_X^{\otimes n}}{\Gamma_X^{\otimes n}}
    &= \DD{\sigma_X}{\Gamma_X} \ ; \\
    \lim_{n\to\infty}\frac1n
    \Dmax[\epsilon]{\rho_{X'}^{\otimes n}}{\Gamma_{X'}^{\otimes n}}
    &= \DD{\rho_{X'}}{\Gamma_{X'}} \ ,
  \end{align}
\end{subequations}
noting that terms which scale sublinearly in $n$, for instance
$\log(1-\epsilon)$, disappear because of the factor $1/n$ in the limit
$n\to\infty$.

\begin{proof}[*prop:min-max-rel-entropy-to-DHyp]
  Let $\tilde\sigma$ be optimal for $\Dminz[\epsilon]{\sigma}{\Gamma}$, i.e.,
  $\Dminz[\epsilon]{\sigma}{\Gamma} = \Dminz{\tilde\sigma}{\Gamma} = -\log
  \tr`[\Pi^{\tilde\sigma}\,\Gamma]$ with
  $P(\sigma,\tilde\sigma)\leqslant\epsilon$.  As
  $\sigma\geqslant\tilde\sigma-\Delta$ for some $\Delta\geqslant0$ with
  $\tr\Delta\leqslant\epsilon$, we have that
  $\tr`[\Pi^{\tilde\sigma}\sigma]\geqslant 1 -
  \tr(\Pi^{\tilde\sigma}\Delta)\geqslant 1 - \epsilon$.  Then
  $\Pi^{\tilde\sigma}$ is feasible in the primal program for
  $2^{-\DHyp[1-\epsilon]{\sigma}{\Gamma}}$, achieving the value
  $(1-\epsilon)^{-1}\tr`(\Pi^{\tilde\sigma}\Gamma)$.  Hence, for any
  $0<\epsilon<1$,
  \begin{align}
    \Dminz[\epsilon]{\sigma}{\Gamma}
    \leqslant \DHyp[1-\epsilon]{\sigma}{\Gamma} - \log`(1-\epsilon)\ .
  \end{align}
  Conversely, for any $0<\epsilon'<1/2$ to be fixed later, let $Q$ be primal
  optimal for
  $2^{-\DHyp[1-\epsilon']{\sigma}{\Gamma}} = (1-\epsilon')^{-1}\tr`(Q\Gamma)$
  with $\tr`(Q\sigma)\geqslant 1-\epsilon'$.  For $\eta=\epsilon'$, Let $P^\eta$
  be the projector onto the eigenspaces of $Q$ associated to eigenvalues greater
  than or equal to $\eta$, and hence satisfying $\eta P^\eta \leqslant Q$.  It
  follows that $\tr`(Q\Gamma)\geqslant \eta \tr`(P^\eta\Gamma)$.  Now, define
  $\tilde\sigma = P^\eta\,\sigma\,P^\eta / \tr`(P^\eta\sigma)$, noting that
  $\tr`(P^\eta\sigma)\geqslant\tr`(P^\eta Q P^\eta \sigma) \geqslant
  \tr`(Q\sigma) - \tr`( (\Ident-P^\eta)\,Q\,(\Ident-P^\eta)\,\sigma ) \geqslant
  1 - \epsilon' - \eta$ (recall that all eigenvalues of
  $(\Ident-P^\eta)\,Q\,(\Ident-P^\eta)$ are less than $\eta$).
  Using~\autoref{util:purified-distance-to-normalized-projected-normalized-state},
  we see that
  $P`(\tilde\sigma,\sigma)\leqslant \sqrt{2`(\epsilon'+\eta)} /
  \sqrt{1-\epsilon'-\eta} = \sqrt{4\epsilon'/(1-2\epsilon')}$.  Now
  $\tilde\sigma$ is a valid candidate for the smoothing in
  $\Dminz[\sqrt{4\epsilon'/(1-2\epsilon')}]{\sigma}{\Gamma}$, and hence
  $ \Dminz[\sqrt{4\epsilon'/(1-2\epsilon')}]{\sigma}{\Gamma} \geqslant
  -\log\tr`(P^\eta \Gamma) \geqslant -\log`[\eta^{-1}\tr`(Q\Gamma)] =
  -\log`[`(`(1-\epsilon')/\epsilon')\,`(1-\epsilon')^{-1}\tr`(Q\Gamma)] =
  \DHyp[1-\epsilon']{\sigma}{\Gamma} - \log`[`(1-\epsilon')/\epsilon']$.  Now
  consider any $0<\epsilon<1/2$ and assume that
  $0<\epsilon'\leqslant\epsilon^2/(4+2\epsilon^2)$, noting that
  $0<\epsilon'<1/2$.  We have
  $\epsilon'\,(4+2\epsilon^2)\leqslant\epsilon^2 \;\Leftrightarrow\;
  4\epsilon'\leqslant\epsilon^2\,(1-2\epsilon') \;\Leftrightarrow\; \epsilon
  \geqslant \sqrt{4\epsilon'/(1-2\epsilon')}$,
  and thus
  $\Dminz[\epsilon]{\sigma}{\Gamma}\geqslant
  \Dminz[\sqrt{4\epsilon'/(1-2\epsilon')}]{\sigma}{\Gamma}$.  Hence, for any
  $0<\epsilon<1/2$ and for any $0<\epsilon'\leqslant\epsilon^2/(4+2\epsilon^2)$,
  we have:
  \begin{align}
    \Dminz[\epsilon]{\sigma}{\Gamma}
    \geqslant \DHyp[1-\epsilon']{\sigma}{\Gamma} - \log\frac{1-\epsilon'}{\epsilon'}\ .
  \end{align}
  
  For the max-relative entropy, for any $\rho,\Gamma$ and for any
  $0<\epsilon<1/2$, let $\tilde\rho$ be a normalized quantum state such that
  $\Dmax[\epsilon]{\rho}{\Gamma} = \Dmax{\tilde\rho}{\Gamma}$.  Let $Q$ be
  primal optimal for
  $2^{-\DHyp[2\epsilon]{\rho}{\Gamma}}=(2\epsilon)^{-1}\,\tr`(Q\Gamma)$, such
  that $\tr(Q\rho)\geqslant 2\epsilon$.  But $\tilde\rho\geqslant\rho-\Delta$
  for a $\Delta\geqslant0$ with $\tr\Delta\leqslant\epsilon$, since
  $D(\tilde\rho,\rho)\leqslant\epsilon$, and thus
  $\tr`(Q\tilde\rho)\geqslant 2\epsilon-\epsilon=\epsilon$.  Then $Q$ is primal
  feasible also for $\DHyp[\epsilon]{\tilde\rho}{\Gamma}$ and
  $2^{-\DHyp[\epsilon]{\tilde\rho}{\Gamma}} \leqslant
  \epsilon^{-1}\,\tr`(Q\Gamma) = 2\cdot 2^{-\DHyp[2\epsilon]{\rho}{\Gamma}}$.
  Then, using~\cite[Prop.~4.1]{Dupuis2013_DH},
  $\Dmax{\tilde\rho}{\Gamma} \geqslant \DHyp[\epsilon]{\tilde\rho}{\Gamma}
  \geqslant \DHyp[2\epsilon]{\rho}{\Gamma} - 1$, and hence
  \begin{align}
    \Dmax[\epsilon]{\rho}{\Gamma} \geqslant \DHyp[2\epsilon]{\rho}{\Gamma} - 1\ .
  \end{align}
  For a lower bound on $\Dmax[\epsilon]{}{}$, we
  invoke~\cite[Prop.~4.1]{Dupuis2013_DH}; however the quantity called
  $\Dmax[\epsilon]{}{}$ there optimizes over subnormalized states whereas we
  optimize over normalized states only, so we have to work a little more.  For
  any subnormalized state $\tilde\rho$ with $\tr\tilde\rho\geqslant 1-\epsilon$,
  we have by definition that
  $2^{\Dmax{\tilde\rho}{\Gamma}} =
  \norm{\Gamma^{-1/2}\,\tilde\rho\,\Gamma^{-1/2}}_\infty =
  \tr`(\tilde\rho)\,2^{\Dmax{\tilde\rho/\tr\tilde\rho}{\Gamma}} \geqslant
  `(1-\epsilon)\cdot 2^{\Dmax{\tilde\rho/\tr\tilde\rho}{\Gamma}}$, and hence
  \begin{multline}
    \min_{\substack{\tilde\rho:~ \tr\tilde\rho\leqslant 1 \\
    P(\tilde\rho,\rho)\leqslant\epsilon}} \Dmax{\tilde\rho}{\Gamma}
    \geqslant 
    \min_{\substack{\tilde\rho:~ \tr\tilde\rho\leqslant 1 \\
    P(\tilde\rho,\rho)\leqslant\epsilon}} \Dmax{\tilde\rho/\tr\tilde\rho}{\Gamma}
    + \log\,`(1-\epsilon)
    \\
    = \Dmax[\epsilon]{\rho}{\Gamma} + \log\,`(1-\epsilon)\ .
  \end{multline}
  Then, invoking~\cite[Prop.~4.1]{Dupuis2013_DH} for any $0<\epsilon<1$, and
  chaining with the above inequality,
  \begin{align*}
    \DHyp[\epsilon^2/2]{\rho}{\Gamma} \geqslant \Dmax[\epsilon]{\rho}{\Gamma}
    + \log`(1-\epsilon)\ .
    \tag*\qedhere
  \end{align*}
\end{proof}

\begin{proof}[*prop:coh-rel-entr-AEP]
  We start by upper bounding the coherent relative entropy
  $\DCohx[\epsilon]`\big{*\rho^{\otimes{}n}_{X'R_X}}{X^n}{X'^n}%
  {\Gamma_X^{\otimes{}n}}{\Gamma_{X'}^{\otimes{}n}}$.  Thanks to
  \autoref{prop:cohrelentr-smooth-upper-bound-min-max}, choosing
  $\tilde{\epsilon}=\tilde\epsilon'= (1-2\epsilon)/137414920$ 
  with $\tilde{\bar\epsilon}=\tilde{\epsilon}+\tilde{\epsilon}'+2\epsilon$, and
  then using~\eqref{eq:AEP-min-max-rel-entropy},
  \begin{align}
    &\lim_{n\to\infty}\frac1n
    \DCohx[\epsilon]`\big{*\rho_{X'R_X}^{\otimes n}}{X^n}{X'^n}%
    {\Gamma_X^{\otimes n}}{\Gamma_{X'}^{\otimes n}}
    \nonumber\\
    &\qquad \leqslant
      \lim_{n\to\infty}\frac1n`*[
      \Dmax[\tilde\epsilon]`\big{\rho_{R_X}^{\otimes n}}{\Gamma_{R_X}^{\otimes n}}
      -
      \Dminz[\tilde\epsilon']`\big{\rho_{X'}^{\otimes n}}{\Gamma_{X'}^{\otimes n}}
      -
      \log\,`*[\epsilon`(1-\tilde{\bar\epsilon})]
      ]
      \nonumber\\
    &\qquad = 
      \DD{\rho_{R_X}}{\Gamma_{R_X}}
      -
      \DD{\rho_{X'}}{\Gamma_{X'}}\ .
  \end{align}
  The lower bound is given by
  \autoref{prop:coh-rel-entr-smooth-lower-bound-diff-Dmin-Dmax}: Choosing
  $\hat\epsilon'=\hat\epsilon''=\hat\epsilon'''=\epsilon^2/197334000868$,
  \begin{align*}
    &\lim_{n\to\infty}\frac1n
    \DCohx[\epsilon]`\big{*\rho_{X'R_X}^{\otimes n}}{X^n}{X'^n}%
    {\Gamma_X^{\otimes n}}{\Gamma_{X'}^{\otimes n}}
    \\
    &\qquad \geqslant 
    \lim_{n\to\infty}\frac1n `*[
    \Dminz[\hat\epsilon'']`\big{\rho_{R_X}^{\otimes n}}{\Gamma_{R_X}^{\otimes n}}
    - \Dmax[\hat\epsilon']`\big{\rho_{X'}^{\otimes n}}{\Gamma_{X'}^{\otimes n}}
    + \log\frac{\hat\epsilon'''^2}{2 + \hat\epsilon'''^2}
    ]\
    \\
    &\qquad =
    \DD{\rho_{R_X}}{\Gamma_{R_X}}
    - \DD{\rho_{X'}}{\Gamma_{X'}}\ .
  \end{align*}

  Equation~\eqref{eq:prop-coh-rel-entr-AEP-2} follows directly
  from~\eqref{eq:prop-coh-rel-entr-AEP}, using the relations given by
  \autoref{prop:cohrelentr-2-equiv}.
\end{proof}

\begin{proof}[*prop:coh-rel-entr-AEP-moderate]
  Moderate deviation analysis provides a full characterization of the
  second-order asymptotic behavior of the hypothesis testing
  entropy~\cite{Chubb2017CMP_moderate,Cheng2017arXiv_moderate} in cases where
  $\epsilon\to0$ simultaneously with $n\to\infty$.  For our purposes and for
  simplicity we consider the leading order only: For any sequence
  $(\hat{\epsilon}_n)$ such that $\hat\epsilon_n\to 0$ and
  $-(1/n)\ln(\hat{\epsilon}_n)\to 0$, it holds that
  \begin{subequations}
    \label{eq:AEP-DHyp-moderate}
    \begin{align}
      \lim_{n\to\infty}\frac1n
      \DHyp[\hat{\epsilon}_n]{\sigma^{\otimes n}}{\Gamma^{\otimes n}}
      &= \DD{\sigma}{\Gamma} \ ;
        \label{eq:AEP-DHyp-moderate-1}
      \\
      \lim_{n\to\infty}\frac1n
      \DHyp[1-\hat{\epsilon}_n]{\sigma^{\otimes n}}{\Gamma^{\otimes n}}
      &= \DD{\sigma}{\Gamma} \ .
        \label{eq:AEP-DHyp-moderate-2}
    \end{align}
\end{subequations}
  So, we proceed analogously to the proof of \autoref{prop:coh-rel-entr-AEP} via
  the bounds we determined on the coherent relative entropy in terms of the min-
  and max-relative entropies.

  We invoke \autoref{prop:cohrelentr-smooth-upper-bound-min-max} choosing
  $\tilde{\epsilon}_n=\tilde\epsilon'_n= \min(\epsilon_n, (1-2\epsilon_n)/3)$
  and
  $\tilde{\bar\epsilon}_n = \tilde{\epsilon}_n+\tilde{\epsilon}'_n+2\epsilon_n$,
  further observing that $\tilde{\bar\epsilon}_n<1$ and
  $\tilde{\bar\epsilon}_n\leqslant 4\epsilon_n$.  Then
  \begin{align}
    &\lim_{n\to\infty}\frac1n
    \DCohx[\epsilon_n]`\big{*\rho_{X'R_X}^{\otimes n}}{X^n}{X'^n}%
    {\Gamma_X^{\otimes n}}{\Gamma_{X'}^{\otimes n}}
    \nonumber\\
    &\qquad \leqslant
      \lim_{n\to\infty}\frac1n`*[
      \Dmax[\tilde\epsilon_n]`\big{\rho_{R_X}^{\otimes n}}{\Gamma_{R_X}^{\otimes n}}
      -
      \Dminz[\tilde\epsilon'_n]`\big{\rho_{X'}^{\otimes n}}{\Gamma_{X'}^{\otimes n}}
      -
      \log\,`*[\epsilon_n `(1-\tilde{\bar\epsilon}_n)]
      ]
      \nonumber\\
      &\qquad \leqslant
      \lim_{n\to\infty}\frac1n \Bigl[
      \DHyp[\tilde\epsilon_n^2/2]%
      `\big{\rho_{R_X}^{\otimes n}}{\Gamma_{R_X}^{\otimes n}}
      -
      \DHyp[1- \tilde{\epsilon}'^2_n /6]%
      `\big{\rho_{X'}^{\otimes n}}{\Gamma_{X'}^{\otimes n}}
        \nonumber \\
      & \qquad\hphantom{\leqslant\lim_{n\to\infty}\frac1n}\hspace{1ex}
        - \log(1-\tilde\epsilon_n)
        + \log\,`*[\frac{1-\tilde{\epsilon}'^2_n/6}{\tilde{\epsilon}'^2_n/6}]
        - \log\,`*[\epsilon_n `(1-\tilde{\bar\epsilon}_n)]
      \Bigr]
        \nonumber\\
      &\qquad \leqslant
      \lim_{n\to\infty}\frac1n \Bigl[
      \DHyp[\tilde\epsilon_n^2/2]%
      `\big{\rho_{R_X}^{\otimes n}}{\Gamma_{R_X}^{\otimes n}}
      -
      \DHyp[1- \tilde{\epsilon}'^2_n /6]%
      `\big{\rho_{X'}^{\otimes n}}{\Gamma_{X'}^{\otimes n}}
        \nonumber \\
      & \qquad\hphantom{\leqslant\lim_{n\to\infty}\frac1n}\hspace{1em}
        + \log(\poly(\epsilon_n) / \poly(\epsilon_n)) \Bigr]
        \nonumber\\
    &\qquad = \DD{\rho_{R_X}}{\Gamma_{R_X}}
      -
      \DD{\rho_{X'}}{\Gamma_{X'}}\ ,
  \end{align}
  where we used \autoref{prop:min-max-rel-entropy-to-DHyp} in the second
  inequality, where $\poly(\epsilon_n)$ denotes a polynomial in $\epsilon_n$ of
  arbitrary but constant degree, and where we used~\eqref{eq:AEP-DHyp-moderate}
  for the last equality, noting that $(1/n)\log(\poly(\epsilon_n))\to0$ as
  $n\to\infty$, using the assumption in the claim that
  $-(1/n)\ln(\epsilon_n)\to0$ as $n\to\infty$.

  The other direction follows similarly: We apply
  \autoref{prop:coh-rel-entr-smooth-lower-bound-diff-Dmin-Dmax} choosing
  $\bar\epsilon'_n=\bar\epsilon''_n=\bar\epsilon'''_n=\epsilon_n^2/64$, such
  that
  $2\sqrt{2\bar\epsilon'_n}+2\sqrt{2(\bar\epsilon''_n+\bar\epsilon'''_n)} =
  2\sqrt{2\epsilon_n^2/64} + 2\sqrt{4\epsilon_n^2/64} \leqslant \epsilon_n$;
  then
  \begin{align}
    &\lim_{n\to\infty}\frac1n
    \DCohx[\epsilon_n]`\big{*\rho_{X'R_X}^{\otimes n}}{X^n}{X'^n}%
    {\Gamma_X^{\otimes n}}{\Gamma_{X'}^{\otimes n}}
      \nonumber \\
    &\qquad \geqslant 
    \lim_{n\to\infty}\frac1n `*[
    \Dminz[\bar\epsilon''_n]`\big{\rho_{R_X}^{\otimes n}}{\Gamma_{R_X}^{\otimes n}}
    - \Dmax[\bar\epsilon'_n]`\big{\rho_{X'}^{\otimes n}}{\Gamma_{X'}^{\otimes n}}
    + \log\frac{\bar\epsilon'''^2_n}{2 + \bar\epsilon'''^2_n}
    ]\
      \nonumber \\
    &\qquad \geqslant 
    \lim_{n\to\infty}\frac1n \Bigl[
      \DHyp[1-\bar\epsilon''^2_n/6]%
      `\big{\rho_{R_X}^{\otimes n}}{\Gamma_{R_X}^{\otimes n}}
    - \DHyp[\bar\epsilon'^2_n/2]`\big{\rho_{X'}^{\otimes n}}{\Gamma_{X'}^{\otimes n}}
      \nonumber \\
      & \qquad\hphantom{\leqslant\lim_{n\to\infty}\frac1n}\hspace{1em}
        - \log\,`*[\frac{1-\bar{\epsilon}''^2_n/6}{\bar{\epsilon}''^2_n/6}]
        + \log(1-\bar\epsilon'_n)
        + \log(\bar\epsilon'''^2_n/3)
        \Bigr]
    \nonumber\\
    &\qquad \geqslant 
    \lim_{n\to\infty}\frac1n \Bigl[
      \DHyp[1-\bar\epsilon''^2_n/6]%
      `\big{\rho_{R_X}^{\otimes n}}{\Gamma_{R_X}^{\otimes n}}
    - \DHyp[\bar\epsilon'^2_n/2]`\big{\rho_{X'}^{\otimes n}}{\Gamma_{X'}^{\otimes n}}
      \nonumber \\
      & \qquad\hphantom{\leqslant\lim_{n\to\infty}\frac1n}\hspace{1em}
        + \log\,`*(\poly(\epsilon_n)/\poly(\epsilon_n))
        \Bigr]
      \nonumber \\
    &\qquad =
      \DD{\rho_{R_X}}{\Gamma_{R_X}}
      - \DD{\rho_{X'}}{\Gamma_{X'}}\ .
  \end{align}

  Equation~\eqref{eq:prop-coh-rel-entr-AEP-moderate-2} follows directly
  from~\eqref{eq:prop-coh-rel-entr-AEP-moderate}, using the relations given by
  \autoref{prop:cohrelentr-2-equiv}.
\end{proof}

\section{Robustness of battery states to smoothing}
\appendixlabel{appx:sec:robustness-battery-states}

Because the {battery} system is a part of the physical implementation of the
process, we may ask why it is not included in the definition of the smooth
coherent relative entropy~\eqref{eq:coh-rel-entr-def} in a way which would allow
the physical implementation to fail to produce the appropriate output battery
state with a small probability.
Remarkably, there would have been no difference had we chosen to smooth the
battery states as well.  This follows from the following proposition, which
asserts that optimization candidates which include {smoothing} on the {battery
  states} are in fact already included in the optimization in the definition
above.  This holds for the general battery states of the form
$P_A\Gamma_A P_A /\tr`*(P_A\Gamma_A)$, for a projector $P_A$ commuting with the
$\Gamma_A$ of the battery (see
\autoref{item:prop-equiv-battery-models-equivstatement-projGamma} of
\autoref{prop:equiv-battery-models}).

\begin{proposition}[Smoothing battery states]
  \label{prop:coh-rel-entr-smoothing-battery-states}
  Let $A,A'$ be quantum systems with corresponding $\Gamma_A$, $\Gamma_{A'}$.
  Let $P_A$, $P'_{A'}$ be projectors such that $[P_A,\Gamma_A]=0$ and
  $[P_{A'},\Gamma_{A'}]=0$, and let $\Phi_{XA\to X'A'}$ be a trace
  nonincreasing, completely positive map such that
  $\Phi_{XA\to X'A'}`(\Gamma_X\otimes\Gamma_A) \leqslant
  \Gamma_{X'}\otimes\Gamma_{A'}$, and such that
  \begin{align}
    P\,\bigg[\Phi_{XA\to X'A}`*(
    \sigma_{XR}\otimes\frac{P_A\Gamma_A P_A}{\tr P_A\Gamma_A})
    \,,
    \; \rho_{X'R}\otimes\frac{P'_{A'}\Gamma_{A'}P'_{A'}}{\tr P'_{A'}\Gamma_{A'}} \bigg]
    \leqslant \epsilon
    \ ,
  \end{align}
  Then there exists a trace-nonincreasing, completely positive map
  $\mathcal{T}_{X\to X'}$ such both the following conditions hold:
  \begin{subequations}
    \begin{gather}
      P`(\mathcal{T}_{X\to X'}`(\sigma_{XR}), \rho_{X'R})\leqslant\epsilon\ ;  \\
      \mathcal{T}_{X\to X'}`(\Gamma_X) \leqslant
      \frac{\tr`(P'_{A'}\Gamma_{A'})}{\tr`(P_{A}\Gamma_{A})}\,\Gamma_{X'}\ .
    \end{gather}
  \end{subequations}
\end{proposition}
\begin{proof}[*prop:coh-rel-entr-smoothing-battery-states]
  Define, for any $\omega_X$,
  \begin{align}
    \mathcal{T}_{X\to X'}`(\omega_X)
    = \tr_{A'}\,`*[ P'_{A'}\;
    \Phi_{XA\to X'A'}`*(\omega_X\otimes\frac{P_A\Gamma_A P_A}{\tr P_A\Gamma_A})]\ .
    \label{eq:prop-coh-rel-entr-smoothing-battery-states-def-mathcalT}
  \end{align}
  Then
  \begin{multline}
    \mathcal{T}_{X\to X'}`(\sigma_{XR})
    = \tr_{A'}\,`*[ P'_{A'}\;
    \Phi_{XA\to X'A'}`*(\sigma_{XR}\otimes\frac{P_A\Gamma_A P_A}{\tr P_A\Gamma_A})]
    \\[0.5ex]
    = \tr_{A'}\,`*[ P'_{A'}\, \tilde\rho_{A'X'R} ]\ ,
  \end{multline}
  where
  $\tilde\rho_{A'X'R} := \Phi_{XA\to X'A'}`( \sigma_{XR}\otimes\frac{P_A\Gamma_A
    P_A}{\tr P_A\Gamma_A})$ satisfies
  $P`(\tilde\rho_{A'X'R}, \rho_{X'R}\otimes\frac{P'_{A'}\Gamma_{A'}P'_{A'}}{\tr
    P'_{A'}\Gamma_{A'}}) \leqslant\epsilon$ by assumption.  Using the
  monotonicity of the purified distance \cite{PhDTomamichel2012} in particular
  under the trace-nonincreasing completely positive map
  $\tr\,`\big[P'_{A'}\,`(\cdot)]$, we have
  \begin{align}
    P`\big(\mathcal{T}_{X\to X'}`(\sigma_{XR}), \rho_{X'R}) \leqslant \epsilon\ .
  \end{align}

  We also have
  \begin{multline}
    \mathcal{T}_{X\to X'}`(\Gamma_X)
    = \tr_{A'}\,`*[ P'_{A'} \, \Phi_{XA\to X'A'}`(
    \Gamma_X\otimes\frac{P_A\Gamma_A P_A}{\tr P_A\Gamma_A}) ]
    \\
    \leqslant \frac1{\tr P_A\Gamma_A}\cdot \tr_{A'}\,`*[
    P'_{A'} \, \Gamma_{X'}\otimes\Gamma_{A'} ]\ ,
  \end{multline}
  using the fact that
  $P_A\Gamma_A P_A = \Gamma_A^{1/2} P_A\Gamma_A^{1/2}\leqslant\Gamma_A$
  (because
  $[P_A,\Gamma_A]=0$) and also with the fact that $\Phi_{XA\to X'A'}$ is
  $\Gamma$-sub-preserving.  Then
  \begin{align}
    \mathcal{T}_{X\to X'}`(\Gamma_X)
    \leqslant \frac{\tr P'_{A'}\Gamma_{A'}}{\tr P_A\Gamma_A}\, \Gamma_{X'}\ ,
  \end{align}
  which completes the proof.
\end{proof}

This means that the processes which also allow \qq{fuzziness} on the battery
states are \emph{de facto} already included in the optimization defining the
smooth coherent relative entropy~\eqref{eq:coh-rel-entr-def}.  This is
formulated explicitly in the following corollary.

\begin{corollary}
  \label{cor:coh-rel-entr-smooth-expr-with-process-on-battery}
  Let $\rho_{X'R_X}$ be a subnormalized state, let
  $\Gamma_X,\Gamma_{X'}\geqslant0$ and let $\epsilon>0$.  Then
  \begin{multline}
    \DCohz[\epsilon]\rho{X}{X'}{\Gamma_X}{\Gamma_{X'}}
    \\
    = \max_{A,A', P_A,P'_{A'}, \Phi_{XA\to X'A'}}
    -\log\frac{\tr P'_{A'}\Gamma_{A'}}{\tr P_{A}\Gamma_{A}}\ ,
    \label{eq:cor-coh-rel-entr-smooth-expr-with-process-on-battery-max}
  \end{multline}
  where the optimization is performed over all systems $A$, $A'$, all operators
  $\Gamma_A$, $\Gamma_{A'}$, and all projectors $P_{A}$, $P'_{A'}$ such that
  $[P_A,\Gamma_A]=0$ and $[P'_{A'},\Gamma_{A'}]=0$, for which there is a trace
  nonincreasing, completely positive map $\Phi_{XA\to X'A'}$ satisfying
  $\Phi_{XA\to X'A'}`*( \Gamma_{X}\otimes\Gamma_A ) \leqslant
  \Gamma_{X'}\otimes\Gamma_{A'}$ as well as
  \begin{multline}
    P\,\bigg[
    \Phi_{XA\to X'A'}`*(\sigma_{XR}\otimes\frac{P_A\Gamma_A P_A}{\tr P_A\Gamma_A} )\, ,
    \\ \;
    \rho_{X'R}\otimes\frac{P'_{A'}\Gamma_{A'} P'_{A'}}{\tr P'_{A'}\Gamma_{A'}}
    \bigg] \leqslant\epsilon \ .
    \label{eq:cor-coh-rel-entr-smooth-expr-with-process-on-battery-condition-on-Phi}
  \end{multline}
\end{corollary}
\begin{proof}[*cor:coh-rel-entr-smooth-expr-with-process-on-battery]
  First, let $A$, $A'$, $P_A$, $P'_{A'}$, $\Gamma_A$, $\Gamma_{A'}$ and
  $\Phi_{XA\to X'A'}$ satisfy the conditions of the
  maximization~\eqref{eq:cor-coh-rel-entr-smooth-expr-with-process-on-battery-max}.
  Let $\mathcal{T}_{X\to X'}$ the mapping given by
  \autoref{prop:coh-rel-entr-smoothing-battery-states}.  Observe that
  $\norm[\big]{
    \Gamma_{X'}^{-1/2}\mathcal{T}_{X\to{}X'}`(\Gamma_X)\,\Gamma_{X'}^{-1/2}
  }_\infty \leqslant `(\tr{P'_{A'}\Gamma_{A'}})/`(\tr{P_{A}\Gamma_{A}})$.  Note
  also that
  $P(\mathcal{T}_{X\to X'}`(\sigma_{XR}),\rho_{X'R_X}) \leqslant \epsilon$ as
  guaranteed by our previous use of
  \autoref{prop:coh-rel-entr-smoothing-battery-states}.  Hence,
  $\mathcal{T}_{X\to X'}$ is a valid candidate in the optimization given
  by \autoref{prop:coh-rel-entr-rewrite-primal-problem-infinity-norm} for
  $\DCohz[\epsilon]{\rho}{X}{X'}{\Gamma_X}{\Gamma_{X'}}$.  Hence
  \begin{align}
    \DCohz[\epsilon]{\rho}{X}{X'}{\Gamma_X}{\Gamma_{X'}}
    \geqslant -\log \frac{\tr{P'_{A'}\Gamma_{A'}}}{\tr{P_{A}\Gamma_{A}}}\ .
  \end{align}

  To show that equality is achieved
  in~\eqref{eq:cor-coh-rel-entr-smooth-expr-with-process-on-battery-max}, let
  $\mathcal{T}_{X\to X'}$ be a valid optimization candidate
  in~\eqref{eq:coh-rel-entr-def} for
  $\DCohz[\epsilon]{\rho}{X}{X'}{\Gamma_X}{\Gamma_{X'}}$ which achieves the
  optimal value
  $y=\DCohz[\epsilon]{\rho}{X}{X'}{\Gamma_R}{\Gamma_{X'}} = -\log\,\norm[\big]{
    \Gamma_{X'}^{-1/2} \mathcal{T}_{X\to{}X'}`*(\Gamma_X)\,
    \Gamma_{X'}^{-1/2}}_\infty$, with
  $P`(\mathcal{T}_{X\to X'}`(\sigma_{XR_X}), \rho_{X'R_X})\leqslant\epsilon$.
  Then $\mathcal{T}_{X\to X'}`*(\Gamma_X)\leqslant 2^{-y}\,\Gamma_{X'}$, and
  this mapping satisfies the conditions of
  \autoref{item:prop-equiv-battery-models-equivstatement-criterionMapEnorm} of
  \autoref{prop:equiv-battery-models}.  Let $A=A'$ be a qubit system with
  $P_A=\proj 0_A$, $P'_{A'}=\proj 1_{A'}$, and
  $\Gamma_A=\Gamma_{A'}=g_0\,\proj 0_A + g_1\,\proj 1_A$, with $g_0/g_1 = 2^y$.
  In virtue of \autoref{item:prop-equiv-battery-models-equivstatement-wit} of
  \autoref{prop:equiv-battery-models}, there exists a trace-nonincreasing,
  completely positive map $\Phi_{XA\to X'A'}$ such that
  $\Phi_{XA\to X'A'}`(\Gamma_{X}\otimes\Gamma_A)\leqslant
  \Gamma_{X'}\otimes\Gamma_{A'}$ and which satisfies
  $\Phi_{XA\to X'A'}`*(`*(\cdot)\,\otimes\proj0_A) =
  \mathcal{T}_{X\to{}X'}`*(\cdot)\otimes\proj1_{A'}$.  Then
  \begin{align}
    \Phi_{XA\to X'A'}`*(\sigma_{XR_X}\otimes\proj0_A)
    = \mathcal{T}_{X\to X'}`(\sigma_{XR_X})\otimes\proj1_{A'}\ ,
  \end{align}
  and hence
  \begin{multline}
    P`(\Phi_{XA\to X'A'}`*(\sigma_{XR_X}\otimes\proj0_A),
    \rho_{X'R_X}\otimes\proj1_{A'})
    \\
    = P`(\mathcal{T}_{X\to X'}`(\sigma_{XR_X}), \rho_{X'R_X})
    \leqslant \epsilon\ .
  \end{multline}
  Hence, all the conditions of the
  maximization~\eqref{eq:cor-coh-rel-entr-smooth-expr-with-process-on-battery-max}
  are satisfied, and the achieved value is indeed
  $\;-\log`[`(\tr P'_{A'}\Gamma_{A'})/`(\tr P_A\Gamma_A)] = -\log`(g_1/g_0) =
  y$.
\end{proof}

\section{Technical Utilities}
\appendixlabel{appx:sec:technical-utilities}

\begin{lemma}
  \label{util:blow-muA-operator-to-try-cover-B-limit}
  Let $A\geqslant 0$, $B\geqslant 0$ and let $\Pi$ be the projector onto the support
  of $A$. Let $\mu>0$. Define $P$ as the projector onto the eigenspaces associated to
  nonnegative eigenvalues of the operator $\left(\mu A - B\right)$. Then there exists
  a constant $c$ which is independent of $\mu$ such that
  \begin{align}
    \label{eq:util-blow-muA-operator-to-try-cover-B-limit-norm}
    \norm{\Pi - P\Pi P}_\infty \leqslant \frac{c}{\mu}\ .
  \end{align}
  In particular,
  \begin{align}
    \label{eq:util-blow-muA-operator-to-try-cover-B-limit-operator-inequality}
    \Pi \leqslant P + \frac{c}{\mu}\Ident\ .
  \end{align}
\end{lemma}
\begin{proof}[*util:blow-muA-operator-to-try-cover-B-limit]
  This lemma follows from a result of perturbation of matrix
  eigenspaces~\cite{BookBhatiaMatrixAnalysis1997}. We'll consider the operators
  $A-\frac1\mu B$ and $A$. Let $Q=\Ident-P$ be the projector on the eigenspaces
  associated to the strictly negative eigenvalues of $A-\frac1\mu B$. Let
  $a_\mathrm{min}=\norm{A^{-1}}_\infty^{-1}$ be the smallest nonzero eigenvalue
  of $A$.  Recall that $\Pi$ projects onto the eigenspaces of $A$ associated to
  eigenvalues larger or equal to $a_\mathrm{min}$. We may now invoke
  \cite[\theoremname~VII.3.1]{BookBhatiaMatrixAnalysis1997}, which asserts that
  for any unitarily invariant norm $\norm{\cdot}_\bullet$,
  \begin{align}
    \norm{Q\Pi}_\bullet \leqslant \frac{1}{\mu a_\mathrm{min}}\norm{Q B\Pi}_\bullet
    \leqslant \frac1{\mu a_\mathrm{min}}\norm{B}_\bullet\ .
  \end{align}
  (The gap $\delta$ in \cite[\theoremname~VII.3.1]{BookBhatiaMatrixAnalysis1997}
  is here the gap between $0$ and $a_\mathrm{min}$.)  In particular, we have
  $\norm{Q\Pi}_\infty\leqslant \left(\mu
    a_\mathrm{min}\right)^{-1}\norm{B}_\infty$.  We then have
  \begin{multline}
    \norm{\Pi - P \Pi P}_\infty
    \leqslant \norm{\Pi - P\Pi}_\infty + \norm{P\Pi - P\Pi P}_\infty
    \\
    \leqslant \norm{\Pi - P\Pi}_\infty + \norm{P}_\infty\norm{\Pi - \Pi P}_\infty
    \\
    = 2\norm{\Pi - P\Pi}_\infty = 2\norm{Q\Pi}_\infty \leqslant \frac{c}\mu\ ,
  \end{multline}
  with $c=2\left(a_\mathrm{min}\right)^{-1}\norm{B}_\infty$. This
  implies~\eqref{eq:util-blow-muA-operator-to-try-cover-B-limit-operator-inequality}
  because
  \begin{align}
    \Pi - P\Pi P \leqslant \frac{c}{\mu}\Ident
    \qquad\Rightarrow\qquad
    \Pi \leqslant P\Pi P + \frac{c}{\mu}\Ident \leqslant P + \frac{c}{\mu}\Ident\ .
    \tag*\qedhere
  \end{align}
\end{proof}

\begin{lemma}
  \label{util:trace-distance-DeltaPlusMinus}
  \index{trace distance}
  Let $\rho$ and $\sigma$ be quantum states.  The trace distance $D\,`(\rho,\sigma)$
  between $\rho$ and $\sigma$ can be written as the semidefinite program in terms of
  the variables $\Delta^\pm\geqslant0$:
  \begin{subequations}
    \begin{align}
      \mathrm{minimize:}\quad & \frac12\,\tr\,`(\Delta^+ + \Delta^-) \\
      \mathrm{subject~to:}\quad & \sigma = \rho +  \Delta^+ - \Delta^-\ .
      \label{eq:util-trace-distance-DeltaPlusMinus-sdp1-constraint}
    \end{align}
  \end{subequations}
  Furthermore, $\tr\Delta^+=\tr\Delta^-=D\,`(\rho,\sigma)$ for the optimal solution.  The
  dual to this program is an alternate expression of the same quantity, in terms of the
  Hermitian variable $Z$:
  \begin{subequations}
    \begin{align}
      \mathrm{maximize:}\quad & \frac12\,\tr\,`(Z\,`(\rho - \sigma)) \\
      \mathrm{subject~to:}\quad & -\Ident \leqslant Z \leqslant \Ident\ .
    \end{align}
  \end{subequations}
\end{lemma}
\begin{proof}[*util:trace-distance-DeltaPlusMinus]
  Write $D\,`(\rho,\sigma)=\frac12\norm{\rho-\sigma}_1$ and recall that for any Hermitian
  $A$, $\norm{A}_1 = \tr\abs{A}$. Choosing $\Delta_\pm\geqslant0$ as the positive and
  negative parts of $\rho-\sigma$, i.e.\@ such that $\rho-\sigma=\Delta+-\Delta_-$, yields
  feasible candidates for the primal problem and
  $\frac12\tr`(\Delta_+ + \Delta_-)= \frac12\tr\abs{\rho-\sigma} = D(\rho,\sigma)$.  Now
  let $\Pi_\pm$ be the projectors onto the strictly positive and strictly negative parts
  of $\rho-\sigma$, respectively, and choose $Z=\Pi_+ - \Pi_-$.  Observe that
  $\Pi_\pm(\rho-\sigma) = \pm\Delta_\pm$.  Then
  $\frac12\tr(Z(\rho-\sigma)) = \frac12\tr`(\Delta_+ + \Delta_-) = D(\rho,\sigma)$.  We
  have exhibited primal and dual candidates achieving the value $D(\rho,\sigma)$, and
  hence this is the optimal solution of the semidefinite program.
  Furthermore~\eqref{eq:util-trace-distance-DeltaPlusMinus-sdp1-constraint} implies that
  $\tr\Delta_+=\tr\Delta_-$ and hence
  $\tr\Delta_+ = \tr\Delta_- = \frac12\tr`(\Delta_+ +\Delta_-) = D(\rho,\sigma)$.
\end{proof}

\begin{lemma}[Gentle measurement lemma for the purified distance]
  \label{util:purified-distance-to-normalized-projected-normalized-state}
  \index{purified distance}
  Let $\rho\geqslant 0$ with $\tr\rho=1$. Let $\epsilon\geqslant 0$. Let $\Pi$ be a
  projector such that $\tr\left(\Pi\rho\right)\geqslant 1-\epsilon$. Then
  \begin{align}
    P`*(\rho,\frac{\Pi\rho\Pi}{\tr\left(\Pi\rho\right)})
    \leqslant \frac{\sqrt{2\epsilon}}{\sqrt{1-\epsilon}}\ .
  \end{align}
\end{lemma}
\begin{proof}[*util:purified-distance-to-normalized-projected-normalized-state]
  Calculate
  \begin{align}
    P^2`*(\rho,\frac{\Pi\rho\Pi}{\tr\left(\Pi\rho\right)})
    &= 1 - F^2`*(\rho,\frac{\Pi\rho\Pi}{\tr\left(\Pi\rho\right)})
    \nonumber\\
    &= \frac1{\tr`(\Pi\rho)}\,`*[\tr`(\Pi\rho) - F^2`*(\rho, \Pi\rho\Pi)]
    \nonumber\\
    &\leqslant \frac1{\tr`(\Pi\rho)}\,`*[1 - F^2`*(\rho, \Pi\rho\Pi)]
    \nonumber\\
    &\leqslant \frac1{1-\epsilon}P^2`*(\rho,\Pi\rho\Pi)\ ,
  \end{align}
  noting that the generalized fidelity is
  $F`(\rho,\sigma)=\norm[\big]{\sqrt\rho\sqrt\sigma}_1$ as long as one of the
  states is normalized, and hence for $a>0$ we have
  $F^2`(\rho,a\sigma) = a\,F^2`(\rho,\sigma)$ if $\tr\rho=1$.  Now, applying
  \cite[\lemmaname~7]{Berta2010_uncertainty}, we have
  \begin{align}
    P`*(\rho,\frac{\Pi\rho\Pi}{\tr\left(\Pi\rho\right)})
    &\leqslant 
      \frac{\sqrt{2\epsilon - \epsilon^2}}{\sqrt{1-\epsilon}}
      = \frac{\sqrt{\epsilon\,`(2-\epsilon)}}{\sqrt{1-\epsilon}}
      \leqslant \frac{\sqrt{2\epsilon}}{\sqrt{1-\epsilon}}\ .
      \tag*\qedhere
  \end{align}
\end{proof}

\begin{lemma}[Smoothing \qq{part of} a state]
  \label{util:smoothing-part-of-bipartite-state-rest-exact}
  \index{trace distance}\index{purified distance}
  Let $\rho_{AB}$ be a bipartite normalized quantum state and let $\tilde\rho_A$ be a
  normalized quantum state such that $D(\tilde\rho_A,\rho_A)\leqslant\delta$.  Then
  there exists a normalized quantum state $\hat\rho_{AB}$ such that
  $\tr_B\hat\rho_{AB} = \tilde\rho_A$, $\tr_A\hat\rho_{AB}=\rho_B$ and
  $P(\hat\rho_{AB},\rho_{AB})\leqslant 2\sqrt{2\delta}$.
\end{lemma}
\begin{proof}[*util:smoothing-part-of-bipartite-state-rest-exact]
  Because $\tilde\rho_A$ and $\rho_A$ are $\delta$-close in trace distance, by
  \autoref{util:trace-distance-DeltaPlusMinus} there exists
  $\Delta^\pm_A\geqslant 0$ such that
  $\tr\Delta^-_A=\tr\Delta^+_A=D\left(\tilde\rho_A,\rho_A\right) \leqslant
  \delta$ and
  \begin{align}
    \tilde\rho_A = \rho_A + \Delta^+_A - \Delta^-_A\ .
  \end{align}

  Let $A = \tilde\rho_A + \Delta^-_A \geqslant 0$ and let
  $M_A = \tilde\rho_A^{1/2}\,A^{-1/2}$.  Observe that
  $M_A^\dagger M_A = A^{-1/2}\,\tilde\rho_A\,A^{-1/2} \leqslant\Ident$ since
  $\tilde\rho_A\leqslant A$.  Now define the completely positive map
  \begin{align}
    \mathcal{M}_{A\to A}`(\cdot)
    = M_A\,`(\cdot)\,M_A^\dagger
    + \tr`\big[`(\Ident-M_A^\dagger{}M_A)`(\cdot)]\,
    \xi_A\ ,
  \end{align}
  with
  $\xi_A := (M_A\,\Delta^+_A\,M_A^\dagger)/\tr`(M_A\,\Delta^+_A\,M_A^\dagger)
  \geqslant 0$ except if $\tr`(M_A\,\Delta^+_A\,M_A^\dagger)=0$, in which case
  we set $\xi_A:=\Ident_A/\abs{A}$.  In any case $\tr\xi_A=1$ and
  $\tr`(M_A\,\Delta^+_A\,M_A^\dagger)\,\xi_A = M_A\,\Delta^+_A\,M_A^\dagger$.
  The mapping $\mathcal{M}_{A\to A}$ is trace preserving:
  \begin{align}
    \mathcal{M}^\dagger`(\Ident_A)
    = M_A^\dagger\,M_A + `(\Ident - M_A^\dagger M_A)\,\tr\xi_A
    = \Ident_A\ .
  \end{align}
  We now show that $\mathcal{M}_{A\to A}`(\rho_A) = \tilde\rho_A$.  On one hand,
  using $A=\tilde\rho_A+\Delta^-_A=\rho_A+\Delta^+_A$, we have
  \begin{align}
    M_A\,\rho_A\,M_A^\dagger
    = M_A\,A\,M_A^\dagger - M_A\,\Delta^+_A\,M_A^\dagger
    = \tilde\rho_A - M_A\,\Delta^+_A\,M_A^\dagger\ .
    \label{eq:util-smoothing-part-of-bipartite-state-rest-exact-calc-1}
  \end{align}
  while noting that $\rho_A$ lies within the support of $A$ since
  $A=\rho_A+\Delta^+_A$.  We deduce that
  $\tr`(M_A\,\rho_A\,M_A^\dagger) = 1 - \tr`(M_A\,\Delta^+_A\,M_A^\dagger)$.  On
  the other hand,
  \begin{align}
    \tr`\big[`(\Ident- M_A^\dagger M_A)\,\rho_A]\,\xi_A
    &= `\big(1 - \tr`(M_A\,\rho_A\,M_A^\dagger))\,\xi_A
      \nonumber\\
    &= \tr`(M_A\,\Delta^+_A\,M_A^\dagger)\,\xi_A
      \nonumber\\
    &= M_A\,\Delta^+_A\,M_A^\dagger\ ,
    \label{eq:util-smoothing-part-of-bipartite-state-rest-exact-calc-2}
  \end{align}
  and hence,
  combining~\eqref{eq:util-smoothing-part-of-bipartite-state-rest-exact-calc-1}
  with~\eqref{eq:util-smoothing-part-of-bipartite-state-rest-exact-calc-2}
  \begin{align}
    \mathcal{M}_{A\to A}`(\rho_A)
    &= M_A\,\rho_A\,M_A^\dagger
      + \tr`\big[`(\Ident- M_A^\dagger M_A)\,\rho_A]\,\xi_A
      = \tilde\rho_A\ .
  \end{align}
  Define now the state $\hat\rho_{AB}$ as
  \begin{align}
    \hat\rho_{AB}
    = \mathcal{M}_{A\to A}\left[\rho_{AB}\right]
  \end{align}
  where the identity mapping is understood on system $B$.  By properties of
  quantum channels the state on $B$ is preserved, i.e.\@
  $\tr_A\hat\rho_{AB} = \rho_B$ (and in particular we have
  $\tr\hat\rho_{AB}=1$), and we showed above that
  $\tr_B\hat\rho_{AB}=\tilde\rho_A$.
  It remains to see that $\hat\rho_{AB}$ and $\rho_{AB}$ are close in purified
  distance. Let $\ket\rho_{ABC}$ be a purification of
  $\rho_{AB}$. Apply~\cite[\lemmaname~A.4]{Dupuis2013_DH}---itself a
  reformulation of \cite[\lemmaname~15]{Tomamichel2009IEEE_AEP}---with
  $\rho_\mathrm{Lem~A.4} = \rho_{A}$,
  $\sigma_\mathrm{Lem~A.4} = \tilde\rho_{A}$,
  $\Delta_\mathrm{Lem~A.4} = \Delta^-_A$, $G_\text{Lem~A.4}=M_A$ and
  $\ket{\psi_\mathrm{Lem~A.4}} = \ket\rho_{ABC}$ to obtain
  \begin{align}
    P`\big(M_A\rho_{ABC}M_A^\dagger, \rho_{ABC})
    \leqslant \sqrt{\left(2-\tr\Delta^-_A\right)\tr\Delta^-_A}
    \leqslant \sqrt{2\delta}\ .
  \end{align}
  This distance can only decrease if we trace out the system $C$, and thus
  $P`\big(M_A\rho_{AB}M_A^\dagger, \rho_{AB}) \leqslant \sqrt{2\delta}$.
  On the other hand, we have by definition
  \begin{align}
    \hat\rho_{AB} = M_A\,\rho_{AB}\,M_A^\dagger + \Delta'_{AB} \ ,
  \end{align}
  with
  $\Delta'_{AB} = \tr_A`[`(\Ident_A-M_A^\dagger M_A)\,\rho_{AB}]\otimes\xi_{A}
  \geqslant 0$.  Calculate
  $\tr\Delta'_{AB} = \tr`\big[`(\Ident_A-M_A^\dagger M_A)\,\rho_A] =
  \tr`(M_A\,\Delta^+_A\,M_A^\dagger) \leqslant \tr\Delta^+_A \leqslant \delta$,
  and hence $D`\big(M_A\,\rho_{AB}\,M_A^\dagger, \hat\rho_{AB})\leqslant\delta$.
  Finally, by triangle inequality and using
  $P`(\rho,\rho') \leqslant \sqrt{2 D`(\rho,\rho')}$,
  \begin{align*}
    P`(\hat\rho_{AB},\rho_{AB})
    \leqslant P`\big(\hat\rho_{AB}, M_A\,\rho_{AB}\,M_A^\dagger)
    + P`\big(M_A\,\rho_{AB}\,M_A^\dagger, \rho_{AB})
    \leqslant 2\sqrt{2\delta}\ .
    \tag*\qedhere
  \end{align*}
\end{proof}

\begin{lemma}
  \label{lemma:cohrelentr-2-std-purifs-close}
  Let $\sigma_X$, $\hat\sigma_X$ be two states.  Consider another system
  $R\simeq X$.  Then
  \begin{align}
    P(\sigma_X^{1/2}\,\Phi_{X:R}\,\sigma_X^{1/2},
    \hat\sigma_X^{1/2}\,\Phi_{X:R}\,\hat\sigma_X^{1/2})
    \leqslant 2\,\sqrt{D(\sigma_X,\hat\sigma_X)}\ .
  \end{align}
\end{lemma}
\begin{proof}[*lemma:cohrelentr-2-std-purifs-close]
  Let $\epsilon=D(\sigma_X,\hat\sigma_X)$.  Using the properties of the trace
  distance, let $\Delta_X^\pm\geqslant 0$ satisfy
  $\hat\sigma_X = \sigma_X + \Delta^+_X - \Delta^-_X$ with
  $\tr\Delta^+=\tr\Delta^-=\epsilon$.  Let
  $\ket\psi = (1 + \epsilon)^{-1/2}\,(\sigma_X + \Delta^+_X)^{1/2}\ket\Phi_{X:R}
  = (1 + \epsilon)^{-1/2}\,(\hat\sigma_X + \Delta^-_X)^{1/2}\ket\Phi_{X:R}$,
  noting that $\braket\psi\psi = \tr(\sigma_X+\Delta_X^+) / (1-\epsilon) = 1$.
  For any two pure states $\ket\phi, \ket\chi$ we know that
  $P(\proj\phi,\proj\chi) = `\big(1 - \abs{\braket\phi\chi}^2)^{1/2}$.  Our
  strategy for proving the claim is the following: We show that both
  \begin{subequations}
    \label{eq:lemma-cohrelentr-2-std-purifs-close-overlaps-with-psi}
    \begin{align}
      \abs{\matrixel{\psi_{XR}}{\sigma_X^{1/2}}{\Phi_{X:R}}}
      &\geqslant (1+\epsilon)^{-1/2}\ ; \\
      \abs{\matrixel{\psi_{XR}}{\hat\sigma_X^{1/2}}{\Phi_{X:R}}}
      &\geqslant (1+\epsilon)^{-1/2}\ ,
    \end{align}
  \end{subequations}
  and the claim will then follow by triangle inequality for the purified
  distance:
  $P(\sigma_X^{1/2}\,\Phi_{X:R}\,\sigma_X^{1/2},
  \hat\sigma_X^{1/2}\,\Phi_{X:R}\,\hat\sigma_X^{1/2}) \leqslant
  P(\sigma_X^{1/2}\,\Phi_{X:R}\,\sigma_X^{1/2}, \proj\psi) +
  P(\hat\sigma_X^{1/2}\,\Phi_{X:R}\,\hat\sigma_X^{1/2}, \proj\psi) \leqslant
  2\sqrt{1-1/`(1+\epsilon)} \leqslant 2\sqrt{\epsilon/`(1+\epsilon)} \leqslant
  2\sqrt\epsilon$.
  It remains to show the
  properties~\eqref{eq:lemma-cohrelentr-2-std-purifs-close-overlaps-with-psi}.
  We have
  $\matrixel{\psi_{XR}}{\sigma_X^{1/2}}{\Phi_{X:R}} =
  \matrixel{\Phi_{X:R}}{`(\sigma_X^{1/2}+\Delta_X^+)^{1/2}\sigma_X^{1/2}}{\Phi_{X:R}}
  / \sqrt{1+\epsilon} = \tr`\big[(\sigma_X+\Delta^+_X)^{1/2}
  \sigma_X^{1/2}]/\sqrt{1+\epsilon} \geqslant 1/\sqrt{1+\epsilon}$, noting that
  $`(\sigma_X+\Delta^+_X)^{1/2}\geqslant`(\sigma_X)^{1/2}$, and hence
  $\abs{\matrixel{\psi_{XR}}{\sigma_X^{1/2}}{\Phi_{X:R}}}\geqslant(1+\epsilon)^{-1/2}$.
  Similarly,
  $\matrixel{\psi_{XR}}{\hat\sigma_X^{1/2}}{\Phi_{X:R}} =
  \matrixel{\Phi_{X:R}}{`(\hat\sigma_X^{1/2}+\Delta_X^-)^{1/2}
    \hat\sigma_X^{1/2}}{\Phi_{X:R}}/\sqrt{1+\epsilon} \geqslant
  (1+\epsilon)^{-1/2}$.
\end{proof}

\begin{lemma}[Continuity of the relative entropy in its first argument]
  \label{util:continuity-of-relative-entropy-in-1st-arg}
  \index{relative entropy}
  Let $\Gamma\geqslant 0$.  Let $\rho,\sigma$ lie within the support of $\Gamma$.  Assume
  that $D`*(\rho,\sigma)\leqslant\epsilon$.  Then
  \begin{multline}
    \abs*{ \DD{\rho}{\Gamma} - \DD{\sigma}{\Gamma} }
    \\
    \leqslant
    \epsilon\log`*(\rank\Gamma - 1) + h`(\epsilon) + \epsilon\,\norm{\log\Gamma}_\infty\ ,
  \end{multline}
  where $h`(\epsilon)=-\epsilon\log\epsilon-(1-\epsilon)\log(1-\epsilon)$ is the binary
  entropy.
\end{lemma}
\begin{proof}[*util:continuity-of-relative-entropy-in-1st-arg]
  First, write
  \begin{align}
    \DD{\rho}{\Gamma} = \tr\,`*[ \rho\log\rho - \rho\log\Gamma ]
     = - \HH{\rho} - \tr\,`*[\rho\log\Gamma]\ ,
  \end{align}
  and so
  \begin{align}
    \abs{\DD\rho\Gamma-\DD\sigma\Gamma}
    &\leqslant \abs*{\HH\sigma-\HH\rho}
    + \abs*{\tr\,`*[\sigma\log\Gamma] - \tr\,`*[\rho\log\Gamma]}\ .
  \end{align}
  Using the continuity bound of Audenaert~\cite{Audenaert2007JPA_sharp}, we have
  \begin{align}
    \abs*{\HH{\rho} - \HH{\sigma}} \leqslant \epsilon\,\log`*(\rank\Gamma-1) + h(\epsilon)\ ,
  \end{align}
  where the states $\rho$ and $\sigma$ can be seen as living in a subspace of the full
  Hilbert space of dimension at most $\Gamma$ (because they must both lie within the
  support of $\Gamma$), and where
  $h(\epsilon)=-\epsilon\ln\epsilon-(1-\epsilon)\ln(1-\epsilon)$ is the binary entropy.
  On the other hand,
  \begin{align*}
    \tr\rho\log\Gamma - \tr\sigma\log\Gamma
    &= \norm{\log\Gamma}_\infty\,\tr\,`\bigg[(\rho-\sigma) \frac{\log\Gamma}{\norm{\log\Gamma}_\infty} ]
    \\
    &\leqslant \norm{\log\Gamma}_\infty\,D`*(\rho,\sigma)\ ,
  \end{align*}
  as $\log\Gamma/\norm{\log\Gamma}_\infty$ is a valid candidate for $Z$ in
  \autoref{util:trace-distance-DeltaPlusMinus}.  Inverting the roles of $\rho$ and
  $\sigma$ in the equation above we finally obtain:
  \begin{align}
    \abs*{\tr\rho\log\Gamma - \tr\sigma\log\Gamma}
    \leqslant \norm{\log\Gamma}_\infty\,D`*(\rho,\sigma)
    \leqslant\norm{\log\Gamma}_\infty\cdot\epsilon\ . \tag*\qedhere
  \end{align}
\end{proof}


%

\end{document}